\documentclass{article}
\usepackage[english]{babel}

\usepackage[a4paper,marginparwidth=1.75cm]{geometry}

\usepackage{mathtools}
\usepackage{amsmath,amsfonts}
\usepackage{mathabx}
\usepackage{amssymb,latexsym,amsthm,enumerate,fullpage}
\usepackage[utf8]{inputenc}
\usepackage{color}
\usepackage{graphicx}
\usepackage{tikz}
\usetikzlibrary{math,calc}
\usetikzlibrary{snakes}
\usepackage{graphicx}
\usepackage[colorlinks=true, allcolors=blue]{hyperref}
\usepackage{ytableau}
\usepackage{framed}
\usepackage{tocloft}
\usepackage{pdfpages}

\newtheorem{thm}{Theorem}[section]
\newtheorem{prop}[thm]{Proposition}
\newtheorem{lem}[thm]{Lemma}
\newtheorem{lemma}[thm]{Lemma}
\newtheorem{cor}[thm]{Corollary}

\newtheorem{conj}[thm]{Conjecture}
\theoremstyle{definition}
\newtheorem{definition}[thm]{Definition}
\newtheorem{algorithm}[thm]{Algorithm}
\newtheorem{nn}[thm]{Notation}

\newtheorem{rmk}[thm]{Remark}
\newtheorem{summary}[thm]{Summary}
\newtheorem{illustration}[thm]{Illustration}
\newtheorem{ex}[thm]{Example}
\newtheorem{obs}[thm]{Observation}

\newcommand{\R}{\mathbb{R}}
\newcommand{\RR}{\mathbb{R}}

\newcommand{\A}{\mathcal{A}}
\newcommand{\Ampl}{\mathcal{A}}
\newcommand{\D}{D}
\newcommand{\CD}{\mathcal{CD}}
\newcommand{\C}{\mathcal{C}}
\newcommand{\Mat}{\mathrm{Mat}}
\DeclareMathOperator{\Supp}{support}
\DeclareMathOperator{\SIGN}{{sign^{\forall}}}
\newcommand{\Z}{\widetilde{Z}}
\newcommand{\Gr}{\mathrm{Gr}}
\newcommand{\Grp}[2]{\mathrm{Gr}^{>}_{#1,#2}}
\newcommand{\Grnn}[2]{\mathrm{Gr}^{\ge}_{#1,#2}}

\newcommand{\Span}{\mathrm{span}}
\newcommand{\pre}{{\mathrm{pre}}}
\newcommand{\inc}{{\mathrm{inc}}}
\newcommand{\emb}{\overleftrightarrow{\iota}}
\newcommand{\lemb}{\overleftrightarrow{\;\textsc{low}\;}}
\newcommand{\uemb}{\overleftrightarrow{\;\textsc{upp}\;}}
\newcommand{\embilr}[1]{\emb_{#1}(\textbf{t},\textbf{s})}
\newcommand{\midemb}{{\Upsilon}}
\newcommand{\Id}{\mathrm{Id}}
\newcommand{\E}{\mathrm{E}}
\newcommand{\e}{\mathrm{e}}
\newcommand{\GL}{\mathrm{GL}}

\newcommand{\sgn}{\mathrm{sgn}}

\newcommand{\N}{N}
\newcommand{\K}{K}
\newcommand{\dd}{\mathrm{below}}
\newcommand{\ddp}{\mathrm{beyond}}
\newcommand{\behind}{\mathrm{behind}}
\newcommand{\mi}{{\,\;\mathclap{\circ}\mathclap{\relbar}\;\,}}
\newcommand{\pl}{{\pluscirc}}

\newcommand{\p}{\hat{\varepsilon}}
\newcommand{\Sl}{\alpha}
\newcommand{\Sr}{\beta}
\newcommand{\El}{\hat{\gamma}}
\newcommand{\Er}{\hat{\delta}}
\newcommand{\ee}{\eta}
\newcommand{\es}{\theta}

\newcommand{\rshs}[2]{\mathrm{rshift}^{\mathrm{start}}_{#1}(#2)}
\newcommand{\rshe}[2]{\mathrm{rshift}^{\mathrm{end}}_{#1}(#2)}
\newcommand{\lshs}[2]{\mathrm{lshift}^{\mathrm{start}}_{#1}(#2)}
\newcommand{\lshe}[2]{\mathrm{lshift}^{\mathrm{end}}_{#1}(#2)}
\newcommand{\Var}{Var}
\newcommand{\tVar}{\widetilde{Var}}
\newcommand{\Cont}{\mathrm{Contradiction}}
\newcommand{\SA}{S_{\partial\Ampl}}
\newcommand{\generate}{\textsc{construct-matrix}}
\newcommand{\Dom}{\mathcal{D}}
\newcommand{\favorite}[4]{{\left\langle\!\left\langle\, #1 \;{\big\vert}\; #2 \;{\big\vert}\; #3 \;{\big\vert}\; #4\,\right\rangle\!\right\rangle}}
\newcommand{\smin}[1]{\underline{#1}}
\newcommand{\smax}[1]{\overline{#1}}

\title{The Amplituhedron BCFW Triangulation}

\author{Chaim Even-Zohar\thanks{Faculty of Mathematics, Technion, Haifa, \texttt{chaime@technion.ac.il}} \and Tsviqa Lakrec\thanks{Institute of Mathematics, University of Z\"urich, \texttt{tsviqa@gmail.com}} \and Ran J.~Tessler\thanks{Department of Mathematics, Weizmann Institute of Science, \texttt{ran.tessler@weizmann.ac.il}}}
\date{}

\begin{document}

\maketitle

\begin{abstract}
The amplituhedron $\Ampl_{n,k,4}$ is a geometric object, introduced by Arkani-Hamed and Trnka (2013) in the study of scattering amplitudes in quantum field theories. They conjecture that $\Ampl_{n,k,4}$ admits a decomposition into images of BCFW positroid cells, arising from the Britto--Cachazo--Feng--Witten recurrence (2005). We prove that this conjecture is true.
\end{abstract}

\section{Introduction}

The amplituhedron $\Ampl_{n,k,m}$ is the image of a positive map $\Grnn{k}{n}\xrightarrow{\;\Z\;}\Gr_{k,(k+m)}$ which ``projects'' the nonnegative real Grassmannian to a smaller Grassmannian. This paper targets the fundamental problem of triangulating the amplituhedron into images of positroid cells of~$\Grnn{k}{n}$. We first briefly review the definitions of all these objects.

\subsubsection*{Grassmannians, Positroid Cells, and Amplituhedra}

The real \emph{Grassmannian} $\Gr_{k,n}$ is the variety of $k$-dimensional linear subspaces in $\RR^n$. It is concretely represented as the quotient space $\Gr_{k,n} = \GL_k(\RR)\,\backslash\,\Mat^{\ast}_{k\times n}(\RR)$. This means full-rank $k \times n$ real matrices modulo row operations performed by invertible $k \times k$ matrices. For a representative $k \times n$ matrix~$C$, every set $I \subseteq \{1,\dots,n\}$ of $k$ columns defines a \emph{Pl\"ucker coordinate} $P_I(C)$ as the determinant of the $k \times k$ minor corresponding to~$I$. Using all Pl\"ucker coordinates, the Grassmannian embeds in the $\left(\tbinom{n}{k}-1\right)$-dimensional real projective space. The \emph{positive Grassmannian} $\Grp{k}{n}$ is the subset of points in the Grassmannian with all $P_I >0$. Its closure is the \emph{nonnegative Grassmannian} $\Grnn{k}{n}$, with all $P_I \geq 0$.

The theory of positivity for algebraic groups and partial flag varieties, and in particular the positive Grassmannian, was developed by Lusztig~\cite{lusztig1994total}. It was further studied by Rietsch~\cite{rietsch1998total,rietsch2006closure}, Marsh and Rietsch \cite{marsh2004parametrizations}, Fomin and Zelevinsky~\cite{fomin1999double, fomin2002cluster, fomin2003cluster}, and Postnikov~\cite{postnikov2006total, postnikov2018positive}. Postnikov developed a rich combinatorial picture of the nonnegative Grassmannian and its cell decomposition into positroid cells. An open \emph{positroid cell} $S \subseteq \Grnn{k}{n}$
comprises the points where a certain subset of the Pl\"ucker coordinates are positive, i.e., all $C$ such that $P_I(C)>0$ for $I \in M$ and $P_I(C)=0$ for $I \not\in M$, for some set~$M$. These positroid cells admit explicit parametrizations by $(0,\infty)^d$  for some $d={\dim S}$. They correspond to various combinatorial objects such as graphs, tableaux, and permutations, which will play a role below. The positive Grassmannian bears relations and applications to diverse areas, including cluster algebras, tropical geometry, and integrable systems \cite{speyer2005tropical, kodama2013combinatorics, kodama2014kp, lukowski2023positive, speyer2021positive}. A~recent application to scattering amplitudes in theoretical physics by Arkani-Hamed et~al.~\cite{arkani2016Grassmannian} motivates the definition of the amplituhedron by Arkani-Hamed and Trnka~\cite{arkani2014amplituhedron}, which is the focus of this work.

The amplituhedron depends on an additional parameter $m \leq n-k$ and  a real matrix $Z \in \Mat^{>}_{n\times(k+m)}$. This notation means that $Z$ is assumed to be \emph{positive}, in the sense that each one of its $\tbinom{n}{k+m}$ maximal minors has a positive determinant. The right multiplication by $Z$ induces a well-defined map $\Z:\Grnn{k}{n} \to \Gr_{k,(k+m)}$. In terms of representative matrices, the image of a point $C \in \Grnn{k}{n}$ is $\Z(C) = CZ \in \Gr_{k,(k+m)}$, and the image of a cell $S \subseteq \Grnn{k}{n}$ is $\Z(S) = \{CZ : C \in S\}$. The tree \emph{amplituhedron} is defined as the image of the entire nonnegative Grassmannian: 
$$ \Ampl_{n,k,m}(Z) \;=\; \left\{ \, CZ \,:\, C \in \Grnn{k}{n} \, \right\} \;\subset\; \Gr_{k,(k+m)} $$ 
The amplituhedron is known to be a proper subspace of full dimension~$km$. Since many of its structural properties do not seem to depend on the choice of~$Z$, we often denote the amplituhedron by~$\A_{n,k,m}$ and argue for any fixed~$Z$. See~\cite{arkani2014amplituhedron, arkani2016Grassmannian, hodges2013eliminating, bourjaily2018amplituhedron} for an exposition of the topic and physical background. The amplituhedron has been much studied in recent years \cite{bao2019m, karp2019amplituhedron,  lukowski2019cluster, galashin2020parity, karp2020decompositions, lukowski2021boundaries, mohammadi2021triangulations, parisi2021m}. Generalizations to non-tree level or other physical theories have been considered \cite{arkani2015shell, arkani2018scattering, arkani2019positive, arkani2021eft, trnka2021towards,  arkani2021stringy}. Its mathematical foundations are being studied under the name \emph{positive geometries}~\cite{arkani2017positive}.

\subsubsection*{Triangulating the Amplituhedron}

The structure of the amplituhedron is less understood than that of the nonnegative Grassmannian. It is desirable to establish ways to subdivide $\mathcal{A}_{n,k,m}$ into simple pieces, in analogy to triangulations of a polytope for example. The following definition suggests that this can be achieved using homeomorphic images under~$Z$ of positroid cells of the appropriate dimension.

\begin{definition}[Bao and He \cite{bao2019m}]
\label{triangulation}
A~\emph{triangulation} of the amplituhedron $\Ampl_{n,k,m}$ is a collection~$\mathcal{T}$, of $km$-dimensional open positroid cells of the nonnegative Grassmannian~$\Grnn{k}{n}$, that satisfies the following properties for every~$Z \in \Mat^{>}_{n\times(k+m)}$:
\begin{samepage}
\begin{itemize}
\itemsep0.12em
\item \emph{Injectivity}: $S \to \Z(S)$ is an injective map for every cell $S \in \mathcal{T}$.
\item \emph{Separation}: $\Z(S)$ and $\Z(S')$ are disjoint for every two cells $S \neq S'$ in $\mathcal{T}$.
\item \emph{Surjectivity}: $\bigcup_{S \in \mathcal{T}} \Z(S)$ is an open dense subset of $\Ampl_{n,k,m}(Z)$.
\end{itemize}
\end{samepage}
\end{definition}

The case $m=4$ is the most relevant to physics, being applicable to scattering amplitudes in planar $\mathcal{N}{=}4$ supersymmetric Yang--Mills theory (SYM). Recurrence relations for computing scattering amplitudes arise from the work of Britto, Cachazo, Feng, and Witten~\cite{britto2005new, britto2005direct}. The BCFW recurrence translates into a recursive definition of a collection of $4k$-dimensional positroid cells in the nonnegative Grassmannian~$\Grnn{k}{n}$. These cells correspond to different terms that add up to the scattering amplitude, see~\cite{arkani2014amplituhedron, arkani2016Grassmannian}. A~direct definition by Karp, Williams, and Zhang~\cite{karp2020decompositions} will be elaborated in Section~\ref{bcfw} below. It gives rise to a collection of $\tfrac{1}{k+1}\tbinom{n-3}{k}\tbinom{n-4}{k}$ cells, named the \emph{BCFW cells} and denoted here by~$\mathcal{BCFW}_{n,k}$, see Definition~\ref{maindefbcfw} below. This definition corresponds to a standard way of applying the BCFW recurrence in planar $\mathcal{N}=4$ SYM theory. Other ways are treated in a subsequent paper with Parisi, Sherman-Bennett and Williams. The relation between scattering amplitudes and the amplituhedron is based on the following conjecture where it suffices to fix any BCFW triangulation.

\begin{conj}[Arkani-Hamed and Trnka \cite{arkani2014amplituhedron}]
For every $k \geq 1$ and $n \geq k+4$, the cells $\mathcal{BCFW}_{n,k}$ form a triangulation of the amplituhedron $\Ampl_{n,k,4}$.
\label{conj:main}
\end{conj}

In this paper we prove this conjecture. We develop new machinery for analyzing how positroid cells map into the amplituhedron. We devise a scheme of procedures for recursively constructing subsets of the positive Grassmannian, including all BCFW cells. We derive distinctive characteristics of their images, such as functions having constant sign on a given cell, and keep track of their evolution as the construction proceeds. These techniques let us locate preimages in order to show injectivity,  compare two cells to tell their images apart, and analyze boundaries between cells. We demonstrate our approach in the case of Conjecture~\ref{conj:main}, showing that the BCFW cells triangulate the $m=4$ amplituhedron in the sense of Definition~\ref{triangulation}. Namely, we prove the following three properties for every $k \geq 1$, $n \geq k+4$ and positive $Z \in \Mat^{>}_{n\times(k+4)}$.

\begin{thm}\label{thm:injectiveness}
The map $S \to \Z(S)$ is injective for every cell $S \in \mathcal{BCFW}_{n,k}$.
\end{thm}

\begin{thm}\label{thm:separation}
The images $\Z(S)$ and $\Z(S)$ are disjoint for every two different cells $S,S' \in \mathcal{BCFW}_{n,k}$.
\end{thm}
\begin{thm}\label{thm:surj}
The union of $\Z(S)$ over $S \in \mathcal{BCFW}_{n,k}$ is an open dense subset of~$\Ampl_{n,k,4}(Z)$.
\end{thm}

Another outcome of our analysis is a characterization of the boundary of the amplituhedron. We also show that this is a \emph{good} triangulation, in the sense that every internal wall is an image of a boundary stratum between two BCFW cells. See Corollaries~\ref{cor:bdries} and~\ref{good-triangulation} for the precise formulation of these results. An additional property that we deduce is that the interior of the amplituhedron is homeomorphic to an open ball, see Theorem~\ref{thm:int_is_ball}.

Many of the techniques we develop generalize to other values of~$m$ and other triangulations and families of positroid cells. These are useful for manipulating functions of the amplituhedron's coordinates, showing injectivity of the amplituhedron map, separating between cells, and boundary cancellations. 

\subsubsection*{Related Work}

Since Conjecture~\ref{conj:main} was posed, triangulations of $\Ampl_{n,k,m}$ have been studied for various values of $m$ and~$k$. Karp and Williams give a triangulation of $\Ampl_{n,k,1}$~\cite{karp2019amplituhedron}. The case of $\Ampl_{n,1,m}$ is a cyclic polytope in projective space~\cite{sturmfels1988totally}. Galashin and Lam introduce the parity duality, which relates triangulations of $\Ampl_{n,k,m}$ with triangulations of $\Ampl_{n,n-m-k,m}$~\cite{galashin2020parity}. Karp, Williams and Zhang prove injectivity and separation in the special case $\Ampl_{n,2,4}$ using domino forms and an exhaustive case analysis~\cite{karp2020decompositions}. Bao and He prove a triangulation of $\Ampl_{n,k,2}$ based on BCFW-like cells~\cite{bao2019m}. Parisi, Sherman-Bennett and Williams use twistor coordinates to establish many triangulations of $\Ampl_{n,k,2}$, and relate them to triangulations of the hypersimplex~\cite{parisi2021m}. For other related results, see a very recent survey by Williams~\cite{williams2022positive}. 

Our methods are inspired by several ideas from previous works on Conjecture~\ref{conj:main}. We use matrix operations similar to Lam \cite{lam2014totally} and Bao and He~\cite{bao2019m} in order to manipulate cells. Our basic ingredient in separating cells is the twistor coordinates, similarly to Parisi, Sherman-Bennett and Williams~\cite{parisi2021m}. 
Our analysis of boundaries is reminiscent with that of Agarwala and Marcott~\cite{agarwala2023cancellation}.
We use the domino form by Karp, Williams, Zhang, and Thomas~\cite{karp2020decompositions} to represent BCFW cells. In particular, we settle their Conjecture~A.7 and prove the following. 
\begin{thm}
\label{thm:domino_non_formal}
Every point in a BCFW cell has a representative matrix in the domino form of that cell. Conversely, every domino matrix represents a point in the corresponding BCFW cell.
\end{thm}

See Section~\ref{subsec-domino} for the definition of domino matrices and their sign rules, and see Theorem~\ref{thm:domino} 
for a more detailed formulation of this result.

\subsubsection*{Proof Overview}

In Section~\ref{bcfw} we define the main three structures that we use to represent BCFW positroid cells: chord diagrams, domino matrices, and decorated permutations. We also show that these correspond to the BCFW cells from previous works.

In Section~\ref{sec:domino} we describe explicit algorithms for constructing domino representatives for positroid cells. We analyze this process by algebraic and combinatorial methods, and deduce Theorem~\ref{thm:domino_non_formal} on the matrix form of BCFW cells. Our proofs in later sections rely on the recursive nature of this construction scheme, which follows the structure of the chord diagram.

In Section~\ref{sec:promotion} we move to discuss the amplituhedron map. We review the twistor coordinates of the amplituhedron, and present key ideas and tools for the next sections. We introduce the study of functionaries, which are polynomials in the twistors, and how they evolve under natural algebraic operations that preserve nonnegativity.

In Section~\ref{sec:inj} we show that the BCFW cells map injectively into the amplituhedron. We use the hierarchical structure of chord diagrams to find a preimage. This requires showing that certain twistors are nowhere vanishing on a cell's image. Employing the tools developed in Sections~\ref{sec:domino}-\ref{sec:promotion}, we prove that they do not vanish by induction on chord diagrams.

In Section~\ref{sec:separation} we separate the images of the BCFW positroid cells using functionaries, which are combinatorially determined by their chord diagrams. Given two BCFW cells, we use the tools of Sections~\ref{sec:domino}-\ref{sec:promotion} to provide a functionary that has a fixed opposite sign on each image, independently of the positive matrix~$Z$.

Section~\ref{sec:precise_ineqs} revisits the BCFW cells, and characterizes their boundaries via inequalities in the entries and $2 \times 2$ minors of their domino matrices. We combinatorially identify all pairs of chord diagrams whose positroid cells share a codimension one boundary stratum.

In Section~\ref{sec:surj} we prove surjectivity. It follows by a topological argument from the above characterization of internal and external boundaries of the triangulation.

In Section~\ref{consequences} we prove an additional, topological result that the interior of the amplituhedron is homeomorphic to an open ball, for every positive $Z$. Then, we further investigate the structure of the amplituhedron, and describe its decomposition into images of some type of products of two positive Grassmannians.

\subsubsection*{Acknowledgments}

C.E.~thanks the Lloyds Register Foundation's Alan Turing Institute programme on Data-Centric Engineering for their support during part of the research. 
T.L.~was supported by ERC 2020 grant HomDyn grant no.~833423 and SNSF grant Dynamical Systems grant no.~188535.
R.T., incumbent of the Lillian and George Lyttle Career Development Chair, was supported by the ISF grant no.~335/19 and by a research grant from the Center for New Scientists of Weizmann Institute. 

The authors thank Nima Arkani-Hamed, Matteo Parisi, Melissa Sherman-Bennett, and Lauren Williams for interesting discussions in relation to this work.
The authors thank the IT team in the Computer Science and Mathematics Department of the Weizmann Institute, and especially Unix engineer Amir Gonen, who helped us to run our computations. R.T.~thanks Jian-Rong Li, Xavier Blot and Yoel Groman for helpful discussions.

The authors sincerely thank the anonymous referee for their thorough and helpful review of the submitted manuscript and for providing many thoughtful suggestions that improved the clarity and quality of the article.

\setlength{\cftaftertoctitleskip}{10pt}
\setlength{\cftbeforesecskip}{2pt}
\renewcommand{\cfttoctitlefont}{\normalsize\bfseries}
\setcounter{tocdepth}{1}
\tableofcontents

\section{Definitions for BCFW Cells}
\label{bcfw}
\label{sec::bcfw}

This section is devoted to the combinatorial structures that represent BCFW cells. First, we introduce the two central definitions of \emph{chord diagrams} and \emph{domino matrices}, in the particular forms devised and used in this work. Then, we give two equivalent definitions for the collection of BCFW cells, one based on lattice walks and combinatorial representations of positroids from the literature, and a new one directly from chord diagrams.

In the seminal work of Postnikov~\cite{postnikov2006total}, several families of combinatorial structures are employed to define and study the cells of the nonnegative Grassmannian. A few objects that are identified with positroid cells are: certain bicolored planar networks, known as \emph{plabic graphs}; certain $0/1$-filled Young tableaux, known as \emph{$\oplus$-diagrams}~\cite{lam2008total}; and permutations with two kinds of fixed points, known as \emph{decorated permutations}. 

More combinatorial representations are available for the special class of BCFW cells. These are given by a particular family of plabic graphs in~\cite[with a certain rotation by two]{arkani2016Grassmannian}, and by pairs of \emph{noncrossing lattice walks} in~\cite{karp2020decompositions}, where additional equivalent viewpoints are described using \emph{binary trees}, \emph{Dyck paths} and \emph{domino matrices}. For our purposes, it is convenient to start from yet another representation, a \emph{chord diagram}, which is a variant of the \emph{Wilson loop diagrams} used by Agarwala, Marin-Amat and Marcott~\cite{agarwala2017wilson,agarwala2023cancellation}. We start with a definition of these combinatorial structures.

\subsection{Chord Diagrams}
\label{subsec-chords}

\begin{samepage}
\begin{definition}\label{cds}
A \emph{chord diagram} is a circle, containing $n$ \emph{markers} labeled $\{1,2,\dots,n\}$ counterclockwise, and $k$ \emph{chords}, always oriented increasingly with respect to the markers, such that:
\begin{enumerate}
\itemsep0em
\item[(a)]
No chord starts or ends on a marker.
\item[(b)]
No two chords intersect.
\item[(c)]
No chord starts before the marker $1$ or ends after $n-1$. 
\item[(d)] 
No chord starts and ends on the same segment between markers, nor on adjacent segments.
\item[(e)]
No two chords \emph{start} on the same segment between two markers.
\end{enumerate}
The set of all chord diagrams with $n$ markers and $k$ chords is denoted $\mathcal{CD}_{n,k}$. 
\end{definition}
\end{samepage}

\begin{rmk}
Agarwala et~al.~\cite{agarwala2017wilson,agarwala2023cancellation} use the term \emph{propagators} for chords, and label the markers by $Z$'s rows: $Z_1,\dots,Z_n$. The constraints they put on the chords, and the relation to positroid cells, are different from ours.
\end{rmk}

\begin{samepage}
\begin{ex}
\label{example3chords}
Below is a chord diagram in $\mathcal{CD}_{14,3}$ and an equivalent linear drawing of it. In this paper we always draw chord diagrams in the linear form, and the circular drawing is only included for comparison with previous works, such as~\cite{agarwala2023cancellation}.
\end{ex}

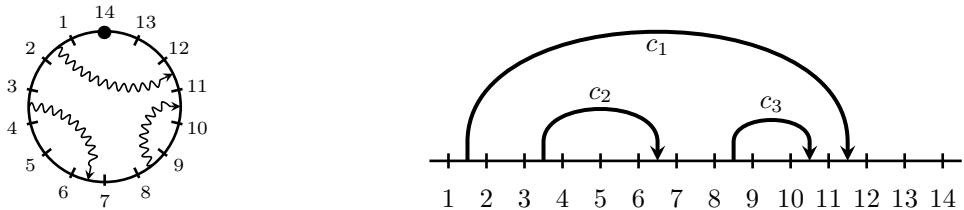
\begin{center}
\tikz[line width=1]{
\draw (0,0) circle (1);
\foreach \i in {1,2,...,14}{
\def\angle{90+\i*360/14}
\draw (\angle:0.9)--(\angle:1.05);
\node at (\angle:1.25) {\scriptsize\i};}
\node at (0,0.975) {\Large$\bullet$};
\foreach \i/\j in {1/11,3/6,8/10}{
\def\angla{103+\i*360/14}
\def\anglb{103+\j*360/14}
\def\anglc{283+\i*360/14}
\def\angld{283+\j*360/14}
\draw[decorate,decoration={snake,amplitude=1.5},segment length=3.6,line width=0.6,segment aspect=0] (\angla:1) to[out=\anglc,in=\angld] (\anglb:0.875);
\draw[-stealth,line width=0.6] (\anglb:0.875) -- (\anglb:1);}
}
\hspace{1in}
\tikz[line width=1]{
\draw (0.25,0) -- (7.25,0);
\foreach \i in {1,2,...,14}{
\def\x{\i/2}
\draw (\x,-0.1)--(\x,+0.1);
\node at (\x,-0.5) {\i};}
\foreach \i/\j in {1/11,3/6,8/10}{
\def\x{\i/2+0.25}
\def\y{\j/2+0.25}
\draw[line width=1.5,-stealth] (\x,0) -- (\x,0.25) to[in=90,out=90] (\y,0.25) -- (\y,0);
}
\node at(3.25,1.5) {$c_1$};
\node at(2.5,0.875) {$c_2$};
\node at(4.75,0.75) {$c_3$};
}
\end{center}
\end{samepage}

Combinatorially, a chord diagram is represented by $D=([n],(c_1,\dots,c_k)) $. First comes the set of all markers $[n] = \{1,\dots,n\}$. Second, we have $k$ quadruples of markers of the form $c = (i,i+1,j,j+1)$, corresponding to the $k$ chords, where the pairs $(i,i+1)$ and $(j,j+1)$ are the two segments incident to the chord~$c$. By convention, each $c_l$ is increasing, and $c_1,\dots,c_k$ are listed in lexicographic order, i.e., increasingly according to their starts. For example, the diagram drawn above is represented by:
$$\left([14],\,\left((1,2,11,12),(3,4,6,7),(8,9,10,11)\right)\right)$$

\begin{definition}
\label{terminology}
The following terminology for chord diagrams is self-explanatory.
\begin{itemize}
\item 
A chord $c = (i,i+1,j,j+1)$ is said to \emph{start} at the pair $(i,i+1)$ and to \emph{end} at the pair $(j,j+1)$. These two segments are respectively the \emph{start} and the \emph{end} of $c$. 
\item
As we order the chords lexicographically, if two chords $c = (i,i+1,j,j+1)$ and $c' = (i',i'+1,j',j'+1)$ satisfy $i < j \leq i' < j'$, then $c'$ comes \emph{after} $c$, and $c$ comes \emph{before}~$c'$. 
\item
If two chords $c = (i,i+1,j,j+1)$ and $c' = (i',i'+1,j',j'+1)$ satisfy $i < i' < j' \leq j$, then $c'$ is a \emph{descendant} of $c$, and $c$ is an \emph{ancestor} of~$c'$. 
\item
A~chord with no ancestor is a \emph{top} chord. A chord with no descendant is a \emph{lowest} chord. 
A~top chord that starts first and ends last is \emph{long}. 
A~lowest chord $(i,i+1,i+2,i+3)$ is \emph{short}.
\item
The \emph{chain} from $c$ to $c'$ is the 
longest descendent sequence $(c,\dots,c')$, in which every chord is a descendant of the previous one. If $(c,c')$ is a chain of two chords, then $c$ is the \emph{parent} of $c'$, and $c'$ is a \emph{child} of~$c$.
\item
Two chords are \emph{siblings} if they are either top chords or children of a common parent. The terms \emph{next/previous/first/last sibling} refer to their position in a maximal sequence $(c,\dots,c')$ of siblings ordered by the before/after relation. 
\item
If $c$ starts at $(i,i+1)$ and $c'$ starts at $(i+1,i+2)$ then $c$ is a \emph{sticky} parent and $c'$ is a \emph{sticky} child. A~chain of chords $(c,\dots,c')$ is \emph{sticky} if all consecutive pairs are sticky. For every chord~$c_l$, we denote by $c_{\ast l}$ and $c_{l\ast}$ the first and last chords in the maximal sticky chain $(c_{\ast l},\dots,c_{l-1},c_l,c_{l+1},\dots,c_{l\ast})$.
\item
If two chords or more end at the same pair $(j,j+1)$ then this pair is a \emph{same-end}. Such chords relate to each other as \emph{same-end} child/descendant/parent/ancestor/chain.
\item
If $c'$ starts at the same pair $(j,j+1)$ where $c$ ends, then $(c,c')$ are \emph{head-to-tail} chords. 
\end{itemize}
\end{definition}

See Example~\ref{example9} below for a more complicated chord diagram than the above one, which illustrates some cases of sticky, same-end, and head-to-tail chords. We note that a previous version of this paper used the terms ``tail'' and ``head'' instead of ``start'' and ``end''.

\subsection{Domino Matrices}
\label{subsec-domino}

Karp, Williams, Zhang and Thomas~\cite[Appendix~A]{karp2020decompositions} suggest representing the points in the BCFW cells by special matrices, whose rows are called \emph{domino bases}. This form is especially useful for analyzing the amplituhedron map. Here we redefine these matrices via their one-to-one correspondence to chord diagrams.

\begin{samepage}
\begin{definition}
\label{def:domino_entries}
The \emph{domino matrix} $C$ of a chord diagram $D \in \mathcal{CD}_{n,k}$ is a $k \times n$ real matrix depending on $5k$ real variables. For each $l \in \{1,\dots,k\}$, construct the row $C_l$ of $C$ corresponding to the chord $c_l = (i_l,i_l+1,j_l,j_l+1)$ as follows.
\begin{itemize}
\itemsep0.125em
\item[(a)]
Write the four variables $(\alpha_l, \beta_l, \gamma_l, \delta_l)$ at the respective positions $(i_l, i_l+1, j_l, j_l+1)$. 
\item[(b)]
If $c_l$ is a top chord, write another variable $\varepsilon_l$ at the last position $n$.
\item[(c)]
If $c_l$ is a child of~$c_m$, add $(\varepsilon_l\alpha_m,\varepsilon_l\beta_m)$ at the start positions $(i_m,i_m+1)$ of the parent.
\item[(d)]
Elsewhere, write zeros.
\end{itemize}
We denote by $\mathcal{DM}_{n,k}$ the set of all $k \times n$ domino matrices so obtained from $\mathcal{CD}_{n,k}$.
\end{definition}
\end{samepage}

\begin{ex}
\label{example3domino}
Here is the domino matrix of the chord diagram from Example~\ref{example3chords}.
$$ \left(\begin{array}{*{14}{c}}
\alpha_1 & \beta_1 & \;0\; & \;0\; & \;0\; & \;0\; & \;0\; & \;0\; & \;0\; & \;0\; & \gamma_1 & \delta_1 & 0 & \varepsilon_1 \\
\varepsilon_2\alpha_1 & \varepsilon_2\beta_1 & \alpha_2 & \beta_2 & 0 & \gamma_2 & \delta_2 & 0 & 0 & 0 & \;0\; & \;0\; & \;0\; & \;0\; \\
\varepsilon_3\alpha_1 & \varepsilon_3\beta_1 & 0 & 0 & 0 & 0 & 0 & \alpha_3 & \beta_3 & \gamma_3 & \delta_3 & 0 & 0 & 0
\end{array}\right) $$
\end{ex}

\begin{rmk}
It follows from the definition that if the chord $c_{l}$ is a sticky child then at position $i_l$ of~$C_l$ appears the sum $\alpha_l + \varepsilon_l\beta_{l-1}$. See Example~\ref{example9} below for such cases.
\end{rmk}

\begin{definition}
\label{rmk:support}
We refer to the pairs $(\alpha_l,\beta_l)$ and $(\gamma_l,\delta_l)$ as the \emph{start domino} and \emph{end domino} of the row~$C_l$. The pair $(\varepsilon_l\alpha_m,\varepsilon_l\beta_m)$ is the domino that $C_l$ \emph{inherits} from~$C_m$. The support of the row~$C_l$, denoted $\mathrm{support}(C_l)$, is the set of five or six positions in $[n]$ where $C_l$ is nonzero. {We refer to the elements $\alpha_l,\ldots,\varepsilon_l$ as the \emph{domino entries} or \emph{domino variables} of the matrix $C.$}
\end{definition} 

\begin{rmk}
Since the rows of~$C$ correspond to chords in~$D$, and dominoes correspond to starts and ends of chords, we borrow terms from chord diagrams to domino matrices and vice versa. For example, we refer to the \emph{dominoes of a chord~$c_l$}, and write $\Supp(c_l)$. Conversely, we use \emph{children}, \emph{parents} and the other terms from Definition~\ref{terminology} to matrix rows. The four \emph{markers of $C_l$} refer to those of~$c_l$.
\end{rmk}

Without any restrictions on the $5k$ variables in a domino matrix, it might represent points outside the nonnegative Grassmannian. The following conditions will be shown to guarantee nonnegativity. Moreover, we later show that under these rules, a domino matrix provides a parametrization of an appropriate BCFW positroid cell.

\begin{definition}\label{def:domino_signs}
The \emph{sign rules} of the domino matrix $C$ that corresponds to a chord diagram $D$ is the following set of conditions on its variables.
\begin{samepage}
\begin{enumerate}
\itemsep0.125em
\item 
For every chord $c_l$: \; $\alpha_l > 0$, \; $\beta_l > 0$
\item 
For every chord $c_l$: \; $(-1)^{\mathrm{below}(c_l)}\gamma_l > 0$, \; $(-1)^{\mathrm{below}(c_l)}\delta_l > 0$ 
\item 
If $c_l$ is a top chord: \; $(-1)^{\mathrm{behind}(c_l)}\varepsilon_l > 0$
\item
If $c_l$ is not a top chord: \; $(-1)^{\mathrm{beyond}(c_l)} \varepsilon_l > 0$
\item
If $c_l$ is a same-end child of $c_m$: \; $\delta_l/\gamma_l < \delta_m/\gamma_m$ 
\item
If $c_l$ is head-to-tail after $c_m$: \; $\beta_l/\alpha_l > \delta_m/\gamma_m $
\end{enumerate}
\end{samepage}
where
\begin{samepage}
\begin{itemize}
\itemsep0.125em
\item[--] 
$\mathrm{below}(c_l)$ is the number of descendants of $c_l$.
\item[--] 
$\mathrm{behind}(c_l) = k-l$ is the number of chords that start after the start of~$c_l$.
\item[--] 
$\mathrm{beyond}(c_l) = l-m-1$ is the number of descendants of $c_l$'s parent $c_m$, that come before~$c_l$.
\end{itemize}
\end{samepage}
\end{definition}

Note that according to the sign rules, the two entries of a domino always have the same sign, and the ratios in rules 5 and 6 are always positive. These two rules are not explicitly specified in~\cite[Appendix~A]{karp2020decompositions}. See Example~\ref{example9} for a case where they are applicable.

\subsection{BCFW Cells in the Combinatorial Literature}
\label{subsec-equivalence}

In principle, the BCFW recurrence can be performed in several ways, and our convention follows a canonical choice that was described using plabic graphs by Karp, Williams, and Zhang~\cite[Section~5]{karp2020decompositions} based on Arkani-Hamed et~al.~\cite[Section~16]{arkani2016Grassmannian}. In Sections 6-7 of~\cite{karp2020decompositions}, they introduce several other equivalent combinatorial structures that index the set of BCFW cells in a nonrecursive form. One of these equivalent objects is a noncrossing pair of lattice walks, that gives rise to a an $\oplus$-diagrams or a decorated permutation, which can be used to define the corresponding BCFW positroid cell. 

In this section we describe this definition of BCFW cells following~\cite{karp2020decompositions}, and explain how they are related to chord diagrams. In the next section we give another equivalent combinatorial definition of the BCFW cells from chord diagrams, which we actually use in the rest of the paper. The definition by lattice paths and $\oplus$-diagrams given in this section is not be needed in other sections of the paper.

\begin{definition}
\label{deflattice}
A \emph{lattice walk} in a $k \times l$ rectangle is a path $W$ from its upper right corner to its lower left corner that takes $k$ vertical unit steps and $l$ horizontal unit steps. A pair of lattice walks $(W_A,W_B)$ is \emph{noncrossing} if $W_A$ stays weakly above $W_B$. The set of all noncrossing pairs of lattice walks in a $k \times (n-k-4)$ rectangle is denoted~$\mathcal{LW}_{n,k}$.
\end{definition}

We clarify that ``weakly above'' in this definition means that the horizontal edges of $W_A$ can not be strictly below those of $W_B$. The two walks may occupy the same grid edges.

\begin{samepage}
\begin{ex}
\label{example3lattice}
Here is a noncrossing pair of lattice walks $(W_A,W_B) \in \mathcal{LW}_{14,3}$. The grid's dimensions are $3 \times (14-3-4) = 3 \times 7$. The vertical steps of $W_A$ are at positions $J = \{2,8,10\} \subseteq [10]$, and those~of $W_B$ are at $I = \{1,3,8\}$. Equivalently, the walk $W_A$ makes $(a_0,\dots,a_3) = (1,5,1,0)$ horizontal steps at the different rows of the lattice, whereas $W_B$ makes $(b_0,\dots,b_3) = (0,1,4,2)$ horizontal steps.

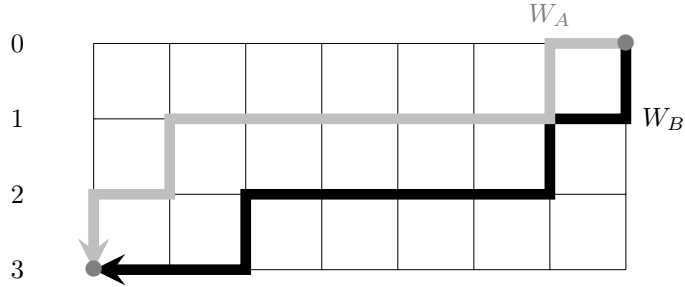
\begin{center}
\tikz[]{
\draw[black,line width=0.25] (0,0) grid (7,3);
\draw[black,line width=4,-stealth] (7,3) -- (7,2) -- (6,2) -- (6,1) -- (2,1) -- (2,0) -- (0,0);
\draw[lightgray,line width=4,-stealth] (7,3) -- (6,3) -- (6,2) -- (1,2) -- (1,1) -- (0,1) -- (0,0);
\node[gray] at(0,0) {\LARGE$\bullet$};
\node[gray] at(7,3) {\LARGE$\bullet$};
\node[gray] at(6,3.375) {$W_A$};
\node[black] at(7.5,2) {$W_B$};
\node at(-1,0) {3};
\node at(-1,1) {2};
\node at(-1,2) {1};
\node at(-1,3) {0};
}
\end{center}
\end{ex}
\end{samepage}

\begin{definition}
\label{lp2cd}
Let the map $\Phi : \mathcal{LW}_{n,k} \to \mathcal{CD}_{n,k}$ define \emph{the chord diagram corresponding to a noncrossing pair of lattice walks} by the following procedure. Consider $(W_A,W_B) \in \mathcal{LW}_{n,k}$. The walk $W_B$ determines the chords' starts as follows. Let $I = \{i_1,\dots,i_k\} \subseteq [n-4]$ be the set of $k$ vertical steps in the walk~$W_B$. For every step $i_l \in I$ we put a chord start on the segment $(i_l,i_l+1)$. Then, a reverse walk along $W_A$ determines the ends of $c_k,\dots,c_1$. For $l \in \{k,\dots,1\}$, let $J_l\subset\{i_l+2,\dots,n-2\}$ be the set of $j$ such that $(j,j+1)$ is a possible end for $c_l$, i.e.~excluding intervals that lie under another chord $c_i$ for $i>l$. The end of $c_l$ is the $(a_l+1)$-th smallest option in $J_l$, where $a_l$ is the number of $W_A$'s horizontal steps in row~$l$. After an example, we show that this map is well defined.
\end{definition}

\begin{ex}
\label{example3walkwalk}
We apply $\Phi$ on the lattice walks $(W_A,W_B) \in \mathcal{LW}_{14,3}$ from Example~\ref{example3lattice}. Since $W_B$ steps vertically at $I = \{i_1,i_2,i_3\} = \{1,3,8\}$ we put chord starts at $(1,2)$, $(3,4)$, and $(8,9)$. The options for ending $c_3$ are $J_3=\{10,11,12\}$. Since $W_A$ makes $a_3=0$ steps at row~3, we select the smallest $j_3 = 10$ and end $c_3$ at $(10,11)$. The options for ending $c_2$ are $J_2=\{5,6,7,8,10,11,12\}$. Since $a_2=1$ we skip $5$, select $j_2=6$, and end $c_2$ at $(6,7)$. Finally, $c_1$ can end at any of $J_1=\{3,6,7,8,10,11,12\}$. Since $a_1=5$ we skip the first five options, select $j_1=11$, and end $c_1$ at $(11,12)$. This resulting chord diagram $\Phi(W_A,W_B)$ is the one in Example~\ref{example3chords}.
\end{ex}

\begin{prop}
The map $\Phi$ is a well-defined bijection between the noncrossing pairs of lattice walks~$\mathcal{LW}_{n,k}$ and the chord diagrams~$\mathcal{CD}_{n,k}$.
\end{prop}

\begin{proof}
Clearly, any choice of $k$ start segments in $\{(1,2),\dots,(n-4,n-3)\}$ is obtained from some walk~$W_B$, and every walk uniquely determines such starts. It is left to match the options for chord ends given starts with the possible walks $W_A$ given~$W_B$. 

The map $\Phi$ starts assigning chords' ends from the last chord $c_k$. The set $J_k$ always contains the possible marker~$n-2$ representing the segment $(n-2,n-1)$. If $c_k$ starts at $(n-4,n-3)$ then $j_k=n-2$ is the only option. This is the case $b_k=0$, where no horizontal steps are taken by $W_B$ and hence also by $W_A$ at row~$k$ of the lattice. Otherwise, $c_k$~starts at some $i_k = (n-4)-b_k$ and the $b_k+1$ available ends are $J_k = \{i_k+2,\dots,n-2\}$. In this case, $W_B$ takes some $b_k>0$ horizontal steps at row~$k$, and $W_A$ can take any $a_k \in \{0,\dots,b_k\}$ horizontal steps without crossing~$W_B$, so $W_A$ has exactly $|J_k|$ options as desired. The definition of $\Phi$ asserts that the end of $c_k$ lies on the segment $(j_k,j_k+1)$ selected by skipping the $a_k$ smallest options in~$J_k$ and letting $j_k$ be the next one. The skipped options will not be included in~$J_{k-1}$, because segments under the chord~$c_k$ are no longer available for ending other chords.

We continue with $c_{k-1},c_{k-2},\dots,c_1$ in the same fashion. Suppose that the ends of $c_{l+1},\dots,c_k$ are already determined. If $c_l$ is sticky with relation to $c_{l+1}$ then $J_l$ contains the remainder of $J_{l+1}$ as no new end segments become available. Otherwise, $c_l$ has some new end options $\{(i_l+2,i_l+3),\dots,(i_{l+1},i_{l+1}+1)\}$. These are $b_l$ segments, same as the horizontal steps taken by $W_B$ at row~$l$ of the lattice. These steps make room for $W_A$ to potentially take $b_l$ extra steps without crossing~$W_B$, beyond the steps not taken in the previous row~$l+1$. Thus $|J_l|=|J_{l+1}|-a_{l+1}+b_l$ is again the number of options for~$a_l$, the horizontal steps at row~$l$. By the definition of $\Phi$, the end of $c_l$ is selected by skipping the $a_l$ smallest options in $J_l$. These $a_l$ segments will not appear in $J_{l-1}$, to prevent intersection of chords.

In conclusion, throughout the procedure, the ways to continue the walk $W_A$ in row~$l$ are in bijection with the ways to pick the end of~$c_l$ in the chord diagram. 
\end{proof}

The next combinatorial objects are $\oplus$-diagrams, which are used to represent positroid cells in general~\cite{lam2008total}. Here we only consider the special case of $\oplus$-diagrams that are equivalent to chord diagrams, and refer to~\cite{postnikov2006total, lam2008total,karp2020decompositions} for the general case.

\begin{definition}
\label{bcfw-oplus}
A \emph{BCFW $\oplus$-diagram} of type $(n,k)$ is a Young diagram with $k$ rows and at most $n-k$ columns, filled with $\bigcirc$ and $\bigplus$ according to the following rules. In every row of the diagram, 
\begin{samepage}
\begin{itemize}
\item[(a)]
Four boxes contain $\bigplus$ and the rest contain $\bigcirc$.
\item[(b)]
No $\bigcirc$ appears before the 1st $\bigplus$ from the left or after the 4th $\bigplus$.
\item[(c)]
No $\bigcirc$ between the 3rd and 4th $\bigplus$ is directly below a $\bigplus$.
\item[(d)]
Each $\bigcirc$ between the 2nd and 3rd $\bigplus$ is directly above an $\bigcirc$ between the 3rd and 4th $\bigplus$ in its row.
\end{itemize}
The set of all BCFW $\oplus$-diagrams of type $(n,k)$ is denoted $\mathcal{OP}_{n,k}$.
\end{samepage}
\end{definition}

\begin{samepage}
\begin{ex}
\label{example3oplus}
Here is a BCFW $\oplus$-diagram of type $(14,3)$. The role of the labels along the rim will be explained below.
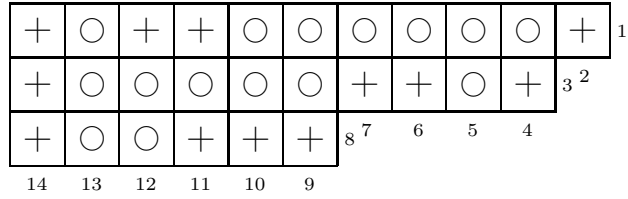
\def\O{\bigcirc}
\def\X{\bigplus}
$$\ytableausetup{boxsize=2em}
\begin{ytableau}
\X & \O & \X & \X & \O & \O & \O & \O & \O & \O & \X & \none[{\scriptstyle 1}\;\;\;\;] \\
\X & \O & \O & \O & \O & \O & \X & \X & \O & \X & \none[{\scriptstyle 3\;}^2\;\,] \\
\X & \O & \O & \X & \X & \X & \none[{\scriptstyle 8\;}^7\;\,] & \none[^6]  & \none[^5]  & \none[^4] \\
\none[^{14}]  & \none[^{13}]  & \none[^{12}]  & \none[^{11}]  & \none[^{10}]  & \none[^9]
\end{ytableau}$$
\end{ex}
\end{samepage}

\begin{definition}
\label{chord2oplus}
Let the map $\Psi : \mathcal{CD}_{n,k} \to \mathcal{OP}_{n,k}$ define \emph{the BCFW $\oplus$-diagram that corresponds to a chord diagram} by the following procedure. Consider $D \in \mathcal{CD}_{n,k}$ with chords $c_1,\dots,c_k$. The shape of the Young diagram $\Psi(D)$ is determined by an $n$-step lattice walk that outlines its rim. We take a unit step down at every $i \in [n]$ such that $(i,i+1)$ starts a chord, and take a step left otherwise. We assign the corresponding labels $1,\dots,n$ to the rows and columns of the diagram, as in Example~\ref{example3oplus}. We then fill $\Psi(D)$ from bottom to top, iterating over the chords $c_k,\dots,c_1$. In the row of~$c_l = (i_l,i_l+1,j_l,j_l+1)$, we first identify the suffix of boxes whose column label is smaller than~$j_l$. Then,
\begin{samepage}
\begin{enumerate}
\itemsep0.25em
\item 
Write \fbox{$\bigcirc\bigcirc\cdots\bigcirc\bigplus$} in that suffix, though some of the $\bigcirc$ may already be there.
\item
Fill with \fbox{$\bigcirc$} the column of boxes above each $\bigcirc$ in the suffix. 
\item
Write \fbox{$\bigplus\bigcirc\cdots\bigcirc\bigplus\bigplus$} in the remaining \emph{empty} boxes of the row, which may be nonconsecutive.
\end{enumerate}
\end{samepage}
This map is demonstrated by Examples~\ref{example3chords} and~\ref{example3oplus}, or the more complicated Example~\ref{example9} below.
\end{definition}

\begin{prop}
The map $\Psi$ from chord diagrams to BCFW $\oplus$-diagrams is well defined and one-to-one.
\end{prop}

\begin{proof}
Items~(a) and (b) of Definition~\ref{bcfw-oplus} are straightforward to verify, so we check (c) and~(d). Since there are $k$ segments where chords start, the walk determining the rim of the diagram fits in a $k \times (n-k)$ rectangle. The suffix considered for each row contains at least the rightmost box, because $(i_{l\ast}+1,i_{l\ast}+2)$ is a segment under the chord $c_l$ where no chord starts. Clearly, that box is still empty, so writing~$\bigplus$ there does not violate item~(c) of Definition~\ref{bcfw-oplus}. At least three empty boxes are available outside the suffix, because markers $\{n-2,n-1,n\}$ are never smaller than~$j_m$ for any chord $c_m \in \{c_l,\dots,c_k\}$. Item~(d) follows since any~$\bigcirc$ between the 2nd and 3rd~$\bigplus$ has arisen from a suffix of a lower row. Hence, $\Psi$ is well-defined. 

To see that $\Psi$ is one-to-one, first note that the segments $(i_l,i_l+1)$ where chords start can be recovered from the shape of $\Psi(D)$. The segment $(j_l,j_l+1)$ where a chord $c_l$ ends given $c_{l+1},\dots,c_k$ is revealed by the column label of the 3rd~$\bigplus$ which is either~$j_l$, or~$i_{m\ast}+1$ if some~$i_m=j_l$. The assumption that $c_l$ can not intersect such $c_m$ uniquely determines $j_l$ in the latter case.
\end{proof}

Karp, Williams, and Zhang describe a procedure $\Omega : \mathcal{LW}_{n,k} \to \mathcal{OP}_{n,k}$ that assigns an $\oplus$-diagram to a given noncrossing pair of lattice walks. See~\cite[Definition~6.2]{karp2020decompositions} for the details. It can be verified that $\Omega$ produces BCFW $\oplus$-diagrams satisfying Definition~\ref{bcfw-oplus}. A close examination of the three procedures shows that $\Omega = \Psi \circ \Phi$. 

We go on and describe another combinatorial object -- a decorated permutation. This is one of the several combinatorial structures that were given by Postnikov~\cite{postnikov2006total} for representing general positroid cells. There is a one-to-one correspondence between positroid cells in the nonnegative Grassmannian $\Grnn{k}{n}$ and decorated permutations of $[n]$ with $k$ anti-excedances.

\begin{definition}
\label{decorated}
A \emph{decorated permutation of $[n]$} is a one-to-one map $\pi : [n] \to [n]$, whose fixed points are classified into two types, black and white. A~\emph{black} fixed point is written $\pi(m)=m$ as usual, while a white fixed point is written $\pi(m)=\overline{m}$. An element $m \in [n]$ is an \emph{anti-excedance} of $\pi$ if either $\pi^{-1}(m) > m$ or $\pi(m) = \overline{m}$. 
\end{definition}

For concreteness, we briefly include here the relation between decorated permutations and positroid cells, as follows from Postnikov~\cite[Section~17]{postnikov2006total}, see e.g.~\cite[Lemma~2.3]{muller2017twist}. 

\begin{definition}
\label{perm2cell}
Given a decorated permutation $\pi$ of $[n]$ with $k$ anti-excedances, the \emph{positroid cell $S \subset \Grnn{k}{n}$ associated to $\pi$} is uniquely determined by the following condition. If $\pi(i)=j$, then every $C \in S$ satisfies that $j$ is the label of the first column $C^j$ after $C^i$, such that $C^i$ is spanned by $C^{i+1},C^{i+2},\dots,C^j$ with indices added modulo~$n$. If $\pi(i)=i$, then $C^i$ is zero if $i$ is a black fixed point and nonzero if white.
\end{definition}

\begin{rmk}
White fixed points do not appear in the decorated permutations corresponding to the positroid cells appearing in this paper. Hence, we freely use the standard notation and operations for permutations, such as the composition of two permutations. Any unspecified fixed point is assumed to be black.
\end{rmk}

In order to complete our review of the definition of BCFW cells from lattice walks, it is only left to recall how a decorated permutation is associated to an $\oplus$-diagram. This requires the following definition.

\label{sec:pipes}
\def\O{\tikz[baseline=-2.5pt,line width=1.5pt,<-]{\draw (-0.35,0) -- (0.35,0); \draw (0,0.35) -- (0,-0.35);}}
\def\X{\tikz[baseline=-2.5pt,line width=1.5pt,<-]{\draw (-0.35,0) to[bend left] (0,-0.35); \draw (0,0.35) to[bend right] (0.35,0);}}

\begin{definition}
\label{pipedream}
A \emph{pipe dream} of type $(n,k)$ is an equivalent way to draw an $\oplus$-diagram of type $(n,k)$. The content of the boxes is replaced by the rule:
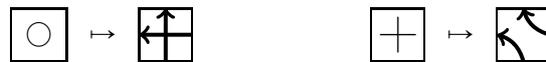
\vspace{0.5em}
$$ \setlength{\fboxsep}{5.25pt}
\boxed{\bigcirc}
\;\;\mapsto\;\; \setlength{\fboxsep}{-1pt}\boxed{\O} \;\;\;\;\;\;\;\;\;\;\;\;\;\;\;\;\;\;\;\;\;\;\;\; \setlength{\fboxsep}{2.1875pt} \boxed{\bigplus} \;\;\mapsto\;\; \setlength{\fboxsep}{-1pt}\boxed{\X} \vspace{0.5em}
$$
The \emph{labels} $1,\dots,n$ are written along the rim of the diagram, from top right to bottom left, and copied to the opposite side of every row or column. The \emph{decorated permutation of $[n]$ that corresponds to the pipe dream} assigns, to each rim label, the left or top label obtained by flowing through the pipes. Fixed points are black or white if they arise from a column or a row, respectively.   
\end{definition}

\begin{samepage}
\begin{ex}
\label{example3pipe}
The pipe dream of type (14,3) of the BCFW $\oplus$-diagram in Example~\ref{example3oplus} is the following one. The induced decorated permutation maps $1 \mapsto 2$, $2 \mapsto 11$, $3 \mapsto 4$, $4 \mapsto 6$, and so on.
$$\ytableausetup{boxsize=2em}
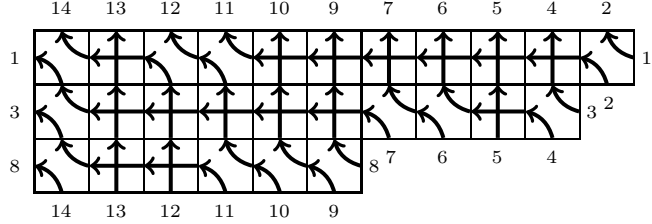
\begin{ytableau}
\none & \none[_{14}]  & \none[_{13}]  & \none[_{12}]  & \none[_{11}]  & \none[_{10}]  & \none[_9] & \none[_7]  & \none[_6]  & \none[_5]  & \none[_4]  & \none[_2] \\
\none[\;\;{\scriptstyle 1}] & \X & \O & \X & \X & \O & \O & \O & \O & \O & \O & \X & \none[{\scriptstyle 1}\;\;\;\;] \\
\none[\;\;{\scriptstyle 3}] & \X & \O & \O & \O & \O & \O & \X & \X & \O & \X & \none[{\scriptstyle 3\;}^2\;\,] \\
\none[\;\;{\scriptstyle 8}] & \X & \O & \O & \X & \X & \X & \none[{\scriptstyle 8\;}^7\;\,] & \none[^6]  & \none[^5]  & \none[^4] \\
\none & \none[^{14}]  & \none[^{13}]  & \none[^{12}]  & \none[^{11}]  & \none[^{10}]  & \none[^9]
\end{ytableau}$$
\end{ex}
\end{samepage}

\subsection{BCFW Cells from Chord Diagrams}
\label{subsec-perms}

In this section, we give our main definition of the collection of BCFW positroid cells. To each chord diagram with $n$ markers and $k$ chords, we directly associate a decorated permutation of $[n]$ with $k$ anti-excedances, which in turn defines a positroid cell in $\Grnn{k}{n}$. We also prove that this definition coincides with the one that appears in the literature and is described in the previous section.

\begin{definition}
\label{defdecperm}
\label{eq:explicit_sigma_alpha}
The \emph{decorated permutation of $[n]$ corresponding to a chord diagram} $D \in \mathcal{CD}_{n,k}$ is the following product of 5-cycles
$$ \pi \;=\; (T_1~U_1~V_1~W_1~n)\,(T_2~U_2~V_2~W_2~n)\cdots(T_k~U_k~V_k~W_k~n) $$
\begin{samepage}
where for every chord $c_l = (i_l,i_l+1,j_l,j_l+1)$, ordered lexicographically 
\begin{itemize}
\itemsep0.5em
\item 
$T_l = i_l$
\item 
$U_l = i_{l\ast}+1$
\item 
$V_l = \begin{cases} 
j_l & \text{if no chord starts at }(j_l,j_l+1) \\ 
i_{m\ast}+1 & \text{if some chord }c_m\text{ starts at }(j_l,j_l+1)
\end{cases}$
\item 
$W_l = \begin{cases}
j_l+1 & \text{if no chord starts at }(j_l,j_l+1)\text{, nor does any chord starts at }(j_l+1,j_l+2)\\
i_{m\ast} + 1 & \text{if no chord starts at }(j_l,j_l+1)\text{, and }c_m\text{ starts at }(j_l+1,j_l+2) \\
j_m & \text{if }c_m\text{ starts at }(j_l,j_l+1)\text{, and no chord starts at }(j_m,j_m+1)\\
i_{h*}+1 & \text{if }c_m\text{ starts at }(j_l,j_l+1)\text{, and }c_h\text{ starts at }(j_m,j_m+1)
\end{cases}$
\end{itemize}
\end{samepage}
where as in Definition~\ref{terminology}, $c_{l\ast} = (i_{l\ast},i_{l\ast}+1,j_{l\ast},j_{l\ast}+1)$ is the last chord in a maximal sticky chain descendent from $c_l$, and it is understood that $c_l=c_{l\ast}$ if no chord starts at $(i_l+1,i_l+2)$. The set of decorated permutations that correspond to chord diagrams in~$\mathcal{CD}_{n,k}$ is denoted~$\mathcal{DP}_{n,k}$.
\end{definition}

\begin{ex}
\label{example3perm}
The decorated permutation of $[14]$ that corresponds to the chord diagram given in Example~\ref{example3chords} is
$$ \pi \;=\; (1~ 2~ 11~ 12~ 14)~ (3~ 4~ 6~ 7~ 14)~ (8~ 9~ 10~ 11~ 14) \;\in\; \mathcal{DP}_{14,3}$$
Writing $\pi$ in two-line notation, i.e., $\pi(m)$ below $m$, we obtain:
$$\begin{tabular}{|p{1em}p{1em}p{1em}p{1em}p{1em}p{1em}p{1em}p{1em}p{1em}p{1em}p{1em}p{1em}p{1em}p{1em}|}\hline
1 & 2 & 3 & 4 & 5 & 6 & 7 & 8 & 9 & 10 & 11 & 12 & 13 & 14 \\
2 & 11 & 4 & 6 & 5 & 7 & 1 & 9 & 10 & 12 & 3 & 14 & 13 & 8 \\ \hline
\end{tabular}$$
This permutation has 3 anti-excedances, at $\{1,3,8\}$. Hence it defines a positroid cell in $\Grnn{3}{14}$. This example demonstrates the first cases for $V_l$ and~$W_l$. See Example~\ref{example9} below for some of the other cases.
\end{ex}

The common theme of the various cases in Definition~\ref{defdecperm} is that $(T_l,U_l,V_l,W_l)$ are preferably the four markers in~$c_l$, unless they coincide with some first markers $i_m$ of other chords~$c_m$. If they do, then we advance them to subsequent elements of those chords in an order-preserving way. It follows that each 5-cycle $(T_l~U_l~V_l~W_l~n)$ has a single anti-excedance, namely $n \mapsto T_l$. The element $T_l=i_l$ is also an anti-excedance in the product of the 5-cycles of $c_l,\dots,c_k$ because it is the smallest element there. Since $i_l$ is a first marker, it does not repeat in the 5-cycles of $c_1,\dots,c_{l-1}$ and remains an anti-excedance of the whole decorated permutation~$\pi$. We conclude the following.

\begin{cor}
\label{perm2cell-bcfw}
Every decorated permutation $\pi \in \mathcal{DP}_{n,k}$ has exactly $k$ anti-excedances. Therefore, $\pi$~corresponds to a positroid cell $S \subseteq \Grnn{k}{n}$.
\end{cor}

We emphasize that the correspondence we use between decorated permutations and positroid cells in the nonnegative Grassmannian is as in Definition~\ref{perm2cell}. The following main definition is based on this correspondence, composed on Definition~\ref{defdecperm} of the decorated permutations corresponding to a chord diagram. 

\begin{definition}
\label{maindefbcfw}
The \emph{BCFW cells}, denoted $\mathcal{BCFW}_{n,k}$ are the set of positroid cells represented by~$\mathcal{DP}_{n,k}$, i.e., by the decorated permutations of $[n]$ that correspond to the chord diagrams in~$\mathcal{CD}_{n,k}$.
\end{definition}

Using other equivalent combinatorial objects, Postnikov~\cite[Section~6]{postnikov2006total} gives a concrete algorithm for computing a parameterized family of matrices that represent the elements of the positroid cell of~$\pi$. These representatives are \emph{not} the domino matrices given above in Section~\ref{subsec-domino}. As mentioned above in Theorem~\ref{thm:domino_non_formal}, we show that they are indeed equivalent under the left $\GL_k(\mathbb{R})$ action.

\medskip
As described in Section~\ref{subsec-equivalence}, Karp, Williams and Zhang~\cite{karp2020decompositions} show that the recursively defined BCFW positroid cells are those that correspond to $\oplus$-diagrams that arise from noncrossing pairs of lattice walks. According to that definition, the BCFW cells also correspond to the decorated permutations of the pipe dreams of these $\oplus$-diagrams, which may be associated to chord diagrams via the map~$\Psi$. Here, in Definition~\ref{maindefbcfw}, we define the BCFW positroid cells differently, via the decorated permutations that arise from chord diagrams directly. The following proposition implies that the two definitions agree, yielding the same collection of cells.

\begin{prop}
\label{prop:bcfw dperm}
Let $D \in \mathcal{CD}_{k,n}$ be a chord diagram. The decorated permutation of $[n]$ that corresponds to~$D$, via Definition~\ref{defdecperm}, is equal to the decorated permutation of $[n]$ that corresponds to the pipe dream of~$\Psi(D)$, via Definitions~\ref{chord2oplus} and~\ref{pipedream}.
\end{prop}

\begin{proof}
The pipes in a pipe dream always flow upwards. Therefore, one can break a pipe dream with $k$ rows into $k$ smaller pipe dreams, of one row each. The decorated permutation that corresponds to the original pipe dream is the composition of the $k$ decorated permutations that correspond to its one-row pipe dreams. The labels are induced from the original pipe dream. Labels that do not appear in a certain row can safely be regarded as fixed points.

Consider a pipe dream of type $(n,k)$ that arises from a BCFW $\oplus$-diagram, so it satisfies the rules in Definition~\ref{bcfw-oplus}. Each of its rows has the following form, with some labels $T' < U' < V' < W' < n$.
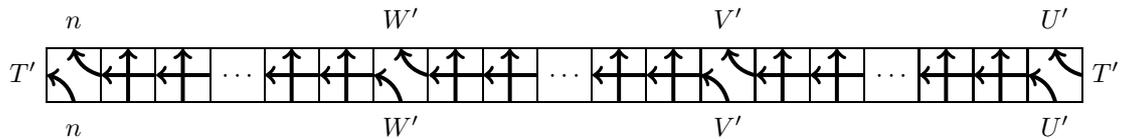
$$\ytableausetup{boxsize=2em}
\begin{ytableau}
\none & \none[{n}] & \none & \none & \none & \none & \none & \none[{W'}] & \none & \none & \none & \none & \none & \none[{V'}] & \none & \none & \none & \none & \none & \none[{U'}] \\
\none[\;{T'}] & \X & \O & \O & \cdots & \O & \O & \X & \O & \O & \cdots & \O & \O & \X & \O & \O & \cdots & \O & \O & \X & \none[{T'}\;] \\
\none & \none[{n}] & \none & \none & \none & \none & \none & \none[{W'}] & \none & \none & \none & \none & \none & \none[{V'}] & \none & \none & \none & \none & \none & \none[{U'}]
\end{ytableau}$$
The decorated permutation of this one-row pipe dream is the 5-cycle $(T'~U'~V'~W'~n)$. Overall, the decorated permutation of $[n]$ that corresponds to the pipe dream of $\Psi(D)$ is
$$ \pi' \;=\; (T'_1~U'_1~V'_1~W'_1~n)~(T'_2~U'_2~V'_2~W'_2~n)~\cdots~(T'_k~U'_k~V'_k~W'_k~n) $$
The element $T'_l$ is the row label of the $l$th row of the pipe dream of $\Psi(D)$. The four elements $n,W'_l,V'_l,U'_l$ are the labels of the four columns where there is a $\bigplus$ in the $l$th row of $\Psi(D)$. These labels are determined by $D$ via the procedure described in Definition~\ref{chord2oplus}. 

Not incidentally, this expression is similar to the decorated permutation of $[n]$ that corresponds to the chord diagram~$D$. Definition~\ref{defdecperm} writes it as
$$ \pi \;=\; (T_1~U_1~V_1~W_1~n)~(T_2~U_2~V_2~W_2~n)~\cdots~(T_k~U_k~V_k~W_k~n) $$
where $T_l, U_l, V_l, W_l$ are defined there as markers in the chord diagram~$D$.
In order to show that $\pi' = \pi$, it is sufficient to verify that the labels
$T'_l,U'_l,V'_l,W'_l$ agree with $T_l, U_l, V_l, W_l$ for every~$l \in \{1,\dots,k\}$.

Consider the $l$th row of $\Psi(D)$, which corresponds to the chord $c_l=(i_l,i_l+1,j_l,j_l+1)$ of~$D$. When considering the following cases, it is illuminating to examine some chords in Example~\ref{example9} below, and see how the 5-cycles in its decorated permutation correspond to the positions of $\bigplus$ in its $\oplus$-diagram.
\begin{itemize}
\item[(T)]
By Definition~\ref{chord2oplus}, the label $T'_l$ of row~$l$ of~$\Psi(D)$ is the $l$th vertical step in the rim walk, which is the first marker of $c_l$, which is~$i_l = T_l$ by Definition~\ref{defdecperm}. Therefore $T'_l = T_l$.
\item[(U)]
The label $U'_l$, of the column of the 4th $\bigplus$ in row~$l$, is the next horizontal step after~$i_l$. If no chord starts at $(i_l+1,i_l+2)$ then $U'_l = i_l+1$. Otherwise, it is $U'_l = i_{l\ast}+1$ because all markers from $i_l+1$ to $i_{l\ast}$ are first markers of chords in the sticky chain $c_l,\dots,c_{l\ast}$. In any case, $U'_l=U_l$.
\item[(V)] 
The label $V'_l$, of the  3rd $\bigplus$ in row~$l$, is determined as the rightmost box of column label at least~$j_l$ that is still empty when we fill the $l$th row. If no other chord starts at $(j_l,j_l+1)$ then $j_l$ is a column label. If there exists another chord $c_m$ with $i_m=j_l$ then the smallest column label after $j_l$ is the next horizontal step~$i_{m\ast}+1$. We verify that in either case this box is still empty and thus filled with the 3rd $\bigplus$. 

Indeed, step~2 of the procedure in Definition~\ref{bcfw-oplus} in row $h>l$ only fills columns whose label $L$ satisfies $i_{h\ast}+1<L<j_h$ for some chord $c_h=(i_h,i_h+1,j_h,j_h+1)$. The case $L=j_l$ is impossible, since then $c_h$ would intersect $c_l$. The case $L=i_{m\ast}+1$ is also impossible because only chord in the sticky chain that ends at~$c_{m\ast}$ may start between $j_l$ and $L$ and these have their 2nd $\bigplus$ in column~$L$. It follows that $V'_l$ is the column label of the right empty box, and~$V'_l=V_l$.  
\item[(W)]
First, consider the case that no chord starts at $(j_l,j_l+1)$. Then, we have seen that the 3rd $\bigplus$ fills an empty box in column $j_l$, and the next box may be either in column $j_l+1$, or $i_{m\ast}+1$ if some chord $c_m$ starts at $(j_l+1,j_l+2)$. This box is empty and filled with the 2nd $\bigplus$ exactly for the same reasons as in the two cases for $V'_l$ above. It follows that $W'_l=W_l$ in these cases. 

Now, suppose that a chord $c_m$ starts at $(j_l,j_l+1)$. Then the 3rd $\bigplus$ went to column $i_{m\ast}+1$. All the subsequent columns $L < j_m$ are filled with $\bigcirc$ in step~2 of row~$m$. If no chord starts at $(j_m,j_m+1)$, then the next empty box is at $j_m$ since it labels a horizontal rim step. This one is empty because a chord that would fill it with $\bigcirc$, would also intersect~$c_l$. Hence the 2nd $\bigplus$ lands at column $W_l' = j_m = W_l$. 

In the remaining case that $c_m$ starts at $(j_l,j_l+1)$ and $c_h$ starts at $(j_m,j_m+1)$ the next column label without an $\bigcirc$ is $i_{h\ast}+1$. Again, it is empty by a nonintersection argument, and thus filled with the 2nd $\bigplus$. So also in this case $W_l' = i_{h\ast}+1 = W_l$.
\end{itemize}
We conclude $\pi'=\pi$.
\end{proof}

\begin{rmk}
An $\oplus$-diagram is \emph{reduced} if no two pipes of its pipe dream cross twice. It follows from the conditions in Definition~\ref{bcfw-oplus} that BCFW $\oplus$-diagrams are reduced. This implies that the dimension of the corresponding positroid cell in $\Grnn{k}{n}$ equals the number of~$\bigplus$ in the diagram, which is~$4k$ . See~\cite[Lemma~6.4]{karp2020decompositions} for a detailed proof.
\end{rmk}

\subsection{Example}
\label{subsec-example}

\begin{samepage}
\begin{ex}
\label{example9}
We summarize this section with a larger example of a chord diagram and its various derived combinatorial structures. Consider the following chord diagram. \end{ex}

\begin{align*}
D \;=\; ([18],((1,2,6,7),(2,3,4,5),&(4,5,6,7),(6,7,10,11),(7,8,9,10),\\&(10,11,16,17),(11,12,16,17),(13,14,16,17)))  \in \mathcal{CD}_{18,8}
\end{align*} 
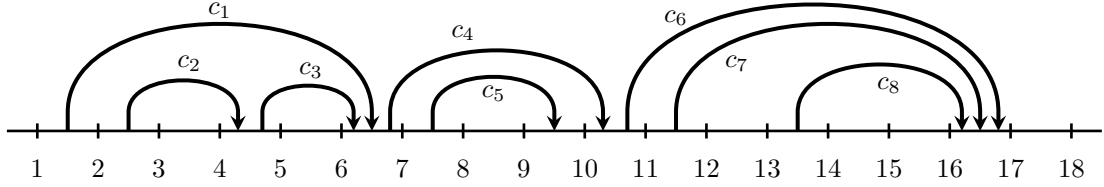
\begin{center}
\tikz[line width=1]{
\draw (0.4,0) -- (0.8*18.5,0);
\foreach \i in {1,2,...,18}{
\def\x{\i*0.8}
\draw (\x,-0.1)--(\x,+0.1);
\node at (\x,-0.5) {\i};}
\foreach \i/\j in {1.5/6.5,2.5/4.3,4.7/6.2,6.8/10.3,7.5/9.5,10.7/16.8,11.5/16.5,13.5/16.2}{
\def\x{\i*0.8}
\def\y{\j*0.8}
\draw[line width=1.5,-stealth] (\x,0) -- (\x,0.25) to[in=90,out=90] (\y,0.25) -- (\y,0);
}
\node at(4*0.8,1.6) {$c_1$};
\node at(3.5*0.8,0.9) {$c_2$};
\node at(5.5*0.8,0.8) {$c_3$};
\node at(8*0.8,1.3) {$c_4$};
\node at(8.5*0.8,0.5) {$c_5$};
\node at(11.5*0.8,1.5) {$c_6$};
\node at(12.5*0.8,0.9) {$c_7$};
\node at(15*0.8,0.6) {$c_8$};
}
\end{center}
\end{samepage}

\begin{samepage}
\medskip
\noindent
To practice some terminology: $c_1$, $c_4$, and $c_6$ are the top chords. $c_1$ is a sticky parent of $c_2$, since their starts $(1,2)$ and $(2,3)$ overlap. The siblings $c_2$ and $c_3$ are head-to-tail. The three chords in the chain $c_6, c_7, c_8$ are same-end, and their common end is $(16,17)$. The chords $c_2$, $c_3$, and $c_5$ are short. There is no long chord in this diagram. 

\medskip
\noindent
The domino matrix $C \in \mathcal{DM}_{18,8}$ corresponding to $D$ is
$$
\begin{array}{*{18}{c}}
\color{lightgray} 1 & 
\color{lightgray} 2 & 
\color{lightgray} 3 & 
\color{lightgray} 4 & 
\color{lightgray} 5 & 
\color{lightgray} 6 & 
\color{lightgray} 7 & 
\color{lightgray} 8 & 
\color{lightgray} 9 & 
\color{lightgray} 10 & 
\color{lightgray} 11 & 
\color{lightgray} 12 & 
\color{lightgray} 13 & 
\color{lightgray} 14 & 
\color{lightgray} 15 & 
\color{lightgray} 16 & 
\color{lightgray} 17 & 
\color{lightgray} 18 \\
\hline \alpha_1 & \beta_1 & 0 & 0 & 0 & \gamma_1 & \delta_1 & 0 & 0 & 0 & 0 & 0 & 0 & 0 & 0 & 0 & 0 & \varepsilon_1 \\
\varepsilon_2\alpha_1 & \!\!\!\varepsilon_2\beta_1 {+} \alpha_2\!\!\! & \beta_2 & \gamma_2 & \delta_2 & 0 & 0 & 0 & 0 & 0 & 0 & 0 & 0 & 0 & 0 & 0 & 0 & 0  \\
\varepsilon_3\alpha_1 & \varepsilon_3\beta_1 & 0 & \alpha_3 & \beta_3 & \gamma_3 & \delta_3 & 0 & 0 & 0 & 0 & 0 & 0 & 0 & 0 & 0 & 0 & 0  \\
0 & 0 & 0 & 0 & 0 & \alpha_4 & \beta_4 & 0 & 0 & \gamma_4 & \delta_4 & 0 & 0 & 0 & 0 & 0 & 0 & \varepsilon_4 \\
0 & 0 & 0 & 0 & 0 & \!\varepsilon_5\alpha_4 & \!\!\!\varepsilon_5\beta_4{+}\alpha_5\!\!\! & \beta_5 & \gamma_5 & \delta_5 & 0 & 0 & 0 & 0 & 0 & 0 & 0 & 0 \\
0 & 0 & 0 & 0 & 0 & 0 & 0 & 0 & 0 & \alpha_6 & \beta_6 & 0 & 0 & 0 & 0 & \gamma_6 & \delta_6 & \varepsilon_6 \\
0 & 0 & 0 & 0 & 0 & 0 & 0 & 0 & 0 & \!\varepsilon_7\alpha_6 & \!\!\!\varepsilon_7\beta_6{+}\alpha_7\!\!\! & \beta_7 & 0 & 0 & 0 & \gamma_7 & \delta_7 & 0\\
0 & 0 & 0 & 0 & 0 & 0 & 0 & 0 & 0 & 0 & \varepsilon_8\alpha_7 & \varepsilon_8\beta_7\! & \!\alpha_8\! & \!\beta_8\! & 0 & \gamma_8 & \delta_8 & 0 \\ \hline
\end{array} 
\vspace{0.5em}
$$
According to the sign rules, the variables $\varepsilon_1, \varepsilon_3, \gamma_4, \delta_4, \gamma_7, \delta_7$ are negative and the rest are positive. The ratio relations are: $\beta_3/\alpha_3 > \delta_2/\gamma_2$, $\,\beta_4/\alpha_4 > \delta_1/\gamma_1 > \delta_3/\gamma_3$, $\,\beta_6/\alpha_6 > \delta_4/\gamma_4$, and $\,\delta_6/\gamma_6 > \delta_7/\gamma_7  > \delta_8/\gamma_8$.
\end{samepage}

\medskip\noindent
The decorated permutation corresponding to $D$ is
\begin{samepage}
\begin{align*}
\pi \;=\; (1~3~8~12~18)~(2~3~&5~8~18)~(4~5~8~12~18)~(6~8~12~16~18)~\\&(7~8~9~12~18)~(10~12~16~17~18)~(11~12~16~17~18)~(13~14~16~17~18) \;\in\; \mathcal{DP}_{18,8}
\end{align*}
$$ \;\;=\;\; \begin{tabular}{|p{1.0em}p{1.0em}p{1.0em}p{1.0em}p{1.0em}p{1.0em}p{1.0em}p{1.0em}p{1.0em}p{1.0em}p{1.0em}p{1.0em}p{1.0em}p{1.0em}p{1.0em}p{1.0em}p{1.0em}p{1.0em}|}\hline
1 & 2 & 3 & 4 & 5 & 6 & 7 & 8 & 9 & 10 & 11 & 12 & 13 & 14 & 15 & 16 & 17 & 18 \\
3 & 8 & 5 & 12 & 1 & 18 & 2 & 9 & 16 & 6 & 4 & 17 & 14 & 7 & 15 & 10 & 11 & 13 \\ \hline
\end{tabular}
\vspace{0.5em}
$$
with 8 anti-excedances at $1,2,4,6,7,10,11,13$.
\end{samepage}

\begin{samepage}
\smallskip\noindent
Here is the BCFW $\oplus$-diagram $\Psi(D)$. Its pipe dream is the same, with pipe tiles instead of $\bigcirc$ and $\bigplus$. One may verify that the pipe flow on this diagram yields the same decorated permutation~$\pi$. \def\O{\bigcirc}
\def\X{\bigplus}
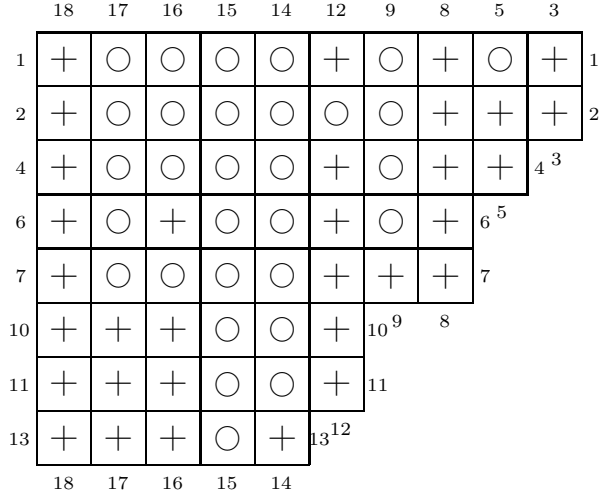
$$\ytableausetup{boxsize=2em}
\begin{ytableau}
\none & \none[_{18}] & \none[_{17}] & \none[_{16}] & \none[_{15}] & \none[_{14}] & \none[_{12}] & \none[_{9}] & \none[_{8}] & \none[_{5}] & \none[_{3}] \\
\none[\;\;\;{\scriptstyle 1}] & \X & \O & \O & \O & \O & \X & \O & \X & \O & \X & \none[{\scriptstyle 1}\;\;\;\;] \\
\none[\;\;\;{\scriptstyle 2}] & \X & \O & \O & \O & \O & \O & \O & \X & \X & \X & \none[{\scriptstyle 2}\;\;\;\;] \\
\none[\;\;\;{\scriptstyle 4}] & \X & \O & \O & \O & \O & \X & \O & \X & \X & \none[{\scriptstyle 4\;}^3\;\,] \\
\none[\;\;\;{\scriptstyle 6}] & \X & \O & \X & \O & \O & \X & \O & \X & \none[{\scriptstyle 6\;}^5\;\,] \\
\none[\;\;\;{\scriptstyle 7}] & \X & \O & \O & \O & \O & \X & \X & \X & \none[{\scriptstyle 7}\;\;\;\;] \\
\none[\;\;\,{\scriptstyle 10}] & \X & \X & \X & \O & \O & \X &  \none[{\scriptstyle 10\,}^{9}\;\,] & \none[^8] \\
\none[\;\;\,{\scriptstyle 11}] & \X & \X & \X & \O & \O & \X & \none[{\scriptstyle 11}\,\;\;\;] \\
\none[\;\;\,{\scriptstyle 13}] & \X & \X & \X & \O & \X & \none[{\scriptstyle 13}^{12}\;\,] \\
\none & \none[^{18}]  & \none[^{17}]  & \none[^{16}]  & \none[^{15}]  & \none[^{14}]
\end{ytableau}$$
\end{samepage}

\section{Structure of BCFW Cells}
\label{sec:domino}

In this section, we prove Theorem~\ref{thm:domino}, that every point in a BCFW cell has an explicit matrix representative in the domino form of the corresponding chord diagram, as conjectured in~\cite[Conjecture~A.7]{karp2020decompositions}. We give an algorithmic construction of domino matrices based on simple row and column operations, and analyze its properties mainly in terms of decorated permutations. The inductive nature of our construction is the key to our approach in later sections, where these cells are shown to form a triangulation.

\subsection{Matrix Operations}
\label{operations}

We define certain operations on matrices that let one manipulate subsets of the nonnegative Grassmannian. These definitions are close to those of Bao and He~\cite{bao2019m}. We then use them to define two embeddings from a given Grassmannian to a larger one that play a main role in this work. 

\begin{definition}
\label{nonconsecutive}
First, we set up some convenient notation for indexing vectors and matrices. We use arbitrary \emph{index sets} $\K,\N \subset \mathbb{N} = \{1,2,3,\dots\}$ to index rows and columns. Thus, the space of matrices is $\Mat_{\K \times \N}$ rather than $\Mat_{k \times n}$ which stands for the special case where $N=[n]$ and $K=[k]$. Similarly, $\Gr_{k,\N}$ is the Grassmannian of $k$-dimensional vector spaces in $\R^\N$. These spaces are clearly equivalent to those with $n=[|N|]$, but they are more convenient for performing insertion operations. When working with general index sets, we write $i +_{N} 1$ for the next-largest element after $i \in N$, and similarly $i +_{N} 2$ and $i -_{N} 1$, etc. We denote the maximum and minimum indices in a set by $\smax{N}=\max N$ and $\smin{N}=\min N$. The ordering of the index set~$N$ is understood to be cyclic, so that $\smax {N} +_N 1 = \smin {N}$. If $N$ is clear from the context, then we write $i \pl 1$ and $i \mi 1$ to emphasize that these are successors and predecessors with respect to a general index set with circular ordering. The definitions of chord diagrams, domino matrices, and decorated permutations in~Section~\ref{sec::bcfw} naturally extend from $[n]$ to this setting. We usually denote by $M_j^l$ the $(j,l)$ entry of the matrix~$M$, by $M_J$ its restriction to rows indexed by~$J$, and $M^I$ for columns indexed~$I$. We often state and prove results for $[n]$ and then use their natural extension to general index sets.
\end{definition}

\begin{samepage}
\begin{definition}
\label{def:pre}
The following map \emph{inserts a zero column} at position~$i$. 
$$ \pre_i\;:\;\Mat_{\K \times \N}\;\;\to\;\;\Mat_{\K \times (\N\cup\{i\})} $$
where $K,N \subset \mathbb{N}$ and $i \notin N$. For a matrix $A = (A_{j}^{i})_{j \in K}^{i \in N}$ in the domain,
$$ \pre_i(A) \;\;=\;\; 
\begin{pmatrix}
A_{\smin {K}}^{\smin {N}} & \cdots & A_{\smin {K}}^{i\mi1} & 0 & A_{\smin {K}}^{i\pl1} & \cdots & A_{\smin {K}}^{\smax {N}}\\
\vdots  & \ddots & \vdots & \vdots & \vdots & \ddots & \vdots\\
A_{\smax {K}}^{\smin {N}} & \cdots & A_{\smax {K}}^{i\mi1} & 0 & A_{\smax {K}}^{i\pl1} & \cdots & A_{\smax {K}}^{\smax {N}}
\end{pmatrix}
$$
where $i\mi1$ and $i\pl1$ relate of course to the index set $N \cup \{i\}$. This map induces an embedding between Grassmannians, which restricts also to positive Grassmannians:
$$ \pre_i\;:\;\Gr_{k,\N}\;\;\to\;\; \Gr_{k,\N\cup\{i\}} $$ 
$$ \pre_i\;:\;\Gr^{\geq}_{k,\N}\;\;\to\;\; \Gr^{\geq}_{k,\N\cup\{i\}} $$
\end{definition}
\end{samepage}

\begin{definition}
\label{def:inc}
The following map \emph{increments the number of rows and inserts a unit column},
$$ \inc_{i;j}\;:\;\Mat_{\K \times \N}\;\;\to\;\;\Mat_{(\K\cup\{j\}) \times (\N\cup\{i\})} $$
where $K,N \subset \mathbb{N}$ and $i\notin\N,~j\notin\K$. The new column is at~$i$ and the new row at~$j$, with $1$ at the entry $(j,i)$ and zeros elsewhere. This map also flips the signs of all entries $(j',i')$ with either $i'>i$ or $j'>j$ but not both. For $A = (A_{j}^i)_{j \in K}^{i \in N}$ in the domain,
$$
\inc_{i;j}(A) \;\;=\;\;
\begin{pmatrix}
A_{\smin {K}}^{\smin {N}} & \cdots & A_{\smin {K}}^{i\mi1} & 0 & -A_{\smin {K}}^{i\pl1} & \cdots & -A_{\smin {K}}^{\smax {N}}\\
\vdots  & \ddots & \vdots & \vdots & \vdots & \ddots & \vdots\\
A_{j\mi1}^{\smin {N}} & \cdots & A_{j\mi1}^{i\mi1} & 0 & -A_{j\mi1}^{i\pl1} & \cdots & -A_{j\mi1}^{\smax {N}}\\[0.25em]
0       & \cdots & 0         & 1 & 0       & \cdots & 0      \\
-A_{j\pl1}^{\smin {N}} & \cdots & -A_{j\pl1}^{i\mi1} & 0 & A_{j\pl1}^{i\pl1} & \cdots & A_{j\pl1}^{\smax {N}}\\
\vdots  & \ddots & \vdots & \vdots & \vdots & \ddots & \vdots\\
-A_{\smax {K}}^{\smin {N}} & \cdots & -A_{\smax {K}}^{i\mi1} & 0 & A_{\smax {K}}^{i\pl1} & \cdots & A_{\smax {K}}^{\smax {N}}
\end{pmatrix}
$$
This map induces an embedding between Grassmannians, which preserves positivity:
$$ \inc_i\;:\;\Gr_{k,\N}\;\;\to\;\;\Gr_{k+1,\N\cup\{i\}}$$
$$ \inc_i\;:\;\Grnn{k}{\N}\;\;\to\;\;\Grnn{k+1}{\N\cup\{i\}}$$
\end{definition}

\begin{definition}
\label{def:x_i y_i}
The next operations are \emph{adding to one column a multiple of an adjacent one}. This is done by right multiplication with the following $N \times N$ matrices, where $N \subset \mathbb{N}$, $i \in N \setminus \{\smax{N}\}$, and $t$ is a real variable.
$$
\begin{matrix} &
\begin{matrix}
    \;\;\;{\color{lightgray}i} & {\color{lightgray}i\pl1}
\end{matrix} & \\
[x_i(t)] \;= &
\begin{pmatrix}
    1 & \cdots & 0 & 0 & \cdots & 0 \\
    \vdots & \ddots & \vdots & \vdots & \ddots &  \vdots \\
    0 & \cdots & 1 & t & \cdots & 0 \\
    0 & \cdots & 0 & 1 & \cdots & 0 \\
    \vdots & \ddots & \vdots & \vdots & \ddots & \vdots \\
    0 & \cdots & 0 & 0 & \cdots & 1
\end{pmatrix} &
\begin{matrix}
    {\color{lightgray}i}\\
    {\color{lightgray}i\pl1}
\end{matrix}
\end{matrix}
\;\;\;\;\;\;\;\;\;\;\;\;\;\;\;\;\;\;\;\;
\begin{matrix} &
\begin{matrix}
    \;\;\;{\color{lightgray}i} & {\color{lightgray}i\pl1}
\end{matrix} & \\
[y_i(t)] \;= &
\begin{pmatrix}
    1 & \cdots & 0 & 0 & \cdots & 0 \\
    \vdots & \ddots & \vdots & \vdots & \ddots & \vdots \\
    0 & \cdots & 1 & 0 & \cdots & 0 \\
    0 & \cdots & t & 1 & \cdots & 0 \\
    \vdots & \ddots & \vdots & \vdots & \ddots & \vdots \\
    0 & \cdots & 0 & 0 & \cdots & 1
\end{pmatrix} &
\begin{matrix}
    {\color{lightgray}i}\\
    {\color{lightgray}i\pl1}
\end{matrix}
\end{matrix}
$$
These two matrices act on $\Mat_{\K\times\N}$ by right multiplication and on $\Mat_{\N\times\K}$ by left multiplication, for index sets $K,N \subset \mathbb{N}$. They induce well-defined actions on the Grassmannian $\Gr_{k,\N}$, which are defined by matrix multiplication $x_i(t) \,C = C\cdot[x_i(t)]$ and $y_i(t)\,C = C\cdot[y_i(t)]$ for a representative matrix $C \in \Gr_{k,\N}$. Observe that if $t \geq 0$ then the positive Grassmannian $\Grnn{k}{\N}$ maps to itself under $x_i(t)$ and~$y_i(t)$. This definition extends as follows to the last index $i= \max N$ by introducing a sign that depends on~$k$ in order to preserve positivity.
\begin{align*}
[x^k_i(t)] \;&=\; \Id_\N + (-1)^{(k-1)\cdot\delta[i=\max N]}\,t\,\E_{i}^{i{\pl1}} \\ 
[y^k_i(t)]\;&=\;\Id_\N+ (-1)^{(k-1)\cdot\delta[i=\max N]}\,t\,\E_{i{\pl1}}^{i}
\end{align*}
where $\Id_\N \in \Mat_{\N\times\N}$ is the identity matrix, $\E_{j}^{i} \in \Mat_{\N\times\N}$ is the matrix whose $(j,i)$ entry is 1 and the rest are~0, and $\delta[...] = 1$ if its argument holds and 0 otherwise. The dimension $k=|K|$ is usually implied from the context of $x_i(t)$ and $y_i(t)$ and omitted. 
\end{definition}

\begin{rmk}
These definitions of matrix operations include the trivial case where $K = \varnothing$ and $k=0$. The nonnegative Grassmannian $\Grnn{0}{N}$ is a single point, and hence $\pre_i$, $x_i(t)$ and $y_i(t)$ are trivial maps between one-point spaces.
The matrix operation $\inc_i$ maps the one point of $\Grnn{0}{\N}$ to an element of~$\Grnn{1}{\N\cup\{i\}}$, represented by a unit row with $1$ at $i$ and zeros elsewhere. 
\end{rmk}

\begin{definition}
\label{def::emb}
We compose previous operations, and introduce a useful map that \emph{increments with a new unit column, and then successively adds adjacent multiples from both sides}. For $N \subset \mathbb{N}$, $i \notin N$, and two sequences of real variables $s_1,\dots,s_r$  and $t_1,\dots,t_l$ such that $r+l \leq |N|$, we define a map between Grassmannians:
$$ \emb_{i,l,r}(t_1,\ldots,t_l,s_1,\ldots,s_r) \;:\; \Gr_{k-1,\N}\;\;\to\;\;\Gr_{k,\N\cup\{i\}} $$
as the composition of maps
$$ y_{i\mi l}(t_l)\circ\cdots\circ y_{i \mi 1}(t_1) \,\circ\, x_{i \pl (r-1)}(s_r)\circ\cdots\circ x_{i \pl 1}(s_2)\circ x_i(s_1) \,\circ\, \inc_{i} $$
We also denote this map by $\embilr{i,l,r}$ for short, where $(\textbf{t},\textbf{s})=(t_1,\dots,t_l,s_1,\dots,s_r)$.

More explicitly, the map $\emb_{i,l,r}$ first takes a representative matrix $C \in \Mat_{(k-1)\times\N}$, increments it by a row and a unit column to obtain the matrix $\inc_{i;k}(C)$. Then it adds a multiple of the new column to its right neighbor, a multiple of that neighbor to its right neighbor, and so on, and similarly to the left. Note that all the $r+l$ column additions can be performed at once by right multiplication with one $(N \cup \{i\}) \times (N \cup \{i\})$ matrix: $\left(\inc_{i;k}(C)\right)\left[x_{i}\cdots x_{i \pl r\mi 1} y_{i \mi 1}\cdots y_{i \mi l}\right]$. In some cases, we equivalently use $\inc_{i;j}$ with another choice of new row~$j \in [k]$.
\end{definition}

In this work, we usually apply the composed map~$\emb_{i,l,r}$ with $r+l=m=4$. One main example of this map, which corresponds to a possible step in the BCFW recursion, is defined by the following choice of parameters.

\begin{definition}
\label{def:upper emb}
Let $i < \min(N)$, and denote $n=\max N$. For $(t,u,v,w) \in \mathbb{R}^4$,
the \emph{upper embedding} matrix operation is
$$ \uemb_i(w,v,u,t) \;=\; \emb_{i,3,1}(w,v,u,t) \;=\; y_{n \mi 2}(u) \circ y_{n \mi 1}(v) \circ y_{n}(w) \circ x_i(t) \circ \inc_i $$
In matrix form, the upper embedding of $C \in \Mat_{(k-1)\times N}$ is
$$ \begin{pmatrix}
{\;\;1\;\;} & t & 0 & \cdots & 0 &  (-1)^{k-1}uvw &  (-1)^{k-1}vw &  (-1)^{k-1}w \\
0 & C^{\smin {N}}_{1} & C^{\smin {N} \pl 1}_{1} & \cdots & C^{n \mi 3}_{1} & C^{n \mi 2}_{1} + uC^{n \mi 1}_{1} + uv C^{n}_{1} &
C^{n \mi 1}_{1} + v C^{n}_{1} &
C^{n}_{1} \\
\vdots & \vdots & \vdots & \ddots & \vdots & \vdots & \vdots & \vdots \\
0 & C^{\smin {N}}_{k-1} & C^{\smin {N} \pl 1}_{k-1} & \cdots & C^{n \mi 3}_{k-1} & C^{n \mi 2}_{k-1} + uC^{n \mi 1}_{k-1} + uv C^{n}_{k-1} &
C^{n \mi 1}_{k-1} + v C^{n}_{k-1} &
C^{n}_{k-1}
\end{pmatrix}
\vspace{0.5em}
$$
\end{definition}

\begin{definition}
\label{def:lower emb}
While the upper embedding is useful for building cells bottom to top, occasionally we have to work top to bottom, and insert a lowest chord somewhere in a chord diagram. We demonstrate how this can be implemented with a composed map as well. Let $n=\max N$ and assume $n \mi 2 < i < n \mi 1$. For $(t,u,v,w) \in \mathbb{R}^4$, the \emph{lower embedding} matrix operation is
$$ \lemb_i(t,u,v,w) \;=\; \emb_{i,2,2}(t,u,v,w) \;=\; y_{i \mi 2}(t) \circ y_{i \mi 1}(u) \circ x_{i \pl 1}(w) \circ x_i(v) \circ \inc_i .$$
In matrix form, the lower embedding of $C \in \Mat_{(k-1)\times N}$ is
$$ \begin{pmatrix}
C^{\smin {N}}_{1} & \cdots & C^{n \mi 4}_{1} & C^{n \mi 3}_{1} + t C^{n \mi 2}_{1} & C^{n \mi 2}_{1} & 0 &
-C^{n \mi 1}_{1}&
-C^{n}_{1} - w C^{n \mi 1}_{1} \\
\vdots & \ddots & \vdots & \vdots & \vdots & \vdots & \vdots & \vdots \\
C^{\smin {N}}_{k-1} & \cdots & C^{n \mi 4}_{k-1} & C^{n \mi 3}_{k-1} + t C^{n \mi 2}_{k-1} & C^{n \mi 2}_{k-1} & 0 &
-C^{n \mi 1}_{k-1}&
-C^{n}_{k-1} - w C^{n \mi 1}_{k-1} \\
0 & \cdots & 0 & tu & u & {\;\;1\;\;} & v & vw
\end{pmatrix}
$$
\end{definition}

\begin{prop}
\label{positroid maps}
Let $S$ be a positroid cell in a nonnegative Grassmannian. The following properties of the matrix operations follow from their definitions.
\begin{itemize}
\item 
$\pre_i(S)$ and $\inc_i(S)$ map $S$ bijectively to a positroid cell in a larger Grassmannian.
\item
$x_i$ and $y_i$ map $S \times (0,\infty)$ to a positroid cell in the same Grassmannian.
\item
$\emb_{i,l,r}$ maps $S \times (0,\infty)^{l+r}$ to a positroid cell in a larger Grassmannian.
\end{itemize}
\end{prop}

\begin{rmk}
We have special interest in the cases of upper and lower embeddings given in Definition~\ref{def:upper emb} and Definition~\ref{def:lower emb}. These two maps turn out to be embeddings of the BCFW cells of $\Grnn{k-1}{n-1}$ into BCFW cells of~$\Grnn{k}{n}$, as demonstrated by Corollaries~\ref{cor:generation_left}-\ref{cor:generation_top} at the end of Section~\ref{sec:domino}. The names \emph{upper} and \emph{lower} derive from the representation of image cells in terms of chord diagrams. The upper embedding adds a long, 
top chord from $(1,2)$ to $(n-2,n-1)$. The lower embedding adds a lowest, 
short chord from $(n-4,n-3)$ to $(n-2,n-1)$.
\end{rmk}

We end this section with a brief summary about the representation of the matrix operations in terms of planar bicolored graphs. \emph{Plabic graphs} are frequently used in the literature as one of the combinatorial structures that Postnikov \cite{postnikov2006total} introduced to encode positroid cells. Our discussion of plabic graphs in this paper is limited and does not require going into details. Therefore, we do not present here their full definitions and properties and refer to~\cite{postnikov2006total, lam2014totally, fomin2021introduction}.

\begin{summary}
\label{subsec:plabic}
We consider $\{1,3\}$-valent plabic graphs whose boundary vertices are indexed by a set~$N \subseteq [n]$, with the inherited cyclic ordering, and use the following two terms. A \emph{lollipop} is a degree-one black or white internal vertex connected to a boundary vertex. A \emph{bridge} is a boundary-parallel edge, weighted by a variable, from white to black internal vertices that are neighbors of boundary vertices.

One way to associate positroid cells, or their points, with plabic graphs is the following sequential process, following \cite[Section 7]{lam2014totally},~\cite[Section 3]{arkani2016Grassmannian}. Start from an empty graph and a zero vector space, and simultaneously generate a plabic graph $G$ and a vector space $V$ by steps of the following types:
\begin{enumerate}
\itemsep0.125em
\item Add to $G$ a black lollipop with a new boundary vertex labeled $i \in [n]$, and apply $\pre_i$ to~$V$.
\item Add to $G$ a white lollipop with a new boundary vertex labeled $i \in [n]$, and apply $\inc_i$ to~$V$.
\item Add to $G$ a bridge of weight $t$ from $i$ to $i \pl 1$, and apply $x_i(t)$ to~$V$. 
\item Add to $G$ a bridge of weight $t$ from $i \pl 1$ to~$i$, and apply $y_i(t)$ to~$V$. 
\end{enumerate}
This construction gives a map $(0,\infty)^d\to\Gr_{k,n}^{\geq}$, where $d$ is the number of added bridges, $k$ is the number of added white lollipops, and $n$ is the total number of added lollipops. The graph is \emph{reduced} if the dimension of the image, as a subvariety of the Grassmannian, equals $d$, the number of faces of the graph minus one. A~plabic graph $G$ is associated a decorated permutation $\pi_G$ defined by certain trips between boundary vertices~\cite[Section 13]{postnikov2006total}. An explicit construction of the permutation $\pi_G$ from the sequence of $\pre_i,\inc_i,x_i(t),y_i(t)$ is given in Definition~\ref{sigma_and_the_algorithm} below. If $G$ is reduced then the map from $(0,\infty)^d$ to the nonnegative Grassmannian is a homeomorphism on a positroid cell (\cite[Theorem 12.7]{postnikov2006total},\cite[Theorem 7.12]{lam2014totally}), and $\pi_G$ is its decorated permutation via the standard correspondence used in Definition~\ref{maindefbcfw} (\cite[Sections 19-20]{postnikov2006total} and \cite[Lemma 2.6]{karp2020decompositions}). This homeomorphism is easily seen to be a diffeomorphism. If $G$ is not reduced, then $\pi_G$ might not be the decorated permutation corresponding to the image. 
\end{summary}

\begin{illustration}
We show the effect of the four matrix operations on the corresponding plabic graph. For example, the composed map $\emb_{i,l,r}$ corresponds to a white lollipop at~$i$, then $r$ bridges from $i$ to $i \pl 1$ to $i \pl 2$ etc., and $l$ bridges to the other side with the appropriate weight variables. Here, on the right hand side, we show how the upper embedding $\uemb_i$ acts on a given plabic graph~$D$. 
    
\begin{center}
~

\begin{tabular}{ccccccc}
\begin{tikzpicture}
\draw (0,0) -- (0,1);
\draw[line width=2pt] (-0.5,0) -- (0.5,0);
\filldraw[black] (0,1) circle (3pt);
\node at (0,-0.25) {$i$};
\end{tikzpicture}
&&
\begin{tikzpicture}
\draw (0,0) -- (0,1);
\draw[line width=2pt] (-0.5,0) -- (0.5,0);
\filldraw[fill=white] (0,1) circle (3pt);
\node at (0,-0.25) {$i$};
\end{tikzpicture}
& &
\begin{tikzpicture}
\draw (0,0) -- (0,1.25);
\draw (0.75,0) -- (0.75,1.25);
\draw (0,0.75) -- (0.75,0.75);
\draw[line width=2pt] (-0.25,0) -- (1,0);
\filldraw[fill=black] (0,0.75) circle (3pt);
\filldraw[fill=white] (0.75,0.75) circle (3pt);
\node at (0.75,-0.25) {$i$};
\node at (0,-0.25) {$i{\pl}1$};
\end{tikzpicture}
& &
\begin{tikzpicture}
\draw (0,0) -- (0,1.25);
\draw (0.75,0) -- (0.75,1.25);
\draw (0,0.75) -- (0.75,0.75);
\draw[line width=2pt] (-0.25,0) -- (1,0);
\filldraw[fill=white] (0,0.75) circle (3pt);
\filldraw[fill=black] (0.75,0.75) circle (3pt);
\node at (0.75,-0.25) {$i$};
\node at (0,-0.25) {$i{\pl}1$};
\end{tikzpicture}
\\
$\pre_i$ && $\inc_i$ && $x_i(t)$ && $y_i(t)$ \\ 
\end{tabular}
\;\;\;\;\;\;\;\;\;\;\;\;\;\;\;\;
\raisebox{-55pt}{\begin{tikzpicture}
\draw[gray, thick] (0,0) -- (0:25pt);
\filldraw[gray!30] (0,0) circle (25pt);
\draw (0:0) node{$D$};
\draw[thick] (0,0) circle (50pt);
\draw[thick] (0,0) circle (25pt);

%markers
\foreach \ang/\name/\dist/\anchor in {270/$i$/31.25pt/north, 
 270+20/$i\mi1$/25pt/west,
 270+40/$i\mi2$/25pt/west,
 270+60/$i\mi3$/25pt/west,
 270-30/$i\pl1$/25pt/east} {
\filldraw[black, thick] (\ang:47.5pt) -- (\ang:52.5pt) node[anchor=\anchor]{\name};
\draw (\ang:\dist) -- (\ang:50pt);
}

%edges
\draw (270:37.5pt) -- (270-30:37.5pt);
\draw (270:43.75pt) -- (270+20:37.5pt);
\draw (270+20:43.75pt) -- (270+40:37.5pt);
\draw (270+40:43.75pt) -- (270+60:37.5pt);

%white
\foreach \ang/\dist in {270/31.25pt, 270/37.5pt, 270/43.75pt, 
270+20/43.75pt, 270+40/43.75pt} {
    \filldraw[white] (\ang:\dist) circle (2pt);
\draw[black] (\ang:\dist) circle (2pt);
}

%black
\foreach \ang/\dist in {270-30/37.5pt, 
270+20/37.5pt, 270+40/37.5pt, 270+60/37.5pt} {
    \filldraw[black] (\ang:\dist) circle (2pt);
}

\draw (360:37.5pt) node{$\vdots$};
\draw (180:37.5pt) node{$\vdots$};
\end{tikzpicture}}
\end{center}
\end{illustration}

\subsection{Generating Domino Matrices}
\label{subsec:first_alg}

Consider a chord diagram $D \in \mathcal{CD}_{n,k}$ as defined in~Section~\ref{subsec-chords}. It corresponds to a domino matrix $C \in \mathcal{DM}_{n,k}$ as defined in~Section~\ref{subsec-domino}. The first step towards showing that the BCFW cell that correspond to $D$ has the domino form of $C$, is an alternative description of the relation of $C$ to~$D$. In this section, we define an algorithm that gradually constructs a domino matrix from $D$ based on the matrix operations defined above. We then show that its output is indeed~$C$. 

\begin{algorithm}
\label{construct-matrix}
The following algorithm \textsc{construct-matrix} maintains a matrix $M \in \Mat_{K \times N}$ with $4k$ real variables $s_l,u_l,v_l,w_l$ for $l \in [k]$. We start with an empty row index set $K$, and $N = \{n\}$. The process iteratively updates $K$, $N$, and $M$ using the matrix operations. Eventually $K = [k]$, $N = [n]$, and $M \in \Mat_{k \times n}$ is the output. The index-set notation $i \pl 1$, $i \mi 1$ always relates to the current set~$N$. When indices are meant to be consecutive numbers, we use the regular notation $i+1$, $i-1$.

\begin{framed}
\noindent
\begin{minipage}{\textwidth}
\noindent \textsc{construct-matrix}(chord diagram $D \in \mathcal{CD}_{n,k}$,\, variables $s_l,u_l,v_l,w_l \in \mathbb{R}$ for $l \in [k]$)

\noindent \hspace{0.9cm} \textsc{initialize:} Let $N = \{n\}$, $K = \emptyset$, and an empty matrix $M \in \Mat_{K \times N}$

\noindent \hspace{0.9cm} \textbf{for} $m$ \textbf{in} $(n-1,n-2,\dots,1)$ \textbf{do}

\noindent \hspace{1.8cm} \textsc{fill}($m$): \textbf{if} $m \notin N$ \textbf{then} apply $\pre_m$ to $M$ and add $m$ to $N$

\noindent \hspace{1.8cm} \textbf{if} there exists in $D$ a chord $c_l = (m,m+1,j,j+1)$ \textbf{then} 

\noindent \hspace{2.7cm} \textsc{start}($c_l$): \textbf{if} $c_l$ is not a sticky child, and $(c_l,\dots,c_{l+h})$ is a maximal sticky chain 

\noindent \hspace{4.5cm} \textbf{then} apply $y_{m+h}(u_{l+h}) \circ \dots \circ y_{m+1}(u_{l+1}) \circ y_{m}(u_l)$ to $M$

\noindent \hspace{1.8cm} \textbf{for} every chord $c_l = (i,i+1,m,m+1)$ \textbf{in} $D$, 
in ascending order in $i$,
\textbf{do}

\noindent \hspace{2.7cm} \textsc{end}($c_l$): apply $y_{(i+1)\!\mi\! 1}(s_l) \circ x_m(w_l) \circ x_{i+1}(v_l) \circ \inc_{i+1;l}$ to $M$, 

\noindent\hspace{4.1cm} add $i{+}1$ to $N$, and add $l$ to $K$

\noindent \hspace{0.9cm} \textbf{return} $M$
\end{minipage}
\end{framed}
\end{algorithm}

\begin{rmk}
When applying $y_{(i+1)\!\mi\! 1}(s_l)$ in $\textsc{end}(c_l)$, the index $i$ may not be in the current index set~$N$, and hence $(i+1)\mi 1$ can be different than~$i$. In more detail, if $c_l$ is a top chord then $i+1$ is the smallest index in $N$ we apply $y_n(s_l)$. Otherwise, denoting by $c_p = (i',i'+1,j',j'+1)$ the parent of $c_l$, we apply $y_{i'{+}1}(s_l)$ as this index has been added in $\textsc{end}(c_p)$.
\end{rmk}

\begin{illustration}
We demonstrate the subroutines of the $\generate$ algorithm by their action on the plabic graph $D$ that corresponds to the current cell. This illustration is not needed for the rest of the paper.

\begin{center}
\begin{tabular}{ccc}
\;\;\;\;\;\; \textsc{start}(m) &
\textsc{fill}(m) & 
\textsc{end}(m) \;\;\;\; \\
\begin{tikzpicture}
\draw[gray, thick] (0,0) -- (0:25pt);
\filldraw[gray!30] (0,0) circle (25pt);
\draw (0:0) node{$D$};
\draw[thick] (0,0) circle (50pt);
\draw[thick] (0,0) circle (25pt);
\filldraw[black, thick] (270+60:47.5pt) -- (270+60:52.5pt) node[anchor=west]{$m$};
\draw (270+60:25pt) -- (270+60:50pt);
\filldraw[black, thick] (270+40:47.5pt) -- (270+40:52.5pt) node[anchor=west]{$m{+}1$};
\draw (270+40:25pt) -- (270+40:50pt);
\filldraw[black, thick] (270+20:47.5pt) -- (270+20:52.5pt) node[anchor=west]{$m{+}2$};
\draw (270+20:25pt) -- (270+20:50pt);
\filldraw[black, thick] (270-30:47.5pt) -- (270-30:52.5pt) node[anchor=east]{$m{+}h$};
\draw (270-30:25pt) -- (270-30:50pt);
\filldraw[black, thick] (270-50:47.5pt) -- (270-50:52.5pt) node[anchor=east]{$m{+}h{+}1$};
\draw (270-50:25pt) -- (270-50:50pt);
\draw (270+60:40pt) -- (270+40:30pt);
\draw (270+40:40pt) -- (270+20:30pt);
\draw (270+20:40pt) -- (270+10:35pt);
\draw (270-30:30pt) -- (270-20:35pt);
\draw (270-30:40pt) -- (270-50:30pt);
\filldraw[white] (270+40:30pt) circle (2pt);
\filldraw[white] (270+20:30pt) circle (2pt);
\filldraw[white] (270-30:30pt) circle (2pt);
\filldraw[white] (270-50:30pt) circle (2pt);
\filldraw[black] (270+60:40pt) circle (2pt);
\filldraw[black] (270+40:40pt) circle (2pt);
\filldraw[black] (270+20:40pt) circle (2pt);
\filldraw[black] (270-30:40pt) circle (2pt);
\draw (270+40:30pt) circle (2pt);
\draw (270+20:30pt) circle (2pt);
\draw (270-30:30pt) circle (2pt);
\draw (270-50:30pt) circle (2pt);
\draw (265:37.5pt) node{$\dots$};
\draw (360:37.5pt) node{$\vdots$};
\draw (180:37.5pt) node{$\vdots$};
\node[white] at (270:57.5pt) {123};
\end{tikzpicture} \;
&
\begin{tikzpicture}
\draw[gray, thick] (0,0) -- (0:25pt);
\draw[thick] (0,0) circle (50pt);
\filldraw[gray!30] (0,0) circle (25pt);
\draw[thick] (0,0) circle (25pt);
\filldraw[black, thick] (270+30:47.5pt) -- (270+30:52.5pt) node[anchor=north]{$m\mi1$};
\filldraw[black, thick] (270:47.5pt) -- (270:52.5pt) node[anchor=north]{$m$};
\filldraw[white, thick] (90:51.5pt) -- (90:62.5pt) node[anchor=north]{$haha$};
\filldraw[black, thick] (270-30:47.5pt) -- (270-30:52.5pt) node[anchor=north]{$m+1$};
\draw (270+30:25pt) -- (270+30:50pt);
\draw (270:32.5pt) -- (270:50pt);
\filldraw[black] (270:32.5pt) circle (2pt);
\draw (270-30:25pt) -- (270-30:50pt);
\draw (0:0) node{$D$};
\draw (330:37.5pt) node[rotate=90]{$\ddots$};
\draw (210:37.5pt) node{$\ddots$};
\end{tikzpicture} 
&
\begin{tikzpicture}
\draw[gray, thick] (0,0) -- (0:25pt);
\filldraw[gray!30] (0,0) circle (25pt);
\draw (0:0) node{$D$};
\draw[thick] (0,0) circle (50pt);
\draw[thick] (0,0) circle (25pt);
\filldraw[black, thick] (270+60:47.5pt) -- (270+60:52.5pt) node[anchor=west]{$(i{+}1)\mi 1$};
\draw (270+60:25pt) -- (270+60:50pt);
\filldraw[black, thick] (270+20:47.5pt) -- (270+20:52.5pt) node[anchor=west]{$i{+}1$};
\draw (270+20:31.25pt) -- (270+20:50pt);
\filldraw[black, thick] (270-20:47.5pt) -- (270-20:52.5pt) node[anchor=east]{$m$};
\draw (270-20:25pt) -- (270-20:50pt);
\filldraw[black, thick] (270-60:47.5pt) -- (270-60:52.5pt) node[anchor=east]{$m{+}1$};
\draw (270-60:25pt) -- (270-60:50pt);
\draw (270+60:43.75pt) -- (270+20:43.75pt);
\draw (270+20:37.5pt) -- (270-20:35pt);
\draw (270-20:43.75pt) -- (270-60:37.5pt);
\filldraw[black] (270+60:43.75pt) circle (2pt);
\filldraw[black] (270-20:35pt) circle (2pt);
\filldraw[black] (270-60:37.5pt) circle (2pt);
\filldraw[white] (270+20:37.5pt) circle (2pt);
\filldraw[white] (270+20:43.75pt) circle (2pt);
\filldraw[white] (270-20:43.75pt) circle (2pt);
\filldraw[white] (270+20:31.25pt) circle (2pt);
\draw[black] (270+20:37.5pt) circle (2pt);
\draw[black] (270+20:43.75pt) circle (2pt);
\draw[black] (270-20:43.75pt) circle (2pt);
\draw[black] (270+20:31.25pt) circle (2pt);
\draw (360:37.5pt) node{$\vdots$};
\draw (180:37.5pt) node{$\vdots$};
\node[white] at (270:57.5pt) {123};
\end{tikzpicture} 
\end{tabular}
\end{center}

\end{illustration}

\begin{ex}
\label{example3-algo}
We run \textsc{construct-matrix} on the chord diagram from Example~\ref{example3chords}, $$ D \;=\; ([14],((1,2,11,12),(3,4,6,7),(8,9,10,11))) $$ The pointer
$m$ goes from 13 to~1. In steps \textsc{fill}(13), \textsc{fill}(12), \textsc{fill}(11) it applies $\pre_{13}$, $\pre_{12}$, $\pre_{11}$. Then the \textsc{end}($c_1$) statement applies $y_{14}(s_1) \circ x_{11}(w_1) \circ x_2(v_1) \circ \inc_{2;1}$. The current matrix $M$ is
$$ \begin{array}{*{5}{c}}
\color{lightgray} 2 & \color{lightgray} 11 & \color{lightgray} 12 & \color{lightgray} 13 & \color{lightgray} 14 \\ \hline
1 & v_1 & w_1v_1 & 0 & s_1 \\ \hline \end{array} \vspace{0.5em} $$
Then at step $m=10$ the algorithm applies $\pre_{10}$ in \textsc{fill}(10) and \textsc{end}($c_3$): $y_{2}(s_3) \circ x_{10}(w_3) \circ x_9(v_3) \circ \inc_{9;3}$.
At step $m=9$ nothing happens. At $m=8$, $\pre_8$ is applied and then \textsc{start}($c_3$): $y_8(u_3)$. Now $M=$
$$ \begin{array}{*{8}{c}}
\color{lightgray} 2 & \color{lightgray} 8 & \color{lightgray} 9 & \color{lightgray} 10 & \color{lightgray} 11 & \color{lightgray} 12 & \color{lightgray} 13 & \color{lightgray} 14 \\ \hline
1 & 0 & 0 & 0 & -v_1 & -w_1v_1 & 0 & -s_1 \\
s_3 & u_3 & 1 & v_3 & w_3v_3 & 0 & 0 & 0 \\\hline
\end{array} \vspace{0.5em} $$
In steps $m=7,6,5$, the algorithm applies \textsc{fill}(7): $\pre_7$ and \textsc{fill}(6): $\pre_6$ followed by \textsc{end}($c_2$): $y_2(s_2) \circ x_6(w_2) \circ x_4(v_2) \circ \inc_{4;2}$, then \textsc{fill}(5): $\pre_5$. Nothing is done at $m=4$. In $m=3$, it applies \textsc{start}($c_2$): $y_3(u_2) \circ \pre_3$, and again nothing at $m=2$. Finally at $m=1$, it applies $\pre_1$ and \textsc{start}($c_1$): $y_1(u_1)$. The resulting matrix $M$ is
$$ \begin{array}{*{14}{c}}
\color{lightgray} 1 & \color{lightgray} 2 & \color{lightgray} 3 & \color{lightgray} 4 & \color{lightgray} 5 & \color{lightgray} 6 & \color{lightgray} 7 & \color{lightgray} 8 & \color{lightgray} 9 & \color{lightgray} 10 & \color{lightgray} 11 & \color{lightgray} 12 & \color{lightgray} 13 & \color{lightgray} 14 \\ \hline
u_1 & 1 & 0 & 0 & 0 & 0 & 0 & 0 & 0 & 0 & v_1 & w_1v_1 & 0 & s_1 \\
u_1s_2 & s_2 & u_2 & 1 & 0 & v_2 & w_2v_2 & 0 & 0 & 0 & 0 & 0 & 0 & 0
\\
-u_1s_3 & -s_3 & 0 & 0 & 0 & 0 & 0 & u_3 & 1 & v_3 & w_3v_3 & 0 & 0 & 0 \\\hline
\end{array} \vspace{0.5em} $$
This matrix has the same domino form as the one in Example~\ref{example3domino} up to change of variables, and for positive $s_l,u_l,v_l,w_l$ it satisfies the sign rules.
\end{ex}

\begin{ex}
\label{fuad}
As another demonstration of \textsc{construct-matrix}, we run the algorithm on a more condensed version of~$D$, that features the sticky chain, same-end, and head-to-tail situations. $$D'\;=\;([8],((1,2,6,7),(2,3,4,5),(4,5,6,7)))$$
After \textsc{fill}(7), \textsc{fill}(6), and \textsc{end}($c_1$):
$$ \begin{array}{*{4}{c}}
\color{lightgray} 2 & \color{lightgray} 6 & \color{lightgray} 7 & \color{lightgray} 8 \\ \hline
1 & v_1 & w_1v_1 & s_1 \\ \hline
\end{array} \vspace{0.5em} $$
Then \textsc{end}($c_3$) affects both rows in the application of $x_6(w_3)$, and after \textsc{fill}(4) and \textsc{start}($c_3$):
$$ \begin{array}{*{6}{c}}
\color{lightgray} 2 & \color{lightgray} 4 & \color{lightgray} 5 & \color{lightgray} 6 & \color{lightgray} 7 & \color{lightgray} 8 \\ \hline
1 & 0 & 0 & -v_1 & -w_1v_1-w_3v_1 & -s_1 \\
s_3 & u_3 & 1 & v_3 & w_3v_3 & 0 \\ \hline
\end{array} \vspace{0.5em} $$
Then \textsc{end}($c_2$) affects also the row of $c_3$ in the application of $x_4(w_2)$:
$$ \begin{array}{*{7}{c}}
\color{lightgray} 2 & \color{lightgray} 3 & \color{lightgray} 4 & \color{lightgray} 5 & \color{lightgray} 6 & \color{lightgray} 7 & \color{lightgray} 8 \\ \hline
1 & 0 & 0 & 0 & v_1 & w_1v_1+w_3v_1 & s_1 \\
s_2 & 1 & v_2 & w_2v_2 & 0 & 0 & 0 \\
-s_3 & 0 & u_3 & 1+w_2u_3 & v_3 & w_3v_3 & 0 \\ \hline
\end{array} \vspace{0.5em} $$
Nothing happens at \textsc{start}($c_2$) because it sticks to~$c_1$. After \textsc{fill}(1) and we apply \textsc{start}($c_1$) with the sticky chain $(c_1,c_2)$, which yields $y_2(u_2) \circ y_1(u_1)$:
$$ \begin{array}{*{8}{c}}
\color{lightgray} 1 & \color{lightgray} 2 & \color{lightgray} 3 & \color{lightgray} 4 & \color{lightgray} 5 & \color{lightgray} 6 & \color{lightgray} 7 & \color{lightgray} 8 \\ \hline
u_1 & 1 & 0 & 0 & 0 & v_1 & w_1v_1+w_3v_1 & s_1 \\
u_1s_2 & s_2+u_2 & 1 & v_2 & w_2v_2 & 0 & 0 & 0 \\
-u_1s_3 & -s_3 & 0 & u_3 & 1+w_2u_3 & v_3 & w_3v_3 & 0 \\ \hline
\end{array} \vspace{0.5em} $$
This result satisfies the domino form $D'$ with the sign rules as defined in~Section~\ref{subsec-domino}.
\end{ex}

\begin{rmk}
\label{rows-change-rows}
As demonstrated in Examples~\ref{example3-algo} and~\ref{fuad}, the step $\textsc{end}(c_l)$ followed by $\textsc{start}(c_l)$ always create the row indexed by $l$ and fill it with up to five nonzero entries. However, as demonstrated as well in these examples, sometimes a step associated with one row cause nontrivial changes to another one. It is useful to record the main ways it can happen:
\begin{itemize}
\item 
If $c_p$ and $c_l$ are same-end chords, $p<l$, then at $\textsc{end}(c_l)$ the second end entry in row $p$ is modified.
\item 
If $c_l$ and $c_j$ are head-to-tail with $l<j$, then at $\textsc{end}(c_l)$ the second start entry in row $j$ is modified.
\item 
If $c_i$ is a child of $c_l$, then $\textsc{start}(c_l)$ creates a nonzero entry in row $i$ at the first start entry of~$c_l$.
\end{itemize}
In particular, note that if $c_i$ is a sticky child of $c_l$, then $c_l$ is responsible for creating the fifth nonzero entry in row $i$, because $\textsc{start}(c_i)$ only modifies an existing entry.
\end{rmk}

\begin{rmk}
\label{same-cell}
For a chord diagram $D \in \mathcal{CD}_{n,k}$
the name of the Algorithm~\ref{construct-matrix} denotes its output matrix,
$$ M \;=\; \textsc{construct-matrix}\left(D,\{s_l,u_l,v_l,w_l\}_{l=1}^k\right) \;\in\; \Mat_{k \times n} $$
where, unless stated otherwise, the $4k$ real variables are assigned positive values. Given~$D$, the row-span of the resulting~$M$ lies in a single positroid cell $S$, regardless of this assignment in $(0,\infty)^{4k}$, by Proposition~\ref{positroid maps}, and we later show that this is a homeomorphism to~$S$. By common abuse of notation, here and throughout this paper, we regard a matrix as the point in the Grassmannian that it represents, and thus we can write: $M \in \Gr_{k,n}$. For a point in the nonnegative Grassmannian, it is assumed that $M$ has nonnegative Pl\"ucker coordinates, rather than nonpositive. We occasionally omit the $4k$ variables and write \textsc{construct-matrix}(D). See~\cite{evenzohar2022bcfw} for a Sage implementation of this algorithm.
\end{rmk}

Chord diagrams retain the key property that no two chords cross. Hence, once the algorithm starts handling the end of a chord, it processes all its descendants before handling its start and moving on. This means that the construction actually admits a recursive nature. This viewpoint is useful for the analysis. We therefore describe a recursive formulation of the algorithm, equivalent to the iterative one above.

\begin{algorithm}
\label{recursive-algorithm}
The following subroutine \textsc{sub-construct-matrix} is defined using the statements \textsc{end}, \textsc{start}, and \textsc{fill} from Algorithm~\ref{construct-matrix}, which are not restated here. The input is a list of sibling chords $(c_1',\dots,c_g')$ between two given points \emph{parent-start} and \emph{parent-end}. Always $c'_g$ is the first sibling and $c'_1$ the last one, so they are handled in right-to-left order. We use the notation $c_h'$ because our indexing is not the standard one $c_1,\dots,c_k$ of a chord diagram, and their order is decreasing. 

\begin{framed}
\noindent
\begin{minipage}{\textwidth}
\noindent
\textsc{sub-construct-matrix}(chords $(c'_1,\dots,c'_g)$ in $D$,\, \emph{parent-start}, \emph{parent-end})

\noindent
\hspace{0.9cm} 
\textbf{for} every $c'_h = (i_h,i_h+1,j_h,j_h+1)$ \textbf{in} $(c'_1,\dots,c'_g)$ \textbf{do}

\noindent
\hspace{1.8cm} 
\textbf{for} $m$ \textbf{in} $(i_{h-1}-1,i_{h-1}-2,\dots,j_h)$  \textbf{do} \textsc{fill}($m$)
\hfill
//~ where $i_0=\textit{parent-end}$

\noindent
\hspace{1.8cm} 
\textsc{end}($c'_h$)

\noindent
\hspace{1.8cm} 
\textsc{sub-construct-matrix}(\emph{children}($c'_h$), $i_h$, $j_h$)
\hfill
//~ ordered last to first

\noindent
\hspace{1.8cm} 
\textsc{start}($c'_h$)

\noindent
\hspace{0.9cm} 
\textbf{for} $m$ \textbf{in} $(i_{g}-1,i_{g}-2,\dots,\textit{parent-start})$  \textbf{do} \textsc{fill}($m$) 
\end{minipage}
\end{framed}
\end{algorithm}

\medskip
With this definition, running \textsc{sub-construct-matrix}(\emph{top}(D), 1, $n$) is equivalent to the above algorithm \textsc{construct-matrix}(D), as the two algorithms can be seen to apply the same matrix operation in the same order. The input \emph{top}(D) is the sequence of top chords, ordered last to first.

\medskip
We now prove that the algorithm generates a domino matrix. Specifically, it generates the domino matrix that corresponds to the given chord diagram according to Definition~\ref{def:domino_entries}. Moreover, this domino matrix satisfies the sign rules of Definition~\ref{def:domino_signs}.

\begin{prop}\label{prop:domino}
Let $D \in \mathcal{CD}_{n,k}$ be a chord diagram, and $C \in \mathcal{DM}_{n,k}$ the corresponding domino matrix. Then
$\textsc{construct-matrix}(D)$ has the matrix form of~$C$ and satisfies the domino sign rules.
\end{prop}

\begin{proof}
Denote $M=\textsc{construct-matrix}(D)$. We have to show that for every assignment in $M$ of the variables $\{s_l,u_l,v_l,w_l\}_{l=1}^k \in (0,\infty)^{4k}$, there exists an assignment in~$C$ of $\{\alpha_l,\beta_l,\gamma_l,\delta_l,\varepsilon_l\}_{l=1}^k \in \mathbb{R}^{5k}$, that satisfies~$M=C$ and the sign rules. The proof goes by tracking the course of algorithm and verifying all the details required by Definitions~\ref{def:domino_entries} and~\ref{def:domino_signs}. We recover the $5k$ domino variables and their signs row by row, where we separately consider rows that correspond to top chords, nonsticky children, and sticky children.

\textbf{Top:}
Let $c_l = (i,i+1,j,j+1)$ be a top chord. The row $M_l$ corresponding to~$c_l$ is created at $\textsc{end}(c_l)$ when $m=j$, with nonzero entries at four positions: $1$ at position $i+1$, $v_l$ at position $(i+1)\pl1=j$, $w_lv_l$ at position $m\pl 1=j+1$, and $\pm s_l$ at position $(i+1)\mi1=n$. At later stages, the three entries at positions $j,j+1,n$ of row~$l$ occasionally flip their signs, whenever an $\inc_{i_h+1;h}$ operation is invoked for some descendant $c_h$ of~$c_l$ starting at $(i_h,i_h+1)$. In any case, the two end entries of the chord at $(j,j+1)$ keep having equal signs. This is important, because every same-end descendant $c_h$ of~$c_l$ invokes an application of $x_j(w_h)$ with $w_h>0$, and the equal signs guarantee that the $j+1$ entry does not vanish or flip sign, as it gains a positive multiple of the entry at~$j$. At $\textsc{start}(c_l)$ when $m=i$, a positive entry $u_l$ is created via $y_i(u_l)$ at position~$i$ of row~$l$, next to the existing $1$ at position~$i+1$. Later, only $x_i(w_h)$ applications might affect the row $M_l$, in the case that some chords are head-to-tail with~$c_l$. These changes add $w_hu_l$ to the entry at $i+1$. The positivity of the two entries at positions~$i,i+1$ is preserved. All the remaining steps for $m<i$ do not affect the row~$M_l$, and its support remains $\{i,i+1,j,j+1,n\}$. These are exactly the positions of $\{\alpha_l,\beta_l,\gamma_l,\delta_l,\varepsilon_l\}$ in the row~$C_l$ corresponding to a top chord. To summarize, the algorithm produces the following expressions for the domino variables:
\begin{align*}
\alpha_l = u_l \;\;\;\;\;\;\;\; \beta_l = 1 + u_l\sum_{j_h=i}w_h \;\;\;\;\;\;\;\;\gamma_l = \pm v_l \;\;\;\;\;\;\;\;\delta_l = \pm v_l\sum_{\substack{j_h=j \\ h \geq l}} w_h \;\;\;\;\;\;\;\;\varepsilon_l = \pm s_l &&    
\label{abcde} \hfill (\ast)
\end{align*}
Note that the sum in $\beta_l$ is over all chords $c_h$ that are head-to-tail with $c_l$, and the sum in $\delta_l$ is over the chord $c_l$ itself together with all chords $c_h$ that are same-end with $c_l$. The variables $\alpha_l$ and $\beta_l$ are positive, as required by the sign rules of Definition~\ref{def:domino_signs}(1). Indeed, their sign never flips as all $\inc_{i_h+1;h}$ operations occur either to their lower right or to their upper left. The entries $\gamma_l$ and $\delta_l$ are created positive by two $x$ operations, and then they flip sign due to $\inc$ operations for every chord below $c_l$, hence their final sign is $(-1)^{\mathrm{below}(c_l)}$ as required by Definition~\ref{def:domino_signs}(2). The last entry $\varepsilon_l$ is created by a wrap-around $y_n^{k'}(s_l)$ operation, where the $k' = k-l-\mathrm{below}(c_l)+1$ existing rows in the matrix correspond to $c_l$ and the chords after it. The parity of their number determines its initial sign $(-1)^{k'-1}$ by Definition~\ref{def:x_i y_i}. Then, the sign of this entry flips $\mathrm{below}(c_l)$ times for each chord below $c_l$, similar to $\gamma_l$ and $\delta_l$. Hence, the final sign is $(-1)^{k-l} = (-1)^{\mathrm{behind}(c_l)}$ as claimed, satisfying Definition~\ref{def:domino_signs}(3).

\textbf{Nonsticky:} Now, let $c_l = (i,i+1,j,j+1)$ be a child of some other chord $c_p = (i',i'+1,j',j'+1)$, and suppose that $c_l$ is a sticky child, and in particular $i-i' \geq 2$. The row $M_l$ is created at $\textsc{end}(c_l)$ when $m=j$ with four positive entries $s_l,1,v_l,v_lw_l$ at positions $i'+1,i+1,j,j+1$ respectively. As in the top chord case, same-end descendants $c_h$ of~$c_l$ contribute at $x_j(w_h)$ a same-sign $\pm v_lw_h$ to the entry at~$j+1$. Then, an application of $y_i(u_l)$ at $\textsc{start}(c_l)$ puts a positive multiple of entry~$i+1$ at~$i$, namely~$u_l$. Then, as in the case of a top chord again, $x_i(w_h)$ for head-to-tail chords $c_h$ contribute $u_lw_h$ to the entry at $i+1$. Also, every $\inc_{i_h+1;h}$ operation for a chord $c_h$ ending between $i'+1$ and $i$ flips the sign of the entry $\pm s_l$ at $i'+1$. Suppose for a moment that the parent chord $c_p$ is not a sticky child. Then the application of $y_{i'}(u_p)$ during $\textsc{start}(c_{p})$ puts a same-sign multiple of entry~$i'+1$ at position~$i'$, namely $\pm u_ps_l$. The only remaining operations that affect row $l$ are $x_{i'}(w_h)$ during $\textsc{end}(c_h)$ for chords $c_h$ that are head-to-tail with $c_p$. These operations add a same-sign multiple of the entry at~$i'$, namely $\pm s_l u_p w_h$, to the entry at~$i'+1$. Also in the case that $c_p$ is a sticky child $y_{i'}(u_p)$ occurs, now in $\textsc{start}(c_{\ast p})$ rather than $c_p$, where $c_{\ast p}$ is the ancestor at the top of the sticky chain ending in~$c_p$.  Due to the order of $y$ operation in this $\textsc{start}$ statement, row~$l$ cannot be modified by any other operations. In conclusion, the support of $M_l$ is exactly $\{i',i'+1,i,i+1,j,j+1\}$, same as for the row $C_l$ in the corresponding domino matrix. The six nonzero entries at these positions are respectively given by the following six expressions:$$ 
\pm s_l u_p
\;\;\;\;\;\;\;\;
\pm s_l \left(1 + u_p \sum_{j_h = i'} w_h \right)
\;\;\;\;\;\;\;\;
u_l \;\;\;\;\;\;\;\; 
1 + u_l\sum_{j_h=i}w_h \;\;\;\;\;\;\;\;
\pm v_l \;\;\;\;\;\;\;\;
\pm v_l\sum_{\substack{j_h=j \\ h \geq l}} w_h $$
where the last sum is over $c_l$ and its same-end descendants.
According to the domino form, these expressions should equal $(\varepsilon_l \alpha_p, \varepsilon_l \beta_p, \alpha_l, \beta_l, \gamma_l, \delta_l)$ respectively. This is achieved by setting $\varepsilon_l = \pm s_l$ and noting that the resulting $\alpha_p$ and $\beta_p$ are consistent with their definitions for the parent row~$p$. The signs of the variables $\alpha_l, \beta_l, \gamma_l, \delta_l$ satisfy the domino sign rules, exactly as in the case of a top chord. The sign of $\varepsilon_l$ is the parity of the number of $\inc$ operations to the upper right of the entry at position $i'+1$ of row~$l$, which is the number of chords ending between the starts of $c_l$ and its parent~$c_p$, which yields $(-1)^{\mathrm{beyond}(c_l)}$ as required by Definition~\ref{def:domino_signs}(4).

\textbf{Sticky:}
Finally, let $c_l = (i,i+1,j,j+1)$ be a sticky child of some other chord $c_{p} = (i-1,i,j',j'+1)$ where $p=l-1$. The effect of $\textsc{end}(c_{l})$ and the descendants of $c_l$ is similar to the nonsticky case. When we reach the start of $c_l$, the nonzero terms are at positions $i,i+1,j,j+1$ of row~$l$, since $(i+1)\mi 1=i$. The entry at position $i+1$ is $1$ as before, and the entry at position $i$ is $s_l$ with a positive sign in this case, because there is no other chord between the two starts, so no sign flip could occur. Since $c_l$ is a sticky child, the next operations occur in $\textsc{start}(c_{\ast l})$ at $c_l$'s sticky ancestor. There, the application of $y_{i-1}(u_p)$ creates a new nonzero entry $u_ps_l$ at position $i-1$, and the subsequent application of $y_{i}(u_l)$ adds $u_l$ to the existing $s_l$ at position~$i$. Other $y$ operations in the sticky chain, before or after these two, do not affect row~$l$ since they can be seen to add a zero to some entry. In the case that the row is not further modified by the algorithm, the three entries at positions $i-1,i,i+1$ should be, respectively, according to the domino form:
$$ \varepsilon_l \alpha_p = s_lu_p \;\;\;\;\;\;\;\;
\varepsilon_l \beta_p + \alpha_l = s_l + u_l
\;\;\;\;\;\;\;\;
\beta_l = 1 $$
This follows readily from setting $\varepsilon_l = s_l$, $\alpha_l = u_l$, $\beta_l = 1$, $\alpha_p = u_p$, and $\beta_p = 1$, which is consistent which our choices for $c_p$ in cases handled above. The only remaining case where row $l$ is further modified is when $c_p$ is at the top of the sticky chain, and it is head-to-tail with other chords $c_h$ ending at $(i-1,i)$. The $x_{i-1}(w_h)$ operations at these ends add $s_lu_pw_h$ terms to the middle entry at~$i$. These terms also result from the above definition of $\beta_p$, which contains $u_pw_h$ terms in this head-to-tail case.

So far, we have shown that there exist $\{\alpha_l,\beta_l,\gamma_l,\delta_l,\varepsilon_l\}_{l=1}^k \in \mathbb{R}^{5k}$ such that the output of the algorithm $M = \generate(D,\{s_l,u_l,v_l,w_l\}_{l=1}^k)$ is equal to the domino matrix~$C$, and these $5k$ real variables satisfy the sign rules in Definition~\ref{def:domino_signs}(1)-(4). It remains to verify that they also satisfy the two remaining $2 \times 2$ sign rules. For verifying Definition~\ref{def:domino_signs}(5), let $c_l$ be a same-end child of $c_p$. We have seen that $\delta_l = \pm v_l \sum_h w_h$ summing over $c_l$ and its descendants, and $\gamma_l = \pm v_l$ with the same sign, and similarly for~$c_p$. Hence, 
$$ \delta_l/\gamma_l \;=\; w_l + \sum_hw_h \;<\;  w_p + w_l + \sum_hw_h \;=\; \delta_p/\gamma_p $$
where the possibly-empty sum is over same-end descendants $c_h$ of~$c_l$. Similarly, we verify Definition~\ref{def:domino_signs}(6), for $c_l$ and $c_r$ head-to-tail. Recall $\alpha_r = u_r$ and $\beta_r = 1 + u_r\sum_hw_h$ summing over all $c_h$ ending at the start of~$c_r$. Hence,
$$ \beta_r/\alpha_r \;=\; 1/u_r + w_l + \sum_h w_h \;>\; w_l + \sum_h w_h \;=\; \delta_l/\gamma_l $$
as required. In conclusion, the output of $\generate$ has the domino form and satisfies the sign rules.
\end{proof}

The above proof gives formulas for the $5k$ domino variables in terms of the $4k$ algorithm parameters.
One may invert these formulas, and recover the parameters from the domino variables as follows.

\begin{lem}
\label{inverse-domino}
Let $D \in \mathcal{CD}_{n,k}$ be a chord diagram, and $C \in \mathcal{DM}_{n,k}$ the corresponding domino matrix.
Let $\{\alpha_l,\beta_l,\gamma_l,\delta_l,\varepsilon_l\}_{1 \leq l \leq k}$ be a real assignment to~$C$ that satisfies the domino sign rules. Define the parameters $\{s_l,u_l,v_l,w_l\}_{1 \leq l \leq k}$ recursively as follows. 
$$ u_l = \left( \frac{\beta_l}{\alpha_l} - \sum_{j_h=i_l} w_h \right)^{-1} \;\;\;\;\;\;\;\;\;\;
v_l = \left|\frac{\gamma_l}{\alpha_l}\right|u_l
\;\;\;\;\;\;\;\;\;\;
w_l = \left( \frac{\delta_l}{\gamma_l} - \sum_{\substack{j_h=j \\ h > l}} w_h \right)
\;\;\;\;\;\;\;\;\;\;
s_l = \left|
\frac{\varepsilon_l}{\alpha_l}\right|u_l
$$
Here, we always handle next the leftmost lowest chord so-far unhandled. The first sum is over chords $c_h$ head-to-tail with $c_l$, and the second sum is over same-end descendants $c_h$ of~$c_l$, excluding $c_l$ itself. 

Then, these $4k$ parameters are positive, and their assignment to $\generate(D)$ yields the domino matrix~$C$ up to rescaling each set $(\alpha_l,\beta_l,\gamma_l,\delta_l,\varepsilon_l)$ by a positive factor for every~$l$.
\end{lem}

\begin{proof}
First, by induction, for every $c_l$, the ratio $\delta_l/\gamma_l$ equals the sum $\sum_h w_h$ over $c_l$ and its same-end descendants. Hence $u_l$ is positive by the sign rule $\beta_l/\alpha_l > \delta_h/\gamma_h$, and $s_l$ is positive by the sign rule $\delta_l/\gamma_l > \delta_h/\gamma_h$. Clearly, $v_l$ and $s_l$ are positive as well. It is now straightforward that the formulas $(\ast)$ in the proof of Proposition~\ref{prop:domino} recover the variables $\alpha_l,\beta_l,\gamma_l,\delta_l,\varepsilon_l$ and their signs, up to a common positive factor.
\end{proof}

\subsection{The Generated Cell}
\label{analysis}

The algorithm in Section~\ref{subsec:first_alg} takes a chord diagram $D \in \mathcal{CD}_{n,k}$ and generates a matrix $C \in \mathcal{DM}_{n,k}$, shown to have the domino form corresponding to~$D$ as defined in~Section~\ref{subsec-domino}. By our definitions in~Section~\ref{subsec-perms}, the chord diagram $D$ also directly corresponds to a decorated permutation $\pi \in \mathcal{DP}_{n,k}$. This permutation~$\pi$ corresponds in turn to a positroid cell $S_{D} \in \mathcal{BCFW}_{n,k}$ in a unique and standard manner, see Definitions~\ref{perm2cell} and~\ref{maindefbcfw} and the references in~Section~\ref{subsec-equivalence}. Besides, the generated matrix~$C$, for any assignment of $4k$ positive numbers, lies in a particular positroid cell~$S$, as observed in Remark~\ref{same-cell}. If $\sigma$ is the decorated permutation of~$S$, then our goal is to show that $\sigma = \pi$, and therefore $S=S_D$. To that end, we first turn to analyze Algorithm~\ref{construct-matrix} and describe its output in terms of decorated permutations.

\begin{definition}
\label{sigma_and_the_algorithm}
Let $D \in \mathcal{CD}_{n,k}$ be a chord diagram. Consider the following construction running in parallel to the algorithm $\textsc{construct-matrix}$ on $D$:
\begin{samepage}
\begin{enumerate}
\item 
Start with the one-element permutation $\sigma : \{n\} \to \{n\}$.
\item 
Whenever $\pre_m$ is applied, augment~$\sigma$ by the fixed point, $\sigma(m)=m$.  
\item 
Whenever $\inc_{i+1;l}$ is applied, augment~$\sigma$ by the white fixed point, $\sigma(i+1)=\overline{i+1}$.  
\item 
Whenever $x_j(t)$ is applied, update $\sigma$ to be $\sigma \cdot (j~~j \pl 1)$.
\item 
Whenever $y_j(t)$ is applied, update $\sigma$ to be $(j~~j \pl 1) \cdot \sigma$.
\end{enumerate}
The resulting decorated permutation $\sigma \in S_n$ is called the \emph{algorithmic permutation} of~$D$.
\end{samepage}
\end{definition}

\begin{rmk}
At any point of this process, we have a permutation $\sigma:N \to N$, for the current index set~$N$ maintained by the algorithm. Thus, the adjacent transpositions $(j~~j \pl 1)$ are with respect to the current state of the index set $N$ at the moment of application. The apparent distinction between black and white fixed points will never be important, due to transpositions applied later at these points.
\end{rmk}

\begin{rmk}\label{rmk:upper_as_plabic_and_perm}
One can also describe how the composed operation $\emb_{i,l,r}$ acts on permutations. It corresponds to adding a white fixed point at $i$, and then multiplying by the cycle $(i~(i{\pl 1})~\cdots~(i{\pl r}))$ from the right and by $(i~(i{\mi1})~\cdots~(i{\mi l}))$ from the left. For example, the upper embedding $\uemb_1=\emb_{1,3,1}$ 
corresponds to $(1~(n-2)~(n-1)~n)$ from the left and $(1~2)$ from the right.
\end{rmk}

We also give the following, equivalent but more explicit, definition for the algorithmic permutation of a given chord diagram. For stating this definition, recall that the chords in a diagram are always noncrossing. Thus, for head-to-tail chords, the end of the first is said to come before the start of the second one, even though they are both in the same segment. Similarly, if we order the chords increasingly by their ends, then same-end chords are ordered with decreasing starts. 

\begin{definition}
\label{sigma}
\label{algorithmic}
\label{def:nicer_form_pi_alpha}
Let $D \in \mathcal{CD}_{n,k}$ and $C = \textsc{construct-matrix}(D, \{s_l,u_l,v_l,w_l\})$. Then the \emph{algorithmic permutation} corresponding to $D$ and $C$ is the following composition of 2-cycles and 3-cycles: 
$$ \sigma \;=\; (p_1~q_1)~(p_2~q_2) \cdots (p_{2k}~q_{2k}) \,\cdot\, (i_k+1~~j_k~~j_k+1) \cdots (i_1+1~~j_1~~j_1+1) \;\in\; S_n$$
where
\begin{itemize}
\itemsep0.125em
\item
The terms $(i_l+1~~j_l~~j_l+1)$ correspond to the $k$ chords $c'_l = (i_l,i_l+1,j_l,j_l+1)$, where $c'_1,\dots,c'_k$ are indexed increasingly by the occurrence of their ends along~$D$.
\item 
The terms $(p_m~q_m)$ correspond to the sequence of \emph{endpoints}, both starts and ends, in order of occurrence along the diagram~$D$, and determined as follows:
\begin{itemize}
\item If the $m$th endpoint is a start of some chord $c_l = (i_l,i_l+1,j_l,j_l+1)$ then $(p_m~q_m)=(i_l~~i_{l\ast}+1)$, using the notation $c_{l\ast}$ for a maximal sticky descendant of $c_l$, and as usual possibly $c_l = c_{l\ast}$.
\item 
If the $m$th endpoint is a end of some chord $c_l$ then this term swaps the following elements:
\begin{samepage}
\begin{itemize}
\itemsep0.25em
\item 
$p_m = i_{l}+1$, i.e., the second marker at the \emph{start} of~$c_l$. \vspace{0.25em}
\item
$q_m = 
\begin{cases}
n & \text{if }c_l\text{ is a top chord,} \\
i_{h}+1 & \text{if }c_l\text{ is a child of another chord }c_h. 
\end{cases}$
\end{itemize}
\end{samepage}
\end{itemize}
\end{itemize}
\end{definition}

\begin{ex}
\label{ex:perm_of_alg_for_3_chords}
For $D = \left([14],\,\left((1,2,11,12),(3,4,6,7),(8,9,10,11)\right)\right)$ of Examples~\ref{example3chords} and~\ref{example3-algo}: 
$$
\sigma \;=\;
(1~2)~(3~4)~(2~4)~(8~9)~(2~9)~(2~14)~(2~11~12)~(9~10~11)~(4~6~7)
$$
This permutation coincides with Example~\ref{example3perm}, as we later prove in general.
\end{ex}

\begin{ex}
\label{ex:perm_of_alg_for_3_chords_B}
For $D = \left([8],\,\left((1,2,6,7),(2,3,4,5),(4,5,6,7)\right)\right)$ of Example~\ref{fuad}: 
$$ \sigma \;=\; (1~3)~(2~3)~(2~3)~(4~5)~(2~5)~(2~8)~(2~6~7)~(5~6~7)~(3~4~5)
$$
Note that the leftmost transposition is $(1~3)$ rather than $(1~2)$, because the chord has a sticky child.
\end{ex}

\begin{rmk}
\label{rmk:pi_alpha_according_to_top_levels}
If there are several top chords $c_{(1)},\dots,c_{(g)}$ in~$D$, indexed left to right, then the product in Definition~\ref{algorithmic} admits a useful factorization, grouping together the descendants of each top chord:
$$ \sigma \;=\; \tau_1 \cdot \tau_2 \cdots \tau_g \,\cdot\, \rho_g \cdot \rho_{g-1} \cdots \rho_1 $$
where $\tau_h$ is the product over transpositions $(p_m~q_m)$ that arise from starts and ends of $c_{(h)}$ and its descendants, and $\rho_h$ is the product of 3-cycles that arise from their ends.
\end{rmk}

\begin{lem}
\label{newlemma}
Definitions~\ref{sigma_and_the_algorithm} and~\ref{algorithmic} are equivalent.
\end{lem}

\begin{proof}
The computation is straightforward. The contribution of the $x_j(t)$ applications happens at the end of every chord, and thus according to the reverse order of these ends, which is $c'_k,\dots,c'_1$ as in~Definition~\ref{sigma}. At each end of a chord $c_l' = (i,i+1,j,j+1)$, the algorithm applies $x_{i+1}(v_l)$ and then~$x_j(w_l)$. These operations right-multiply $\sigma$ by $(i+1~~j)(j~~j+1) = (i+1~~j~~j+1)$, noting that these three indices are adjacent in~$N$ at that moment. Repeating for all chord ends, this yields the 3-cycles in the right half of the algorithmic permutation~$\sigma$ of Definition~\ref{algorithmic}.

For the 2-cycles, first assume that there are no sticky chords in~$D$. In this case, the applications of $y_j(t)$ happen at every end and start of a chord according to their reverse order of occurrence along the diagram~$D$. This is indeed the order used in the definition of the terms $(p_m~q_m)$ in Definition~\ref{sigma}. At every start of a chord $c_l = (i,i+1,j,j+1)$, an application of $y_i(u_l)$ contributes the left-factor $(i~~i+1)$. At every end the contribution of $y_{(i+1) \mi 1}(s_l)$ gives the left-factor $(i+1~~n)$ if $c_l$ is a top chord, and $(i+1~~i'+1)$ otherwise where $c_{h} = (i',i'+1,j',j'+1)$ is $c_l$'s parent. This gives the product of 2-cycles in the left half of the algorithmic permutation in Definition~\ref{algorithmic}.

What if there are sticky chains in~$D$? In this case the $y_i(u_l)$ applications at starts are applied in the opposite order of their occurrence when going back on the diagram. Suppose that $(c_{l},\dots,c_{l+h})$ is a maximal descending sticky chain in~$D$. Then we apply the operations $y_{i+h}(u_{l+h}) \circ \dots \circ y_{i+1}(u_{l+1}) \circ y_{i}(u_l)$. This yields the product of transpositions
$$ (i+h~~i+h+1)~\cdots~(i+1~~i+2)~(i~~i+1) \;=\; (i~~i+h+1)~(i+1~~i+h+1)~\cdots~(i+h~~i+h+1)
$$
where the verification of this equality is routine. Now the 2-cycles are listed in order of occurrence, though the contribution of every start is $(i~~i'+1)$ where $c_{l\ast} = c_{l+h}$ is its maximal sticky descendant. In conclusion, the construction described in the Definition~\ref{sigma_and_the_algorithm} yields the algorithmic permutation~$\sigma$ of~$D$ as in Definition~\ref{algorithmic}.
\end{proof}

Note that Lemma~\ref{newlemma} does not yet claim that the matrices generated by \textsc{construct-matrix} lie in the positroid cell corresponding to the decorated permutation~$\sigma$. As a first step, we show that they lie in some $4k$-dimensional cell.

\begin{lemma}
\label{lem:reduced_rep}
Let $D \in \mathcal{CD}_{n,k}$ be a chord diagram, and let $S$ be the positroid cell that arises from $\textsc{construct-matrix}\left(D,\, \{s_l,u_l,v_l,w_l\}_{l=1}^k \right)$. Every point in $S$ is obtained from a unique choice of the $4k$ positive variables. Hence, $S$ is $4k$-dimensional.
\end{lemma}

In the proof of lemma~\ref{lem:reduced_rep}, we need the following useful lemma, which serves us for other purposes later in the paper. It is stated here and proven below, at the end of this Section~\ref{analysis}. We note that $\binom{[n]}{k}$ is the set of $k$-element subsets of~$[n]$.

\begin{lemma}
\label{lem:more than 4}
Let $\D \in \CD_{n,k}$ be a chord diagram with a top chord $c = (i,i+1,j,j+1)$, and let $S \subseteq \Grnn{k}{n}$ be all the points representable by a domino matrix of the form corresponding to~$D$. Then,
\begin{itemize}
\itemsep0.125em
\item 
For every $V \in S$, the Pl\"ucker coordinate $P_I(V) \neq 0$ only if $I$ intersects $A=\{i,i+1,j,j+1,n\}$.
\item
For every $p\in A$, there exists $I \in \tbinom{[n]}{k}$ such that $I\cap A=\{p\}$ and $P_I(V) \neq 0$ for every $V \in S$.
\item
In particular, no $V \in S$ contains a nonzero vector supported on a proper subset of~$A$.
\end{itemize}
\end{lemma}

\begin{proof}[Proof of Lemma~\ref{lem:reduced_rep}]
The proof goes by establishing an inverse map. We proceed by induction on $k$, the number of chords. For a chord diagram $D_a$, we prove that that given an output of \textsc{construct-matrix}$(D_a)$ in the corresponding positroid cell $S_a$, all the real parameters of the $x_i$ and $y_i$ operations in the algorithm can be uniquely determined. The claim is obviously true when $k=1$.

Suppose we have shown the claim for less than $k$ chords. 
If $\D_a$ has at least two top chords, denote by $(i,i+1)$ the end of the first top chord $c_1$, and let $\D_b$ be the subdiagram made of chords which end no later than $(i,i+1),$ so that all its chords are $c_1$ and its descendants, and the marker set is $[i+1]\cup\{n\}$. 
Let $\D_c$ be the subdiagram made of the other chords, with markers $[n]\setminus[i-1].$
Suppose that $\D_b$ has $k'$ chords.

Let $C$ be a domino representative of $V\in S_a.$ 
$C$ can be written as \[C=\begin{pmatrix}C'\\C''\end{pmatrix},\]
where $C'$ is the submatrix which consists of the upper $k'$ rows and $C''$ is the submatrix which consists of the remaining rows. Denote by $V',V''$ the row spans of $C',C''$ respectively. 
Algorithm~\ref{construct-matrix} is structured so that the construction of $C'$ is independent of the construction of $C''$ which precedes it, as can be immediately seen from its recursive form, Algorithm~\ref{recursive-algorithm}.
In other words,
$$ C'\;=\;\pre_{n-1}\cdots\pre_{i+2}\;\generate\left(D_b,\,\{s_h,u_h,v_h,w_h\}_{h\in[k']}\right). $$
Similarly,
$C''$ is almost untouched by the construction of $C'.$ They are only affected when operations of the form $x_{i}(w_h)$ are performed in \textsc{end}$(c_h)$ for $c_h$ which ends at $(i,i+1),$ $h\in[k']$. 
If $w$ is the sum of all those $w_h,$ then
$$
C''\;=\;\pre_{i-1}\cdots\pre_1\;x_i(w)\;\generate\left(D_c,\,\{s_h,u_h,v_h,w_h\}_{h\in[k]\setminus[k']}\right)
$$
We claim that $V',V''$ are uniquely determined from $V.$
Indeed, if some linear combination of rows, which involves at least one of the $k-k'$ rows nontrivially lies in $V',$ then since all vectors in $V'$ have zero entries between $i+2$ and $n-1,$ there must be some nontrivial linear combination $v$ of rows from $V''$ which has zeros in these places as well. But all vectors in $V''$ have zeros at entries $1,\ldots,i-1,$ thus, this linear combination must be supported on $\{i,i+1,n\}.$ Lemma~\ref{lem:more than 4} applied to $c_1$ tell us that there is a set $I\in\binom{[n]\setminus\{i,i+1,n\}}{k}$ with $P_I(V)\neq 0.$ But if we take a representative matrix $M$ for $V$ which contains $v$ as one of the rows, expanding $\det(M^I)$ with respect to the row of $v$ shows that $P_I(V)$ must be zero, which is a contradiction. A similar claim holds for $V''.$

Applying the induction hypothesis on $V'$ allows us to uniquely calculate all $4k'$ positive parameters $\{s_h,u_h,v_h,w_h\}_{h\in[k']}.$ It also allows us to calculate the parameter $w$ in $x_i(w)$ above, since it is a sum of some of these $w_h$ variables.
Applying the induction hypothesis on $x_i(-w).V''$ allows us to uniquely calculate all parameters $\{s_h,u_h,v_h,w_h\}_{h\in[k]\setminus[k']}.$

The remaining case is that $\D_a$ has a single top chord $c_1$. 
We may assume it is a long chord, starting at $(1,2)$ and ending at $(n-2,n-1),$ or otherwise we can erase zero columns, and recover the algorithm variables for the smaller problem, which does satisfy these assumptions.
Let $\D_b$ be the subdiagram made of $c_1$'s descendants with index set $[n]\setminus[1].$

Take $V\in S_a,$ and let $C$ be a domino representative of $V$. 
We can write \[C=\begin{pmatrix}C_1\\C'\end{pmatrix}.\]
The submatrix $C'$ obtained from $C$ by erasing the first row is not a domino matrix for $S_b.$ But it equals
$$\pre_n\mathrm{rescale}_1 (u_1)\;x_1(1)\;\mathrm{rotate}_{k-1}\;\generate\left(D_b,\,\{s_h,u_h,v_h,w_h\}_{h\in[k]\setminus[1]}\right)$$
Here, $\mathrm{rotate}_{k}$ is the circular shift that takes the last column, puts it first, relabels it as $1$, and multiplies it by $(-1)^{k-1}$ in order to preserve positivity. $\mathrm{rescale}_1(t)$ multiplies the first column by $t$.

Both $\Span(C_1)$ and $V'$ are uniquely determined by $V,$ by an argument similar to that used to distinguish $V'$ and $V''$ in the previous case. Thus, $C_1$ is the unique vector in $\Span(C_1)$ whose second entry is $1.$ The variables $\{s_1,u_1,v_1\}$ are read from its $n$th, $1$st and $(n-2)$th coordinates respectively. $w_1$ cannot be determined yet, since there may applications of $x_{n-2}(w_l)$ for same-end descendants $c_l$ of $c_1.$ 
Similarly, starting from $V',$ projecting on the first $n-1$ columns and undoing the reversible $\mathrm{rescale}_1 (u_1),~x_1(1),~\mathrm{rotate}_{k-1}$ operations we obtain a vector space in $S_b$ which has a domino representation 
\[\generate\left(D_b,\,\{s_h,u_h,v_h,w_h\}_{h\in[k]\setminus[1]}\right).\] By induction this representation is unique, and we can determine its parameters. 
Now, if we define $w$ to be the sum of $w_l,$ over same-end descendants $c_l$ of $c_1,$ a simple calculation reveals that $w_1$ is precisely the ratio between the $(n-1)$th entry of $x_{n-2}(-w).C_1$ and its $(n-2)$th entry. Thus, all parameters are now found, as required.
\end{proof}

\begin{ex}
\label{ex:for_reduced}
Recall the two chord diagrams considered in Examples~\ref{example3-algo}-\ref{fuad}.
Consider first the subdiagrams which contain only $c_2,c_3.$
The matrix of the subdiagram
$(\{2,\ldots,14\},((3,4,6,7),(8,9,10,11))),$ is
\[\begin{array}{*{13}{c}}
\color{lightgray} 2 & \color{lightgray} 3 & \color{lightgray} 4 & \color{lightgray} 5 & \color{lightgray} 6 & \color{lightgray} 7 & \color{lightgray} 8 & \color{lightgray} 9 & \color{lightgray} 10 & \color{lightgray} 11 & \color{lightgray} 12 & \color{lightgray} 13 & \color{lightgray} 14 \\ \hline
0&u_2&1&0&v_2&w_2v_2&0&0&0&0&0&0&-s_2\\
0&0&0&0&0&0&u_3&1&v_3&w_3v_3&0&0&s_3 \\\hline
\end{array} \vspace{0.5em}\]
and the matrix of the subdiagram $(\{2,\ldots,8\},((2,3,4,5),(4,5,6,7)))$ is
\[\begin{array}{*{7}{c}}
\color{lightgray} 2 & \color{lightgray} 3 & \color{lightgray} 4 & \color{lightgray} 5 & \color{lightgray} 6 & \color{lightgray} 7 & \color{lightgray} 8 \\ \hline
u_2&1&v_2&v_2w_2&0&0&-s_2\\
0&0&u_3&1+w_2u_3&v_3&w_3v_3&s_3 \\\hline
\end{array} \vspace{0.5em}\]
respectively.
In both cases it is easy to see that each row vector is uniquely determined, after fixing its second entry to be $1,$ from the two dimensional vector space. 
In the first case the construction by the algorithm of each row is independent of the other. In the second the construction of the top row is independent of the bottom row. After performing $x_4(-w_2)$ on the bottom row, where $w_2$ is determined from the first row, and throwing zeros, we are left with the row vector which is the output of the generating algorithm to the chord diagram made of $c_3.$

If we compare the matrices considered in this example, to the matrices constructed in Examples~\ref{example3-algo}-\ref{fuad}, which correspond to adding $c_1$ we first observe that indeed the row which corresponds to $c_1$ is uniquely determined in both cases. We also see that the submatrices made of the rows which correspond to $c_2,c_3$ are indeed related to the matrices considered in this example by first shifting the last column to be the first, and fixing the sign. Then acting with $x_1(1),$ scaling the first column by $u_1$ and adding a zero column as a last column.
\end{ex}

\begin{rmk}\label{rmk:domino_coords_as_pluckers}
The nonzero entries of a matrix in domino form can be written, after one fixes some nonzero entry in each row to be $1$ as rational functions in the Pl\"ucker coordinates. This can be seen by undoing Algorithm~\ref{construct-matrix}, every time we undo a $x_i$ or $y_i$ operation we use \cite[Proposition 7.10]{lam2014totally} to write its parameter using the Pl\"ucker coordinates of the matrix at the time of the action, and to update the new Pl\"ucker coordinate in terms of the previous ones. We omit the details.
\end{rmk}

In the beginning of this section Section~\ref{analysis}, we have defined the algorithmic permutation~$\sigma$ of~$D$, made of $4k$ transpositions that record steps along $\generate$.  By Lemma~\ref{lem:reduced_rep}, the positroid cell~$S$ generated by the same $\generate(D)$ attains its largest possible dimension,~$4k$. By the discussion in Summary~\ref{subsec:plabic}, we can deduce that these two algorithmic outcomes bear the standard correspondence between a decorated permutation and a positroid cell.

\begin{cor}
\label{cor:reduced_rep}
Let $D \in \mathcal{CD}_{n,k}$ be a chord diagram. The decorated permutation of the positroid cell $S$ that arises from \textsc{construct-matrix($D$)} is the algorithmic permutation $\sigma$ of~$D$.
\end{cor}

\begin{rmk}
Our approach for showing that the BCFW cells are $4k$-dimensional is directly from the matrix construction and self-contained. Since we also show that the decorated permutations of these cells are given by other equivalent combinatorial descriptions, and coincide with BCFW cells defined elsewhere in the literature, we could also deduce this fact differently, e.g.~as these permutations correspond to reduced plabic graphs with $4k+1$ faces. Our direct approach in Lemmas~\ref{lem:reduced_rep} and~\ref{lem:more than 4} will serve us well in Section~\ref{sec:precise_ineqs}.
\end{rmk}

We now give the proof of Lemma~\ref{lem:more than 4}.

\begin{proof}[Proof of Lemma~\ref{lem:more than 4}]
The first part is straightforward. Consider the domino form $C$ of $V.$ Then by expanding $\det(C^I)$ with respect to the row which corresponds to $c$ we see that this determinant vanishes unless $I$ intersects the support of this row, which is precisely $A.$

For the proof of the second part, we first introduce the following terminology. A sequence of chords $(c,c',c''',\dots)$ is a \emph{Loch Ness Monster} if all consecutive pairs of chords are head-to-tail. For example, the sequence $(c_2,c_3,c_4,c_6)$ in Example~\ref{example9} is a Loch Ness Monster.

We prove the second part of the lemma by constructing a set $I$ of size~$k$, consists of one element for each chord, such that $\det(C^I)\neq 0$. The set $I$ has the form $I_1\cup I_2\cup I_3\cup\{p\},$ where $I_1\subseteq [i-1],~I_2\subseteq\{i+2,\ldots,j-1\},$ and $I_3\subseteq\{j+2,\ldots,n-1\}$ are defined as follows.
\begin{enumerate}
\item For $I_1$ we take the set of first markers of all chords that start before $c$.
\item $I_2$ is the set of second markers for all sticky descendants $c_h$ of $c$, together with the first markers of all other descendant of $c$. Recall that the first two markers are the start of the chord, and not the markers inherited from its parent, which may come earlier.
\item 
We construct $I_3$ as follows. 
Denote by $c_1$ the first chord which starts after $c$ ends.
\begin{itemize}
\item If $c_1$ starts at $(j',j'+1)$ for $j'>j+1$ we take for $I_3$ the set of first markers of each chord which starts $c$ ends.
\item
If $c_1$ starts at $(j+1,j+2),$ then we act as we did for $I_2,$ namely, for $c_1$ and any other chord in a maximal sticky chain from $c_1$ we add to $I_3$ the second marker, and for all the remaining chords which start afterwards we choose the first marker.
\item 
Finally, if $c_1$ starts at $(j,j+1),$ we write $c_0=c$ and we initiate $I_3$ to be the empty set. Write $c_2,c_3,\ldots,c_h$ for a maximal Loch Ness Monster which starts at $c_0.$ 
Put $j_0=j$ and $(j_h,j_{h}+1)$ for the end of $c_h.$ Observe that each $c_l$ 
must be a top chord. 

For each $l\in[h]$ first we add $j_l$ to $I_3.$ 
Then, for any descendant in a sticky chain from $c_l$, add the second marker. For any other chord which descends from $c_l$ we add the first marker. 

Denote by $\tilde{c}$ the first chord which starts after $c_h$ ends. $\tilde{c}$ is a top chord, and we denote its start by $(\tilde{j},\tilde{j}+1).$ $\tilde{j}>j_h,$ since $\tilde{c}$ does not belong to the Loch Ness Monster. If $\tilde{j}>j_h+1,$ add $\tilde{j}$ to $I_3,$ as well as the first marker of every chord that comes after it.

If $\tilde{j}=j_h+1$ then for each sticky descendant $\tilde{c}$ of it, take the second marker, while for the remaining chords take the first marker. 
\end{itemize}
\end{enumerate}

Note that all added entries are different, since chords start at different markers, and the third marker of a chord $c_l$ can never be the first or second marker of any other chord, except $c_{l+1}.$ Therefore, the size of $I$ is $k$, and is easily seen not to intersect~$A$.
Moreover, if $l\leq h$ then $j_{l}+1\notin I$ by construction.

\smallskip
We now change the matrix $C,$ by applying row operations, to another matrix $\bar{C}$ which is more convenient to analyze, and $\det(C^I)\neq 0$ if and only if $\det(\bar{C}^I)\neq 0.$ 
We define $\bar{C}$ in two steps.
\begin{itemize}
\item First apply row operations that cancel the domino inherited by the parent, and make the first marker of every chord $c'$ be the first nonzero position in the row of $c'.$ 
This is done by going parent-to-child, for each chord $c_l$ we subtract $\varepsilon_l$ times its parent's chord, which cancels the domino inherited to that chord by its parent. These operations result in adding nonzero entries, but only to the right of the starting domino of each chord's column. Moreover, the first domino of each row is untouched, and rows that correspond to top chords, like $c,$ are not affected. 

\item Then for every chord $c'\in\{c,c_1,\ldots,c_h\}$ or $c'=\tilde{c},$ and $\tilde{j}=j_l+1,$ we act iteratively on the sticky chain which descends from $c',$ in parent-to-child order. For every parent we subtract a multiple of its row from its sticky child in order to cancel the first domino entry of the latter. 
The second domino entries remain unchanged during this process.

These actions may add nonzero entries at positions which are after the end of each chord. They may also add nonzero entries before the start of a chord, but in a limited way:
If $c'=c$ they may add nonzero entries at column $i.$ 
If $c'=c_l,$ they may add nonzero entries at column $j_{l-1}.$
If $c'=\tilde{c},$ the may add nonzero entries at column $j_h.$
The rows of $\{c,c_1,\ldots,c_h,\tilde{c}\}$ and non-sticky descendants do not change. 
\end{itemize}
Consider the matrix $\bar{C}^I.$ 
Its restriction to the first $|I_1|$ rows and columns is upper triangular with nonzero diagonal. Thus, it is nonsingular precisely if the matrix $\tilde{C}=\bar{C}^{I\setminus I_1}_{[k]\setminus[|I_1|]},$ obtained from it by erasing the first $|I_1|$ rows and columns is nonsingular.

The first row of this matrix corresponds to $c.$ In this row there is a single nonzero entry, column indexed $p.$ Thus, $\tilde{C}$ is nonsingular precisely if the matrix $C',$ obtained from $\tilde{C}$ by erasing the first row and the column indexed $p$ is nonsingular.

Now, by the construction of $\bar{C}$ and of $I_2,$ which misses $i,i+1$
the restriction of $C'$ to its first $|I_2|$ rows and columns is again upper triangular, hence $C'$ is nonsingular precisely if $C^{(1)},$ obtained from $C'$ by erasing these rows and columns is nonsingular.

We iteratively construct $C^{(l+1)}$ from $C^{(l)}$ which is nonsingular precisely if $C^{(l)}$ is nonsingular in a similar fashion. We first note that in $C^{(l)}$ the row which corresponds to $c_l$ has nonzero entry only in the location which corresponds to column $j_l$. Indeed $c_l=(j_{j-1},j_{j-1}+1,j_l,j_l+1,n)$, as remarked above $j_l+1\notin I,$ and $\{j_{j-1},j_{j-1}+1,n\}$ either do not belong to $I,$ or were removed from the column set in previous steps. Erasing the row which corresponds to $c_l$ and column $j_l$ we obtain a matrix whose restriction to the first $\dd(c_l)$ rows and columns is upper triangular with nonzero diagonal, where we remind that $\dd(c_l)$ is the number of chords which descend from $c_l.$ Erasing also these rows and columns we obtain $C^{(l+1)}.$
Continuing like this, we end at $C^{(h+1)}.$ 

By the same reasoning this matrix is upper triangular with nonzero diagonal, hence invertible. As needed. 

For the third part of the lemma, suppose that there is a vector in $V$ spanned by $\{\e_h\}_{h\in A\setminus\{p\}}.$ Take $I\in\binom{[n]}{k}$ with $I\cap A= \{p\}$ and $P_I(V)\neq 0.$ Its existence implies that the support of every nontrivial linear combination of the elements of $V$ must contain an element out of $A\setminus\{p\}.$
\end{proof}

\subsection{Agreement of Permutations}
\label{agree}

Recall Definition~\ref{defdecperm} of $\pi$, the decorated permutation of $[n]$ associated directly to a chord diagram $D \in \mathcal{CD}_{n,k}$, which was used in Definition~\ref{maindefbcfw} of the BCFW positroid cell of~$D$. On the other hand, $D$~gives rise to a positroid cell via the algorithm $\generate(D)$, and by Corollary~\ref{cor:reduced_rep} the decorated permutation of that cell is the algorithmic permutation $\sigma$ of~$D$, from Definition~\ref{def:nicer_form_pi_alpha}. The next two lemmas show that $\pi$ and $\sigma$ admit an identical inductive characterization. It follows that $\sigma=\pi$ and the algorithm $\generate$ indeed produces the BCFW cell of~$D$.

\begin{lemma}\label{lem:properties_of_pi_alpha}
Let $D_a \in \mathcal{CD}_{n,k}$. The permutation
$\sigma_a$ satisfies the following properties, and is uniquely determined by them.
\begin{enumerate}
    \item If $k=1$ and $c_1=(i,i+1,l,l+1)$ then $\sigma_a = (i~(i+1)~l~(l+1)~n).$
    \item If $\D_a$ contains a top chord $c_j=(i,i+1,l,l+1)$, then let $\D_c$ be the subdiagram made of $c_j$ and the chords that start after it, including its descendants, with index set $\N_c = [n]\setminus[i-1],$ and let $\D_b$ be the subdiagram made of all the remaining chords, with index set $\N_b=[i-1]\cup\{i_{j\ast}+1,l,n\},$ where $(i_{j\ast},i_{j\ast}+1)$ is the start of the last descendant of $c_j$ in a sticky chain, and we relabel markers so that $i$ is relabeled as $i_{j\ast}+1$ and $i+1$ as $l$, so that chord which ends at $(i,i+1)$ in $\D_a$ ends in $\D_b$ at $(i_{j\ast}+1,l)$ while every chord that ends at $(i-1,i)$ in $\D_a$ ends at $(i-1,i_{j\ast}+1)$ in $\D_b$.
    Then $\sigma_a = \sigma_b\sigma_c,$ where the permutations are extended to $[n]$ by adding black fixed points outside of $\N_b,\N_c$.
    \item\label{it:upper_emb} Suppose that $D_a$ is obtained from $D_b$ by a \emph{combinatorial upper embedding}: adding a new long 
    top chord $c_1$ which starts at $(1,2)$ and ends at $(n-2,n-1).$ Then 
    \[\sigma_a=(1~(n-2)~(n-1)~n)\sigma_b(1~2),\]
    where $\sigma_b$ is the decorated permutation of $D_b,$ whose index set is $[n]\setminus\{1\}.$
\end{enumerate}
\end{lemma}

\begin{proof}
The first item is straightforward. For the second item, note that by Remark~\ref{rmk:pi_alpha_according_to_top_levels} the algorithmic permutation can be written as
\[\sigma_a \;=\; \tau_b\,\tau_c\,\rho_c\,\rho_b,\]
If in $D_b$ all chords end before $i,$ then the two permutations commute and the effect of the replacements $i\to i_{j\ast}+1,~i+1\to l$ in $D_b$ is trivial, and the claim holds.

If in $D_b$ there are some chords which end at $(i-1,i)$ but no chord ends at $(i,i+1)$ then $\rho_c,\tau_c$ commute, since their periods do not contain common indices: in $\rho_c$ appear only elements of the form $i_h+1,l_h,l_h+1$ for chords $c_h=(i_h,i_h+1,l_h,l_h+1)$ in $D_c$, and these are always larger than the corresponding terms for $D_b.$
Also most transpositions of $\tau_c$ are of indices which do not appear in $D_b.$ The only exclusion is the leftmost transposition in $\tau_c$ which is $(i~i_{j\ast}+1).$ It has a common index only with triplets of $D_b$ of the form $(p+1~i-1~i).$
Commuting to such elements gives
\[(i~i_{j\ast}+1)(p+1~i-1~i)=(p+1~i-1~i_{j\ast}+1)(i~i_{j\ast}+1).\]
Thus, the commutation of $(i,i_{j\ast}+1)$ with $\rho_b$ has the prescribed effect. Since $i$ does not appear in $\tau_b$ the claim also follows in this case.

The final case is when there are chords in $D_b$ which end at $(i,i+1)$ and possibly also chords which ends at $(i-1,i).$ Note that the latter appears in $\rho_b$ to the right of the former. 
The only element in $\rho_c$ which shares an index with elements of $\rho_b$ is $((i+1)~l~(l+1))$ which does not commute only with elements of the form $((p+1)~i~(i+1)).$
Commuting such elements results in replacing each appearance of $i+1$ in $\rho_b$ by $l,$ since
\[((i+1)~l~(l+1))((p+1)~i~(i+1))=((p+1~ i ~l)((i+1)~l~(l+1)).\]
After moving $\rho_c$ to the right of $\rho_b$ the only term of $\tau_c$ which does not commute with elements in its right is again $(i~i_{j\ast}+1)$ which does not commute with elements of the form $((p+1)~(i-1)~i)$ or $((p+1)~i~l).$ Again the effect of commuting it with them is $i\to i_{j\ast}+1$ in the terms of $b.$
Since again $i,i+1$ do not appear in $\tau_b$ also in this case the result follows. 

For the third item, we can write
\[\sigma_a = (1~(i_{1\ast}+1))\tau_b(2~n)(2~(n-2)~(n-1))\rho_b,\]
where $\tau_b,\rho_b$ are as in the previous item. 
The required permutation $\sigma_b$ is obtained from $\tau_b\rho_b$ by replacing every term of the form $(2~(i_h+1)),$ for children $c_h=(i_h,i_h+1,l_h,l_h+1)$ of $c_1$, by $(2~n).$

We first move $(1~(i_{1\ast}+1))$ to the right. If in $D_b$ no chord starts at $2$ then $i_{1\ast}+1=2$ and $(1~2)$ commutes with $\tau_b.$ 
Otherwise $\tau_b$ is of the form 
\[(2~(i_{1\ast}+1))(3~(i_{1\ast}+1))\cdots(i_{1\ast}~(i_{1\ast}+1))\sigma',\]
$\sigma'$ does not contain $1$ but it contains $(2~(i_h+1))$ for every child $c_h=(i_h,i_h+1,l_h,l_h+1)$ of $c_1.$

Commuting with $(1~(i_{1\ast}+1))$ we obtain that \[(1~(i_{1\ast}+1))\tau_b=\tilde{\tau}_b(1~ 2)\]
where $\tilde{\tau}_b$ is obtained from $\tau_b$ by replacing every appearance of $2$ by $1,$ where we have used
\[(1~(i_{1\ast}+1))(2~(i_{1\ast}+1))=(2~(i_{1\ast}+1))(1~2)~\] to commute with $(2~(i_{1\ast}+1))(3~(i_{1\ast}+1))\cdots(i_{1\ast}~(i_{1\ast}+1))$ and
\[
(1~2)(2~(i_h+1))=(1~(i_h+1))(1~2)\]
to commute with $\sigma'.$

Then we commute it with $(2~n)(2~(n-2)~(n-1))=(2~(n-2)~(n-1)~n)$ and obtain \[(1~2)(2~(n-2)~(n-1)~n)=(1~(n-2)~(n-1)~n)(1~2),\] and then further commute $(1~2)$ to the right, which is possible since $\rho_b$ does not contain $1,2$ in the indices on which it acts nontrivially.

We now wish to commute $(1~(n-2)~(n-1)~n)$ to the left of $\tilde{\tau}_b.$ Among the indices $1,n-2,n-1,n$ only $1$ can appear in $\tilde{\tau}_b,$ and only in $(1~(i_h+1))$ terms.
Since 
\[(1~(i_h+1))(1~(n-2)~(n-1)~n)=(1~(n-2)~(n-1)~n)((i_h+1)~n),\]
the result follows.

The fact that the permutation is uniquely determined by these properties follows from a simple induction on the number of chords.
\end{proof}

\begin{lemma}
\label{lem:properties_of_sigma_alpha}
$\pi_a$ satisfies the following properties, and is uniquely determined by them.
\begin{enumerate}
    \item If $k=1$ then $\pi_a = (a_1~(a_1+1)~b_1~(b_1+1)~n).$
    \item If $\D_a$ contains a top chord $c_h=(i,i+1,l,l+1),$ let $\D_c$ be the subdiagram made of $c_h,$ and the chords which start after it, including its descendants, on the markers set of is $\N_c = [n]\setminus[i-1].$ Let $\D_b$ be the subdiagram made of all the remaining chords, with index set $\N_b=[i-1]\cup\{i_{h\ast}+1,l,n\},$ where $(i_{h\ast},i_{h\ast}+1)$ is the end of the last chord in the sticky chain from $c_h,$ and we relabel the markers so that a chord that ends at $(i,i+1)$ in $\D_a$ ends in $\D_b$ at $(i_{h\ast}+1,l)$ while every chord that ends at $(i-1,i)$ in $\D_a$ ends at $(i-1,i_{h\ast}+1)$ in $\D_c$.
    Then $\pi_a = \pi_b\pi_c,$ where the permutations are extended to $[n]$ by adding black fixed points.
    \item Suppose that $\D_a$ is obtained from $\D_b$ by a combinatorial upper embedding: adding a new long chord which starts at $(1,2)$ and ends at $(n-2,n-1).$ Then 
    \[\pi_a=(1~(n-2)~(n-1)~n)\pi_b(1~2),\]
    where $\pi_b$ is the decorated permutation of $D_b,$ whose markers set is $[n]\setminus\{1\}.$
\end{enumerate}
\end{lemma}
\begin{proof}
The first item is immediate from the definition of $T,U,V,W.$
The third item is also immediate: $\pi_b$ and $\pi_c$ are just the partial products made of the $5$-cycles which correspond to the chords of $\D_b,~\D_c$ respectively.

For the first item, by definition we have
\[\pi_a=(T_1~U_1~V_1~W_1~n)\pi_b,\]and it is easily seen that $(T_1~U_1~V_1~W_1~n)=(1~R~(n-2)~(n-1)~n)$ where $R=2$ if the next chord starts at $(j,j+1)$ for $j>2,$ and otherwise $R=U_2.$
In the former case, $1,2$ do not appear in $\pi_b$ and we can split
\[(1~2~(n-2)~(n-1)~n)=(1~(n-2)~(n-1)~n)(1~2),\]and commute $(1~2)$ to the rightmost position.
In the latter case, $1$ does not appear in $\pi_b,$ and $2$ appears there exactly once - in the leftmost term $(2~R_2).$
We split \[(1~R_2~(n-2)~(n-1)~n)=(1~(n-2)~(n-1)~n)(1~R_2).\]
Since $(1~R_2)(2~R_2)=(2~R_2)(1~2),$ we can again commute $(1~2)$ to the rightmost position. As needed.

Simple induction shows that these properties determine the permutation uniquely.
\end{proof}

\begin{cor}
\label{pi_equals_sigma}
Let $D \in \mathcal{CD}_{n,k}$ be a chord diagram. The decorated permutation $\pi$ of $D$ is equal to the algorithmic permutation $\sigma$ of $D$.
\end{cor}

\begin{proof}
This is immediate from Lemmas~\ref{lem:properties_of_sigma_alpha} and~\ref{lem:properties_of_pi_alpha}. Both permutations are determined by the same recursion.
\end{proof}

\subsection{Another Algorithm}
\label{subsec:second_alg}

We describe a variant of Algorithms~\ref{construct-matrix} and~\ref{recursive-algorithm}, which is useful for analyzing the effect of extending a chord diagram by adding chords to the right of the given ones. The variant algorithm iterates the top chords in opposite order, increasing in their occurrence along $D$, rather than decreasing. The descendants of each chord are scanned in decreasing order as before.

\begin{algorithm}
\label{rightwards-algorithm}
Given a chord diagram~$D$, The algorithm
\textsc{construct-matrix-rightwards}
is the following variation of \textsc{construct-matrix} from Algorithm~\ref{construct-matrix}. The variables $M, N, K, \{s_l,u_l,v_l,w_l\}^k_{l=1}$ are as before. The statements \textsc{start}, \textsc{fill} and the subroutine \textsc{sub-construct-matrix} are as in Algorithms~\ref{construct-matrix} and~\ref{recursive-algorithm} and not restated here. Denote the top chords of~$D$ by $c_{(1)},\dots,c_{(g)}$ in \emph{increasing} order, from first to last. The notation $c_1,\dots,c_l,\dots,c_k$ is reserved for the usual startwise increasing ordering.

\begin{framed}
\noindent
\begin{minipage}{\textwidth}
\noindent
\noindent \textsc{construct-matrix-rightwards}(chord diagram $D \in \mathcal{CD}_{n,k}$,\, variables $s_l,u_l,v_l,w_l \in \mathbb{R}$ for $l \in [k]$)

\noindent \hspace{1cm} \textsc{initialize:} Let $N = \{n\}$, $K = \emptyset$, and an empty matrix $M \in \Mat_{K \times N}$

\noindent
\hspace{1cm} 
\textbf{for} every $c_{(h)} = (i_h,i_h+1,j_h,j_h+1) = c_l$ \textbf{in} $(c_{(1)},\dots,c_{(g)})$ \textbf{do}

\noindent
\hspace{2cm} 
\textbf{for} $m$ \textbf{in} $(j_{h-1}+2,j_{h-1}+3,\dots,i_h,i_h+1,j_h+1)$  \textbf{do} \textsc{fill}($m$)
\hfill
//~ where $j_0=-1$

\noindent
\hspace{2cm} 
\textsc{end-right}($c_{(h)}$): apply $y_{i_h+1}(v_l) \circ x_{j_h+1}(s_l) \circ x_{j_h}(w_l) \circ \inc_{j_h;l}$ to $M$ 

\noindent
\hspace{4.8cm} 
add $j_h$ to $N$, and add $l$ to $K$

\noindent
\hspace{2cm} 
\textsc{sub-construct-matrix}(\emph{children}($c_{(h)}$), $i_h+2$, $j_h$)
\hfill
//~ ordered last to first

\noindent
\hspace{2cm} 
\textsc{start}($c_{(h)}$)

\noindent
\hspace{1cm} 
\textbf{for} $m$ \textbf{in} $(j_{g}+2,\dots,n-1)$  \textbf{do} \textsc{fill}($m$) 

\noindent \hspace{1cm} \textbf{return} $M$
\end{minipage}
\end{framed}
\end{algorithm}

The analysis of \textsc{construct-matrix-rightwards} is similar to the original variant. Let $S'$ be the positroid cell constructed by this algorithm. Its dimension is again at most~$4k$. As in Definition~\ref{def:nicer_form_pi_alpha} and Remark~\ref{rmk:pi_alpha_according_to_top_levels} for the case of \textsc{construct-matrix}, the chord diagram $D$ is associated via the left-to-right algorithm with \emph{another algorithmic permutation},
$$ \sigma' \;=\; \tau_g' \cdot \tau_{g-1}' \cdots \tau_1' \,\cdot\, \rho_1' \cdot \rho_2' \cdots \rho_g' $$ 
where $c_{(1)},\dots,c_{(g)}$ are the top chords in increasing order, and the factors $\tau_f'$ and $\rho'_f$ are the following products of 2-cycles and 3-cycles corresponding to a chord $c_{(f)} = (i,i+1,j,j+1)$:
\begin{samepage}
\begin{itemize}
\item 
$\tau_f'$ is the same as $\tau_f$ except for the rightmost transposition, which is $(i+1~~j)$ instead of $(i+1~~n)$.   
\item 
$\rho_f'$ is the same as $\rho_f$ except for the leftmost 3-cycle, which is $(j~~j+1~~n)$ instead of $(i+1~~j~~j+1)$.
\end{itemize}
\end{samepage}

The following two lemmas are analogous to Proposition~\ref{prop:domino} and 
Corollary~\ref{cor:reduced_rep}
of the original algorithm, and their proofs are omitted.

\begin{lemma}
\label{lem:domino_second_alg}
Let $D \in \mathcal{CD}_{n,k}$ be a chord diagram, and let $C \in \mathcal{DM}_{n,k}$ be the domino matrix that corresponds to~$D$. Then
$\textsc{construct-matrix-rightwards}(D)$ has the same form of $C$, and it satisfies the domino sign rules.
\end{lemma}

\begin{lemma}
\label{lem:reduced_second_alg}
Let $D \in \mathcal{CD}_{n,k}$ be a chord diagram. The decorated permutation of the positroid cell $S'$ that arises from \textsc{construct-matrix-rightwards($D$)} is the other algorithmic permutation $\sigma'$ of~$D$.
\end{lemma}

In order to show that the two algorithms parameterize the same positroid cell, we prove that their decorated permutations agree.
\begin{lemma}
\label{lem:equivalence_of_algs}
Let $D \in \mathcal{CD}_{n,k}$ be a chord diagram. Let $S$ the positroid cell arising from \textsc{construct-matrix}$(D)$, and associated with the algorithmic permutation~$\sigma$, and let $S'$ be the positroid cell arising from \textsc{construct-matrix-rightwards}$(D)$, and associated with the other algorithmic permutation~$\sigma'$. Then $\sigma=\sigma'$ and hence $S=S'$.
\end{lemma}

\begin{proof}
The two algorithmic permutations are
\begin{align*}
\sigma \;&=\; \tau_1 \cdot \tau_2 \cdots \tau_g \,\cdot\, \rho_g \cdot \rho_{g-1} \cdots \rho_1
\\
\sigma' \;&=\; \tau_g' \cdot \tau_{g-1}' \cdots \tau_1' \,\cdot\, \rho_1' \cdot \rho_2' \cdots \rho_g'
\end{align*}
We claim first that for every top chord $c_{(f)}$ 
$$ \tau_f' \cdot \rho_f' \;=\; \tau_f \cdot \rho_f $$
Indeed, if $c_{(f)} = (i,i+1,j,j+1)$ then these factors differ by
\begin{align*}
\tau_f \;&=\; \cdots (i+1~~n) & \rho_f\;&=\; (i+1~~j~~j+1) \cdots \\
\tau_f' \;&=\; \cdots (i+1~~j) & \rho_f'\;&=\; (j~~j+1~~n) \cdots
\end{align*}
Our claim follows from the identity
$$ (i+1~~n)~(i+1~~j~~j+1) \;=\; (i+1~~j)~(j~~j+1~~n) $$
We apply this claim to the factors of the first top chord~$c_1 = (i,i+1,j,j+1)$ in $\sigma'$, and obtain
$$ \sigma' \;=\; \tau_g' \cdots \tau_2' \cdot \tau_1 \,\cdot\, \rho_1 \cdot \rho_2' \cdots \rho_g'
$$
The factor $\rho_1$ commutes with all $\rho_f'$ for $f \geq 2$ as they do not share any elements other than fixed points. Indeed, the least element in $\rho_2$ is either $j_2$ the third marker of $c_2$ or the second marker of one of its descendants, and these are strictly greater than $j+1$ the fourth marker of $c_1$. Similarly, $\tau_1$ commutes with all $\tau_f'$ for $f \geq 2$. Indeed, all the transpositions in $\tau_f'$ are in the range $j$ to $n-1$, while the rightmost transposition in $\tau_1$ swaps its second marker $i+1$ with $n$, and other transposition only involve starts of $c_1$ and its descendants which are smaller than $j$. In conclusion,
$$ \sigma' \;=\; \tau_1 \cdot \tau_g' \cdots \tau_2'   \,\cdot\,   \rho_2' \cdots \rho_g' \cdot \rho_1
$$
The lemma follows by iterating for all top chords $c_{(2)},\dots,c_{(g)}$ until we obtain~$\sigma$.
\end{proof}

\begin{rmk}
One could also deduce Lemma~\ref{lem:reduced_second_alg} from $\sigma=\sigma'$, based on plabic graphs (Summary~\ref{subsec:plabic}). To sketch the argument, since both corresponding plabic graphs have $4k+1$ faces, which is minimal for a $4k$-dimensional cell, they are reduced and correspond to the same positroid cell $S=S'$.  
\end{rmk}

\subsection{The Domino Theorem}
\label{subsec:domino_proof}

We conclude Section~\ref{sec:domino} by stating key results and corollaries that play a role in later sections. As noted in Remark~\ref{same-cell}, we often denote by $\generate(D)$ the set of all points in $\Gr_{k,n}$ which are obtained from that algorithm, by taking the row-spans of the outputs $\generate(D,\{s_l,u_l,v_l,w_l\}_{l=1}^k)$ where the parameters vary over~$(0,\infty)^{4k}$.
We have seen in Section~\ref{subsec:first_alg} that the algorithm generates domino matrices. By the analysis of Section~\ref{analysis} it indeed gives a positroid cell, shown in Section~\ref{agree} to have the decorated permutation of~$D$ as in Definition~\ref{defdecperm}. An equivalent variant of this algorithm has been given in Section~\ref{subsec:second_alg}. We record these constructions of BCFW cells in the following corollary.

\begin{cor}
\label{cell_equals_cell}
Let $D \in \mathcal{CD}_{n,k}$ be a chord diagram. The BCFW positroid cell that corresponds to $D$ is equal to $\generate(D)$ and to $\textsc{construct-matrix-rightwards}(D)$.
\end{cor}

\begin{proof}
By Lemma~\ref{lem:reduced_rep} and its Corollary~\ref{cor:reduced_rep},  Corollary~\ref{pi_equals_sigma}, Definition~\ref{maindefbcfw}, and Lemma~\ref{lem:equivalence_of_algs}.
\end{proof}

It follows that BCFW cells have domino representations, as formulated in the following theorem. This result is stated above in  Theorem~\ref{thm:domino_non_formal}, and was conjectured by Karp, Williams, Zhang, and Thomas~\cite[Conjecture~A.7]{karp2020decompositions}. 

\begin{thm}
\label{thm:domino}
Let $D \in \mathcal{CD}_{n,k}$ be a chord diagram and $S$ its BCFW positroid cell. Every point in~$S$ has a representative matrix in the domino form of~$D$ that satisfies the domino sign rules, and every such matrix represents a point in~$S$. The domino matrix representation of a point is unique up to rescaling each row by a positive number. 
\end{thm}

\begin{proof}
By Corollary~\ref{cell_equals_cell}, the BCFW cell that corresponds to~$D$ is generated by Algorithm~\ref{construct-matrix}, where the existence of a unique preimage to each $V\in S$ with respect to $\generate$ is guaranteed by Lemma~\ref{lem:reduced_rep}. The output has the domino form of~$D$ and satisfies the sign rules by Proposition~\ref{prop:domino}. Every assignment to the domino matrix~$C$ is obtained from some input to $\generate(D)$ by Lemma~\ref{inverse-domino}, and hence represents a point $V \in S$. As also noted in that lemma, different domino matrices that represent $\generate(D)$ for the same input to the algorithm are equivalent up to positively rescaling each set $(\alpha_l,\beta_l,\gamma_l,\delta_l,\varepsilon_l)$, which amounts to rescaling each row of the representative matrix of the form~$C$ given by Proposition~\ref{prop:domino}.
\end{proof}

The analysis of the BCFW triangulation in the coming sections uses the recursive nature of its algorithmic constructions. It critically depends on which matrix operations generate a BCFW cell~$S$ of some chord diagram~$D$ from a BCFW cell~$S'$ of some subdiagram~$D'$ of~$D$. The relation between $D$ and~$D'$ comes in three different flavors, stated in Corollaries~\ref{cor:generation_left},~\ref{cor:generation_right} and~\ref{cor:generation_top} below.

In these corollaries, we always consider a chord diagram $D = (N, \mathcal{C})$ and a subdiagram $D' = (N', \mathcal{C}')$ with a subset $N' \subset N$ of the markers and a subset $\mathcal{C}' \subset \mathcal{C}$ of the chords. In fact, $N'$ is taken to be a cyclic interval of~$N$ in these cases, and the chords $\mathcal{C}'$ are exactly those whose four markers are in~$N'$. The specific assumptions on the excluded chords $\mathcal{C}\setminus\mathcal{C}'$ differ between the three corollaries. In each case, we relate the respective BCFW positroid cells $S$ and~$S'$. 

We always assume a common largest marker $n = \max N = \max N'$. The matrix operation $x_i = x_i(t)$ is always carried out with a positive variable $t \in (0,\infty)$, and similarly for $y_i$. Writing $y_{i \mi 1}$ implies that $i$~is contained in the index set when the operation is applied at its current precedent. In all diagrams and subdiagrams, the two markers of a chord's start or end are assumed to be consecutive numbers, and hence denoted $(i,i{+}1)$ rather than $(i,i \pl 1)$.

\begin{cor}[left extension]
\label{cor:generation_left}
Let $D' = (N', \mathcal{C}')$ be a subdiagram of $D = (N,\mathcal{C})$ 
with an index set $N' = \{i\in N:i\ge m\}$ for some $m \in N$, such that $\mathcal{C}\setminus \mathcal{C'} = \{(i,i+1,j,j+1)\in\mathcal{C}:j \leq m\}$.
Then the BCFW cell of~$D$ is generated from the BCFW cell of~$D'$ by a sequence of operations from
$$ \left\{ \pre_i : 1 \leq i < m\right\} \;\cup\; \left\{ \inc_i : 1 \leq i < m \right\} \;\cup\; \left\{ x_i : 1 \leq i \leq m \right\} \;\cup\; \left\{ y_{i \mi 1} : 1 \leq i < m \right\} .$$
\end{cor}

\begin{cor}[right extension]
\label{cor:generation_right}
Let $D' = (N', \mathcal{C}')$ be a subdiagram of $D = (N,\mathcal{C})$ with an index set $N' =\{i\in N:i\le m\}\cup\{n\}$ for some $m \in N$, such that $\mathcal{C}\setminus \mathcal{C'} = \{(i,i+1,j,j+1)\in\mathcal{C}:i \geq m-1\}$.
Then the BCFW cell of~$D$ is generated from the BCFW cell of~$D'$ by a sequence of operations from
$$ \left\{ \pre_i : m < i < n\right\} \;\cup\; \left\{ \inc_i : m < i < n \right\} \;\cup\; \left\{ x_i : m < i < n \right\} \;\cup\; \left\{ y_{i \mi 1} : m \leq i < n \right\}. $$
\end{cor}

\begin{cor}[top extension]
\label{cor:generation_top}
Let $D' = (N', \mathcal{C}')$ be a subdiagram of $D = (N,\mathcal{C})$ with an index set $N' = N \setminus \{m\}$ where $m = \min N$, such that $\mathcal{C}\setminus \mathcal{C'} = \{(m,m+1,n \mi 2,n \mi 1)\}$.
Then the BCFW cell of~$D$ is generated from the BCFW cell of~$D'$ by an upper embedding, whose sequence of operations is
$$ y_{n \mi 2} \;\circ\; y_{n \mi 1} \;\circ\; y_n \;\circ\; x_m \;\circ\; \inc_m .$$
\end{cor}

\begin{ex}
We demonstrate these three corollaries, using subdiagrams of the chord diagram $D$ drawn in Example~\ref{example9}.
\begin{itemize}
\item 
Corollary \ref{cor:generation_left} example: Let $D'$ be the subdiagram with markers $\{6,\dots,18\}$ and chords $\{c_4,\dots,c_8\}$. Then, a sequence of operations from $S'$ to $S$ as claimed is given by Algorithm~\ref{construct-matrix}:
$$ y_{3\mi1} \circ y_{2 \mi1}\circ \pre_1 \circ y_{3\mi1} \circ x_4 \circ x_3 \circ \inc_3 \circ y_{5 \mi 1} \circ \pre_4 \circ y_{5 \mi 1} \circ x_6 \circ x_5 \circ \inc_5 \circ y_{2 \mi 1} \circ x_6 \circ x_2 \circ \inc_2 $$
\item 
Corollary \ref{cor:generation_right} example: Let $D'$ be the subdiagram of $D$ with markers $\{1,\dots,11,18\}$ and chords $\{c_1,\dots,c_5\}$. Then, a sequence of operations from $S'$ to $S$ as claimed is given by Algorithm~\ref{rightwards-algorithm}:
$$ y_{12 \mi 1} \circ y_{11 \mi 1} \circ y_{13} \circ \pre_{13} \circ \pre_{15}\circ y_{14\mi1} \circ x_{16} \circ x_{14} \circ \inc_{14} \circ {} $$ $$ {} \circ y_{12\mi1} \circ x_{16} \circ x_{12}  \circ \inc_{12} \circ y_{16\mi1} \circ x_{17} \circ x_{16} \circ \inc_{16} \circ \pre_{17} $$
\item 
Corollary \ref{cor:generation_top} example: Let $D' = (\{10,\dots,18\},\{c_6,c_7,c_8\})$ and $D''=(\{11,\dots,18\},\{c_7,c_8\})$. Then, the upper embedding from $S''$ to $S'$ is: $$ y_{16} \circ y_{17} \circ y_{18} \circ x_{10} \circ \inc_{10} $$
\end{itemize}
\end{ex}

\begin{proof}[Proof of Corollaries~\ref{cor:generation_left},~\ref{cor:generation_right} and~\ref{cor:generation_top}] For all three claims, we consider a sequence of matrix operations that constructs the BCFW cell $S$ of~$D$. This sequence splits into a prefix that constructs the BCFW cell $S'$ of~$D'$, and a suffix that only contains operations in the given ranges. Thus, by applying the prefix of the sequence we obtain~$S$, and then we generate $S'$ from~$S$ as required. 

For Corollary~\ref{cor:generation_left} we use Algorithm~\ref{construct-matrix}. This algorithm scans the diagram~$D$ right to left, such that all starts, ends, and markers in~$D'$ are handled before those not in~$D'$. Hence, if we stop it at the right moment, then we have ran the algorithm for~$D'$ and have generated~$S'$. The exact ranges for the remaining operations are straightforward to verify from the details of the algorithm.

Corollary~\ref{cor:generation_right} is shown similarly using Algorithm~\ref{rightwards-algorithm}, which handles top chords from left to right, each one together with its descendants. Hence it first generates~$S'$, and then $S$ by a sequence of operations in the given ranges. For later reference, we mention additional properties of this sequence in case of a single top chord $(i, i+1, j, j + 1)$ in $\mathcal{C}\setminus\mathcal{C}'$. It starts with $\pre$ operations followed by $x_{j+1} \circ x_j \circ y_{i+1} \circ \inc_j$. Afterwards, there are no $y_{j \mi 1}$ or $y_j$, and only in the case $i=m-1$ there is one~$y_{m-1}$.

Finally, Corollary~\ref{cor:generation_top} is based on the fact that a cell with a single top chord is obtained by an upper embedding. This follows from Remark~\ref{rmk:upper_as_plabic_and_perm} and the last case of Lemma~\ref{lem:properties_of_pi_alpha}. More explicitly, these operations can be shown to produce the cell as in the end of the proof of Lemma~\ref{lem:reduced_rep}.
\end{proof}

We now pause the discussion of the BCFW positroid cells and their algorithmic and domino representations for the next three section, where we turn to discuss their images under the amplituhedron map. We later return to this topic in Section \ref{sec:precise_ineqs}, which builds upon the techniques of this section in order to give a detailed description of the codimension-one boundaries of the BCFW positroid cells. The reader interested in that topic may safely skip from here to Section~\ref{sec:precise_ineqs}.

\section{Twistors and Functionaries}
\label{sec:promotion}

The points of the amplituhedron~$\Ampl_{n,k,m}(Z)$ are $k$-dimensional spaces in the Grassmannian $\Gr_{k, k+m}$. Hence, they may be described by representative $k \times (k+m)$ full rank matrices, or by the Pl\"ucker coordinates of that Grassmannian. Arkani-Hamed and Trnka~\cite{arkani2014amplituhedron} introduced yet another set of coordinates for the amplituhedron, the \emph{twistor coordinates}, which take into account the positive matrix~$Z$. These coordinates were used by Arkani-Hamed, Thomas and Trnka~\cite{arkani2018unwinding} to develop a combinatorial and topological picture of the amplituhedron. Parisi, Sherman-Bennett and Williams~\cite{parisi2021m} used them to characterize the $Z$-images of a large family of positroid cells giving triangulations of~$\Ampl_{n,k,2}$.

\subsection{Definitions}
\label{subsec-twist-func-defs}

We start with some definitions and basic results stated for general~$m$. Later, we focus on the case $m=4$ relevant to this paper, even where our treatment extends to other values of~$m$. As before, our definitions for  $\{1,\dots,n\}$ extend to general index sets $N \subset \mathbb{N}$.

Let $Z\in\Mat^{>}_{n\times(k+m)}$ be a matrix with positive $(k+m) \times (k+m)$ minors where $k+m \leq n$ as in the definition of the amplituhedron~$\Ampl_{n,k,m}(Z)$. Denote the rows of $Z$ by $Z_1,\dots,Z_n \in \RR^{k+m}$. For a point $Y \in \Gr_{k,k+m}$ we also denote by $Y$ a representative matrix in $\Mat_{k \times (k+m)}$ when its choice factors out. When this matrix is written as $Y = \Z(C) = CZ$, the convention is that both $C$ and~$Z$ have nonnegative determinants in their maximal minors, and in cases that determinants might be negative we explicitly say so. Let $Y_1,\dots,Y_k\in\RR^{k+m}$ denote the rows of~$Y$. The determinant of a square matrix $M$ is denoted by~$\langle M \rangle$.

\begin{definition}
\label{twistors}
Consider a matrix $Z\in\Mat_{n\times(k+m)}$ and a representative matrix $Y \in \Gr_{k,k+m}$. For every set $I = \{i_1,\dots,i_m\}$, such that $1 \leq i_1 < i_2 < \dots < i_m \leq n$, the $I$th \emph{twistor coordinate} of $Y$ is the determinant of the $(k+m)\times(k+m)$ matrix whose rows are $Y_1,\dots,Y_k,Z_{i_1},\dots,Z_{i_m}$. We write it using any of the following notations: 
$$ \langle Y~Z_I \rangle \;=\; \langle Y_1 \dots Y_k~Z_{i_1} \dots Z_{i_m} \rangle \;=\;\langle Y~Z_{i_1} \dots Z_{i_m} \rangle \;=\; \langle Y~i_1 i_2 \dots i_m \rangle_Z $$ 
When $Y$ or $Z$ are fixed and understood from the context, we omit one or both of them and write, for example, $\langle i_1 i_2 \dots i_m \rangle$. In the case $k=0$, the twistor $\langle i_1 i_2 \dots i_m \rangle$ is simply the determinant of the corresponding minor of~$Z$. {We rarely also calculate twistor coordinates for matrices $Y$ which are not of full rank, so that they do not represent a point in $\Gr_{k,k+m}.$ In this case the twistor is just $0.$}
\end{definition}

\begin{rmk}
Although the $\tbinom{n}{m}$ different twistors correspond to unordered subsets $I \in \tbinom{[n]}{m}$, the order of indices in a twistor $\langle Y~i_1i_2\dots i_m\rangle_Z$ is important, by the above definition as a determinant. In the forthcoming, we occasionally use the freedom to write indices not in order, for example $\langle1375\rangle = -\langle1357\rangle$. We also write twistors with repeating indices, for example $\langle2466\rangle = 0$. 
\end{rmk}

The twistor coordinates $\langle Y~Z_i~Z_j\rangle$ are instrumental in the study of triangulations of the $m=2$ amplituhedron~\cite{parisi2021m}. It turns out that in triangulations of the $m=4$ amplituhedron, sums of products of twistors serve an important function. We hence give them the following name.

\begin{definition}
\label{functionary}
Let $0 \leq m \leq n$. A~\emph{functionary} is a homogeneous polynomial in the $\tbinom{n}{m}$ twistors $\langle Y Z_I\rangle$. In more detail, a functionary is a real function of $Y \in \Gr_{k,k+m}$ and $Z \in \Mat^{>}_{n \times (k+m)}$ of the form 
$$ F\left(\langle Y Z_I\rangle : {I\in\tbinom{[n]}{m}}\right) \;:=\; F\left.\left((z_I)_{I\in\tbinom{[n]}{m}}\right)\right|_{z_I=\langle Y Z_I\rangle} $$ 
where $F$ is a homogeneous polynomial over~$\mathbb{R}$ of \emph{degree}~$d$ in the variables $\{z_I : I \in \tbinom{[n]}{m}\}$, defined for all applicable $k \in \{0,\dots,n-m\}$. 

A functionary is also denoted $F(\langle Y~Z_I \rangle:I\in\scalebox{0.8}{$\binom{N}{m}$})$
if all the twistors that appear in $F$ are supported on a subset $N \subset [n]$. We denote a polynomial by $F(z_I)$ and a functionary by $F(\langle Y~Z_I \rangle)$ when the range of $I$ is clear from the context, and sometimes we abuse notation and refer to both of them by $F$.
\end{definition}

\begin{ex}
\label{functionary-example}
Here is a functionary of degree 3:
$$ F \;=\; \langle Y 1234\rangle\,\langle Y 1256\rangle\,\langle Y 3489\rangle\;-\;\langle Y 1346\rangle\,\langle Y 1259\rangle\,\langle Y 2348\rangle $$ 
\end{ex}

\begin{definition}
\label{def pure}
A~functionary is \emph{pure} if the multisets of indices occurring in all monomials are the same. This multiset is called the \emph{type} of the functionary. The \emph{multiplicity} of an index~$i$ is denoted~$d_i(F)$. 
\end{definition}

\begin{ex}
The functionary from Example~\ref{functionary-example} is pure of type $112233445689$. Its multiplicities are $(d_1(F),\dots, d_9(F)) = (2,2,2,2,1,1,0,1,1)$. The functionary $F'= \langle 1345\rangle+\tfrac18\langle 2345\rangle$ is not pure.
\end{ex}

\begin{definition}
\label{nn:favorite_functionaries}
For $Y$ and $Z$ of $m=4$, we use a special shorthand for the following two-term pure quadratic functionary:
$$ \favorite{i~i'}{j~j'}{h~h'}{l}_{Y,Z} \;\;=\;\; \left\langle Y~ i~j~j'~l\right\rangle_Z\,\left\langle Y~ i'~h~h'~l\right\rangle_Z\;-\;
\left\langle Y~ i'~ j~j'~l\right\rangle_Z\,\left\langle Y~ i~h~h'~l\right\rangle_Z $$
where $Y$ and $Z$ are omitted from this notation whenever possible.
\end{definition}

\begin{ex}
$\favorite{1\,2}{4\,5}{7\,8}{9} \;=\; \langle1459\rangle\,\langle2789\rangle-\langle2459\rangle\,\langle1789\rangle$
\end{ex}

\begin{rmk} 
This functionary and similar ones have appeared in the literature before, for example in~\cite{lam2014totally} and in~\cite{agarwala2023cancellation}. It is most often applied with $i'=i+1$, $j'=j+1$, and $h'=h+1$.
\end{rmk}

\begin{lemma}
\label{obs:plucker_functionary}
The following identity is straightforward from the Pl\"ucker relations.
\begin{align*}
&\favorite{i~i'}{j~j'}{h~h'}{l}\;\;=\;\;-\,\favorite{i~i'}{h~h'}{j~j'}{l} \\\;\;=\;\;
&\favorite{j~j'}{h~h'}{i~i'}{l}\;\;=\;\;-\,\favorite{h~h'}{j~j'}{i~i'}{l} \\\;\;=\;\;
&\favorite{h~h'}{i~i'}{j~j'}{l}\;\;=\;\;-\,\favorite{j~j'}{i~i'}{h~h'}{l} 
\end{align*}
\end{lemma}

We conclude this introduction to the amplituhedron's coordinates by recalling some well-known properties of twistors. First, we expand a twistor coordinate of a point in the amplituhedron in terms of determinants in~$Z$ and Pl\"ucker coordinates of its preimage.

\begin{lemma}[Lemma 3.6, \cite{parisi2021m}]
\label{Cauchy-Binet}
\label{eq:s(J,I)}
\label{eq:Twistor Cauchy-Binet}
Consider a matrix $Z \in \Mat_{n \times (k+m)}$ and two representative matrices $C\in\Gr_{k,n}$ and $Y = CZ \in \Gr_{k,k+m}$. For every $I \in\binom{[n]}{m}$, the $I$th twistor coordinate is given by
$$ \langle Y~Z_I \rangle \;\;=
\sum_{J\in\binom{[n]}{k}} 
\langle C^J\rangle \,\langle Z_J~Z_I \rangle \;\;= 
\sum_{J\in\binom{[n] \setminus I}{k}} 
s(J,I)\,\langle C^J\rangle \,\langle Z_{I \cup J} \rangle
$$
where 
$s(J,I) \;=\; (-1)^{\displaystyle\;|\{(i,j) : i \in I,\, j \in J,\, i<j \}|}$ .
\end{lemma}

We remark that this expansion is based on the Cauchy-Binet formula. The determinants of $C$ are simply the Pl\"ucker coordinates $\langle C^J \rangle = P_J(C)$ in the Grassmannian $\Gr_{k,n}$. Note that $I$, $J$ and $I \cup J$ are unordered sets, and hence the rows of $Z_I$, $Z_J$ and $Z_{I \cup J}$ are taken in increasing order. In the typical use case, $C$~and~$Z$ are nonnegative matrices, so the sign of the $J$th term is $s(J,I)$.

\smallskip

In the rest of Section~\ref{sec:promotion} we focus on the case $m=4$ for the sake of simplicity. The following definition concerns twistors made of consecutive pairs.

\begin{definition}
Let $\Z : \Gr^{\geq}_{k,n} \to \Gr_{k,k+4}$ as before. The twistor coordinates of the form $\langle i ~ i+1 ~ j ~ j+1 \rangle$ or $\langle 1 ~ i ~ i+1 ~ n \rangle$ are named \emph{boundary twistors}.
\end{definition}

We later show that the points in the amplituhedron $\Ampl_{n,k,4}(Z)$ where a boundary twistor vanishes form the topological boundary of $\Ampl_{n,k,4}(Z)$, as previously conjectured, see e.g.~\cite{arkani2018unwinding}. The next lemma presents a well-known fact, that boundary twistors have a constant sign on the amplituhedron.

\begin{lemma}[e.g.~\cite{arkani2018unwinding}]
\label{obs:SA_and_bdry_twistors}
Let $Z \in \Mat^{>}_{n\times(k+4)}$.
For every $C \in \Gr_{k,n}^{\geq}$ 
\begin{enumerate}
\item 
$\langle \Z(C)~Z_I \rangle \geq 0$~ for every set of four indices of the form $I = \{i,i+1,j,j+1\}$.
\item 
$(-1)^k \langle \Z(C)~Z_I\rangle \geq 0$~ for every set of four indices of the form $I=\{1,i,i+1,n\}$.
\end{enumerate}
with equality if and only if the space $C$ contains a nonzero vector supported on the four indices $I$. 
\end{lemma}

\begin{ex}
In $\Ampl_{8,3,4}$ the twistors $\langle 1234 \rangle$ and $\langle 2367 \rangle$ are nonnegative, and $\langle 1348 \rangle$ is nonpositive.
\end{ex}

\begin{rmk}
As noted above, by writing $\langle \Z(C) ~ Z_I \rangle$ we slightly abuse notation, as this requires considering a specific nonnegative representative matrix $C \in \Mat^{\geq}_{k \times n}$ for $C \in \Gr_{k,n}^{\geq}$, and interpreting its image $\Z(C) = CZ \in \Mat_{k,k+m}$ as the corresponding representative matrix of $\Z(C) \in \Gr_{k,k+m}$. Here, for example, the lemma holds regardless of the choice of~$C$, though it is crucial to avoid a nonpositive representative matrix $C \in \Mat^{\leq}_{k \times n}$. We continue this abuse of notation in the rest of this section, where it is also crucial to consider the same representative~$C$ in all the twistors that occur in a certain functionary or claim.
\end{rmk}

\begin{proof}[Proof of Lemma~\ref{obs:SA_and_bdry_twistors}]
Both types of inequality follow from Lemma~\ref{eq:Twistor Cauchy-Binet}, noting that all the terms in the sum have the same sign as stated. There is equality $\langle i~i\pl1~j~j\pl1 \rangle=0$ exactly when all these terms vanish, which means that the Pl\"ucker coordinates $\langle C^J \rangle=0$ for every $k$ columns $J \subseteq [n] \setminus \{i,i\pl1,j,j\pl1\}$. Equivalently, these $n-4$ columns of $C$ are not of full rank $k$. Since $\mathrm{rank}\, C = k$, this condition amounts to the existence of a nonzero linear combination of $C$'s rows supported on $\{i,i\pl1,j,j\pl1\}$.
\end{proof}

For general $I$, it depends on the choice of~$Z$ which points $C \in \Gr_{k,n}^{\geq}$ satisfy $\langle \Z(C) ~Z_I \rangle = 0$. However, Lemma~\ref{obs:SA_and_bdry_twistors} shows that for some twistors this only depends on the preimage point~$C$ regardless of~$Z$. We denote these points as follows.

\begin{definition}
\label{sda}
For $n \geq k \geq 1$, let
\[\widetilde{\SA} =  {\widetilde{\SA}}^Z_{n,k,4} =\left\{C \in \Gr_{k,n} :\, \Z(C)~\text{is not of full rank, or }\langle\Z(C) ~ Z_{\{i,i\pl1,j,j\pl1\}}\rangle=0 \text{ for some } i \text{ and } j\right\}.\]
Write $\SA = {\SA}_{n,k,4} = \widetilde{\SA}\cap\Gr^{\geq}_{k,n}.$ Equivalently, using Lemma~\ref{obs:SA_and_bdry_twistors},
\begin{align*}
\SA 
&= \left\{C \in \Gr_{k,n}^{\geq} \;:\; \langle\Z(C) ~ Z_{\{i,i\pl1,j,j\pl1\}}\rangle=0 \text{ for some } i \text{ and } j\right\}\\
&= \left\{C \in \Gr_{k,n}^{\geq} \;:\; \text{There exist } i \text{ and } j\text{ such that }\langle C^J\rangle =0~\text{for every }J\subseteq[n]\setminus \{i,i\pl1,j,j\pl1\}\right\}\\
&= \left\{C \in \Gr_{k,n}^{\geq} \;:\; C \cap \mathrm{span}\left\{e_i,e_{i \pl 1},e_j,e_{j \pl 1}\right\} \neq \{0\} \text{ for some } i \text{ and } j \right\}
\end{align*}
where $e_i \in \mathbb{R}^n$ is the $i$th unit vector, and we consider only $i,j$ such that $\{i,i\pl1,j,j\pl1\}$ are four different elements. 
\end{definition}

\begin{rmk}\label{rmk:for_SA_usage}
Note that by Lemma \ref{Cauchy-Binet},
\[\left\{C \in \Gr_{k,n} \;:\; C \cap \mathrm{span}\left\{e_i,e_{i \pl 1},e_j,e_{j \pl 1}\right\} \neq \{0\} \text{ for some } i \text{ and } j \right\}\subseteq \widetilde{\SA}.\]
It follows from the definitions that $\Z(C)$ is not of full rank if and only if $C$ contains a vector which maps to $0$ under multiplication by $Z.$
It follows from the characterization by Pl\"ucker coordinates above, that for every $n$ and~$k$ the set $\SA$ is a union of positroid cells, is closed and is \emph{$Z$-independent}, unlike $\widetilde{\SA}$ which may depend on $Z$.  
Also by the Pl\"ucker characterization, for every positroid cell $S \subseteq \SA$ there is a boundary twistor $\langle i~i\pl1~j~j\pl1 \rangle$ that vanishes on~$\Z(S)$.

We later show that $\SA$ is the preimage of the boundary of the amplituhedron $\Ampl_{n,k,4}(Z)$.
\end{rmk}

\subsection{Promotion of Functionaries}
\label{sec:func promotion}

We turn to analyze how twistors, and thereby functionaries, transform under the matrix operations $\pre$, $\inc$, $x$, $y$, and $\emb$, defined in Section~\ref{operations}. 
Some proofs in this section are technical, and the reader may benefit from skipping them on the first reading and returning to them after seeing the applications in Section \ref{section promotion}.

We start with the embedding $\pre_i : \Gr_{k,N} \to \Gr_{k,N\cup\{i\}}$, which adds a column of zeros at some new index $i \not\in N$.

\begin{lemma}
\label{lem:effect of  pre twistor}
Let $C\in\Gr^{\geq}_{k,N}$ and let $Z\in\Mat^{>}_{(N\cup\{i\}) \times [k+4]}$ where $i\notin N$ and $k\geq 0$. For every $I = \{i_1,i_2,i_3,i_4\} \subseteq N$,
$$ \langle \Z_N(C)~(Z_N)_I \rangle \;=\; \langle \Z(\pre_i C)~Z_I\rangle$$
\end{lemma}

Here we use the notation $\Z_N : \Gr^{\geq}_{k,N} \to \Gr^{\geq}_{k,k+4}$ for the map induced from right multiplication $C \mapsto C Z_N$, where $Z_N \in \Mat^{>}_{N \times [k+4]}$ is obtained from $Z$ by deleting the $i$th row. Later, we write $\Z(C)$ instead of $\Z_N(C)$ if there is no ambiguity.

\begin{proof}
The lemma follows from $(\pre_i C)Z = CZ_N$ and $Z_I=(Z_N)_I$.
\end{proof}

This lemma and similar ones are used to track the signs of functionaries under the matrix operations. We first define an abbreviated notation for a functionary having a fixed sign at a point.

\begin{samepage}
\begin{definition}
\label{sign}
Let $C \in \Gr_{k,N}$ for a finite $N \subset \mathbb{N}$ and $k \geq 0$, and let $F(z_I : I \in \scalebox{0.8}{$\tbinom{N}{4}$})$ be a homogeneous polynomial. If $F\left(\langle \Z(C)~Z_I\rangle \right)$ has the same sign for every $Z \in \Mat^{>}_{N \times [k+4]}$ then we denote this sign by
$$ \SIGN F\left(\langle \Z(C)~Z_I\rangle
\right) \;\in\; \{-1,+1\} $$
and say that the functionary $F$ has a fixed sign at $C$.
Otherwise, if the sign depends on~$Z$, then we say that $F$ does not have a fixed sign at~$C$. Our usage of the notation $\SIGN F$ entails that $F$ has a fixed sign.
Note that this definition $Z$ is not a given matrix, and its dimensions are understood from the context.
\end{definition}
\end{samepage}

\begin{lemma}
\label{lem:effect of  pre}
Let $C\in\Gr^{\geq}_{k,N}$ and $i \not\in N$, and let $F(z_I : I \in \scalebox{0.8}{$\tbinom{N}{4}$})$ be a homogeneous polynomial. 
If $F(\langle \Z(C)~Z_I\rangle)$ has a fixed sign, then 
$$ \SIGN F(\langle \Z(\pre_i C)~Z_I\rangle) \;=\; \SIGN F(\langle \Z( C)~Z_I\rangle)$$
\end{lemma}
\begin{proof}
    Immediate from Lemma \ref{lem:effect of  pre twistor}.
\end{proof}

The next matrix operation is $\inc_{i} : \Gr_{k,N} \to \Gr_{k+1,N\cup\{i\}}$, which adds a unit vector at a new coordinate~$i \not\in N$. For a representative matrix $C \in \Mat^{\geq}_{[k]\times N}$, without loss of generality, we insert the $i$th unit vector as the last row. The resulting representative matrix is $\inc_{i;k+1}\,C \in \Mat^{\geq}_{[k+1]\times (N \cup \{i\})}$ where the columns after~$i$ are negated. 

The next lemma analyzes how the twistor coordinates of~$\Z(C)$ translate to those of $\Z(\inc_i\,C)$. This requires a certain ``projection'' of $Z \in \Mat^{>}_{(N \cup \{i\}) \times [k+4+1]}$ to another positive matrix $Z_{\neg i}^{\neg j} \in \Mat^{>}_{N \times [k+4]}$. 

\begin{lemma}
\label{lem:effect_of_inc twistor}
Let $C\in\Gr^{\geq}_{k,N}$ and $Z\in\Mat^{>}_{(N\cup\{i\}) \times [k+5]}$ where $i\notin N$ and $k\geq 0$, and let $j \in [k+5]$ such that $Z_i^j \neq 0$. For every $I = \{i_1,i_2,i_3,i_4\} \subseteq N$,
$$ \langle {\Z_{\neg i}^{\neg j}}(C) ~ (Z_{\neg i}^{\neg j})_{I}\rangle \;=\; (-1)^{|I \cap [i]|} \;\langle \Z(\inc_i\,C) ~ Z_I\rangle $$
where the matrix $Z_{\neg i}^{\neg j} \in \Mat^{>}_{N\times([k+5]\setminus\{j\})}$ is defined by
$$ (Z_{\neg i}^{\neg j})_p^q \;=\; (-1)^{\delta[p<i]+\delta[q<j]} \left(Z_p^q - \frac{Z_i^q Z_p^j}{Z_i^j}\right) \, (Z_i^j)^{\delta[q = j \pl 1]} $$
\end{lemma}

\begin{proof}
First, we interpret the formula defining~$Z_{\neg i}^{\neg j}$. As $Z_i^j \neq 0$, we subtract multiples of $Z_i$ from all other rows $Z_p$ to cancel $Z_p^j$ and make this column vanish for $p \in N$, and then erase the column $j$ and the row $i$ altogether. The sign in the beginning means that we negate all rows before~$i$, and again negate all columns before~$j$. Since $p \in N$ and $q \in [k+5]\setminus\{j\}$ the $i$th row and $j$th column of $Z$ are actually deleted. The last factor multiplies the arbitrarily chosen column $j \pl 1$ by the removed entry~$Z_i^j$. As an illustration, if $Z_i = (1,0,0,\dots,0)$ then we just delete the $i$th row and 1st column, and negate the rows before~$i$. This example may be regarded as a generic case, since $\Gr_{k+1,k+5}$ can always be rotated by composing a suitable $(k+5) \times (k+5)$ matrix on~$Z$, to turn $Z_i$ into a unit vector. As before, $\Z_{\neg i}^{\neg j}$ denotes the induced map from $\Gr_{k,N}^{\geq}$ to $\Gr_{k,k+4}$.

Examine the maximal determinants in the resulting matrix $Z_{\neg i}^{\neg j}$. Let $J \in \tbinom{N}{k+4}$ and compare the determinants $\langle (Z_{\neg i}^{\neg j})_J \rangle$ and $\langle Z_{J\cup\{i\}}\rangle$. The factors $(-1)^{\delta[p<i]+\delta[q<j]}$ contributes $(-1)^{j+|J\cap[i]|}$, while the column multiplied by $(Z_i^j)^{\delta[q=j \pl 1]}$ contributes $Z_i^j$. The remaining matrix $(Z_p^q - Z_i^q Z_p^j/Z_i^j)$ is exactly the $(i,j)$ minor of $Z_{J \cup \{i\}}$ after subtraction of $Z_i$ from other rows. On the other hand, this subtraction does not affect the determinant of $Z_{J \cup \{i\}}$, so the Laplace expansion by the column~$j$ expresses it as the same $(i,j)$ minor multiplied by $(-1)^{j+|J\cap[i]|}Z_i^j$. In conclusion, both maximal determinants are given by the same product, hence $\langle Z_{J \cup \{i\}} \rangle = \langle (Z_{\neg i}^{\neg j})_J \rangle$. 

The lemma follows by applying Lemma~\ref{Cauchy-Binet} back and forth. Note that all maximal determinants of $\inc_i\,C$ vanish unless one of their $k+1$ columns is~$C^i$. Therefore, we restrict the summation to terms of the form $J \cup \{i\}$.
\begin{align*}
\langle \Z(\inc_i(C)) ~ Z_I\rangle \;&=\; \sum_{J \in \binom{N \setminus I}{k}} s(J \cup \{i\},I) \; \langle (\inc_i\,C)^{J\cup\{i\}} \rangle \;\langle Z_{I \cup J \cup \{i\}} \rangle
\\[0.25em] \;&=\;  \sum_{J \in \binom{N \setminus I}{k}} s(J,I) \; (-1)^{|I\cap[i]|} \; \langle C^J \rangle \; \langle (Z^{\neg j}_{\neg i})_{I \cup J} \rangle 
\\[0.25em] \;&=\; (-1)^{|I \cap [i]|} \;\langle {\Z_{\neg i}^{\neg j}}(C) ~ (Z_{\neg i}^{\neg j})_{I}\rangle
\end{align*}
The second line is obtained by the definitions of the sign $s(J,I)$ and the map $\inc_i$, together with the identity of determinants shown above. Finally, note that this calculation holds in the case $k=0$ as well, when there is a single term $J = \varnothing$, and the $0 \times 0$ determinant $\langle C^J \rangle = 1$ by convention.  
\end{proof}

Lemma~\ref{lem:effect_of_inc twistor} implies the following useful result for functionaries.

\begin{lemma}
\label{lem:effect_of_inc}
Let $C\in\Gr^{\geq}_{k,N}$ and $i \not\in N$, and let $F(z_I : I \in \scalebox{0.8}{$\tbinom{N}{4}$})$ be a homogeneous polynomial. 
Let $F'(z_I)$ be the polynomial obtained from $F$ by the substitution $z_I\mapsto (-1)^{|I \cap [i]|} z_I$.
If~$F(\langle \Z(C)~Z_I\rangle)$ has a fixed sign, then 
$$\SIGN F'(\langle\Z(\inc_i C)~Z_I\rangle
) \;=\; \SIGN F(\langle \Z(C)~Z_I\rangle)$$
\end{lemma}

\begin{proof}
Let $s = \SIGN F(\langle \Z(C)~Z_I\rangle)$.
For any positive matrix $Z \in \Mat^{>}_{(N\cup\{i\}) \times [k+5]}$, let $j \in [k+5]$ be some index such that $Z_i^j \neq 0$. As shown in the proof of the previous lemma, $\langle Z_{J \cup \{i\}} \rangle = \langle (Z_{\neg i}^{\neg j})_J \rangle$ for all $J \in \tbinom{N}{k+4}$, hence $Z_{\neg i}^{\neg j} \in \Mat^{>}_{N \times [k+4]}$ is a positive matrix as well. Thus $F((-1)^{|I \cap [i]|} \langle\Z(\inc_i C)~Z_I\rangle) = F( \langle\Z_{\neg i}^{\neg j}(C)~(Z_{\neg i}^{\neg j})_I\rangle)$, and our assumption implies that the right hand side has the given fixed sign~$s$. 
Clearly $F((-1)^{|I \cap [i]|} \langle\Z(\inc_i C)~Z_I\rangle)=F'(\langle\Z(\inc_i C)~Z_I\rangle)$, so $F'$ has the fixed sign $s$.
\end{proof}

\begin{ex}
If $C \in \Gr^{\geq}_{k,1235678}$ is such that $\langle \Z(C)\,1\,3\,5\,6\rangle\,\langle \Z(C)\,2\,5\,6\,7\rangle > 0$ for all positive $7 \times (k+4)$ matrices~$Z$, then $\langle \Z(\inc_4C)\,1\,3\,5\,6\rangle\,\langle \Z(\inc_4C)\,2\,5\,6\,7\rangle < 0$ for all positive $8 \times (k+5)$ matrices~$Z$.
\end{ex}

We continue with the two matrix operations $x^k_i(t),\,y^k_i(t) : \Gr^{\geq}_{k,N} \to \Gr^{\geq}_{k,N}$, which add to some column a $t$-multiple of an adjacent column, usually applied with $t \in (0,\infty)$. Recall from Definition~\ref{def:x_i y_i} that these maps act as right multiplication $C \mapsto C\cdot[x_i(t)]$ by $N \times N$ matrices, and our notation $\smax N = \max N, \smin N = \min N$. These matrices are $[x_i(t)] = \Id_N + t \,\E_{i}^{i{\pl1}}$ and $ 
[y_i(t)]= \Id_N + t\,\E_{i{\pl1}}^{i}$, with the ``overflow'' exception that if $i= \smax N$ and $k$ is even then $t\,\E_{\max N}^{\min N}$ or $t\,\E_{\min N}^{\max N}$ is subtracted rather than added. The following lemma describes the effect of these transformations on the twistor coordinates. 

\begin{lemma}
\label{lem:effect_of_x_y twistor}
Let $C \in \Gr^{\geq}_{k,N}$, $Z \in \Mat^{>}_{N \times [k+4]}$, $I\in\tbinom{N}{4}$ and $i \in N$. 
\begin{enumerate}
    \item[(X)] 
    Let $s \in [0,\infty)$ let $Z' = [x_i(s)] Z$ and $C' = x_i(s)\,C$. Then 
    $$
    \left\langle \Z'(C) ~ Z'_I \right\rangle \;=\; 
    \begin{cases}
    \left\langle\Z(C') ~ Z_I \right\rangle + s\, \left\langle\Z(C') ~ Z_{I \cup \{i \pl 1\} \setminus \{i\}} \right\rangle &\text{if $i \in I, i \pl 1 \not \in I, i\neq\smax N$}\\
    \left\langle\Z(C') ~ Z_I \right\rangle + (-1)^{(k-1)}\,s\, \left\langle\Z(C') ~ Z_{I \cup \{i \pl 1\} \setminus \{i\}} \right\rangle &\text{if $i \in I, i \pl 1 \not \in I,i=\smax N$}\\
    \left\langle\Z(C') ~ Z_I \right\rangle &\text{else.}
    \end{cases}$$
    \item[(Y)]
    Let $t \in [0,\infty)$ let $Z' = [y_i(t)] Z$ and $C' = y_i(t)\,C$. Then 
    $$
    \left\langle \Z'(C) ~ Z'_I \right\rangle \;=\; 
    \begin{cases}
    \left\langle\Z(C') ~ Z_I \right\rangle + t\, \left\langle\Z(C') ~ Z_{I \cup \{i\} \setminus \{i \pl 1\}} \right\rangle &\text{if $i \not\in I, i \pl 1 \in I,i\neq\smax N$}\\
    \left\langle\Z(C') ~ Z_I \right\rangle + (-1)^{(k-1)}\,t\, \left\langle\Z(C') ~ Z_{I \cup \{i\} \setminus \{i \pl 1\}} \right\rangle &\text{if $i \not\in I, i \pl 1 \in I, i=\smax N$}\\
    \left\langle\Z(C') ~ Z_I \right\rangle &\text{else.}
    \end{cases}
    $$
\end{enumerate}
\end{lemma}

\begin{proof}
The positivity of $Z'$ in both cases follows from the positivity of $Z$ and that we have $s,t \geq 0$, for the same reasons that the nonnegativity of~$C$ is preserved under $x_i$ and $y_i$. Note that this argument relies on the sign coefficient being $(-1)^{k-1} = (-1)^{k+4-1}$ in the overflow case $i = \smax N$. Hence, the induced map $\Z'$ is well-defined from $\Gr^{\geq}_{k,N}$ to~$\Gr_{k,k+4}$ in both cases.

Consider the $x_i$ case first.
The first $k$ rows of the determinant on the left hand side of the equality are $\Z'(C) = C Z' = C\cdot[x_i(t)] Z = (x_i(t)\,C) \cdot Z = \Z(x_i(t)\,C) = \Z(C')$, by the associativity of matrix multiplication. Every row $j \neq i$ remains $Z'_j = Z_j$ while $Z'_i = Z_i \pm t Z_{i\pl1}$, with subtraction if and only if $k$ is even and $i = \max{N}$. If $i \not \in I$ then clearly $Z_I = Z_I'$ and the claim follows. If both $i\in I$ and $i\pl1 \in I$ then $Z_I$ and $Z_I'$ differ by a unimodular row operation, so $\langle Y ~ Z'_I \rangle = \langle Y ~ Z_I \rangle$. The remaining case that $i\in I$ and $i\pl1 \not\in I$ gives rise to the additional term $\pm t \, \langle\Z(x_i(t)\,C) ~ Z_{I \cup \{i \pl 1\} \setminus \{i\}} \rangle$ by the linearity of the determinant in each row. The argument for $y_i$ is analogous, where $Z_{i\pl1}' = Z_{i\pl1} \pm tZ_i$ and the exceptional case is that $i \not \in I$ and $i \pl 1 \in I$.
\end{proof}

\begin{ex}
Take $C \in \Gr^{\geq}_{k,n}$ and $Z \in \Mat^{>}_{n \times (k+4)}$ and $i=3$. Let $Z' = [x_3(t)] Z$ as in the lemma, and denote $Y = \Z'(C) = \Z(x_3(t)\,C)$.
\begin{align*}
& \langle Y\,Z'_2\,Z'_5\,Z'_6\,Z'_7\rangle \;=\; \langle Y\,Z_2\,Z_5\,Z_6\,Z_7\rangle \\
& \langle Y\,Z'_2\,Z'_3\,Z'_4\,Z'_7\rangle \;=\; \langle Y\,Z_2\,Z_3\,Z_4\,Z_7\rangle \\
& \langle Y\,Z'_2\,Z'_3\,Z'_6\,Z'_7\rangle \;=\; \langle Y\,Z_2\,Z_3\,Z_6\,Z_7\rangle + t\; \langle Y\,Z_2\,Z_4\,Z_6\,Z_7\rangle \\
& \langle Y\,Z'_2\,Z'_4\,Z'_6\,Z'_7\rangle \;=\; \langle Y\,Z_2\,Z_4\,Z_6\,Z_7\rangle 
\end{align*}
\end{ex}

As in the previous matrix operations, we state the following useful corollary for functionaries. Unlike the previous cases, here the resulting functionary computed at $Y = \Z(x_i(t)C)$ depends on the preimages in the Grassmannian via the real parameter~$t$. In some cases where we use this lemma, $t$~is expressible using $\langle Y\,Z_I\rangle$ twistors as well.

\begin{samepage}
\begin{lemma}
\label{lem:effect_of_x_y}
Let $C\in\Gr^{\geq}_{k,N}$, $i \in N$, 
and let $F(z_I : I \in \scalebox{0.8}{$\tbinom{N}{4}$})$ be a homogeneous polynomial such that $F(\langle \Z(C) ~ Z_I \rangle)$ has a fixed sign.

\begin{enumerate}
    \item[(X)]
    Let $s\in(0,\infty),\,C'=x_i(s)\,C$ and let $F'$ be the polynomial obtained from $F$ by the substitution
    $$
    z_I \;\mapsto\;
    \begin{cases}
    z_I + s\, z_{I \cup \{i \pl 1\} \setminus \{i\}} &\text{if $i \in I, i \pl 1 \not \in I, i\neq\smax N$}\\
    z_I + (-1)^{(k-1)}\,s\, z_{I \cup \{i \pl 1\} \setminus \{i\}} &\text{if $i \in I, i \pl 1 \not \in I,i=\smax N$}\\
    z_I &\text{else.}
    \end{cases}
    $$
    Then 
    $$\SIGN F'(\langle\Z(C')\;Z_I\rangle)\;=\;\SIGN F(\langle \Z(C) ~ Z_I \rangle)$$ 
    Note that if the functionary $F$ has a representation in which every twistor that contains $i$ in its index set also contains $i\pl 1$ then $F'=F$.
    \item[(Y)]
    Let $t\in(0,\infty),\,C''=y_i(t)\,C$ and let $F''$ be the polynomial obtained from $F$ by the substitution
    $$
    z_I \;\mapsto\;
    \begin{cases}
    z_I + t\, z_{I \cup \{i\} \setminus \{i \pl 1\}} &\text{if $i \not\in I, i \pl 1 \in I,i\neq\smax N$}\\
    z_I + (-1)^{(k-1)}\,t\, z_{I \cup \{i\} \setminus \{i \pl 1\}} &\text{if $i \not\in I, i \pl 1 \in I, i=\smax N$}\\
    z_I &\text{else.}
    \end{cases}
    $$
    Then 
    $$\SIGN F''(\langle\Z(C'')\;Z_I\rangle)\;=\;\SIGN F(\langle \Z(C) ~ Z_I \rangle)$$
    Note that if the functionary $F$ has a representation in which every twistor that contains $i\pl 1$ in its index set also contains $i$ then $F''=F$.
\end{enumerate}

\end{lemma}
\begin{proof}
    Immediate from Lemma \ref{lem:effect_of_x_y twistor}.
\end{proof}
\end{samepage}

\begin{rmk}
\label{rmk:odd_vs_even}
For every even $m$, twistors and functionaries of $\Ampl_{n,k,m}(Z)$ transform similarly to the above lemmas, under the application of the matrix operations $\pre$, $\inc$, $x$ and~$y$ to a preimage $C \in \Gr^{\geq}_{k,n}$ under $Z \in \Mat^{>}_{n \times (k+m)}$. The situation is different for odd~$m$ only in the overflow case, where $x_n(t)$ or $y_n(t)$ acts on the last and first columns. While the matrix $[x^k_n(t)] = \Id_n + (-1)^{(k-1)} \,t\, \E_1^n$ preserves the nonnegativity of~$C$, a different matrix $[x^k_n(-t)] = \Id_n + (-1)^{(k+m-1)} \,t\, \E_1^n$ is required in order to preserve the positivity of~$Z$.
\end{rmk}

We proceed to the matrix operation $\embilr{i,l,r}= \emb_{i,l,r}(t_1,\dots,t_l,s_1,\dots,s_r) : \Gr^{\geq}_{k-1,N} \to \Gr^{\geq}_{k,N \cup \{i\}}$, given in Definition~\ref{def::emb}. Since $\emb_{i,l,r}$ is a composition of a sequence of $x_{i \pl j}$ and $y_{i \mi j}$ and~$\inc_i$, it transforms functionaries as in the above analysis of these operations. Thus, the resulting functionaries depend on the real parameters $t_1,\dots,t_l,s_1,\dots,s_r$. We first analyze the effect of $\emb_{i,l,r}$ for $m=4$ and any $l$ and~$r$, though later we focus on the case $l+r=4$.

\begin{lemma}
\label{lem effect of emb}
\label{prop:general_promotion_and_when_preserved}
Let $i \not\in N$ for some finite $N \subset \mathbb{N}$, let $l \geq 0$ and $r \geq 0$ be such that $l+r \leq |N|$, let $(\textbf{t},\textbf{s})=(t_1,\dots,t_l,s_1,\dots,s_r) \in (0,\infty)^{l+r}$, and let $F(z_I : I \in \scalebox{0.8}{$\tbinom{N}{4}$})$ be a homogeneous polynomial. 

We denote by $F'(z_I:I\in\scalebox{0.8}{$\tbinom{N\cup\{i\}}{4}$})$ the polynomial obtained from $F$ by the following procedure: first substitute $x_I \mapsto (-1)^{|I \cap [i]|} \langle Y~Z'_I\rangle$ where $Y \in \Mat_{k,k+4}$ and $Z' \in \Mat^{>}_{(N \cup \{i\}) \times [k+4]}$, and then expand multilinearly the twistors $\langle Y~Z'_I\rangle$ in terms of $\langle Y~Z_I\rangle$ 
where $$ Z' \;=\; [x_{i}(s_1)]\cdot[x_{i \pl 1}(s_2)]\cdots[x_{i \pl (r-1)}(s_r)]\cdot[y_{i \mi 1}(t_1)]\cdot[y_{i \mi 2}(t_2)]\cdots[y_{i \mi l}(t_l)]\cdot Z $$
so that the $F'(\langle Y~Z_I\rangle:I\in\scalebox{0.8}{$\tbinom{N\cup\{i\}}{4}$})$ is the resulting expression. Note that the rows of the matrices $Y,Z',Z$ are regarded here as formal variables, and that the expansion is carried out by iterating the substitutions from Lemma~\ref{lem:effect_of_x_y twistor}. Note also that $F'$ does not depend on the parameter $k$ except for occasional $(-1)^k$.

Then, for every $C\in\Gr^{\geq}_{k-1,N}$ and $C'=\embilr{i,l,r}\, C$, if $F(\langle \Z_N(C)~Z_I\rangle)$ has a fixed sign, then 
$$\SIGN F'(\langle \Z(C')~Z_I\rangle) \;=\; \SIGN F(\langle \Z_N(C)~Z_I\rangle)$$
\end{lemma}

\begin{proof}
Let $Z,Z' \in \Mat^{>}_{(N \cup \{i\}) \times [k+4]}$ be related as in the definition of~$F'$. Note that the positivity of $Z$ implies the positivity of $Z'$, since multiplying a positive matrix from the left by any single $[x_j(t)]$ or $[y_j(s)]$ preserve positivity, as in the proof of Lemma~\ref{lem:effect_of_x_y twistor}.
By the associativity of matrix multiplication, and the definition of $\emb_{i,l,r}$, 
we get:
\begin{align*}
    \Z'(\inc_i\,C)&\;=\;
    (\inc_i\,C)Z'
    \\&\;=\;(\inc_i\,C) [x_{i}(s_1)]\cdot[x_{i \pl 1}(s_2)]\cdots[x_{i \pl (r-1)}(s_r)]\cdot[y_{i \mi 1}(t_1)]\cdot[y_{i \mi 2}(t_2)]\cdots[y_{i \mi l}(t_l)]\cdot Z
    \\&\;=\; (y_{i\mi l}(t_l)\circ\cdots\circ y_{i \mi 1}(t_1) \,\circ\, x_{i \pl (r-1)}(s_r)\circ\cdots\circ x_{i \pl 1}(s_2)\circ x_i(s_1) \,\circ\, \inc_{i} (C)) \cdot Z
    \\&\;=\;\Z(\embilr{i,l,r}(C))
\end{align*}
and therefore the twistors of $\Z(C')$ in the lemma satisfy $$\langle\Z(C')~Z'_I\rangle \;=\; \langle\Z'(\inc_i\,C)~Z'_I\rangle.$$ 
Therefore,
\begin{equation*}
F\left((-1)^{|I \cap [i]|} \langle\Z\left(C'\right)~Z'_I\rangle \right) \;=\; 
F\left((-1)^{|I \cap [i]|} \langle\Z'(\inc_i\,C)~Z'_I\rangle \right)
\end{equation*}
Combining this equation with Lemma~\ref{lem:effect_of_inc} and the definition of~$F'$, and letting~$s$ be the fixed sign of $F(\langle \Z(C)~Z_I\rangle)$, or equivalently for this case $F(\langle \Z'_N(C)~Z'_I\rangle)$, we find that $F'(\langle \Z(C')~Z_I\rangle)$ has the fixed sign $s$, independently of the arbitrarily chosen $Z$.
\end{proof}

\begin{rmk}\label{rmk:Z' as Z comb}
For later reference, we spell out explicitly the relations of rows of~$Z'$ to rows of~$Z$ in the setting of Lemma \ref{lem effect of emb}. These are used below when computing the functionaries that arise from the lemma. For a general $N \times [k+4]$ matrix $A$, left multiplication by $x_i$ and $y_i$ act on the rows as follows:
\begin{align*}
([x_i(t)] A)_{j} \;&=\; ((\Id_N + t \,\E_{i}^{i{\pl1}}) A)_{j} \;=\; (A+t \E_{i}^{i{\pl1}} A)_{j} \;=\; 
\begin{cases}
A_{i}+tA_{i\pl 1} & j=i\\
A_{j} & j\neq i,
\end{cases} \\
([y_i(t)] A)_{j} \;&=\; ((\Id_N + t\,\E_{i{\pl1}}^{i}) A)_{j} \;=\; (A + t\,\E_{i{\pl1}}^{i} A)_{j} \;=\; \begin{cases}
A_{i\pl1}+tA_{i} & j=i\pl1\\
A_{j} & j\neq i\pl1,
\end{cases}
\end{align*}
Therefore, when $[x_i(s)]$ acts on $Z$, it adds the row $s Z_{i\pl1}$ to the row $Z_{i}$, and $[y_i(t)]$ adds the row $t Z_{i}$ to the row $Z_{i\pl1}$. In the overflow case $i = \max N$, addition is replaced by subtraction if $k$ is even.
Therefore, after applying $[x_{i}(s_1)]\cdot[x_{i \pl 1}(s_2)]\cdots[x_{i \pl (r-1)}(s_r)]\cdot[y_{i \mi 1}(t_1)]\cdot[y_{i \mi 2}(t_2)]\cdots[y_{i \mi l}(t_l)]$ on~$Z$,
\begin{align*}
& Z'_j \;=\; Z_j  &&\!\!\!\!\!\!\!\!\!\!\!\! j \not \in \{i \mi (l-1),\dots,i \pl (r-1)\} 
\\& Z'_{i \mi j} \;=\; Z_{i \mi j} \;\pm\; t_{j+1}\, Z'_{i \mi (j+1)} &&\!\!\!\!\!\!\!\!\!\!\!\! 0 < j < l
\\& Z'_{i \pl j} \;=\; Z_{i \pl j} \;\pm\; s_{j+1}\, Z'_{i \pl (j+1)} &&\!\!\!\!\!\!\!\!\!\!\!\! 0 < j < r
\\& Z'_i \;=\; Z_i \;\pm\; t_1\, Z'_{i\mi 1} \;\pm\; s_1\, Z'_{i \pl 1} &&
\end{align*}
where the $\pm$ sign is $+$ in the consecutive case, and $(-1)^{k-1}$ in the overflow case, i.e., when the two row indices are $\min (N \cup \{i\})$ and $\max (N \cup \{i\})$. The formula for $Z'_{i \mi j}$ follows from the application of $[y_{i \mi (j+1)}(t_{j+1})]$ on~$Z$. The added term is specifically $Z'_{i \mi (j+1)}$ rather than $Z_{i \mi (j+1)}$ because of a previous application of $[y_{i \mi (j+2)}(t_{j+2})]$. The formula for~$Z'_{i \pl j}$ similarly follows from $[x_{i \pl j}(s_{j+1})]$. The case of $Z'_i$ follows from both $[x_i(s_1)]$ and~$[y_{i \mi 1}(t_1)]$. 

Iterating these recursive relations, we expand every row $Z_j'$ as a linear combination of the rows of~$Z$ as follows:
\begin{align*}
& Z'_{i \mi l} \;=\; Z_{i \mi l} \\
& Z'_{i \mi (l-1)} \;=\; Z_{i \mi (l-1)} \;\pm\; t_l \, Z_{i \mi l} \\
& Z'_{i \mi (l-2)} \;=\; Z_{i \mi (l-2)} \;\pm\; t_{l-1} \, Z_{i \mi (l-1)} \;\pm\; t_{l-1} t_l \, Z_{i \mi l} \\ 
& \;\;\;\;\;\;\vdots \\ 
& Z'_{i \mi 1} \;=\; Z_{i \mi 1} \;\pm\; t_2 \, Z_{i \mi 2} \;\pm\; t_2t_3 \, Z_{i \mi 3} \;\pm \dots \pm\; (t_2t_3 \cdots t_l) \, Z_{i \mi l}
\end{align*}
where the sign of the term $\pm Z_{p}$ in the expansion of $Z_{q}'$ is $+$ unless $k$ is even and $p>q$. Similarly,
\begin{align*}
& Z'_{i \pl r} \;=\; Z_{i \pl r} \\
& Z'_{i \pl (r-1)} \;=\; Z_{i \pl (r-1)} \;\pm\; s_r \, Z_{i \pl r} \\
& Z'_{i \pl (r-2)} \;=\; Z_{i \pl (r-2)} \;\pm\; s_{r-1} \, Z_{i \pl (r-1)} \;\pm\; s_{r-1} s_r \, Z_{i \pl r} \\ 
& \;\;\;\;\;\;\vdots \\ 
& Z'_{i \pl 1} \;=\; Z_{i \pl 1} \;\pm\; s_2 \, Z_{i \pl 2} \;\pm\; s_2s_3 \, Z_{i \pl 3} \;\pm \dots \pm\; (s_2s_3 \cdots s_r) \, Z_{i \pl r}
\end{align*}
where the sign of the term $\pm Z_{p}$ in the expansion of $Z_{q}'$ is $+$ unless $k$ is even and $p<q$. Finally,
\begin{align*}
Z'_i \;=\; Z_i & \;\pm\; t_1 \, Z_{i \mi 1} \;\pm\; t_1t_2\, Z_{i \mi 2} \;\pm\; \dots \;\pm\; (t_1t_2\cdots t_l)\, Z_{i \mi l} \\ 
& \;\pm\; s_1 Z_{i \pl 1} \;\pm\; s_1s_2\, Z_{i \pl 2} \;\pm\; \dots \;\pm\; (s_1s_2\cdots s_r)\, Z_{i \pl r} \;\;\;\;\;\;
\end{align*}
where $\pm$ is $(-1)^{k-1}$ for terms with index $i \mi j > i$ or $i \pl j < i$, and $+$ otherwise.

Expanding each twistor using these combinations for its four $Z_I'$ rows, we obtain a functionary of $\embilr{i,l,r}(C)$ with respect to~$Z$ as claimed above.
\end{rmk}

\begin{rmk}
\label{rmk-truncations}
Let $I = \{i_1,i_2,i_3,i_4\} \subset N \cup \{i\}$. In the setting of Lemma~\ref{lem effect of emb}, we point to some cases in which the expansion of $\langle Y ~ Z'_I \rangle$ in the twistors $\langle Y ~ Z_J \rangle$ from Remark~\ref{rmk:Z' as Z comb} can be simplified.
\begin{enumerate}
\item
If $\{i \mi p,i \mi q\} \subset I$ for $0 \leq p < q \leq l$, truncate the expansion of $Z'_{i \mi p}$ from the term $Z_{i \mi q}$.
\item
If $\{i \pl p,i \pl q\} \subset I$ for $0 \leq p < q \leq r$, truncate the expansion of $Z'_{i \pl p}$ from the term $Z_{i \pl q}$.
\item
If $\{i \mi p,i \mi (p+1)\}\subset I$ for $0<p<l$, then replace $Z_{i \mi p}'$ by $Z_{i \mi p}$.
\item
If $\{i \pl p,i \pl (p+1)\}\subset I$ for $0<p<r$, then replace $Z_{i \pl p}'$ by $Z_{i \pl p}$.
\item
If $\{i \mi l,\dots,i, \dots,i \pl r\} \setminus I = \{i \mi l',\dots,i, \dots,i \pl r'\}$ for some $l',r' \geq 0$ then $\langle Y~ Z'_I \rangle = \langle Y~ Z_I \rangle$.
\end{enumerate}
\begin{proof}
The first two cases follow from the general formulas for $Z_j'$ in Remark~\ref{rmk:Z' as Z comb}, noting that twistors where $Z'_{i \mi q}$ or $Z'_{i \pl q}$ appears twice must vanish. The remaining are special cases.
\end{proof}
\end{rmk}

The coefficients of the polynomial $F'(z_I)$ from Lemma~\ref{lem effect of emb} depend on the variables $t_1,\dots, t_l,s_1, \dots, s_r$ as a homogeneous polynomial in $z_I$. 
A~key step in our analysis of the transformation of functionaries under $\emb_{i,l,r}$ is that when $l+r \leq 4$ these variables can generically be recovered from the twistor coordinates of~$Y$, and thus we can express the entire functionary $F'(\langle \Z(C')~Z_I\rangle)$ from Lemma~\ref{lem effect of emb} only with twistors, eliminating the real numbers $(\textbf{t},\textbf{s})$.

\begin{lemma}
\label{lem:twistors_as_emb_params 2}
Let $C = \embilr{i,l,r}\, C' \in \Gr^{\geq}_{k,N}$ for some $(\textbf{t},\textbf{s})=(t_1,\dots,t_l,s_1,\dots,s_r) \in \mathbb{R}^{l+r}$ and some $C' \in \Gr^{\geq}_{k-1,N \setminus \{i\}}$ such that $i \in N$ and $l+r\leq4$. Let $J = \{j_1,j_2,j_3,j_4,j_5\} \subseteq N$ be such that $\{i \mi l, \dots, i, \dots, i \pl r\} \subseteq J$. Then, for every $Z \in \Mat^{>}_{N \times [k+4]}$, 
\begin{align*}
t_{j} \;&=\; \frac{-\left\langle CZ ~ Z_{J \setminus \{i \mi j\}} \right\rangle}{\left\langle CZ ~ Z_{J \setminus \{i \mi (j-1)\}} \right\rangle} \; (-1)^{(k-1)\delta[i \mi j = \smax N]} && 1 \leq j \leq l \\
s_{j} \;&=\; \frac{-\left\langle CZ ~ Z_{J \setminus \{i \pl j\}} \right\rangle}{\left\langle CZ ~ Z_{J \setminus \{i \pl (j-1)\}} \right\rangle} \; (-1)^{(k-1)\delta[i \pl j = \smin N]} && 1 \leq j \leq r
\end{align*}
as long as the denominators are nonzero.
\end{lemma}

These formulas follow by comparing the two determinants, noting that one row of $CZ$ gives essentially $Z_{i \mi (j-1)} + t_j \,Z_{i \mi j}$. We omit further details, and instead state the following lemma, which formulates the relation between variables and twistors suitably for our purposes.

\begin{lemma}
\label{lem:twistors_as_emb_params}
Let $C = \embilr{i,l,4-l}\, C' \in \Gr^{\geq}_{k,N}$ for some $C' \in \Gr^{\geq}_{k-1,N \setminus \{i\}}$ where $i \in N$ and $(\textbf{t},\textbf{s})=(t_1,\dots,t_l,s_1,\dots,s_{4-l}) \in (0,\infty)^4$. Denote $J = \{j_1,j_2,j_3,j_4,j_5\} = \{i \mi l, \dots, i, \dots, i \pl (4-l)\}$ in this order. If $j_1<j_5$ then for every $Z \in \Mat^{>}_{N \times [k+4]}$ the vector
$$ \left( +\langle \Z(C)~Z_{J \setminus \{j_1\}} \rangle,\, -\langle \Z(C)~Z_{J \setminus \{j_2\}} \rangle,\, +\langle \Z(C)~Z_{J \setminus \{j_3\}} \rangle, -\langle \Z(C)~Z_{J \setminus \{j_4\}} \rangle, +\langle \Z(C)~Z_{J \setminus \{j_5\}} \rangle \right) $$
is a nonnegative scalar multiple of the vector
$$ \Big((t_1t_2\cdots t_l),\, \dots,\, t_1t_2,\, t_1,\, 1,\, s_1,\, s_1s_2,\, \dots,\, (s_1s_2\cdots s_{4-l}) \Big) $$
If $j_5<j_1$ then the same holds where each $\pm \langle \Z(C)~Z_{J \setminus \{j_h\}} \rangle$ is multiplied by $(-1)^{k \, |\{\smin N, \dots, j_5\} \setminus \{j_h\}|}$.
\end{lemma}
{The proof of Lemma \ref{lem:twistors_as_emb_params} appears later, after Lemma \ref{obs:solve_params_5}.}

\begin{ex} 
\label{4567}
On the image of $\Z \circ \emb_{5,2,2}(t_1,t_2,s_1,s_2)$ in $\Ampl_{7,k,4}(Z)$,
$$ (+\langle4567\rangle,-\langle3567\rangle,+\langle3467\rangle,-\langle3457\rangle,+\langle3456\rangle) \;\;\propto\;\; (t_1t_2,t_1,1,s_1,s_1s_2)$$
\end{ex}

\begin{ex} 
\label{1267}
On the image of $\Z \circ \emb_{1,3,1}(t_1,t_2,t_3,s_1)$ in $\Ampl_{7,k,4}(Z)$,
$$ (+\langle1267\rangle,-\langle1257\rangle,+\langle1256\rangle,-(-1)^k\langle2567\rangle,+(-1)^k\langle1567\rangle) \;\;\propto\;\; (t_1t_2t_3,t_1t_2,t_1,1,s_1)$$
\end{ex}

Specializing Lemma~\ref{lem:twistors_as_emb_params} even to one coordinate, it already shows that certain twistors are either nonnegative or nonpositive on some parts of the amplituhedron, such as $\langle3567\rangle$, $\langle3457\rangle$, $\langle1257\rangle$, and $\langle2567\rangle$ in the above examples. The other twistors in these examples have a constant sign on the amplituhedron by Lemma~\ref{obs:SA_and_bdry_twistors}. We remark that Lemma~\ref{lem:twistors_as_emb_params} easily generalizes to all even~$m$, with $\emb_{i,l,m-l}$ and $|J|=m+1$, and also to odd $m$ in the case $j_1 < j_{m+1}$. 

\begin{lemma}
\label{obs:solve_params_5}
Let $k,m \geq 1$ and let $Y \in \Mat^{\ast}_{k \times (k+m)}$ and $Z \in \Mat^{\ast}_{(m+1) \times (k+m)}$ be two matrices of full rank. If at least one of the $m+1$ determinants
$$\langle Y \,Z_2\,Z_3\,\cdots\, Z_{m+1}\rangle,~\langle Y\, Z_1\,Z_3\,\cdots\, Z_{m+1}\rangle,~\ldots, ~\langle Y \,Z_1\, Z_2\,\cdots\, Z_{m}\rangle $$ 
is nonzero, then $\Span(Y_1,\ldots, Y_k) \cap \Span(Z_1,\ldots, Z_{m+1})$ is a one-dimensional space, spanned by 
\begin{equation*}
\sum_{j=1}^{m+1} (-1)^{j-1} \,\langle Y\, Z_1\,\cdots\,Z_{j-1}\,Z_{j+1}\,\cdots\,Z_{m+1} \rangle\, Z_j 
\end{equation*}
\end{lemma}

\begin{proof}
Since one of the $m+1$ determinants is nonzero, the space $\Span(Y_1,\dots,Y_k,Z_1,\dots,Z_{m+1})$ is of full dimension~$k+m$. The matrices $Y$ and $Z$ have full ranks $k$ and $m+1$ respectively, and it follows that the subspace $\Span(Y_1,\ldots, Y_k) \cap \Span(Z_1,\ldots, Z_{m+1})$ is exactly one-dimensional. Consider a nonzero vector in this intersection, which is unique up to scaling:
$$ \sum_{i=1}^{k}a_i Y_i\;=\; \sum_{j=1}^{m+1}b_j Z_j $$
{Let $q\in[m+1]$ be such that $\langle Y\, Z_1,\cdots\,Z_{q-1}\,Z_{q+1}\,\cdots\,Z_{m+1}\rangle\neq0$, as we assumed to exist.}
Observe that $b_q\neq 0,$ since otherwise $\Span(Y_1,\ldots, Y_k) \cap \Span(Z_1,\ldots,\,Z_{q-1}\,Z_{q+1}\,\ldots\ Z_{m+1})\neq \{0\}.$
For {an arbitrary $p \in [m+1]$}, apply to both sides the linear functional $f_{pq}(u) = \langle Y \, Z_{[m+1] \setminus \{p,q\}}\,u\rangle$. Since $f_{pq}(Y_i) = 0$ for every~$i \in [k]$, and $f_{pq}(Z_j)=0$ for every $j \not\in \{p,q\}$,
$$ 0 \;=\; b_p \, f_{pq}(Z_p) + b_q f_{pq}(Z_q).$$
Hence, {the coefficients $b_1,\dots,b_{m+1}$ satisfy that for any $p$,
$$ b_p \;=\; -\frac{f_{pq}(Z_q)}{f_{pq}(Z_p)} b_q \;=\; -\frac{\langle Y \, Z_{[m+1] \setminus \{p,q\}}\,Z_q\rangle}{\langle Y \, Z_{[m+1] \setminus \{p,q\}}\,Z_p\rangle} b_q \;=\; (-1)^{p-q}\, \frac{\langle Y \, Z_{[m+1] \setminus \{p\}}\rangle}{\langle Y \, Z_{[m+1] \setminus \{q\}}\rangle} b_q. $$
By multiplying the resulting vector by $\langle Y \, Z_{[m+1] \setminus \{q\}}\rangle$, which is nonzero, the claim follows.
}
\end{proof}

\begin{proof}[Proof of Lemma~\ref{lem:twistors_as_emb_params}]
First consider the case $j_1<j_5$. The point $C \in \Gr_{k,N}$ is as usual regarded as a nonnegative representative $k \times N$ matrix. Since~$C$ arises from $\emb_{i,l,4-l}$ it contains a row $C_k$ with the five cyclically consecutive nonzero entries in the index set $J$,
$$ C_k^J \;=\; \Big((t_1t_2\cdots t_l),\, \dots,\, t_1t_2,\, t_1,\, 1,\, s_1,\, s_1s_2,\, \dots,\, (s_1s_2\cdots s_{4-l}) \Big), $$
with all other entries being zero.
Therefore, after right multiplication by $Z$, the image $Y = \Z(C)$ contains the corresponding row
$$ Y_k \;=\; (CZ)_k \;=\; C_k^{j_1} Z_{j_1} + C_k^{j_2}Z_{j_2} + C_k^{j_3}Z_{j_3} + C_k^{j_4}Z_{j_4} + C_k^{j_5}Z_{j_5} \;\in\; \mathbb{R}^{k+4} $$
{where $(Z_i)_{i\in N}$ are the rows of $Z$. This shows that $Y_k\in \Span(Y_1,\ldots, Y_k) \cap \Span(Z_{j_1},\ldots, Z_{j_5})$, and it is nonzero since $Y$ is in the Grassmannian.}

The $5 \times (k+4)$ matrix $Z_J$ has a full rank by the positivity of~$Z$, and so does the $k \times (k+4)$ matrix~$Y$. The vector $Y_k$ is proportional to the combination of rows $Z_i$ provided by Lemma~\ref{obs:solve_params_5}. Therefore, the five twistors 
$$ \Big({+}\langle Y\, j_2 \, j_3 \, j_4 \, j_5 \rangle_Z,\, -\langle Y\, j_1 \, j_3 \, j_4 \, j_5 \rangle_Z,\, +\langle Y\, j_1 \, j_2 \, j_4 \, j_5 \rangle_Z,\, -\langle Y\, j_1 \, j_2 \, j_3 \, j_5 \rangle_Z,\, +\langle Y\, j_1 \, j_2 \, j_3 \, j_4 \rangle_Z\Big) $$
are proportional to the coefficients $C_k^{j_1},\dots,C_k^{j_5}$. Even if these twistors all vanish, then the required claim holds as well, where the proportion equals zero. Otherwise, the proportion is positive because the first term $(t_1\cdots t_l) >0$ while the first twistor $\langle Y\, j_2 \, j_3 \, j_4 \, j_5 \rangle$ is nonnegative by Lemma~\ref{obs:SA_and_bdry_twistors}. 

The overflow case $j_5<j_1$ is similar with some sign adjustments. In exactly one of the two intervals $J' = \{j_1, \dots, \smax N\}$ or $J'' = \{\smin N ,\dots, j_5\}$, the coefficients $C_k^{j_h}$ gain a factor of $(-1)^{k-1}$ due to an $x_j$ or $y_j$ operation. Thus, by Lemma~\ref{obs:solve_params_5}, the vector $(C_k^{j_1},\dots,C_k^{j_5})$ is proportional to the same five twistors with some $\pm$ signs, that still alternate within each of the two parts. The row reordering that rewrites $\langle Y\, j_1 \, \cdots, \, j_{h-1} \, j_{h+1} \, \cdots \, j_5 \rangle$ as $\langle Y\, J \setminus \{j_h\} \rangle$ contributes another $(-1)$ factor to the twistors in either $J'$ or~$J''$, depending on parity. Note that the parity of $|J''\setminus\{j_h\}|$ detects exactly whether $j_h \in J'$ or~$J''$. In conclusion, the ratio $(-1)^{h-1}(-1)^{k|J''\setminus\{j_h\}|} \langle Y\, J \setminus \{j_h\} \rangle \,/\, C_k^{j_h}$ is equal to a single proportionality constant for all $h \in \{1,2,3,4,5\}$, which is again nonnegative by the two cases of Lemma~\ref{obs:SA_and_bdry_twistors}. 
\end{proof}

We return to the setting of Lemma~\ref{lem effect of emb}, which analyzes the evolution of a functionary~$F$ under the matrix operation $\emb_{i,l,r}(t_1,\dots,s_r)$. In the case $l+r = 4$, the output matrix has a row with five nonzero entries, and one can apply Lemma~\ref{lem:twistors_as_emb_params} in order to decouple the resulting functionary from the apparent dependence on the parameters $t_1,\dots,s_{4-l}$. 
The following lemma gives formulas that generically let us express this functionary only in terms of twistors of the image under~$\Z$.

\begin{lemma}
\label{eq:Z'_(i-j)_special}
\label{twistor coefficients}
Let $Y = \Z(\embilr{i,l,4-l}\;C)$ for some $C \in \Gr^{\geq}_{k-1,N \setminus \{i\}}$ where $i \not \in N$, $Z \in \Mat^{>}_{N \times [k+4]}$ and $(\textbf{t},\textbf{s})=(t_1,\dots,t_l,s_1,\dots,s_{4-l}) \in (0,\infty)^4$. Then, the rows of the matrix
$$ Z' \;=\; [x_{i}(s_1)]\cdot[x_{i \pl 1}(s_2)]\cdots[x_{i \pl (3-l)}(s_{4-l})]\cdot[y_{i \mi 1}(t_1)]\cdot[y_{i \mi 2}(t_2)]\cdots[y_{i \mi l}(t_l)]\cdot Z $$
satisfy the following formulas, regardless of the nonnegative variables $(t_1,\dots,t_l,s_1,\dots,s_{4-l})$.
\begin{enumerate}
\item 
If $j = i \mi l$ or $j = i \pl (4-l)$ or $j \not\in J = \{i \mi l, \dots, i \pl (4-l)\}$ then $ Z'_j = Z_j $.
\item
If $i \mi l \leq j \leq i \mi 1$ then
$$ \langle Y~ Z_{J \setminus \{j\}} \rangle \, Z'_{j} \;=\; \langle Y~ Z_{J \setminus \{j\}} \rangle \, Z_{j} \;-\; \langle Y~ Z_{J \setminus \{j \mi 1\}} \rangle \, Z_{j \mi 1} \;+\; \dots \;\pm\; \langle Y~ Z_{J \setminus \{i \mi l\}} \rangle \, Z_{i \mi l}  $$
\item 
If $i \pl 1 \leq j \leq i \pl (4-l)$ then
$$ \langle Y~ Z_{J \setminus \{j\}} \rangle \, Z'_{j} \;=\; \langle Y~ Z_{J \setminus \{j\}} \rangle \, Z_{j} \;-\; \langle Y~ Z_{J \setminus \{j \pl 1\}} \rangle \, Z_{j \pl 1} \;+\; \dots \;\pm\; \langle Y~ Z_{J \setminus \{i \pl (4-l)\}} \rangle \, Z_{i \pl (4-l)}  $$
\end{enumerate}
\end{lemma}

\begin{proof}
These relations follow immediately from Remark~\ref{rmk:Z' as Z comb} following the proof of Lemma~\ref{lem effect of emb}. The case $Z_j'=Z_j$ is only restated. The second case is obtained by replacing the coefficients $1,t_{j+1},t_{j+1}t_{j+2},\dots$ with proportional twistors, supplied by Lemma~\ref{lem:twistors_as_emb_params}. The third case is obtained similarly by replacing the coefficients $1,s_{j+1},s_{j+1}s_{j+2},\dots,(s_{j+1} \cdots s_{4-l})$. 

We note that a combined formula holds for $Z_i'$ and it is omitted from this lemma. Also the overflow cases are similar with possible sign adjustments, as derived in  Remark~\ref{rmk:Z' as Z comb} and the proof of~\ref{lem:twistors_as_emb_params}, and not repeated here.
\end{proof}

The following procedure, named \emph{promotion}, summarizes how Lemma~\ref{lem effect of emb} together with Lemma~\ref{twistor coefficients} enable us to analyze the effect of $\emb_{i,l,4-l}$ on twistors and functionaries. 

\begin{definition}[Promotion under $\emb_{i,l,4-l}$]\label{def:general promotion}
The \emph{promotion} of a functionary $F(\langle Y Z_I\rangle : I \in \scalebox{0.8}{$\tbinom{N}{4}$})$ under the embedding $\emb_{i,l,4-l}$ is the outcome of the following sequence of operations:
\begin{enumerate}
\item  
Substitute $\langle Y Z_I\rangle \mapsto (-1)^{|I \cap [i]|} \left\langle Y Z'_I\right\rangle$ for every twistor in $F$,
as in Lemma~\ref{lem effect of emb}.
\item 
Express each $Z'_j$ as a combination of $Z_r$-s, with coefficients ratios of twistors, as in Lemma~\ref{twistor coefficients}.
\item 
Expand multilinearly each of the above twistors $\langle Y Z'_I\rangle$ to a rational function in $\langle Y Z_I\rangle$.
\item Multiply the resulting rational function by the greatest common divisor of the denominators.
\end{enumerate}
This yields a functionary $F'(\langle Y Z_I\rangle : I \in \scalebox{0.8}{$\tbinom{N\cup\{i\}}{4}$})$ which is the \emph{promotion} of~$F$.
\end{definition}

Using Lemmas~\ref{lem effect of emb} and Lemma~\ref{twistor coefficients}, we can deduce that the a functionary $F$ of fixed sign at~$C$ promotes to a fixed-sign functionary $F'$ at $C' = \embilr{i,l,4-l}\,C$. This requires that the cleared denominators in step (4) also have a fixed sign. Usually, Lemma \ref{obs:SA_and_bdry_twistors} and Lemma~\ref{lem:twistors_as_emb_params} come to the rescue, and thus the sign of $F'$ at $C'$ is fixed and computable. This is demonstrated in the next section in the special cases $\embilr{i,3,1}$ and $\embilr{n-2,2,2}$ which are most relevant to our purposes, though the method of sign-preserving promotion is applicable to $\emb_{i,l,r}$ all~$i,l$ and $r$ and even other~$m$.

\begin{rmk}
\label{rmk:pureness_and_type}
As operations (1)-(4) above are homogeneous in indices, if $F$ is pure then so is the resulting functionary.
\end{rmk}

\smallskip

In conclusion, the results of this section allow us to compute new fixed-sign functionaries from given ones. Given a pure functionary $F$ of fixed sign at some point $C$,
Lemmas \ref{lem:effect of  pre},~\ref{lem:effect_of_inc},~\ref{lem:effect_of_x_y}, and \ref{lem effect of emb} help us find a functionary $F'$ of fixed sign at the image of $C$ under the matrix operations $\pre_i$, $\inc_i$, $x_i$, $y_i$, and $\embilr{i,l,4-l}$. For the latter case, Lemmas~\ref{lem:twistors_as_emb_params 2},~\ref{lem:twistors_as_emb_params} and~\ref{twistor coefficients} give us tools to derive such a functionary which does not depend on the parameters $(\mathbf{t},\mathbf{s})$, which we refer to as the promotion of~$F$.

\subsection{Upper and Lower Promotion}
\label{section promotion}

Recall the two matrix operations specializing $\embilr{i,l,4-l}$, the upper embedding $\uemb_j=\emb_{j,3,1}$ and the lower embedding $\lemb_{n-2}=\emb_{n-2,2,2}$ from Definitions~\ref{def:upper emb}-\ref{def:lower emb}. In this section, we demonstrate Definition~\ref{def:general promotion} and analyze the evolution of functionaries under these operations, named \emph{upper promotion} and \emph{lower promotion} respectively. We start from the following explicit description.

\begin{prop}[Upper and lower promotion]\label{prop:upromotion} ~

Let $1 \leq i < a < b < c < d \leq n$. The promotion of the twistor $\langle a\,b\,c\,d \rangle$ under $\uemb_{i}$ is given by the following functionary:
\begin{itemize}
\itemsep0pt
\topsep0pt
\parskip0pt
\item if $d\notin\{n-1,n\}$ then it remains $\langle a\,b\,c\,d\rangle$
\item if $d = n-1$ and $c < n-1$, then 
$$\langle a\,b\,c\,n{-}1\rangle\;\langle i\,i{+}1\,n{-}2\,n\rangle-\langle a\,b\,c\,n{-}2\rangle\;\langle i\,i{+}1\,n{-}1\,n\rangle$$
\item if $d=n$ and $c<n-1$, then 
$$\langle a\;b\;c\;n\rangle\;\langle i\;i{+}1\;n{-}2\;n{-}1\rangle-\langle a\;b\;c\;n{-}1\rangle\;\langle i\;i{+}1\;n{-}2\;n\rangle+\langle a\;b\;c\;n{-}2\rangle\;\langle i\;i{+}1\;n{-}1\;n\rangle$$
\item if $d=n$ and $c=n-1$, then 
$$\langle i\,i{+}1\,n{-}2\,n{-}1\rangle \left(
\langle a\,b\,n{-}1\,n\rangle\;\langle i\,i{+}1\,n{-}2\,n\rangle-\langle a\,b\,n-2\,n\rangle\;\langle i\,i{+}1\,n{-}1\,n\rangle
\right)$$
\end{itemize}
We call this functionary the \emph{upper promotion} at $i$ of the twistor $\langle a\,b\,c\,d\rangle$.
The \emph{upper promotion} at $i$ of a pure functionary $F$ on the index set $\{i+1,\dots,n\}$ is the result of replacing each twistor in~$F$ with its upper promotion at~$i$.

Let $a<b<c<d$ be elements of $\{1,\dots,i,i+1,n-1,n\}$ where $i+1 < n-2$. The promotion of $\langle a~b~c~d \rangle$ under $\lemb_{n-2}$ is given by the following functionary:
\begin{itemize}
\itemsep0pt
\topsep0pt
\parskip0pt
\item 
if $i+1\notin\{c,d\}$, then it is $\langle a~b~c~d \rangle$
\item 
if $c=i+1$ and $d=n$, then
\[\langle a~b~i{+}1~d\rangle\;\langle i~n{-}2~n{-}1~n \rangle -  \langle a~b~i~d\rangle\;\langle i{+}1~n{-}2~n{-}1~n \rangle\]
\item 
if $d=i+1$, then
\[\langle a~b~c~i{+}1\rangle\;\langle i~n{-}2~n{-}1~n \rangle -  \langle a~b~c~i\rangle\;\langle i{+}1~n{-}2~n{-}1~n \rangle\]
\item
We forgo the case where $n-1$ appears in the twistor, to avoid unnecessary complicated expressions.
\end{itemize}
We call this functionary the \emph{lower promotion at $n-2$} of $\langle a~b~c~d \rangle$. The \emph{lower promotion} of a pure functionary on the index set $\{1,\dots,i+1,n\}$ is the result of replacing each twistor in~$F$ with its lower promotion.
\end{prop}

\begin{proof}
\label{proof of prop:upromotion}
The upper promotion at $i$ of a functionary $F$, specialized from Definition \ref{def:general promotion}, is performed by doing the following operations:
\begin{itemize}
        \item in all twistors, substitute  $Z_{n}$ with 
        $Z_{n}
        -\frac{\langle i\;i{+}1\;n{-}2\;n\rangle}{\langle i\;i{+}1\;n{-}2\;n{-}1\rangle}Z_{n-1}
        +\frac{\langle i\;i{+}1\;n{-}1\;n\rangle}{\langle i\;i{+}1\;n{-}2\;n{-}1\rangle}Z_{n-2}$ and $Z_{n-1}$ with 
        $Z_{n-1}
        -\frac{\langle i\,i{+}1\,n{-}1\,n\rangle}{\langle i\,i{+}1\,n{-}2\,n\rangle}Z_{n-2}$.
        \item multiply the functionary by the common denominator $\langle i\;i{+}1\;n{-}2\;n{-}1\rangle^{d_{n}(F)}$, where $d_{j}(F)$ is the multiplicity of $j$ in $F$ as in Definition~\ref{def pure}.
        \item multiply the functionary by the common denominator $\langle i\,i{+}1\,n{-}2\,n\rangle^{d_{n-1}(F)}$.
    \end{itemize}
    Now we can expand linearly each twistor containing $n$ or $n-1$ to get a rational function. The denominators will cancel with the multiplicative factors, yielding a polynomial, i.e. a functionary. 
    It is straightforward to verify the equivalence of applying these three operations and applying the four cases in the proposition: any twistor where $Z_{n-1}$ or $Z_n$ is included must be fall into one of the 4 cases of upper promotion at Proposition \ref{prop:upromotion}. In the first 3 options there is up to one substitution, which is exactly as in the 3 alternative substitutions above (including the multiplication by the common denominator). In the last case we substitute both $Z_n$ and $Z_{n-1}$ as above, and after linear expansion and reordering of the twistor coordinates $\langle a\;b\;n-1\;n\rangle$ becomes
    \begin{multline*}
    \langle a\;b\;n-1\;n\rangle
    +\langle a\;b\;n-2\;n\rangle
    \left(-\frac{\langle i\,i{+}1\,n{-}1\,n\rangle}{\langle i\,i{+}1\,n{-}2\,n\rangle}\right)
    \\+\langle\;a\;b\;n-1\;n-2\rangle
    \left(\frac{\langle i\,i{+}1\,n{-}1\,n\rangle}{\langle i\,i{+}1\,n{-}2\,n\rangle}\cdot\frac{\langle i\;i{+}1\;n{-}2\;n\rangle}{\langle i\;i{+}1\;n{-}2\;n{-}1\rangle}-\frac{\langle i\;i{+}1\;n{-}1\;n\rangle}{\langle i\;i{+}1\;n{-}2\;n{-}1\rangle}\right)
    \end{multline*}
    which after multiplying by the common denominators $\langle i\;i{+}1\;n{-}2\;n{-}1\rangle,\langle i\,i{+}1\,n{-}2\,n\rangle$ gives the desired expression.

    For the lower embedding (Definition~\ref{def:lower emb}), Definition \ref{def:general promotion} specializes to performing the following operations.
    We present it in a slightly more general context, in which the functionary may include the index $n-1$ in some twistor.
    \begin{itemize}
    \item in all twistors, substitute  $Z_{i+1} \mapsto Z_{i+1} - \frac{\langle i{+}1~n{-}2~n{-}1~n \rangle}{\langle i~n{-}2~n{-}1~n \rangle} Z_{i}
    $ {and $Z_{n-1} \mapsto  Z_{n-1} - \frac{\langle i~i{+}1~n{-}2~n{-}1 \rangle}{\langle i~i{+}1~n{-}2~n \rangle} Z_{n}.$
   }
        \item multiply the functionary by the common denominator ${\langle i~n{-}2~n{-}1~n \rangle}^{d_{i+1}(F)}$, where $j$ is the degree of $i+1$ in the pure functionary $F$.
        {\item multiply the functionary by the common denominator ${\langle i~i{+}1~n{-}2~n \rangle}^{d_{i}(F)}$, where $j$ is the degree of $n-1$ in the pure functionary $F$.}
    \end{itemize}
        It is again straightforward to see that applying these three operations is equivalent to applying the three cases in the proposition.
\end{proof}

\begin{rmk}
Functionary promotion for both types of embedding is programmed in the companion Sage library for this paper \cite{evenzohar2022bcfw}.
\end{rmk}

Our goal is to show that upper and lower promotions of fixed-sign functionaries yield new fixed-sign functionaries.
As a warm-up, it is helpful to illustrate this phenomenon with a specific example.

\begin{ex}
\label{warmup}
Let $C' \in \Gr^{\geq}_{k-1,234567}$ and suppose that the functionary $\langle 3\,4\,5\,7\rangle$ has a fixed sign $+1$ at~$C'$. This means that $\langle\Z(C')\, Z_3\,Z_4\,Z_5\,Z_7\rangle > 0$ for all $Z \in \Mat^{>}_{234567 \times [k+3]}$, which is for example the case if it arises from $C' = \pre_6\,C''$, by Lemmas~\ref{obs:SA_and_bdry_twistors} and~\ref{lem:effect of  pre}. Suppose also that we are interested in the upper embedding $C = \uemb_1
(t,t',t'',s)\,C'$ for all $(t,t',t'',s) \in (0,\infty)^4$. 

Lemma~\ref{lem effect of emb} says that the functionary in $C,Z$ equal everywhere to $\langle \Z(C)\, Z'_3\,Z'_4\,Z'_5\,Z'_7 \rangle$ is a functionary of fixed sign $+1$. We can express this functionary by expanding $\langle \Z(C)\, Z'_3\,Z'_4\,Z'_5\,Z'_7 \rangle$ in rows of~$Z \in \Mat^{>}_{1234567 \times [k+4]}$ using the relation $Z' = [x_1(s)] \cdot [y_7(t)] \cdot [y_6(t')] \cdot [y_5(t'')] \cdot Z$, with coefficients depending on $t,t',t'',s$. Lemma~\ref{twistor coefficients} allows us to express $t,t',t'',s$ directly using $C,Z$-twistors.
By the first case of Lemma~\ref{twistor coefficients}, we obtain $Z'_3=Z_3$, $Z'_4=Z_4$, and $Z'_5=Z_5$. The second case yields 
$$ \left\langle \Z(C)\,1\,2\,5\,6 \right\rangle \, Z'_7 \;=\; \left\langle \Z(C)\,1\,2\,5\,6 \right\rangle \, Z_7 - \left\langle \Z(C)\,1\,2\,5\,7 \right\rangle \, Z_6 + \left\langle \Z(C)\,1\,2\,6\,7 \right\rangle \, Z_5 $$
Substituting these $Z'_j$ in $\langle \Z(C)\, Z'_3\,Z'_4\,Z'_5\,Z'_7 \rangle$, we obtain
$$ \left\langle \Z(C)\,1\,2\,5\,6 \right\rangle \; \left\langle \Z(C)\, Z'_3\,Z'_4\,Z'_5\,Z'_7 \right\rangle \;=\; \left\langle \Z(C)\,1\,2\,5\,6 \right\rangle \; \left\langle \Z(C)\,3\,4\,5\,7 \right\rangle - \left\langle \Z(C)\,1\,2\,5\,7 \right\rangle \; \left\langle \Z(C)\,3\,4\,5\,6 \right\rangle $$ 
Note that the third terms drops since $\langle3\,4\,5\,5\rangle=0$, which demonstrates a case of Remark~\ref{rmk-truncations}. The twistor $\langle 1\,2\,5\,6 \rangle$ is nonnegative by Lemma~\ref{obs:SA_and_bdry_twistors}. If $C$ is such that this twistor is strictly positive at~$C$, then the above expression is nonzero. In conclusion, the promotion of $\langle 3\,4\,5\,7 \rangle$ under the upper embedding $\uemb_1$ 
is $ \langle 1\,2\,5\,6 \rangle \, \langle 3\,4\,5\,7 \rangle - \langle 1\,2\,5\,7 \rangle \, \langle 3\,4\,5\,6 \rangle $, and it has the same fixed sign. By Definition~\ref{nn:favorite_functionaries} it is written also as $\favorite{6\,7}{1\,2}{3\,4}{5}$, or as $\favorite{1\,2}{3\,4}{6\,7}{5}$ by Lemma~\ref{obs:plucker_functionary}.
\end{ex}

The following two propositions give a general version of the example above. We show under general conditions that upper and lower promotions transform arbitrary fixed-sign pure functionaries on a point to fixed-sign pure functionaries on the embedding of that point in a larger amplituhedron.

\begin{prop}[Fixed sign under upper promotion]
\label{prop:uprom}
Let $F$ be a pure functionary with indices in $N=\{i+1,\dots,n\}$. Let $F'$ be the upper promotion at $i$ of $F$, with indices in $N\cup\{i\}$.
Let $C\in\Grnn{k-1}{N}$ and $t_1,t_2,t_3,s_1>0,$ and denote $C'=\uemb_i(t_1,t_2,t_3,s_1)C$.
Assume that $$\langle \Z(C')\;i\;i{+}1\;n{-}2\;n\rangle_Z,~\langle \Z(C')\;i\;i{+}1\;n{-}2\;n{-}1 \rangle_Z \;\neq\; 0$$ for every $Z \in \Mat^{>}_{N\cup\{i\} \times [k+4]}$, 
and 
that $F(\langle\Z_{N}(C)~Z_I\rangle)$ has a fixed sign. 
Then 
$$\SIGN F'(\langle \Z(C')Z_I\rangle) \;=\; (-1)^{d_{n-1}(F)} \SIGN F(\langle\Z_{N}(C)~Z_I\rangle
)$$
\end{prop}

\begin{proof}
    Let $Z \in \Mat^{>}_{N\cup\{i\} \times [k+4]}$ be arbitrary and denote $Y=\Z(\uemb_i(t_1,t_2,t_3,s_1)C)$.
    
    Apply Lemma \ref{lem effect of emb} with $N$ as in this lemma and $l=3, r=1$, and note that $|I\cap[i]|=0$ for any $I\in\binom{N}{4}$. We get that if $\SIGN F\left( \left\langle\Z_N\left(C\right)~Z_I\right\rangle 
    \right)=s\in\{\pm1\}$ then $\SIGN F''\left( \left\langle Y~Z_I\right\rangle\right) = s,$
    where $F''\left( \left\langle Y~Z_I\right\rangle\right)$ is the functionary such that $F''\left( \left\langle Y~Z_I\right\rangle\right)=F\left( \left\langle Y~Z'_I\right\rangle\right)$, and $Z'$ is the matrix described in that lemma. 

    We can find $F''$ explicitly by linearly expanding the rows of $Z'$ as linear combinations of the rows of $Z$, expressing $F\left( \left\langle Y~Z'_I\right\rangle\right)$ as another polynomial in $Z$-twistors, with some coefficients depending on $t_i,s_i$. Lemma~\ref{twistor coefficients} allows us to express $t_i,s_i$ using twistors, and therefore express the rows of the matrix $Z'$ as linear combinations of rows of $Z$ with coefficients which are rational functions of twistors. 
    
    Apply Lemma \ref{twistor coefficients} with $l=3$, the same $N$, and $J=\{n-2,n-1,n,i,i+1\}$. For $j\notin\{n-1,n,i\}$ we have $Z'_j=Z_j$, and for $j=i$ the vector $Z'_i$ does not appear in the functionary $F$. For the other possibilities for $j$,
    \[
        \langle Y Z_{\{n-2,n,i,i+1\}} \rangle Z'_{n-1} = \langle Y Z_{\{n-2,n,i,i+1\}} \rangle Z_{n-1} - \langle Y Z_{\{n-1,n,i,i+1\}} \rangle Z_{n-2},
    \]
    \[
        \langle Y Z_{\{n-2,n-1,i,i+1\}} \rangle Z'_{n} = \langle Y Z_{\{n-2,n-1,i,i+1\}} \rangle Z_{n} - \langle Y Z_{\{n-2,n,i,i+1\}} \rangle Z_{n-1} + \langle Y Z_{\{n-1,n,i,i+1\}} \rangle Z'_{n-2}.
    \]
    After dividing by $\langle Y Z_{\{n-2,n,i,i+1\}} \rangle$, $\langle Y Z_{\{n-2,n-1,i,i+1\}} \rangle$ respectively (using the assumption that they are nonzero), this is plugged into $F\left( \left\langle Y~Z'_I\right\rangle\right)$ and expanded multilinearly, giving a rational function in twistors with the same fixed sign as $F''$, but independent of $t_i,s_i$. This last step is precisely the substitution operation from Defintion~\ref{def:general promotion} (i.e., steps 1,2,3 there), which we have seen in this form as the first operation in the proof of Proposition~\ref{prop:upromotion} (upper promotion part). 
    
    As mentioned, the resulting expression in $Z$-twistors would have sign $s$.
    The next step in Definition~\ref{def:general promotion} and the proof of Proposition~\ref{prop:upromotion} 
    is multiplication by the common denominators, which are~$\langle Y Z_{\{n-2,n,i,i+1\}} \rangle^{d_{n-1}(F)}$ and $\langle Y Z_{\{n-2,n-1,i,i+1\}} \rangle^{d_{n}(F)}$. This would change the sign by $(-1)^{d_{n-1}(F)},$ since on the amplituhedron $\langle Y Z_{\{n-2,n-1,i,i+1\}} \rangle\ge0$ (by Lemma \ref{obs:SA_and_bdry_twistors}) and thus $\langle Y Z_{\{n-2,n-1,i,i+1\}} \rangle\le0$ (by Lemma \ref{lem:twistors_as_emb_params}), and since both are assumed to be nonzero. Together, this gives us the functionary $F'$ in $Y,Z$, with the sign 
    $(-1)^{d_{n-1}(F)}s$.
\end{proof}

\begin{prop}[Fixed sign under lower promotion]
\label{prop:lprom}
Let $N=\{1,\dots,i+1,n-1,n\}$ for $i\leq n-4$. Let $F$ be a pure functionary with indices in $N\setminus\{n-1\}$. Let $F'$ be the lower promotion of $F$, with indices in $N\setminus\{n-1\}\cup\{n-2\}$.
Let $C\in\Grnn{k-1}{N}$ and $t_1,t_2,s_1,s_2>0,$ and denote $C'=\lemb_{n-2}(t_1,t_2,s_1,s_2)C$. Assume that 
$$\langle \Z(C')\;i\;i{+}1\;n{-}1\;n\rangle_Z,~\langle \Z(C')\;i\;n{-}2\;n{-}1\;n \rangle_Z \;\neq\; 0$$ for every $Z \in \Mat^{>}_{N\cup\{n-2\} \times [k+4]}$, and that $F(\langle\Z_{N}(C)~Z_I\rangle)$ has a fixed sign. 
Then 
$$\SIGN F'(\langle \Z(C')\;Z_I\rangle) \;=\; (-1)^{d_{i{+}1}(F)+d_{n}(F)}
\SIGN F(\langle\Z_{N}(C)~Z_I\rangle)
.$$
\end{prop}
\begin{proof}
    Let $Z \in \Mat^{>}_{(N\cup\{n-2\}) \times [k+4]}$ be arbitrary and denote $Y=\Z(\lemb_i(t_1,t_2,s_1,s_2)C)$.
    
    Apply Lemma \ref{lem effect of emb} with $N$ as in this lemma and $l=r=2$. We get that if 
    $$\SIGN F\left( \left\langle\Z_N\left(C\right)~Z_I\right\rangle 
    \right)=s\in\{\pm1\}$$ 
    then \[\SIGN F\left( (-1)^{|I\cap[n-2]|}\left\langle Y~Z'_I\right\rangle 
    \right) = s,\]
    where $Z'$ is the matrix described in that lemma and $F\left( (-1)^{|I\cap[n-2]|}\left\langle Y~Z'_I\right\rangle\right)$ denotes by abuse of notation the functionary $F''(\langle Y~Z_I\rangle)$ satisfying everywhere $F''(\langle Y~Z_I\rangle) = F\left( (-1)^{|I\cap[n-2]|}\left\langle Y~Z'_I\right\rangle\right)$ from Lemma~\ref{lem effect of emb}. Note that, since $Z_{n-1}$ does not appear in the expression for $F,$ for all relevant $4-$tuples $I,$  $|I\cap[n-2]|$ is odd precisely when $n\in I.$
    Thus, using a similar abuse of notation, \[\SIGN F\left(\left\langle Y~Z'_I\right\rangle \right) = (-1)^{d_n(F)}s.\]
    
    To express $F''$ explicitly, we want to linearly expand the rows of $Z'$, $F\left( \left\langle Y~Z'_I\right\rangle\right)$ using the expression of the rows of $Z'$ are linear combinations of the rows of $Z$. This would give us an expression with coefficients depending on $t_i,s_i$. Lemma \ref{twistor coefficients} will allow us to express these coefficients with twistors, giving us a rational function in $C,Z$-twistors which is equal everywhere to $F\left(\left\langle Y~Z'_I\right\rangle\right)$

    With this in mind, let us apply Lemma \ref{twistor coefficients} with $l=2$, the same $N$, and $J=\{i,i+1,n-2,n-1,n\}$. For $j\notin\{i+1,n-2,n-1\}$ we have $Z'_j=Z_j$, and for $j=n-2,n-1$ we assumed that the vectors $Z'_{n-2},Z'_{n-1}$ do not appear in the functionary $F$. The remaining possibility is $j=i+1$, where
    \[
        \langle Y Z_{\{i,n-2,n-1,n\}} \rangle Z'_{i+1} = \langle Y Z_{\{i,n-2,n-1,n\}} \rangle Z_{i+1} - \langle Y Z_{\{i+1,n-2,n-1,n\}} \rangle Z_{i}.
    \]
    After dividing by the non-zero $\langle Y Z_{\{i,n-2,n-1,n\}} \rangle
    $, this is plugged into $F\left( \left\langle Y~Z'_I\right\rangle\right)$ and expanded multilinearly, giving a rational function in twistors with the same fixed sign $s$ as $F''$, but independent of $t_i,s_i$. This procedure is precisely the substitution operation for the lower promotion from 
    Definition~\ref{def:general promotion} (steps 1,2,3 there), which we have seen in this form in as the first operation in the proof of the lower promotion part of Proposition~\ref{prop:upromotion}. 
    The resulting expression in $Z$-twistors would have sign $(-1)^{d_n(F)} s$.
    
    The next steps in the proof of the lower promotion part of Proposition~\ref{prop:upromotion} is multiplication by the common denominator $\langle Y Z_{\{i,n-2,n-1,n\}} \rangle^{d_{i+1}(F)}$. This twistor is assumed to be non-zero. It is in fact negative, since by Lemma 
    \ref{lem:twistors_as_emb_params}, $-\langle Y Z_{\{i,n-2,n-1,n\}} \rangle/\langle Y Z_{\{i,i+1,n-1,n\}} \rangle>0,$ and the denominator is positive by Lemma \ref{obs:SA_and_bdry_twistors} and the assumption.    
This gives us the functionary $F'$ in $Y,Z$, with the sign $(-1)^{d_{i+1}(F)+d_n(F)}s$. Since the choice of $Z$ was arbitrary, the desired fixed sign follows.
\end{proof}

In the following two corollaries, we give special cases of Propositions \ref{prop:uprom} and \ref{prop:lprom}
for certain important functionaries, that are used later in the paper. These lemmas provides us with a sequence of functionaries that have a fixed sign on the whole images of $\uemb_i$ and $\lemb$.

\begin{cor}
\label{cor:upper_emb}
Let $i < j < n-3$, let $C \in \Gr^{\geq}_{k-1,\{i+1,\dots,n\}}$ and for some positive $t_1,t_2,t_3,s_1$ let $C' = \uemb_i(t_1,t_2,t_3,s_1)\;C \in \Gr^{\geq}_{k,\{i,\dots,n\}}$.
Assume that $C$ and $C'$ are such that $\langle \Z(C)~ j\;j+1\;n-1\;n\rangle_Z\neq0$ for all~$Z$, and $\langle \Z(C')~ i\;i{+}1\;n{-}1\;n\rangle_Z \neq 0$ for all~$Z$. Then
$$ \SIGN \favorite{i~i{+}1}{j~j{+}1}{n{-}2~n{-}1}{n}_{\Z(C'),Z} \;=\; -1$$
\end{cor}

\begin{proof}
For the sake of simplicity, we omit $Y$ and~$Z$ from all twistor and functionary notations. The twistor $F=\langle j\;j+1\;n-1\;n\rangle$ has a positive fixed sign at $C$ by Lemma~\ref{obs:SA_and_bdry_twistors} and the nonzero assumption. We can apply Proposition~\ref{prop:uprom} on $F$, since we assume that $\langle \Z(C')~ i\;i{+}1\;n{-}1\;n\rangle_Z \neq 0$, and thus $\langle \Z(C')~ i\;i{+}1\;n{-}2\;n{-}1\rangle_Z$ is also nonzero by Lemma~\ref{lem:twistors_as_emb_params}. From the proposition we get that $F'$, the upper promotion of $F$, is a negative fixed-sign functionary. From Proposition~\ref{prop:upromotion} and Lemma \ref{obs:plucker_functionary}, 
\begin{align*}
F' &\;=\; \langle i\,i{+}1\,n{-}2\,n{-}1\rangle \left(
\langle j\,j{+}1\,n{-}1\,n\rangle\;\langle i\,i{+}1\,n{-}2\,n\rangle-\langle j\,j{+}1\,n-2\,n\rangle\;\langle i\,i{+}1\,n{-}1\,n\rangle\right) \\
&\;=\; \langle i\,i{+}1\,n{-}2\,n{-}1\rangle\;
\favorite{n{-}2~n{-}1}{i~i{+}1}{j~j{+}1}{n} \\
&\;=\; \langle i\,i{+}1\,n{-}2\,n{-}1\rangle\;
\favorite{i~i{+}1}{j~j{+}1}{n{-}2~n{-}1}{n}
,
\end{align*}
and since we also assume that $\langle i\,i{+}1\,n{-}2\,n{-}1\rangle$ is nonzero at~$C'$, and it is nonnegative by Lemma~\ref{obs:SA_and_bdry_twistors}, then $\favorite{i~i{+}1}{j~j{+}1}{n{-}2~n{-}1}{n}$ is a indeed a negative fixed-sign functionary.
All inequalities hold independently of the choice of $Z$.
\end{proof}

\begin{cor}
\label{cor:promotion_by_lower_emb}
Let $j < i < n-3$, let $C \in \Gr^{\geq}_{k-1,N}$ where $N=\{1,\dots,i+1,n-1,n\}$, and for some positive $t_1,t_2,s_1,s_2$ let $C' = \lemb_{n-2}(t_1,t_2,s_1,s_2)\;C \in \Gr^{\geq}_{k,N\cup\{n-2\}}$.
Assume that $C$ and $C'$ are such that $\SIGN \langle \Z(C) \; j\;j{+}1\;i+1\;n \rangle_{Z} = s \in \{\pm1\}$ and $\langle \Z(C')~ i{+}1\;n{-}2\;n{-}1\;n\rangle_Z \neq 0$ for all $Z$.
Then
\begin{equation*}
\SIGN\,\favorite{j~j{+}1}{i~i{+}1}{n{-}2~n{-}1}{n}_{\Z(C'),Z} \;=\; s \end{equation*}
In particular, if $C=\pre_{n-1} C''$ for some $C'' \;\in\; \Gr_{k-1,N\setminus\{n-1\}}^{\geq}$ then $s=+1$.
\end{cor}

\begin{proof}
Denote by $F=\langle \Z_N(C) \; j\;j{+}1\;i+1\;n \rangle$ the assumed fixed-sign functionary. We apply Proposition~\ref{prop:lprom} where the twistors $\langle i\;i+1\;n-1\;n\rangle$ and $\langle i\;n-2\;n-1\;n\rangle$ are proportional to $\langle i{+}1\;n{-}2\;n{-}1\;n\rangle \neq 0$ by Lemma~\ref{lem:twistors_as_emb_params}. 
From Proposition~\ref{prop:lprom} we get that $F'$, the lower promotion of $F$, is a fixed-sign functionary with the same sign $s$.
From Proposition~\ref{prop:upromotion} and Lemma \ref{obs:plucker_functionary}, 
\begin{align*}
F' &= \langle j~j{+}1~i{+}1~n\rangle\;\langle i~n{-}2~n{-}1~n \rangle -  \langle j~j{+}1~i~n\rangle\;\langle i{+}1~n{-}2~n{-}1~n \rangle \\
& \;=\; \favorite{i~i+1}{n{-}2~n{-}1}{j~j{+}1}{n} \\
& \;=\; \favorite{j~j{+}1}{i~i+1}{n{-}2~n{-}1}{n}
\end{align*}
as the corollary claims.
\end{proof}

We conclude the section by working out explicitly two degenerate cases with $k=1$ of Corollaries \ref{cor:upper_emb} and~\ref{cor:promotion_by_lower_emb}. These examples demonstrate and verify that the fixed-sign functionaries obtained in the corollaries are correct in this end case as well.

\begin{ex}[following Corollary \ref{cor:upper_emb}]
\label{ex:cor:upper_emb}
We focus on the special case $i=1,\,k=1$ of the previous proposition, where the promotion is applied on a $0 \times (n-1)$ matrix. After $\uemb_{1}$ 
the resulting $1 \times n$ matrix $C$ has the form $(a_1,a_2,0,\ldots,0,a_{n-2},a_{n-1},a_n)$, which is the domino matrix of one long top chord $(1,2,n-2,n-1)$. The following calculation shows that the quadratic functionary $\favorite{1~2}{j~j{+}1}{n{-}2~n{-}1}{n}$ is strictly negative, even in the end case $j=2$ where some terms vanish. 
\begin{align*}
& \!\!\!\!\!\!\!\! \langle1~j~j{+}1~n\rangle \,
\langle2~n{-}2~n{-}1~n\rangle \,-\, \langle2~j~j{+}1~n\rangle \,
\langle1~n{-}2~n{-}1~n\rangle \\
& \;=\;  
\big(-a_2\langle Z_{1,2,j,j+1,n}\rangle -a_{n-2}\langle Z_{1,j,j+1,n-2,n}\rangle-a_{n-1}\langle Z_{1,j,j+1,n-1,n}\rangle\big) \cdot \big( +a_1 \langle Z_{1,2,n-2,n-1,n} \rangle \big) \\
& \;\;\;\;\; \;-\; \big( +a_1\langle Z_{1,2,j,j+1,n}\rangle
-a_{n-2}\langle Z_{2,j,j+1,n-2,n}\rangle-a_{n-1}\langle Z_{2,j,j+1,n-1,n}\rangle \big) \cdot \big(-a_2 \langle Z_{1,2,n-2,n-1,n}\rangle \big) \\
&\;=\; \langle Z_{1,2,n-2,n-1,n}\rangle\,\big(-a_1a_{n-2}\langle Z_{1,j,j+1,n-2,n}\rangle-a_1a_{n-1}\langle Z_{1,j,j+1,n-1,n}\rangle \\
&\;\;\;\;\; \;\;\;\;\;\;\;\;\;\;\;\;\;\;\;\;\;\;\;\;\;\;\;\;\;\;\;\;\; -a_2a_{n-2}\langle Z_{2,j,j+1,n-2,n}\rangle-a_2a_{n-1}\langle Z_{2,j,j+1,n-1,n}\rangle\big),
\end{align*}
where the terms $\pm a_1a_2 \langle Z_{1,2,n-2,n-1,n} \rangle\;\langle Z_{1,2,j,j+1,n} \rangle$ are cancelled out in the second equality.
\end{ex}

\begin{ex}[following Corollary \ref{cor:promotion_by_lower_emb}]
\label{ex:cor:promotion_by_lower_emb}
We focus on the special case $i=n-4,\,k=1$ of the previous proposition, where the promotion is applied on a $0 \times (n-1)$ matrix. We demonstrate with direct calculations that there is no exception in this case. After $\lemb_{n-2}$ 
the resulting $1 \times n$ matrix $C$ has the form $(0,\ldots,0,a_{n-4},a_{n{-}3},a_{n-2},a_{n-1},a_n)$, which is the domino matrix of a single chord from $(n-4,n{-}3)$ to $(n-2,n-1)$. For $ j \leq n-5$, we expand in the first row every twistor that participates in the functionary $\favorite{j~j{+}1}{n{-}4~n{-}3}{n{-}2~n{-}1}{n}$.
\begin{align*}
& \!\!\!\!\!\!\!\!\!\!\!\!\!\!\!\!\!\!\!\!\!\!\!\!\!\!\!\!\!\!\!\!\!\!\!\! \langle j~n{-}4~n{-}3~n\rangle \,\langle j{+}1~n{-}2~n{-}1~n\rangle \,-\, \langle j{+}1~n{-}4~n{-}3~n\rangle \, \langle j~n{-}2~n{-}1~n\rangle
\\ 
& \;=\; ( -a_{n-2}\langle Z_{j,n-4,n-3,n-2,n}\rangle-a_{n-1}\langle Z_{j,n-4,n-3,n-1,n}\rangle)\\ & \;\;\;\;\;\;\;\;\;\;\;\;\;\;\;\;\;\; \cdot (-a_{n-4}\langle Z_{j+1,n-4,n-2,n-1,n}\rangle-a_{n-3}\langle Z_{j+1,n-3,n-2,n-1,n}\rangle) \\
& \;-\; (-a_{n-2}\langle Z_{j+1,n-4,n-3,n-2,n}\rangle -a_{n-1}\langle Z_{j+1,n-4,n-3,n-1,n}\rangle) \\ & \;\;\;\;\;\;\;\;\;\;\;\;\;\;\;\;\;\; \cdot (-a_{n-4}\langle Z_{j,n-4,n-2,n-1,n}\rangle -a_{n-3}\langle Z_{j,n-3,n-2,n-1,n}\rangle)
\end{align*}
Expanding and regrouping,
\begin{align*}
& \;=\; a_{n-2}a_{n-4}(\langle Z_{j,n-4,n-3,n-2,n}\rangle\;\langle Z_{j+1,n-4,n-2,n-1,n}\rangle - \langle Z_{j+1,n-4,n-3,n-2,n}\rangle\;\langle Z_{j,n-4,n-2,n-1,n}\rangle)\\
& \;+\; a_{n-2}a_{n-3}(\langle Z_{j,n-4,n-3,n-2,n}\rangle\;\langle Z_{j+1,n-3,n-2,n-1,n}\rangle-\langle Z_{j+1,n-4,n-3,n-2,n}\rangle\;\langle Z_{j,n-3,n-2,n-1,n}\rangle)\\
& \;+\; a_{n-1}a_{n-4}(\langle Z_{j,n-4,n-3,n-1,n}\langle Z_{j+1,n-4,n-2,n-1,n}\rangle-\langle Z_{j+1,n-4,n-3,n-1,n}\rangle\;\langle Z_{j,n-4,n-2,n-1,n}\rangle)\\
& \;+\; a_{n-1}a_{n-3}(\langle Z_{j,n-4,n-3,n-1,n}\rangle\;\langle Z_{j+1,n-3,n-2,n-1,n}\rangle-\langle Z_{j+1,n-4,n-3,n-1,n}\rangle\;\langle Z_{j,n-3,n-2,n-1,n}\rangle)
\end{align*}
Simplifying with the Pl\"ucker relations,
\begin{align*} \;=\;
& a_{n-2}a_{n-4}\langle Z_{j,j+1,n-4,n-2,n}\rangle\;\langle Z_{n-4,n-3,n-2,n-1,n}\rangle \;+\; a_{n-2}a_{n-3}\langle Z_{j,j+1,n-3,n-2,n}\rangle\;\langle Z_{n-4,n-3,n-2,n-1,n}\rangle
\\ \;+\;& a_{n-1}a_{n-4}\langle Z_{j,j+1,n-4,n-1,n}\rangle\;\langle Z_{n-4,n-3,n-2,n-1,n}\rangle \;+\; a_{n-1}a_{n-3}\langle Z_{j,j+1,n-3,n-1,n}\rangle\;\langle Z_{n-4,n-3,n-2,n-1,n}\rangle
\end{align*}
This expression is positive as required, even in the end case $j=n-5$, where some but not all of the terms vanish.
\end{ex}

\section{Injectivity}
\label{sec:inj}

In this section we prove Theorem~\ref{thm:injectiveness}, that each BCFW cell maps injectively onto its image {by solving} an inverse problem {(Proposition~\ref{prop:inverse problem})}: given a point in the image of {a} BCFW cell $S$, we find its unique preimage in that cell.
By Theorem~\ref{thm:domino} there is a chord diagram~$\D$ with chords $\{c_i\}_{i\in[k]}$ such that every point in $S$ has a representative matrix $C$, such that nonzero entries of every row $C_i$ are precisely in the five or six positions $\Supp(c_i)$, and with the right signs.
We can now
calculate the nonzero entries of $C$ iteratively {from $Z$ and $Y=\Z(C)$}. 
We start with the top chords, and after calculating the entries for a row, we calculate the entries of each of its children. 
We use Lemma~\ref{obs:solve_params_5} with $Y$ and $Z_j, Z_{j+1}, Z_l, Z_{l+1}, v$ where either $v=Z_n$ or $v=t_1Z_h+t_2 Z_{h+1}$.
We elaborate on this method in the proof of Proposition~\ref{prop:inverse problem}. In Proposition~\ref{prop:injectiveness_more_general} we generalize this beyond the nonnegative Grassmannian. Remarkably, the resulting preimage can be expressed using a functionary-valued matrix, see Definition~\ref{def:func val mat} and Example~\ref{ex:func val mat}.

\begin{prop}\label{prop:inverse problem}
Let $S \in \mathcal{BCFW}_{n,k}$, $Z\in \Mat^{>}_{n \times(k+4)}$, and $Y \in \Z(S)$. Then, $Y$ has a unique preimage under $\Z$. In {other words}, the restriction of $\Z$ to $S$ is injective.
\end{prop}
We use the following lemma to prove the proposition.
\begin{samepage}
\begin{lemma}
\label{lem:non_zero_dets_for_domino}
Let $S \in \mathcal{BCFW}_{n,k}$, $Z\in \Mat^{>}_{n \times(k+4)}$, and $Y \in \Z(S)$. For every chord $c_i=(j,j+1,l,l+1)$ in the chord diagram of~$S$:
\begin{enumerate}
\item
\label{it:non_zero_twistors_for_rows_with_5} 
If $c_i$ is a top chord then $\langle Y Z_I\rangle\neq 0$ for every $I\subseteq\{j,j+1,l,l+1,n\}$ of size $|I|=4$. 
\item
\label{it:2_for_non_zero_dets_for_domino} 
If $c_i$ is a child of another chord then $\langle Y Z_{\{j,j+1,l,l+1\}}\rangle\neq 0$.
\end{enumerate}
\end{lemma}
\end{samepage}

\begin{proof}[Proof of Proposition~\ref{prop:inverse problem}]

{By Definition~\ref{maindefbcfw}, $S$ comes from some chord diagram $D$, and by Theorem \ref{thm:domino} every point in $S$ must have a unique representation in the domino form of $D$, up to rescaling of rows. Therefore, for any $Y\in\Z(S)$ there is some preimage $C \in S$, which we can uniquely represent as a domino matrix in the form $D$. 
On the other hand, any preimage of $Y$ in $S$ can be represented in the domino form $D$. Therefore, to show the uniqueness of the preimage, we must show that the row entries of a domino matrix $C$ such that $CZ$ has the same row span as $Y$ are uniquely determined by $Y$ and $Z$. We now show this to each row $c_i$ of~$C$.}

If $c_i=(j,j+1,l,l+1)$ is a top level chord 
then on the one hand $Y_i$ belongs to the row span of $Y,$ and on the other hand it is a linear combination of $Z_s$ for $s\in\{j,j+1,l,l+1,n\}.$
The five vectors $Z_s$ for $s\in\{j,j+1,l,l+1,n\}$ are linearly independent, since $k\geq 1$ and $Z$ is a positive $n\times(k+4)$ matrix, so every collection of different $k+4$ rows is linearly independent.
By Lemma~\ref{obs:solve_params_5} and the first case of Lemma~\ref{lem:non_zero_dets_for_domino}, $C_i$ is uniquely determined{, up to rescaling}, as the correct $i$th row of $C$.

Suppose now that $c_i$ is not a top chord, 
but we have found the unique possible representative for the row of its parent $C_{p}$. In this case $k\geq 2.$
Since $C$ is a domino matrix and $C_{p}$ is its $p$th row, the entries $t_1=\alpha_{p}$ and $t_2=\beta_{p}$ are both nonzero.
Therefore if $(h,h+1)$ is the start of $c_{p}$ then  $v=t_1Z_h+t_2Z_{h+1}\neq 0.$ 
As before, $Y_i$ is both in the row span of $Y$ and in the linear span of $\{v, Z_j,Z_{j+1},Z_l,Z_{l+1}\}$. These five vectors are again linearly independent. When $h+1=j$ this follows from the linear independence of the five vectors $\{Z_h,Z_{h+1},Z_{j+1},Z_{l},Z_{l+1}\}$, which in turn follows from the positivity of~$Z$. Otherwise, $h+1<j$ and it follows from the independence of the six vectors $\{Z_h,Z_{h+1},Z_j,Z_{j+1},Z_{l},Z_{l+1}\}$, because $Z$ is positive and $k\geq 2.$ Hence, by Lemma~\ref{obs:solve_params_5} and the second case of Lemma~\ref{lem:non_zero_dets_for_domino} the entries of $C_i$ are uniquely recovered, {up to rescaling,} and we obtain the correct row in the domino matrix~$C$. Repeating for every chord, we recover the whole matrix~$C$.
\end{proof}

\begin{proof}[Proof of Lemma \ref{lem:non_zero_dets_for_domino}]
We start with the first case. Suppose that $Y=CZ,$ where $C$ is the 
domino form of an arbitrary point $V\in S$. In this case, $C_i$ has support $A=\{j,j+1,l,l+1,n\}.$ Let $I=\{i_1,\ldots,i_4\} \subset A$ be as in the statement, and $\{p\}=A\setminus I.$
Then by Lemma~\ref{eq:Twistor Cauchy-Binet} and the fact that $C_i$ is supported on the entries of $A,$
\[\langle Y Z_I\rangle=\sum_{{J\in\binom{[n]\setminus I}{k}},~p\in J}
\langle C^J\rangle \langle Z_{J\cup I} \rangle s(J,I).\]
If $\langle C^J\rangle\neq 0$ for a given $J,$ then there must be a bijection $f:[k]\to J$ such that $C_h^{f(h)}\neq 0,$ or equivalently the entry $f(h)$ is in the support of $C_h.$ Since $c_i$ is a top chord, the supports of the rows that correspond to chords that start before $c_i$ are contained in $[j+1]\cup\{n\}.$ The supports of rows for chords that descend from $c_i$ are contained in $[l+1]\setminus[j+1].$ The support of the remaining rows (other than $C_i$) are contained in $[n]\setminus[l].$
Thus, if $\langle C^J\rangle\neq 0$, then the conditions ${J\in\binom{[n]\setminus I}{k}},~p\in J$ and the fact that $c_i$ is a top chord imply:
\begin{itemize}
\item 
$k_1=|J\cap[j-1]|$ is the number of chords that start before $c_i$, which is $i-1$.
\item 
$k_2=|J\cap\{j+2,\ldots,l-1\}|$ is the number of chords that descend from $c_i$, which is $\dd(c_i)$ 
\item 
$k_3=|J\cap\{l+2,\ldots,n-1\}|$ is the number of chords that start after $c_i$ ends: $\behind(c_i)-\dd(c_i).$
\end{itemize}

Hence, for every $J$ in the above sum that contributes a nonzero determinant, the sign $s(J,A\setminus\{p\})$ as defined in Lemma~\ref{eq:s(J,I)} depends solely on $p$, as follows.
\[s(J,A\setminus\{p\})=\begin{cases}
(-1)^{k_2+k_3}=(-1)^{\behind(c_i)},~\text{if }p=j,\\
(-1)^{1+k_2+k_3}=(-1)^{\behind(c_i)+1},~\text{if }p=j+1,\\
(-1)^{k_3}=(-1)^{\behind(c_i)-\dd(c_i)},~\text{if }p=l,\\
(-1)^{1+k_3}=(-1)^{\behind(c_i)-\dd(c_i)+1},~\text{if }p=l+1,\\
1,~\text{if }p=n.
\end{cases}\]

Thus, it is enough to show that there is a set $J\in\binom{[n]}{k}$
with $J\cap\{j,j+1,l,l+1,n\}=\{p\}$ and $\langle C^J\rangle\neq 0,$ but this follows from Lemma~\ref{lem:more than 4}.

We turn to the second case.
Write $c=c_i=(j,j+1,l,l+1).$ We show that $\langle Y j,j+1,l,l+1\rangle$ is non vanishing for all $Z$ and $Y\in\Z(S)$. The strategy is to start with the subdiagram $D'$ that descends from $c$, that is, made of the chord $c$ and its descendants, erasing the markers outside, except for~$n$. Then one uses Corollaries~\ref{cor:generation_left}-\ref{cor:generation_top} to extend it to the full diagram. We show how the twistor of interest evolves during this process, and that it remains nonzero, for all positive $Z,$ and even keeps its sign. We note that a direct proof for this claim can be produced along the lines of Lemma~\ref{lem:more than 4}, but the proof we now show is simpler, and demonstrates the promotion technique, which is crucial in Section~\ref{sec:separation} below, where direct proofs seem less accessible.

Let $\D'$ be the subdiagram which descends from $c$ and $S'$ the corresponding positroid cell. In $\D'$, $c$ is a top chord, hence Lemma~\ref{lem:non_zero_dets_for_domino},\eqref{it:non_zero_twistors_for_rows_with_5} applies. Thus, for all positive matrices $Z'$ of the appropriate size, on $\Z'(S')$ it holds that 
$$\langle j,j+1,l,l+1\rangle\;>\;0$$
By Corollaries~\ref{cor:generation_left}-\ref{cor:generation_top}, $S$ is obtained from $S'$ by a series of operations of the form:
\begin{enumerate}
\item 
$\inc_h,\pre_h$ for $h\notin\{j,j+1,\ldots,l+1\}\cup\{n\}.$
\item 
$x_h(t),~t\in\R_+,$ for $h\notin\{j+1,\ldots,l+1\}\cup\{n\},$ or $y_h(t),~t\in\R_+$ for $h\notin\{j-1,\ldots,l-1\}\cup\{n{\mi_{\N'} 1}\},$  where $\mi$ is calculated with respect to the index set at the time of the application of $y_h(t).$
\item 
Upper embeddings with respect to a chord which starts before $(j,j+1)$ and ends at $(l,l+1)$ or after. It has the form $\uemb_i(t_1,t_2,t_3,s_1),~t_1,t_2,t_3,s_1\in\R_+,$ where $i{\pl 1}\leq j$ and $i{\mi 3}\geq l.$
\end{enumerate}
All these operations preserve the form and the sign of the twistor $\langle j,j+1,l,l+1\rangle.$ Indeed, for $\pre_h,\inc_h$ this follows from Lemmas~\ref{lem:effect of  pre} and~\ref{lem:effect_of_inc}, where for $\inc_h$ we use that $|I \cap [h]|$ is either $0$ or $4$ for each twistor which participates in the expressions for the determinants, by the restrictions on $h.$
For $x_h,y_h$ this follows from Lemma~\ref{lem:effect_of_x_y} and the restrictions on $h.$
For upper embeddings it follows from Remark~\ref{rmk-truncations} and the constraints on~$h$. Thus, each step keeps this twistor positive, for all~$Z$.
\end{proof}

{
\begin{rmk}
The method from the proof of Proposition~\ref{prop:inverse problem} for calculating the entries of a row that is supported on $m+1$ entries using twistor coordinates has previously appeared in the literature. The appendix of \cite[Appendix A]{agarwala2023cancellation} provides one example for $m=4.$ Another example for $m=2$ appears in \cite[Section 4]{parisi2021m}. The novelty of the above procedure is that it allows, in the context of BCFW cells, to calculate entries in rows of larger support, which is generally the case for descendants.
\end{rmk}
}

Note that we can take the coefficients we get from applying Lemma~\ref{obs:solve_params_5} in the proof {of Proposition~\ref{prop:inverse problem}} to be simply the twistors from the statement of the lemma, and $t_1,t_2$ can be expressed also as functionaries. This enables the following definition.

\begin{definition}
\label{def:func val mat}
    For $S\in\mathcal{BCFW}_{n,k}$ let $D$ be the corresponding chord diagram. The unique preimage $C$ of an arbitrary $Y\in\Z(S)$ can be expressed as a functionary-valued matrix which is a $D$-domino matrix. Denote this representation $C_D(Y,Z)$.
\end{definition}

\begin{ex}
\label{ex:func val mat}
Consider the cell of $D=([8],((1,2,6,7),(2,3,6,7),(4,5,6,7))\in\mathcal{CD}_{8,3}$. It has the following domino form:
$$
\begin{array}{*{8}{c}}
\color{lightgray} 1 & 
\color{lightgray} 2 & 
\color{lightgray} 3 & 
\color{lightgray} 4 & 
\color{lightgray} 5 & 
\color{lightgray} 6 & 
\color{lightgray} 7 & 
\color{lightgray} 8 \\
\hline 
\alpha_1 & \beta_1 & 0 & 0 & 0 & \gamma_1 & \delta_1 & \varepsilon_1 \\
\!\varepsilon_2\alpha_1 & \!\!\!\varepsilon_2\beta_1{+}\alpha_2\!\!\! & \beta_2 & 0 & 0 & \gamma_2 & \delta_2 & 0\\
0 & \varepsilon_3\alpha_2 & \varepsilon_3\beta_2\! & \!\alpha_3\! & \!\beta_3\! & \gamma_3 & \delta_3 & 0 \\ \hline
\end{array} 
\vspace{0.5em}
$$
We derive the functionary-valued matrix $C_D(Y,Z)$. As in the proof of Proposition~\ref{prop:inverse problem}, we fill the domino matrix with functionaries, row by row. For a top chord, such as Row~1, the domino variables $\alpha_i,\beta_i,\gamma_i,\delta_i,\varepsilon_i$ are the five twistors obtained using Lemma~\ref{obs:solve_params_5}. For a child chord, such as Row~3, we use the known $\alpha_p, \beta_p$ of the parent to reduce the problem to five effective entries, and recover $\alpha_i,\beta_i,\gamma_i,\delta_i,\varepsilon_i$ from the lemma. It follows that $\varepsilon_i$ is a twistor, and the other four are combinations of two twistors with the functionary coefficients $\alpha_p, \beta_p$. Functionaries for a sticky child such as Row~2 may be recovered as either of the two cases, and here we demonstrate the latter.

We use the notation $\langle abc | xy | def \rangle = \langle abcx \rangle\,\langle ydef \rangle - \langle abcy \rangle\,\langle xdef \rangle$ for such quadratic functionaries, and 
$\langle abc|xy|de|zw|ghi\rangle = \langle abc|xy|dez \rangle\,\langle wghi\rangle - \langle abc|xy|dew \rangle\,\langle zghi\rangle$ for cubic functionaries of this form. The resulting functionary-valued matrix is the following one, broken into two parts to fit in the page.
$$
\begin{array}{*{8}{c}}
\color{lightgray} 1 & 
\color{lightgray} 2 & 
\color{lightgray} 3 & 
\color{lightgray} 4 \\
\hline 
\langle2678\rangle& -\langle1678\rangle & 0 & 0 \\
\langle2367\rangle\langle2678\rangle & -\langle2367\rangle\langle1678\rangle-\langle678|12|367\rangle& \langle678|12|267\rangle & 0 \\
0 & -\langle4567\rangle\langle678|12|367\rangle & \langle4567\rangle\langle678|12|267\rangle &  -\langle678|12|67|23|567\rangle \\ \hline
\end{array} 
\hspace{1.9cm}
\vspace{0.5em}
$$
$$
\hspace{4.8cm}
\begin{array}{*{8}{c}}
\color{lightgray} 5 & 
\color{lightgray} 6 & 
\color{lightgray} 7 & 
\color{lightgray} 8 \\
\hline 
0 & \langle1278\rangle & -\langle1268\rangle & \langle1267\rangle \\
0 & -\langle678|12|237\rangle & \langle678|12|236\rangle & 0\\
\langle678|12|67|23|467\rangle &  -\langle678|12|67|23|457\rangle & \langle678|12|67|23|456\rangle & 0 \\ \hline
\end{array} 
\vspace{0.5em}
$$
Note that it is possible to simplify the expression in the second row by rescaling by a factor of~$\langle2678\rangle^{-1}$, and obtain a row of five twistors as expected.
\end{ex}

For an arbitrary $Y\in\Gr_{k,k+4}$ the matrix $C_D(Y,Z)$ is always a $D$-domino matrix, even though it might not be of full rank, or nonnegative. The following generalization of Proposition~\ref{prop:inverse problem} will be useful in later sections. Recall the subspace $\widetilde{\SA}$ from Definition~\ref{sda}.

\begin{prop}
\label{prop:injectiveness_more_general}
Let $D$ be a chord diagram and $Z$ a positive matrix. Consider the subspace of the (whole) Grassmannian consisting of vector spaces that have a $D$-domino representative in the sense of Definition~\ref{def:domino_entries}, but do not necessarily satisfy the sign constraints of Definition~\ref{def:domino_signs}. Denote by $S^\diamondsuit=S^\diamondsuit_D$ the subset of this space made of elements which do not belong to $\widetilde{\SA}$. Then the restriction of $\Z$ to $S^\diamondsuit$ is injective. The unique preimage of every $Y\in\Z(S^\diamondsuit)$ is precisely $C_D(Y,Z)$. In particular, the inverse map from $\Z(S^\diamondsuit)$ to $S^\diamondsuit$ is smooth.
\end{prop}

\begin{proof}
First of all, the amplituhedron map is defined on $S^\diamondsuit,$ by the definition of $\widetilde{\SA}$. 

The proof of the proposition is almost identical to the proof of Proposition~\ref{prop:inverse problem}, differing only in the verification of the assumptions of Lemma~\ref{obs:solve_params_5}{, and in that now it is given that the preimage of $Y$ has a representation in the domino form of $D$ (but now with weaker constraints defining $S^\diamondsuit$)}. 
In the previous proof, {the former} relied on Lemma~\ref{lem:non_zero_dets_for_domino}, while here it relies on {the definition of $S^\diamondsuit_D$ and Remark \ref{rmk:for_SA_usage}.}
We only list the changes.

We aim to reconstruct $C\in S^\diamondsuit$ from $Y=\Z(C).$ In the course of calculating the row $C_i,$ which corresponds to the chord $c_i=(j,j+1,l,l+1),$ we first need to check that the twistor $\langle j,j+1,l,l+1\rangle\neq 0.$ 
In the current setting, the non-vanishing of $\langle j,j+1,l,l+1\rangle$ is a consequence of $S^\diamondsuit\cap\widetilde{\SA}=\emptyset.$ 

The only other assumption we need to check is that $v,Z_j,Z_{j+1},Z_l,Z_{l+1}$ are linearly independent, where $v=Z_n$ if $c_i$ is a top chord, and otherwise $v=t_1Z_h+t_2Z_{h+1},$ using the domino $C_i$ inherits from the parent. 
In the case of a top chord, there is no change in the proof. In the other case, the argument used to show these five vectors are linearly independent can only fail if either $v=0$ or $v$ is proportional to $Z_j,$ where the latter option might happen only if $c_i$ is a sticky child. In both situations, using Remark \ref{rmk:for_SA_usage}, $\langle j,j+1,l,l+1\rangle=0,$ which implies $C\in \widetilde{\SA}${, contradicting our assumptions.}

We note that, by construction, all rows that are calculated throughout this process belong to the vector space represented by $C.$ By the assumption $C$ is $D$-domino matrix in the sense of Definition~\ref{def:domino_entries}, but may violate the sign rules of Definition~\ref{def:domino_signs}. 
{
The uniqueness of the preimage follows from Lemma~\ref{obs:solve_params_5} as well, as was the case in Proposition~\ref{prop:inverse problem}. The lemma implies that each of the rows we have constructed is unique up to scaling. On the other hand, if two non-equivalent $C_1,C_2\in S^\diamondsuit$ have the image $Y$, they must have some non-equivalent row, contradicting the aforementioned uniqueness up to scaling. The possible violation of sign rules does not matter for this argument, and since $S^\diamondsuit$ was defined directly by the domino form without reference to positroid cells, we do not require anything parallel to Theorem~\ref{thm:domino}.
}

Finally, the inverse map is smooth since the resulting preimage is given as the functionary-valued matrix~$C_D(Y,Z)$.
\end{proof}

We end this section with a corollary of Lemma~\ref{lem:non_zero_dets_for_domino} and Theorem~\ref{thm:injectiveness} about the signs of different twistors on a BCFW cell. It is stated with the index set $\{1,\dots,n\}$ for convenience.

\begin{lemma}
\label{lem:non_vanishing_boundary_twistors}
Let $S \in \mathcal{BCFW}_{n,k}$ be a BCFW cell.
\begin{enumerate}
\item
\label{it:bdry} 
Let $i,j\in[n]$ such that $i+1<j$ and $|\{i,i+1,j,j{\pl1}\}|=4.$ Then the boundary twistor $\langle i,i+1,j,j{\pl1}\rangle$ has a fixed sign on $\Z(S),$ which is $+1$ unless $j=n$ and $k$ is odd.
\item\label{it:row_of_5}
If $c_l=(i,i+1,j,j+1)$ is a top chord and $\{i_1,\ldots,i_4\}=\{i,i+1,j,j+1,n\}\setminus\{p\}$, 
then the twistor
$\langle Y~Z_{i_1}~Z_{i_2}~Z_{i_3}~Z_{i_4}\rangle$ has a fixed sign on $\Z(S)$, which is $+1$ if $p=n$, $(-1)^{k-l}$ if $p=i$, $(-1)^{k-l+1}$ if $p=i+1$, $(-1)^{k-l-\dd(c_l)}$ if $p=j$, and $(-1)^{k-l-\dd(c_l)+1}$ if $p=j+1$.  
\end{enumerate}
\end{lemma}

\begin{proof}
For the first item, by Remark~\ref{rmk:for_SA_usage} and Lemma~\ref{obs:SA_and_bdry_twistors}, it is enough to show that $S$ is not contained in $\SA,$ where $\SA$ is defined in Definition \ref{sda}. But, by Remark~\ref{rmk:for_SA_usage} again, if $S$ were contained in $\SA,$ then some twistor $\langle i,i\pl1,j,j\pl1\rangle$ would have vanished identically on $S.$ This would imply that $\Z(S)$ is contained in a codimension one subspace of $\Gr_{k,k+4},$ but then $\Z$ cannot be injective on $S,$ in contradiction to Theorem~\ref{thm:injectiveness}.

The second item is an immediate consequence of the proof of  Lemma~\ref{lem:non_zero_dets_for_domino}, where the signs appear in the calculation of Lemma~\ref{lem:non_zero_dets_for_domino},\eqref{it:non_zero_twistors_for_rows_with_5} and we use that $\behind(c_l)=k-l$. 
\end{proof}

\section{Separation}\label{sec:separation}
In this section we prove Theorem~\ref{thm:separation}. We find for each pair $(S_a,S_b)$ of BCFW cells a separating functionary, meaning a functionary that has, for every positive $Z,$ fixed and opposite signs on $\Z(S_a)$ and $\Z(S_b).$
Finding the separating functionary involves analyzing the corresponding chord diagrams, $D_a$ and~$D_b.$ We first analyze cases where the two cells are separated by a twistor. We show that such a twistor exists if $D_a$ and $D_b$ have rightmost top chords with different ends. If the rightmost top chords of $D_a$ and $D_b$ have same ends and different starts, then we find a quadratic functionary that separates the cells. Finally, if the two diagrams have the same rightmost top chord then we use induction.

Both in finding a quadratic separator, and in the induction step, we crucially rely on the results of Section~\ref{sec:promotion} and Corollaries~\ref{cor:generation_left}-\ref{cor:generation_top}. We show a separation for simpler subdiagrams, explicitly or by induction, and analyze their evolution under the extension steps of Corollaries~\ref{cor:generation_left}-\ref{cor:generation_top}.

The first proposition shows that certain twistors have certain fixed signs on some BCFW cells. It is used below for separation of cells by twistors.

\begin{prop}
\label{prop:0 in first and one before last col}
Let $Z$ be a positive $N\times(k+4)$ matrix, and $i=\smax{N}\mi1$, where $\smax{N} = \max N$.
Then 
\[\langle j,j\pl1,\smax{N}\mi2,\smax{N}\rangle\;\geq\; 0\] 
for every point of $\Z(\pre_i(\Gr^{\geq}_{k,\N\setminus\{i\}})),$ and every $j<\smax{N} \mi 3$.

Moreover, this twistor is strictly positive on $\Z(S)$ if $S$ is a BCFW cell whose corresponding chord diagram does not include a chord ending at~$(\smax{N}\mi2,\smax{N}\mi1).$
If $S$ is a BCFW cell whose chord diagram contains a chord ending at $(\smax{N}{\mi2},\smax{N}{\mi1})$ then there exists $j<\smax{N}\mi3$ for which 
$$\langle j,j{\pl1},\smax{N}{\mi2},{\smax{N}} \rangle\;<\;0$$ 
for every point of $\Z(S),$ for all $Z.$
\end{prop}

\begin{proof}
The first claim is known; see for example the closely related \cite[Lemma 4.3]{bao2019m}. It follows immediately from the Cauchy-Binet formula in Lemma~\ref{eq:Twistor Cauchy-Binet} by noting that in this case nonzero terms in the summation must have $i\notin J,$ and therefore $s(J,I)=1.$ 

For the 'More-over' part, note that $S=\pre_i(S'),$ for $S'$ a BCFW cell on the marker set $N\setminus\{i\}.$  On the image of $S'$ in the amplituhedron $\langle j,j\pl1,\smax{N}\mi2,\smax{N}\rangle>0,$ by Lemma~\ref{lem:non_vanishing_boundary_twistors},\eqref{it:bdry} applied to $S'$. Note that here $\mi$ is calculated with respect to $N$. 
Applying $\pre_i$ and using Lemma~\ref{lem:effect of  pre} we obtain the result. When the chord diagram contains a chord that ends at $(\smax{N}\mi2,\smax{N}\mi1)$, there is a unique such chord that is also a top chord. It is supported on $\{j,j\pl1,\smax{N}\mi2,\smax{N}\mi1,\smax{N}\}.$ Applying Lemma~\ref{lem:non_vanishing_boundary_twistors},\eqref{it:row_of_5} establishes the claim in this case. 
\end{proof}

The next proposition shows that certain quadratic functionaries have certain fixed signs on certain BCFW cells. It is used below for separating cells by quadratics.

\begin{prop}
\label{prop:separation_by_nesting_chord}
Let $Z$ be a positive $n\times(k+4)$ matrix. Let $\D\in\CD_{n,k}$ be a chord diagram with a top chord~$c$ whose markers are $(i,i+1,n-2,n-1)$. Then, at every point of $\Z(S_D),$ the functionaries 
$$ \favorite{i,i+1}{j,j+1}{n-2,n-1}{n}<0 $$
for every $j\in[n-4]\setminus\{i\}.$ 
\end{prop}

\begin{rmk}
\label{rmk:separation_by_nesting_chord}
In Proposition~\ref{prop:separation_by_nesting_chord} it is equivalent to say that at every point of  $\Z(S_D)$,
$$ \favorite{j,j+1}{i,i+1}{n-2,n-1}{n}>0 $$ 
for every $j\in[n-4]\setminus\{i\}$. The two statements are equivalent by the Pl\"ucker's relations, see Lemma~\ref{obs:plucker_functionary}. This alternative writing is more natural when $j<i$.
\end{rmk}

\begin{proof}
We first consider $j>i.$
Denote by $\D_c$ the sub diagram of $\D$ which contains the descendants of $c$ and whose index set is $\{i+1,\ldots,n\}.$ By Corollaries~\ref{cor:generation_left} and~\ref{cor:generation_top}, $S_D$ can be constructed as follows.
\begin{itemize}
\item Apply the upper embedding $\uemb_i$ 
to $\R_+^4\times S_c.$
\item Apply to the resulting cell a sequence of operations from: 
\begin{itemize}
\item $\pre_h,~\inc_h$ for $h<i$
\item $x_h(t),~t>0,~ h\leq i$ 
\item $y_h(t)$, $t>0,$ $h<i-1$ or $h=n$
\end{itemize}
\end{itemize}

By Lemma~\ref{lem:non_zero_dets_for_domino},\eqref{it:non_zero_twistors_for_rows_with_5} the twistors $\langle i_1,i_2, i_3, i_4\rangle\neq 0$ for all $\{i_1,\ldots,i_4\}\subset\{i,i+1,n-2,n-1,n\}$ {on $\Z(S_D)$}. By Lemma~\ref{lem:non_vanishing_boundary_twistors},\eqref{it:bdry} all boundary twistors are nonzero.
Thus we can use Corollary~\ref{cor:upper_emb} to deduce that after applying the upper embedding, $$ \favorite{i,i+1}{j,j+1}{n-2,n-1}{n}<0 $$ for every $j\in\{i+1,\ldots,n-4\}.$ 

The iterations of $\pre_h$ do not change the negativity, by Lemma~\ref{lem:effect of  pre}. 
For iterations of $\inc_h,$ we have $I\cap[h]=\emptyset.$ By Lemma~\ref{lem:effect_of_inc} again the sign remains.
Finally, for steps of the form $x_h,y_h,$ our constraint on $h$ guarantees that the form or sign of the functionary do not change, by Lemma~\ref{lem:effect_of_x_y}. 
This is clear for all such steps except $x_i,y_n$, however, in these cases the observation in the end of Lemma~\ref{lem:effect of  pre} holds.
When $y_n$ is applied the index set must have some index lower than $i$, by Algorithm \ref{construct-matrix}, and therefore $n\pl1<i$ would not be in the functionary.
{For applications of $x_i$ we use the fact that $\favorite{i,i+1}{j,j+1}{n-2,n-1}{n}$ has a representation in which every twistor which contains $i$ also contains $i+1$. }
{Since all indices from $F$ are larger than $i$, by the description of upper promotion in Proposition~\ref{prop:upromotion}, the index $i$ may appear in a twistor in $F'$ only if it comes from a substitution of a twistor in $F$ containing $n$ or $n-1$, and in these cases the twistor of $F'$ containing $i$ also contains $i+1$.}

We turn to the case $j<i.$ By Corollary~\ref{cor:generation_right}, $S_D$ can also be constructed as follows.
Let $\D_b$ be the sub diagram of $\D$ made of chords whose support is contained in $[i+1]\cup\{n\}.$ Note that $i\leq n-4.$ 
Let $S_b$ be the corresponding positroid cell, then $S_D$ can be constructed from $S_b$ by
\begin{itemize}
\item Apply to $S_b$ $\pre_{n-1},~\inc_{n-2},$ followed by $y_{i+1},~x_{n-2}$ and $x_{n-1}.$
\item Apply to the resulting cells a sequence of operations from:
\begin{itemize}
\item
$\pre_h,\inc_h$ where $h\in\{i+2,\ldots,n-3\}$
\item
$x_h$ where $h\in\{i+2,\ldots,n-1\}$ 
\item
$y_h$ where $h\in\{i,\ldots, n-4\}$
\end{itemize}
As noted in the proof of Corollary~\ref{cor:generation_right}, this sequence contains a single $y_i(u)$ operation.
\end{itemize}

By Proposition~\ref{prop:0 in first and one before last col} after applying $\pre_{n-1}$ all twistors $\langle j,j+1,i+1,n\rangle$ for $j<i$ are positive. If $\D$ differs from $\D_b$ by exactly one lowest chord, then only the first step and one additional $y_i$ operation are performed. The combination of these operations is precisely the lower embedding $\lemb_{n-2}
(t_1,t_2,s_1,s_2).$
By Lemma~\ref{lem:non_zero_dets_for_domino},\eqref{it:non_zero_twistors_for_rows_with_5} the twistors $\langle i_1 i_2 i_3 i_4\rangle\neq 0$ for all $\{i_1,\ldots,i_4\}\subset\{i,i+1,n-2,n-1,n\}.$
Thus, we can employ Corollary~\ref{cor:promotion_by_lower_emb} and deduce that the twistor $\langle j,j+1,i+1,n\rangle$ is promoted to the functionary $$ \favorite{j,j+1}{i,i+1}{n-2,n-1}{n},$$ which is positive.

In the general case the situation is rather similar, but the treatment requires some modification. Write $c=c_l.$
First note that after the first step, every point of the positroid cell \[S'=x_{n-1}(s_2)x_{n-2}(s_1)y_{i+1}(t_1)\inc_{n-2}\pre_{n-1}(S_b)\] has a matrix representative $C$ whose $l$th has non zero entries only in positions $i+1,n-2,n-1,n.$ Therefore, by the Cauchy-Binet formula of Lemma~\ref{eq:Twistor Cauchy-Binet} the twistor $\langle i,n-2,n-1,n\rangle$ is always negative. Indeed, each contribution to it uses the entry $C_l^{i+1}\neq 0,$ and a minor of {the submatrix} $C_{[l-1]}^{[i]}{=(C_a^b)_{a\in[l-1],b\in[i]}}$ { of $C,$} and all of these contributions appear with the same minus sign. Moreover, at least one term is non zero, the minor $C_{[l-1]}^I$ for $I$ the set of first domino entries of each chord, just as in Lemma~\ref{lem:more than 4}. No operation from the series of operations of the second step affects the form of this twistor, by Lemmas~\ref{lem:effect of  pre},~\ref{lem:effect_of_inc},~\ref{lem:effect_of_x_y}. Its sign changes by $-1$ for every $\inc_h$ operation, since $|I\cap[h]|=|\{i\}|=1.$

Second, note that most operations do not affect the twistors $\langle j,j+1,i+1,n\rangle$. These include all $\pre_h$ and $x_h$
operations, and all $y_h$ operations, except for the single $y_i(u)$ operation, by Lemmas~\ref{lem:effect of  pre} and~\ref{lem:effect_of_x_y} respectively. Each application of $\inc_h$ preserves the form of the twistor, but changes the sign by $-1,$ since $|I\cap[h]|=|\{j,j+1,i+1\}|=3,$ by Lemma~\ref{lem:effect_of_inc}.

Let $k'$ be the number of applications of $\inc_h$ during the process, which come before the application of $y_i(u),$ where we include $\inc_{n-2}$ in the count. Then just before we apply $y_i(u),$
\[\sgn(\langle i,n-2,n-1,n\rangle)=\sgn(\langle j,j+1,i+1,n\rangle)=(-1)^{k'}.\]
By Lemma~\ref{lem:effect_of_x_y}, the application of $y_i(u)$ promotes $\langle j,j+1,i+1,n\rangle$ to 
\begin{equation}\label{eq:partial_lower_emb1}\langle \Z(C), Z_j, Z_{j+1}, (Z_{i+1}+uZ_i), Z_n\rangle=
\langle \Z(C), Z_j, Z_{j+1}, Z_{i+1}, Z_n\rangle+
u\langle \Z(C), Z_j, Z_{j+1}, Z_i, Z_n,\rangle\end{equation}where matrix argument $C$ is the matrix obtained after applying $y_i(u).$ Moreover, this expression has the same sign $(-1)^{k'},$ for every positive $Z$.

Now, $u$ can actually be read from $\Z(C).$
Indeed, before applying $y_i(u),$ the $l$th row of the matrix only had non zero entries at positions $i+1,n-2,n-1,n.$ Therefore, after applying $y_i(u),~C_l$ has non zero entries at positions $i,i+1,n-2,n-1,n$ and we can write
\[u=\frac{C_l^i}{C_l^{i+1}}.\]
We also know that $\langle i,n-2,n-1,n\rangle\neq 0.$ Thus, {we can apply Lemma~\ref{obs:solve_params_5} to $\Z(C)$} and the linearly independent vectors $Z_i,Z_{i+1},Z_{n-2},Z_{n-1},Z_n$ to deduce
\[u=-\frac{\langle i+1, n-2,n-1,n\rangle}{\langle i,n-2,n-1,n\rangle}.\]
Multiplying \eqref{eq:partial_lower_emb1} by $\langle i,n-2,n-1,n\rangle$ and substituting for $u$ we obtain that
\[\langle i,n-2,n-1,n\rangle\langle \Z(C), Z_j, Z_{j+1}, Z_{i+1}, Z_n\rangle-
\langle i+1,n-2,n-1,n\rangle\langle \Z(C), Z_j, Z_{j+1}, Z_i, Z_n\rangle
\]
has sign $(-1)^{k'+k'}=1.$
This expression is precisely the functionary 
\[\favorite{i,i+1}{n-2,n-1}{j,j+1}{n}.\]
By applying Pl\"ucker's relations, Lemma~\ref{obs:plucker_functionary}, we see that this functionary equals 
\[\favorite{j,j+1}{i,i+1}{n-2,n-1}{n},\]
which must therefore be positive. 
Finally, the additional operations of the second step of the construction of $S_D$ from $S_b$ do not change the form of this functionary, by Lemmas~\ref{lem:effect of  pre},~\ref{lem:effect_of_inc} and~\ref{lem:effect_of_x_y}. Only $\inc_h$ operations may affect the sign, but as above, $\inc_h$ alters the sign of each of the four twistors forming the functionary $\favorite{i,i+1}{n-2,n-1}{j,j+1}{n}$, hence does not affect the sign of the functionary, which remains positive.
As claimed.
\end{proof}

\begin{rmk}
\label{rmk:fake_lower_embedding}
The argument in the second part of the above proof can be applied to more general twistors, {and hence to general functionaries, with index set $[i+1]\cup\{n\}$}. A twistor of the form $\langle i_1,\ldots, i_4\rangle$ with $i_1\ldots, i_4\in [i]\cup\{n\}$ keeps its original form. It also keeps its original sign if $n\notin\{i_1,\ldots,i_4\},$ and it changes sign to $(-1)^{k'},$ {where $k'$ is as in the above proof,} times the original sign if $n\in\{i_1,\ldots,i_4\}$.
If one of the indices, say $i_4,$ equals $i+1,$ then the twistor is promoted to the functionary
\[\langle i,n-2,n-1,n\rangle\langle i_1,i_2,i_3, i+1\rangle-
\langle i+1,n-2,n-1,n\rangle\langle i_1,i_2,i_3, i\rangle,\] which has the same sign if $n\in\{i_1,i_2,i_3\},$ or $(-1)^{k'}$ times the same sign otherwise. {Thus, the effect of the construction of the second step on a functionary is precisely the lower promotion from Proposition~\ref{prop:upromotion}.} 
In all cases the resulting functionary is pure. Moreover, as in Remark~\ref{rmk:pureness_and_type}, a pure functionary $F$ with a fixed sign $s$ and multiplicities $d_j(F)$ as in Definition~\ref{def pure}, is promoted to another pure functionary $F'$ with a fixed sign. The multiplicity of each index $i,n-2,n-1,n$ increases by $d_{i+1}(F)$. The sign $s$ changes by~$(-1)^{d_{i+1}(F)+d_n(F)}$.
\end{rmk}

We are now in position to prove the main theorem, Theorem~\ref{thm:separation}.

\begin{proof}[Proof of Theorem~\ref{thm:separation}]
We show that for every two different BCFW cells of a certain~$n$, even with different values of~$k$, there is a \emph{pure separating functionary}. A pure separating functionary {for a pair of positroid cells} is a pure functionary having fixed and opposite signs on the two {images of} cells {(not necessarily in the same amplituhedron)}, independently of the positive {matrices~$Z_a,Z_b$ used for projection to the appropriate amplituhedra. 
}
We proceed by induction on~$n$. The case where $n=4$ and $k=0$ is trivially true, since there is a single cell. Assuming the claim for all $n'<n$, we consider two different BCFW cells, $$ 
S_a, S_b \;\in\; \bigcup_{k=0}^{n-4}\mathcal{BCFW}_{n,k}
$$ Let $D_a$ and $D_b$ denote the corresponding chord diagrams. We consider several cases.

\medskip
\noindent
\textbf{(A) One Diagram is Trivial}

\noindent
Suppose that $\D_b$ is the diagram with no chords. Then all twistors, which are the determinants of $4 \times 4$ minors of {$Z_b$}, are positive on the image of~$S_b$.
Since $\D_a\neq\D_b$, the other diagram $\D_a$ has at least one chord, and hence at least one top chord. Let $c=(i,i+1,j,j+1)$ be the last top chord. Then by Lemma~\ref{lem:non_vanishing_boundary_twistors}~(\ref{it:row_of_5}), the twistor
$\langle i,i+1,j,n\rangle$ is always negative on the image of $S_a$.

\medskip
\noindent
\textbf{(B) Common Unused Markers }

\noindent
If both $D_a$ and $D_b$ do not have a chord in the segments next to some marker $h \in \{1,\dots,n-1\}$, then $S_a$ and $S_b$ have a common zero column at~$h$. It follows that $S_a=\pre_{h}(S_{a'})$ and $S_b=\pre_{h}(S_{b'})$ for some BCFW cells $S_{a'}\neq S_{b'}$ with $n-1$ columns. By the induction on~$n$, there is a separating functionary between $S_{a'}$ and~$S_{b'}$. Applying Lemma~\ref{lem:effect of  pre}, this functionary is promoted to a separating functionary between $S_a$ and~$S_b$.

\medskip
\noindent
\textbf{(C) Different Last Used Marker}

\noindent
Suppose that exactly one of $D_a$ and $D_b$ has a top chord of the form $(i,i+1,n-1,n-2)$. Then the other cell must have a zero column at its $n-1$ coordinate. In this case, the twistor $\langle i,i+1,n-2,n \rangle$ separates them, by Proposition~\ref{prop:0 in first and one before last col}.

\medskip
\noindent
\textbf{(D) Different Last Chords}

\noindent
Suppose that $D_a$ and $D_b$ have last top chords $(i,i+1,n-2,n-1)$ and $(j,j+1,n-2,n-1)$ respectively, for $i\neq j$. Then Proposition~\ref{prop:separation_by_nesting_chord} {and Remark~\ref{rmk:separation_by_nesting_chord} imply} that the sign of the functionary $\favorite{i,i+1}{j,j+1}{n-2,n-1}{n}$ differs between the images of $S_a$ and~$S_b$. {Note that the separating functionaries in cases (A)-(D) are indeed pure.}

\medskip
\noindent
\textbf{(E) Different Last Subdiagrams}

\noindent
The remaining case is $i=j$, meaning that both $\D_a$ and $\D_b$ have the same last top chord $c=(i,i+1,n-2,n-1)$.
Here induction may produce higher degree functionaries.
Since $D_a\neq D_b$, they either differ in the chords descendent from~$c$, or in those before~$c$.
We first consider the former case, {and case (F) below will be the latter}. {The analysis here resembles the first case in the proof of Proposition \ref{prop:separation_by_nesting_chord}.} Let $\D_{a'}$ and $\D_{b'}$ be the respective subdiagrams of $\D_a$ and $\D_b$, induced on the marker set $\{i+1,\dots,n\}$, and let $S_{a'}$ and $S_{b'}$ be the corresponding cells. Note that $\D_{a'}$ and $\D_{b'}$ may consist of a different number of chords, and one of them may even be empty.

{By induction on $k$, there is a pure separating functionary $F$ for the pair of positroid cells $S_{a'}$ and~$S_{b'}$. Here it is crucial that our construction of separating functionaries works for pairs of BCFW cells of the same $n$ also when the values of $k$ are different.}
It follows from Corollary~\ref{cor:generation_top} and Corollary~\ref{cor:generation_left} that $S_a$ and $S_b$ can be constructed from $S_{a'}$ and $S_{b'}$ by the following procedure.
\begin{itemize}
\item Apply the upper embedding $\uemb_i$ 
to both $\R_+^4\times S_{a'}$ and $\R_+^4\times S_{b'}$.
\item Apply to the resulting cells two {possibly different} sequences of operations from: 
\begin{itemize}
\item $\pre_h,~\inc_h$ for $h<i$
\item $x_h(t),~t>0,~ h\leq i$ 
\item $y_h(t)$, $t>0,$ $h<i-1$ or $h=n$
\end{itemize}
\end{itemize}
As in the proof of Proposition~\ref{prop:separation_by_nesting_chord}, by Lemma~\ref{lem:non_zero_dets_for_domino},\eqref{it:non_zero_twistors_for_rows_with_5} all twistors $\langle i_1,\ldots, i_4\rangle\neq 0$ for $\{i_1,\ldots,i_4\}\subset\{i,i+1,n-2,n-1,n\}$ {on both $\Z_a(S_a)$ and $\Z_b(S_b)$}. Thus we can apply Proposition \ref{prop:uprom} to deduce that after the first step, the {pure functionary $F$ is promoted to a pure functionary $F',$ and that $F'$ separates the two cells constructed after in the first step}. Importantly, {Since all indices from $F$ are larger than $i$, by the description of upper promotion in Proposition~\ref{prop:upromotion}, the index $i$ may appear in a twistor in $F'$ only if it comes from a substitution of 
 a twistor in $F$ containing $n$ or $n-1$, and in these cases the twistor of $F'$ containing $i$ also contains $i+1$.}

Applications of $\pre_h$ do not change the form or sign of the functionary, by Lemma~\ref{lem:effect of  pre}. They also maintain the special property that every twistor which includes the index $i$ contains also $i+1.$
A similar claim holds for iterations of $\inc_h,$ where this time Lemma~\ref{lem:effect_of_inc} is called, and we use $I\cap[h]=\emptyset.$
Using Lemma~\ref{lem:effect_of_x_y} and the conditions on $h,$ also steps of $x_h,~y_h$ do not affect the form of the promoted functionary, and do not change the signs of the functionary, besides possibly $x_i$ and $y_n$. However, in these cases the observation in the end of Lemma~\ref{lem:effect of  pre} holds.
For $x_i$, the special property that every twistor in $F$ which contains $i$ also contains $i+1$ guarantees that the functionary is again unaffected by promotion by Lemma~\ref{lem:effect_of_x_y}.
For $y_n$, by Algorithm \ref{construct-matrix}, whenever it is applied we must have $\min(N)<i$, and so $n\pl1 = \min(N)$ does not appear in any twistor in $F'$, and the functionary is unaffected by promotion as well by Lemma~\ref{lem:effect_of_x_y}.
Thus, the separating functionary is promoted in all cases to the same functionary, and its sign is still a witness for separation.

\medskip
\noindent
\textbf{(F) Same Last Subdiagram}

\nopagebreak
\noindent
Finally, suppose that the subdiagram which descends from $c$ is the same for both $\D_a$ and $\D_b.$ {The analysis of this case resembles that of the second case of Proposition \ref{prop:separation_by_nesting_chord}.} Let $\D_{a'}$ and $\D_{b'}$ be the respective chord diagrams obtained by erasing $c$ and its descendants. Their marker set is $\{1,\dots,i,i+1,n\}$, and $\D_{a'}\neq \D_{b'}$ since $\D_a\neq \D_b$, so at least one of them contains chords. By induction, there is a pure separating functionary $F$ between the corresponding images of $S_{a'}$ and~$S_{b'}.$ This time we apply Corollary~\ref{cor:generation_right} to construct the cells $S_a$ and $S_b$ from $S_{a'}$ and $S_{b'}$ respectively.

\begin{itemize}
\item Apply to both cells $\pre_{n-1},~\inc_{n-2},$ followed by $y_{i+1},~x_{n-2}$ and $x_{n-1}.$
\item Apply {to each of the two} resulting cells the same sequence of operations from:
\begin{itemize}
\item
$\pre_h,\inc_h$ where $h\in\{i+2,\ldots,n-3\}$
\item
$x_h$ where $h\in\{i+2,\ldots,n-1\}$ 
\item
$y_h$ where $h\in\{i,\ldots, n-4\}$
\end{itemize}
As noted in the proof of Corollary~\ref{cor:generation_right}, this sequence contains a single $y_i(t)$ operation.
\end{itemize}
{In case the second step comprises only of the single $y_i(t)$ operation, we are in the context of the lower promotion of Proposition \ref{prop:lprom}, and it follows from that proposition that $F$ is promoted to a pure functionary $F'$ which separates the two cells.}
{In the general case, arguing as in the analogous part {in the proof} of Proposition~\ref{prop:separation_by_nesting_chord}, and as {elaborated} in Remark~\ref{rmk:fake_lower_embedding}, the separating {pure} functionary $F$ is promoted to a pure functionary $F'$ which separates $S_{a}$ from $S_{b},$ for all {$Z_a$,$Z_b$. }
}

\medskip
In conclusion, the cases (A)-(F) cover all possible configurations of two different chord diagrams, and the theorem follows.
\end{proof}

\begin{rmk}
\label{rmk:the recursive cases explicitly}
In case~(E), the functionary $F'$ is obtained from $F$ by the substitution 
for the upper promotion in Proposition~\ref{prop:upromotion}.
Similarly, in case~(F), the functionary~$F'$ is obtained from~$F$ by the substitution for the lower promotion in Proposition~\ref{prop:upromotion}.
Suppose that $d_j(F)$ are the multiplicities of indices in~$F$, as in Definition~\ref{def pure}, and $s$ is its sign on the cell. In case~(E), the multiplicities of $\{i,i+1,n-2,n\}$ increase by~$d_{n-1}(F)$, and those of $\{i,i+1,n-2,n-1\}$ increase again by $d_{n}(F)$ by Remark~\ref{rmk:pureness_and_type}. 
{The sign changes by $(-1)^{d_{n-1}(F)}$ by Proposition \ref{prop:uprom}.}
In case~(F), the sign and type change as in Remark~\ref{rmk:fake_lower_embedding}.

This gives a straightforward recursive algorithm to compute a separating functionary between every two given cells, corresponding to two chord diagrams with $k_1,k_2 \in \{0,\dots,n-4\}$. We provide an implementation of this function in~\cite{evenzohar2022bcfw}, and its complete output for $n=6,7,8,9,10$.
\end{rmk}

\section{Boundaries of BCFW Cells}\label{sec:precise_ineqs}
In this section we complete the characterization of the BCFW cells via explicit inequalities on its domino entries and some of their $2\times 2$ minors. We moreover characterize the codimension one boundaries of those cells.
We show that many boundary components are shared by pairs of cells. We characterize those, as well as those which are not. {The techniques used in this section rely on the machinery developed in Section \ref{sec:domino}, and this section can be viewed as a continuation of it.} {The results of this section will be used in Section \ref{sec:surj} to prove Theorem \ref{thm:surj}.}

{
From here until the end of this section, for a chord diagram $D,$ denote the corresponding BCFW cell by $S_D$ and the corresponding decorated permutation by $\pi_D$.}

\subsection{Representation by Inequalities}

Recall Definition~\ref{def:domino_signs}. The following corollary is a simple consequence of Proposition~\ref{prop:domino}.
\begin{cor}\label{cor:signs_of_1by1_2by2}
The domino matrix $C$ of every point in a BCFW cell $S \in \mathcal{BCFW}_{n,k}$ satisfies the following constraints.
\begin{enumerate}
\item 
The entries of the starting domino of each row are positive.
\item 
The entries of the ending domino of $C_i$ have sign $(-1)^{\dd(c_i)}.$
\item 
The $n$th entry, if $c_i$ is a top chord, is of sign $(-1)^{k-i}.$
\item 
If $C_i$ is not a top chord, then the constant which multiplies the domino entries inherited to $C_i$ from its parent is of sign $(-1)^{\ddp(c_i)}.$
\item 
If $c_j$ is a descendant of $c_i$ which ends at the same domino $(l,l+1)$ then the determinant of the $2\times2$ minor $C^{l,l+1}_{i,j}$ is of sign $(-1)^{\dd(c_i)-\dd(c_j)+1}.$
\item 
If $c_j$ starts in the domino $(l,l+1)$ in which $c_i$ ends, then the determinant of the $2\times2$ minor $C^{l,l+1}_{i,j}$ is of sign $(-1)^{\dd(c_i)}.$
\end{enumerate}
Indeed, the first four items are as in Definition~\ref{def:domino_signs}. The last two are obtained by rewriting the last two items in that definition using the previous items.  
\end{cor}

Recall Definitions~\ref{def:domino_entries},~\ref{def:domino_signs}.
\begin{nn}[{Domino Variables}]\label{nn:Var}
For $i=1,\ldots, k$ let 
\[\p_i=(-1)^{\text{behind}(c_i)}\varepsilon_i,\]
if $c_i$ is a top level chord, and otherwise \[\p_i(-1)^{\ddp(c_i)}\varepsilon_i.\] 
Let \[
\El_i=(-1)^{\dd(c_i)}\gamma_i,~\Er_i=(-1)^{\dd(c_i)}\delta_i.\]
{Note that if $c_i$ is a sticky child of $c_p$ and $(h,h+1)$ is the start of $c_i$ then $C_i^h=\Sl_i+\p_i\Sr_p.$
If $c_j$ is a same-end descendant of $c_p,$ and the end of both is $(l,l+1)$ write
\[\ee_{pj}=(-1)^{\dd(c_p)-\dd(c_j)+1}\det(C_{pj}^{l,l+1})={(-1)^{\dd(c_p)-\dd(c_j)+1}(\gamma_p\delta_j-\gamma_j\delta_p)}.\]}
If $(c_l,c_r)$ are head-to-tail chords, whose common markers are $(h,h+1)$ put
\[\es_{lr}=(-1)^{\dd(c_l)}\det(C_{lr}^{h,h+1}){=(-1)^{\dd(c_l)}(\gamma_l\beta_r-\delta_l\alpha_r)}\]
Write $\Var^1=\Var^1_D$ for the set of all elements $\p_i,\Sl_i,\Sr_i,\El_i,\Er_i,~i=1,\ldots,k.$ 
Denote by $\tVar=\tVar_D$ the set
{\[\Var^1_D\cup\{\ee_{pj}:~\text{$c_j$ is a same-end descendant of $c_p$}\}\cup\{\es_{lr}:~\text{$(c_l,c_r)$ are head-to-tail chords}\}.\] }
Finally, let $\Var=\Var_D$ be the subset of $\tVar$ defined by
\begin{itemize}
\item 
$\ee_{pj}$ appears only if $p$ is the parent of $j.$ In this case $\Er_p$ and $\El_j$ do not appear.
\item 
$\es_{lr}$ appears only if $c_l$ and $c_r$ are siblings. In this case $\El_l$ and $\Sr_r$ do not appear.
\item 
All other variables of types $\p_i,\Sl_i,\Sr_i,\El_i,\Er_i$ belong to $\Var.$
\end{itemize}
{For simplicity, when $D$ is fixed, or clear from context, we denote $\Var=\Var_D,$ $\Var^1=\Var^1_D,$ $\tVar=\tVar_D$. The domino variables are always associated with some domino form of a chord diagram $D$, but we leave this implicit in the notation as well, and let the reader deduce $D$ by context, again for simplicity of notation.}

\end{nn}
By Corollary~\ref{cor:signs_of_1by1_2by2} all the variables defined above are positive.
In this section and the next one it is more convenient to work with these domino variables and not the usual ones. 

\begin{prop}
\label{prop:minimal_set_of_ineqs}
Let $C$ be a domino matrix in the sense of Definition~\ref{def:domino_entries} for a chord diagram $D.$ For every $I\in\binom{[n]}{k},~\det(C^I)$ is a polynomial with nonnegative coefficients in the variables of $\tVar_D,$ which is nonzero precisely if the Pl\"ucker coordinate $P_I$ is nonvanishing on $S_D.$ 
Thus, if in addition $C$ satisfies the inequalities of Definition~\ref{def:domino_signs}, or equivalently the inequalities of Corollary~\ref{cor:signs_of_1by1_2by2}, then $C$ represents a point in $S_D.$
 
\end{prop}
\begin{proof}
Let $I=\{i_1,\ldots,i_k\}\in\binom{[n]}{k}$ be an arbitrary set. Expand $\det(C^I)$ to obtain
\[\sum_{f:[k]\to I,~\text{a bijection}}(-1)^{s(f)}\prod_{i=1}^k C_i^{f(i)},\]where 
\[s(f)=|\{(i,j)|i<j \text{ and }f(i)>f(j)\}.\]
Further expand in terms of the domino variables $\alpha_i,\beta_i,\gamma_i,\delta_i,\varepsilon_i$ to obtain
\begin{equation}\det(C^I)=\label{eq:det_before_varphi}\sum_{\substack{f:[k]\to I,~\text{a bijection},\\\mathcal{M}_f\neq \emptyset}}(-1)^{s(f)}\sum_{M\in\mathcal{M}_f}M,\end{equation}
where $\mathcal{M}_f$ is the collection of monomials in the domino variables which appear in the expansion of $\prod_{i=1}^k C_i^{f(i)}.$ The set $\mathcal{M}_f$ may include more than one element {only if for some $i$ $c_i$ is a sticky child of $c_p,$ and $f(i)$ is its first marker. In this case $C_i^{f(i)}=\alpha_i+\varepsilon_i\beta_p.$} 

We first simplify the sub \eqref{eq:det_before_varphi} by noting a collection of cancelling pairs. 
To this end define an involution $\varphi$ on $\sqcup_f\mathcal{M}_f$ as follows. 
Let $M$ be a monomial, and let $i<j\in[k]$ be the smallest indices, if such indices can be found, which satisfy
\begin{enumerate}
    \item  $c_i,c_j$ are either siblings or a parent and a child,
    \item $f(i),f(j)$ are the first two markers of $c_h,$ the parent of $c_j$ (which may be $c_i$).
    \item $M$ is divisible by $\alpha_h\beta_h.$
\end{enumerate}
In case such $i,j$ can be found, define $f':[k]\to I$ to be the bijection which agrees with $f$ everywhere on $[k],$ except at $i,j$ in which $f'(i)=f(j),~f'(j)=f(i).$ 
We note that $M\in\mathcal{M}_{f'}$ and define $\varphi(f,M)=(f',M).$
If there is no such pair, define $\varphi(f,M)=(f,M).$ $\varphi$ is clearly an involution. 
The meaning of the third condition above is that the contribution of $c_i,c_j$ to $M$ comes from the terms inherited by the parent of $c_j.$

Observe that $s(f)=s(f')+1.$ Thus if we change the order of summation in \eqref{eq:det_before_varphi}, to first sum over $\varphi$-equivalence classes, and then summing the terms $(-1)^{s(f)}M$ in each equivalence class, we see that the sum  equals the sum over fixed point of $\varphi,$ that is over monomials which do not contain $\alpha_h\beta_h$ for any~$h.$ Denote by $\mathcal{M}$ this collection of monomials which are fixed by $\varphi.$ 
We can thus write
\begin{equation}\label{eq:det_before_varphi_d}\det(C^I)=\sum_{(f,M)\in\mathcal{M}}(-1)^{s(f)}M.\end{equation}

{The next step is to collect the summands of \eqref{eq:det_before_varphi_d} in a clever way, to obtain a positive expression, at the expense of the need to introduce the variables $\ee_{ij},\es_{ij}$.}  
Let $\Dom$ be the collection of pairs $d=(h,h+1)$ which appear in $\D,$ either as starts or ends of chords. Given a bijection $f:[k]\to I,$ we say that two chords $c_i,c_j$ \emph{$f$-share} $d,$ if both contain $d$ as their start or end, and $\{f(i),f(j)\}=\{h,h+1\}.$
Define an involution $\varphi_d$ on $\mathcal{M},$ for each $d=(h,h+1)\in \Dom$ as follows. Consider $(f,M)\in\mathcal{M}.$ If there are two chords $c_i,c_j$ which $f$-share $d,$ we define $f'$ by $f'(i)=f(j),~f'(j)=f(i)$ and $f(l)=f'(l)$ otherwise. 
Note that for any $d$ there is at most one such pair $\{i,j\},$ {and that since $f$ is a bijection, any $c_i,$ for $i\in[k],$ can $f-$share at most one $d.$}
There is a natural bijection between $\mathcal{M}\cap\mathcal{M}_f$ and $\mathcal{M}\cap\mathcal{M}_{f'}.$ In case $d$ is the common end of $c_i$ and $c_j,$ and, without loss of generality $f(i)<f(j)$ then in terms of minors this bijection is just the multiplication by $\frac{\gamma_j\delta_i}{\gamma_i\delta_j}.$ In case $d$ is the end of $c_i$ and the start of $c_j,$ then it is the multiplication by $\frac{\alpha_j\delta_i}{\beta_j\gamma_i},$ if $f(i)<f(j)$ and $\frac{\beta_j\gamma_i}{\alpha_j\delta_i}$ otherwise. 
We define $\varphi_d(f,M)$ as $(f',M'),$ where $M'$ is the outcome of the aforementioned multiplication.
Note that again $s(f)=s(f')+1.$

Let $\Phi$ be the equivalence relation generated by $\{\phi_d\}_{d\in \Dom}.$
We show that the sum of all monomials in a $\Phi$-equivalence class is a polynomial with nonnegative coefficients in the elements of $\tVar_D,$ and thus prove the claim.

For $i=1,\ldots, k$ define $s_i(f)$ by 
\begin{itemize}
\item 
$\dd(c_i),$ the number of descendants of $c_i,$ if $f(i)$ belongs to the end of $c_i.$
\item 
$\ddp(c_i),$ the number of descendants of the parent of $c_i$ which start before $c_i,$ if $c_i$ is not a top level chord, and $f(i)$ belongs to the start of $c_i$'s parent. 
\item 
$\behind(c_i)=k-i,$ if $c_i$ is a top chord, and $f(i)=n.$
\item Otherwise $s_i(f)=0.$
\end{itemize}

For every $d=(h,h+1)\in\Dom,$ define $\Cont_d(f)$ to be $0,$ except in the following two cases in which it is $1:$
\begin{itemize}
    \item There are $c_i,c_j$ which $f$-share $d,$ $c_j$ is a same-end descendant of $c_i,$ and $f(j)>f(i).$
    \item There are $(c_i,c_j)$ are head-to-tail chords which $f$-share $d$ and $f(i)>f(j).$
\end{itemize}
{If $\Cont_d(f)=1$ we say that $f$ \emph{has a contradiction at $d$}. In this case, if $\varphi_d(f,M)=(f',M')$ we say that $f'$ is obtained from $f$ by \emph{resolving the contradiction at $d$}.}

For a bijection $f$ which is not cancelled by the first involution $\varphi$ we have \begin{equation}\label{sgn:with_contr}s(f)=
\sum_{i\in[k]}s_i(f)+\sum_{d\in \Dom}\Cont_d(f).\end{equation}
Indeed, suppose first that $f$ has no contradictions, then a pair $(i,j)$ where $i<j$ but $f(i)>f(j),$ which contributes to $s(f)$ may appear if only if
\begin{itemize}
    \item $f(i)$ belongs to the end of $c_i$ and $c_j$ is a descendant of $c_i.$ 
    \item $f(j)$ belongs to the start of the parent of $c_j,$ and $c_i$ is a descendant of that parent, which comes before $c_j$. 
    \item $c_i$ is a top chord, and $f(i)=n,$ $c_j$ starts after $c_i.$
\end{itemize}
These cases are enumerated exactly once in the definition of $\{s_i(f)\}_{i\in[k]}.$

{Suppose that there are contradictions in $f$. Note that if $f'$ is obtained from $f$ by resolving the contradiction at $d,$ then
\[s(f)-s(f')=1,\]\begin{align*}
\left(\sum_{i\in[k]}s_i(f)+\sum_{d\in \Dom}\Cont_d(f)\right)&-\left(\sum_{i\in[k]}s_i(f')+\sum_{d\in \Dom}\Cont_d(f')\right)\\&=\Cont_d(f)-\Cont_d(f')=1.\end{align*}
Thus, \eqref{sgn:with_contr} holds for $f$ if and only if it holds for $f'.$
Now, let $\tilde{f}$ be obtained from $f$ by resolving all contradictions. From the previous case we know that \eqref{sgn:with_contr} holds for $\tilde{f}.$ Since $f,\tilde{f}$ differ by a sequence of resolutions of contradictions, and the correctness of \eqref{sgn:with_contr} is preserved under such steps, we deduce that \eqref{sgn:with_contr} holds also for $f.$}

By Corollary~\ref{cor:signs_of_1by1_2by2} and the definition of $s_i$ we have
\begin{equation}\label{eq:sign_of_C_ij}(-1)^{s_i(f)}C_i^{f(i)}>0.\end{equation}

Every $\Phi$-equivalence class $R$ has a unique element $(f_R,M_R)$ where the bijection $f_R$ has the maximal (with respect to inclusion) set of contradictions $D_R\subseteq \Dom.$ Morevoer, any $(f,M)\in R $ is of the form $\varphi_{d_1}\circ\cdots\circ\varphi_{d_l}(f_R,M_R)$ for some $d_1,\ldots,d_l,$ and any such expression is in $R.$ Thus, the size of $R$ is $2^{|D_R|}.$ 
Set $I(D_R)=\bigcup_{d\in D_R}\{l_d,l_{d}+1\}$ where $(l_d,l_d+1)$ are the markers of $d.$ {By the definition of contradictions, for each $d\in D_R$ there are two indices $i_d,j_d\in[k]$ so that the chords $c_{i_d},c_{j_d}$ are responsible for this contradiction, and $\{f(i_d),f(j_d)\}=\{l_d,l_{d}+1\}.$ Since no $i\in[k]$ can participate in two contradictions (any $c_i$ can $f-$share at most one $d,$ as explained above). Thus, the markers pairs $\{l_d,l_{d}+1\}$ for $d\in D_R$ are all disjoint.}

{We will now see that summing the monomials corresponding to elements of $R$ results in a positive expression, which involves the $\ee_{ij},\es_{ij}$ variables. Each of these new variables corresponds to one contradiction in $f_R.$}
Using \eqref{sgn:with_contr},~\eqref{eq:sign_of_C_ij}
we can write $\sum_{(f,M)\in R}(-1)^{s(f)}M$ as 
\[\sum_{(f,M)\in R}\prod_{d\in D_R}(-1)^{\Cont_d(f)}\prod_{i\in[k]}(-1)^{s_i(f)}C_i^{f(i)}=\sum_{(f,M)\in R}\prod_{d\in D_R}(-1)^{\Cont_d(f)}\prod_{i\in[k]}|C_i^{f(i)}|.\]
This sum can be reorganized to
\begin{equation}\label{eq:decomp_into_2by2}
    \prod_{j\notin I(D_R)}|C_{f_R^{-1}(j)}^j|\prod_{d\in D_R}(
    |C_{f_R^{-1}(j_d+1)}^{j_d}||C_{f_R^{-1}(j_d)}^{j_d+1}|
    -|C_{f_R^{-1}(j_d)}^{j_d}||C_{f_R^{-1}(j_d+1)}^{j_d+1}|
    ),
\end{equation}
where we have used that $f_R$ has a contradiction for each $d\in D_R,$ and the above description of $R.$ 

All terms in the first product are clearly positive.
But also all terms in the second product are. Indeed, if $d$ is the end of both $c_i,c_j$ where $i<j$ then by the assumption  $\det(C_{i,j}^{j_d,j_d+1})$ has sign $(-1)^{\dd(c_i)-\dd(c_j)+1}.$ When we scale the $i$th and $j$th rows of the determinant by $(-1)^{\dd(c_i)}$ and $(-1)^{\dd(c_j)}$ respectively,
we obtain
\begin{equation}\label{eq:det1}|C_i^{j_d}||C_j^{j_d+1}|-|C_i^{j_d+1}||C_j^{j_d}|<0.\end{equation}
Now, in this case the contradiction is obtained for $f,$ such as $f_R,$ which maps $i\to j_d,~j\to j_d+1.$ Thus, the left hand side of
\eqref{eq:det1} equals
\[ -\left(|C_{f_R^{-1}(j_d+1)}^{j_d}||C_{f_R^{-1}(j_d)}^{j_d+1}|
    -|C_{f_R^{-1}(j_d)}^{j_d}||C_{f_R^{-1}(j_d+1)}^{j_d+1}|\right).\]
which implies that $|C_{f_R^{-1}(j_d+1)}^{j_d}||C_{f_R^{-1}(j_d)}^{j_d+1}|
    -|C_{f_R^{-1}(j_d)}^{j_d}||C_{f_R^{-1}(j_d+1)}^{j_d+1}|>0.$

Similarly, if
$d$ is the ending domino of both $c_i$ and the starting domino of $c_j$ then $i<j,$ and a contradiction will be obtained from bijections $f,$ like $f_R,$ with $f(i)=j_d+1,~f(j)=j_d.$ 
By the assumption the determinant of $C_{i,j}^{j_d,j_d+1}$ has sign $(-1)^{\dd(c_i)}.$ We scale the $i$th row of the determinant by $(-1)^{\dd(c_i)}$. The result is thus positive. Thus, the corresponding determinant which appears in \eqref{eq:decomp_into_2by2} is 
\begin{align*}
|C_{f_R^{-1}(j_d+1)}^{j_d}||C_{f_R^{-1}(j_d)}^{j_d+1}|
    -|C_{f_R^{-1}(j_d)}^{j_d}||C_{f_R^{-1}(j_d+1)}^{j_d+1}| \;&=\; |C_i^{j_d}||C_j^{j_d+1}|-|C_i^{j_d+1}||C_j^{j_d}| \\\;&=\; (-1)^{\dd(c_i)}\det(C_{i,j}^{j_d,j_d+1})>0.
\end{align*}
All terms in the product \eqref{eq:decomp_into_2by2} belong to $\tVar_D.$ As claimed.

The fact that the nonzero Pl\"ucker coordinates correspond precisely to the nonzero polynomials is straightforward.
Since we only used in the proof the domino form of $C$ and the signs of Corollary~\ref{cor:signs_of_1by1_2by2}, every matrix in that form that satisfies these inequalities has the same nonvanishing determinants, and they are all positive.
\end{proof}

\begin{rmk}\label{rmk:alternative_proof}
Showing that every point that satisfies the constraints of Corollary~\ref{cor:signs_of_1by1_2by2} lies in the positroid cell $S$ also follows from the computations in Lemmas~\ref{inverse-domino} or~\ref{lem:reduced_rep}, by recovering the variables of $\generate$.
The argument we presented above has the advantage of giving the positive decomposition of the Pl\"ucker coordinates, a decomposition we will soon use.
\end{rmk}

The upshot of the next lemma is to show that the vanishing of an element in $\tVar\setminus\Var$ implies the vanishing of at least two elements of $\Var.$
\begin{lemma}\label{lem:eliminate_non_necessary}
{For any chord diagram $D$,}
every element in $\tVar_D\setminus\Var_D$ can be represented as a
linear combination of Laurent monomials in the variables of $\tVar_D,$ 
such that all coefficients are nonnegative, and at least one of the monomials appearing in this sum with a positive coefficient is of the form $\frac{P}{Q},$ where $P$ is a product of elements of $\Var_D,$ and $Q$ is a product of elements of $\Var^1_D.$
\end{lemma}

\begin{proof}
{Throughout the proof, since $D$ is fixed, we omit $D$ from the notations.} We first prove the claim for elements of $\Var^1\setminus\Var.$ There are two cases to consider. The first is that $c_i$ has a same-end child $c_j.$ In this case
$\El_j,\Er_i\in\Var^1\setminus\Var$ and $\ee_{ij}\in\Var.$
From the definitions and Corollary~\ref{cor:signs_of_1by1_2by2} we have:
\[\ee_{ij}=(-1)^{\dd(c_i)-\dd(c_j)+1}\det\begin{pmatrix}
\gamma_i & \delta_i \\
\gamma_j & \delta_j\end{pmatrix}=(-1)^{\dd(c_i)-\dd(c_j)+1}(\gamma_i\delta_j-\gamma_j\delta_i),\]
and 
\[\gamma_i=(-1)^{\dd(c_i)}\El_i,~\delta_i=(-1)^{\dd(c_i)}\Er_i,~\gamma_j=(-1)^{\dd(c_j)}\El_j,~\delta_j=(-1)^{\dd(c_j)}\Er_j.\]
Thus, $\Er_i\El_j=\ee_{ij}+\El_i\Er_j,$ and hence
$\Er_i=\frac{\ee_{ij}+\El_i\Er_j}{\El_j},~~\El_j=\frac{\ee_{ij}+\El_i\Er_j}{\Er_i},$ as claimed.

The second case is when $c_i,c_j$ are head-to-tail siblings. Then $\El_i,\Sr_j\in\Var^1\setminus\Var$ and
$\es_{ij}\in\Var.$ 
Again by Corollary~\ref{cor:signs_of_1by1_2by2},
\[\es_{ij}=(-1)^{\dd(c_i)}\det\begin{pmatrix}
\gamma_i & \delta_i \\
\alpha_j & \beta_j\end{pmatrix}=(-1)^{\dd(c_i)}(\gamma_i\beta_j-\alpha_j\delta_i),\]
and
\[\gamma_i=(-1)^{\dd(c_i)}\El_i,~\delta_i=(-1)^{\dd(c_i)}\Er_i,~\alpha_j=\Sl_j,~\beta_j=\Sr_j.\]
Thus, $\El_i\Sr_j=\es_{ij}+\Er_i\Sl_j,$ and therefore
$\El_i=\frac{\es_{ij}+\Er_i\Sl_j}{\Sr_j},~~\Sr_j=\frac{\es_{ij}+\Er_i\Sl_j}{\El_i},$ showing the claim for this case.

We now find an expression for $\ee_{il}$ when $c_l$ is a same-end descendant which is not the child of $c_i.$
In this case $c_i$ must have a same end child, denote it by $c_j.$ Then $\eta_{ij}\in\Var.$ Write
\[C_{h,h+1}^{i,j,l}=\begin{pmatrix}
\gamma_i & \delta_i \\
\gamma_j & \delta_j\\
\gamma_l & \delta_l\end{pmatrix}.\]
By Corollary~\ref{cor:signs_of_1by1_2by2},
\[\gamma_i=(-1)^{\dd(c_i)}\El_i,~\delta_i=(-1)^{\dd(c_i)}\Er_i,~\gamma_j=(-1)^{\dd(c_j)}\El_j,~\delta_j=(-1)^{\dd(c_j)}\Er_j,\]\[\gamma_l=(-1)^{\dd(c_l)}\El_l,~\delta_l=(-1)^{\dd(c_l)}\Er_l,\]
and
\[\ee_{ij}=(-1)^{\dd(c_i)-\dd(c_j)+1}(\gamma_i\delta_j-\gamma_j\delta_i),~\ee_{il}=(-1)^{\dd(c_i)-\dd(c_l)+1}(\gamma_i\delta_l-\gamma_l\delta_i),\]\[\ee_{jl}=(-1)^{\dd(c_j)-\dd(c_j)+1}(\gamma_j\delta_l-\gamma_l\delta_j).\]
We see that $\delta_l(\gamma_i\delta_j-\gamma_j\delta_i),~\delta_i(\gamma_j\delta_l-\gamma_l\delta_j),~\delta_j(\gamma_i\delta_l-\gamma_l\delta_i)$ all have the same sign $$(-1)^{\dd(c_i)+\dd(c_j)+\dd(c_l)+1},$$ and moreover the sum of the first two equals the third.
Thus, $\frac{\Er_l}{\Er_j}\ee_{ij}  + \frac{\Er_i}{\Er_j}\ee_{jl} = \ee_{il}.$  If $\Er_l\in\Var$ then this case follows. Otherwise we use the first case to write $\Er_l=\frac{\ee_{ll'}+\El_l\Er_{l'}}{\El_{l'}},$ and substitute this in the expression for $\ee_{il}.$

We now express $\es_{jl}$ where $c_j$ has a same-end ancestor $c_i$ and there is a third chord $c_l$ which starts at $(h,h+1)$ where $c_i,c_j$ end. We may take $c_i$ to be a sibling of $c_l,$ and then $\es_{il}\in\Var.$

Write
\[C_{h,h+1}^{i,j,l}=\begin{pmatrix}
\gamma_i & \delta_i \\
\gamma_j & \delta_j\\
\alpha_l & \beta_l\end{pmatrix}.\]
By Corollary~\ref{cor:signs_of_1by1_2by2},
\[\gamma_i=(-1)^{\dd(c_i)}\El_i,~\delta_i=(-1)^{\dd(c_i)}\Er_i,~\gamma_j=(-1)^{\dd(c_j)}\El_j,~\delta_j=(-1)^{\dd(c_j)}\Er_j\]\[\alpha_l=\Sl_l,~\beta_l=\Sr_l,\]
and
\[\ee_{ij}=(-1)^{\dd(c_i)-\dd(c_j)+1}(\gamma_i\delta_j-\gamma_j\delta_i),\es_{il}=(-1)^{\dd(c_i)}(\gamma_i\beta_l-\delta_i\alpha_l),~\es_{jl}=(-1)^{\dd(c_j)}(\gamma_j\beta_l-\delta_j\alpha_l).\]
We see that $\beta_l(\gamma_i\delta_j-\gamma_j\delta_i),~-\delta_j(\gamma_i\beta_l-\delta_i\alpha_l),~-\delta_i(\gamma_j\beta_l-\delta_j\alpha_l)$ have sign $(-1)^{\dd(c_i)+\dd(c_j)+1},$ and moreover the sum of the first two equals the third. 
Thus, $\frac{\Sr_l}{\Er_i}\ee_{ij}  + \frac{\Er_j}{\Er_i}\es_{il} = \es_{jl},$ and we finish as in the previous case. 
\end{proof}

\subsection{Codimension One Boundaries}

\begin{summary}\label{subsec:plabic222}
Recall Summary \ref{subsec:plabic}. 
The stratification of the nonnegative Grassmannian was much studied in the literature~\cite[Section 18]{postnikov2006total}, \cite[Section 6]{postnikov2009matching}. Postnikov showed that the closure of a positroid cell is a disjoint union of positroid cells. Moreover, a positroid cell $S$ is contained in the closure of a another positroid cell $S'$ if and only if all the Pl\"ucker coordinates that vanish on $S'$ also vanish on~$S$. It was later shown that this stratification yields a CW complex structure on the closure of each positroid cell, and in particular on $\Grnn{k}{n}$ \cite[Section 6]{postnikov2009matching}.

Plabic graphs also carry information about gluing the boundary stratification. If a positroid cell $S'$ is a boundary stratum of a positroid cell $S,$ then one can find a reduced plabic graph $G$ for $S,$ together with a perfect orientation, such that after omitting some of its edges, or equivalently put their weights to $0,$ one obtains a reduced plabic graph $G'$ with a perfect orientation for $S',$ and the map $(0,\infty)^d\to S$ extends smoothly to a map $[0,\infty)^d\to \overline{S},$ where $\overline{S}$ is the topological closure of $S$, and $S'$ is the homeomorphic image, and in fact the diffeomorphic image, of the subset of $[0,\infty)^d$ obtained by setting certain coordinates to $0$ (see \cite[Section 18]{postnikov2006total}).
\end{summary}

In the following discussion, we refer to a positroid cell as a \emph{stratum} or \emph{boundary stratum} of another positroid cell, if the former is a cell in the CW decomposition of the closure of the latter. Recall that Corollary~\ref{cor:signs_of_1by1_2by2} and Notation~\ref{nn:Var} allow translating the inequalities of Definition~\ref{def:domino_signs} to the positivity of the elements of $\tVar_D.$

\begin{definition}\label{def:various_generalized_domino_forms}
Let $\D$ be a chord diagram, and $\star\in\Var_D.$ A matrix $C$ is a $(\D,\star)$-\emph{extended domino matrix} if it is in the form of Definition~\ref{def:domino_entries}, satisfies the all inequalities of Definition~\ref{def:domino_signs}, with the exception that $\star=0.$ The matrix $C$ is said to be a $D$-extended domino matrix if it is in the form of Definition~\ref{def:domino_entries}, only that the inequalities of Definition~\ref{def:domino_signs} are relaxed to be weak inequalities.

{For a BCFW cell $S_D$ associated to a chord diagram $D$, we denote by }$\partial_\star S_D$ the subset of the nonnegative Grassmannian whose elements have a representative in the $(\D,\star)$-extended domino form. We also write $\widetilde{\partial_\star S_D}\supseteq\partial_\star S_D$ for the set of vector spaces which have a representative in the $D$ extended domino form with $\star=0$.
\end{definition}

{Note that $\widetilde{\partial_\star S_D}\subseteq\Gr_{k,n}^{\geq}$ by Proposition \ref{prop:minimal_set_of_ineqs}.}
Recall that $\SA\subseteq\Gr_{k,n}^{\geq}$ is defined as the collection of vector spaces which intersect nontrivially $\Span(\e_{i},\e_{i{\pl1}},\e_{j},\e_{j{\pl1}})$ for some $i,j.$

We start by describing generation algorithms for these spaces, in the spirit of Algorithm~\ref{construct-matrix}, in all cases but the case of $\partial_{\Sr_i}S_D,$ when $c_i$ is a sticky child. A generation algorithm for this case can also be written, but will not be needed in what follows.

Recall the algorithm \textsc{construct-matrix} in~Section~\ref{subsec:first_alg}. We now describe the changes needed in order to generate the boundary positroid cells $\partial_\star,~\star\in\Var$ above. Let $c_l=(i,i+1,h,h+1)$ be a chord in~$\D.$
\begin{itemize}
\item$\mathbf{\partial}_{\p_l}:$
Act as in \textsc{construct-matrix}, but do not perform $y_{(i+1){\mi1}}(s_l)$ in \textsc{end}$(c_l)$.
\item$\mathbf{\partial}_{\Sl_l}:$
 Act as in \textsc{construct-matrix}, but do not perform $y_i(u_l)$ in \textsc{start}$(c_{\ast l})$.
\item$\mathbf{\partial}_{\Sr_l}$~(if $c_l$ is not a sticky child and $\Sr_l\in\Var$),$~\mathbf{\partial}_{\es_{ml}}$ (assuming $\es_{ml}\in\Var$ ):
 Act as in \textsc{construct-matrix}, but instead of \textsc{end}$(c_l)$ apply $\inc_i,~y_{i{\mi1}}(s_l),~x_{i}(v_l),~x_{h}(w_l),$ and add $i$ to $N.$
 In \textsc{start}$(c_{\ast l})$ skip $y_{i+1}(u_l).$
\item$\mathbf{\partial}_{\El_l}$~(assuming $\El_l\in\Var$ ):
 Since $\El_l\in\Var,$ then the parent of $c_l$ ends after $c_l,$ and no sibling of $c_l$ starts at $(h,h+1).$ 
 $\pre_h$ is skipped in \textsc{fill}$(h),$ and is performed in \textsc{end}$(c_l)$ instead of doing $x_h(w_l).$ 
\item$\mathbf{\partial}_{\Er_l}$ (if $\Er_l\in\Var$), $\mathbf{\partial}_{\ee_{lm}}$~(assuming $\ee_{lm}\in\Var$ ):
Act as in the algorithm, omitting $x_{h}(w_l)$ in \textsc{end}$(c_l).$
\end{itemize}
{Observe that omitting an $x_h(t)$ or $y_h(t)$ operation is equivalent to performing the corresponding operation with $t=0.$ The reason the cases $\star=\beta_l,\es_{ml}$ require a more significant change in the algorithm, is that in \textsc{construct-matrix}, the first entry that is constructed for a chord $c_l$ is the one which corresponds to $\beta_l,$ and the algorithm substitutes there $1.$ In the cases $\star=\beta_l,\es_{ml},$ we want instead to substitute there~$0,$ and for this we use the entry corresponding to $\alpha_l$ as the first entry of the chord, thus skipping the $\beta_l$ entry for that stage in the algorithm. In particular, when $\beta_l\in \Var,$ the entry corresponding to $\beta_l$ will remain $0$ until the end of the algorithm.}
Denote by $\textsc{construct-matrix}^\star\left(D,\, \{s_l,u_l,v_l,w_l\}_{l=1}^k \right),$ for $\star\in\Var,$ the output of the algorithm, applied with the variables $\{s_*,u_*,v_*,w_*\}.$ Note that one variable is absent, according to the above items, but in order not to make the notation heavier we do not write it explicitly. 
\begin{lem}\label{lem:codim_1_dominoes}
$\textsc{construct-matrix}^\star\left(D,\, \{s_l,u_l,v_l,w_l\}_{l=1}^k \right),$ for a chord diagram $D$ and $\star\in\Var_D,$ generates positroid cells whose elements have representatives in the $(\D,\star)$-extended domino form. 
\end{lem}

\begin{proof}
This lemma is analogous to Proposition~\ref{prop:domino}. The proofs are obtained by minor adaptations to the proof of that proposition. The exception variable $\star = 0$ follows from omitting the corresponding $x_i$ or $y_i$ operation from the algorithm, and positive values for the other variables can be recovered in the same manner.

In the cases $\star = \p_l,\Sl_l,\Sr_l,\El_l,\Er_l$ this is essentially the only change. We note that for $\star = \Sr_l,$ the assumption $\Sr_l\in\Var$ is used to guarantee that not only when the row of $c_l$ is generated, its $i+1$ entry will be zero, but also that this will not change by later $x_i$ operations, as explained in the paragraph preceding this lemma. A similar situation holds for $\El_l$ and $\Er_l$, whose appearance in $\Var$ guarantees they do not gain positive terms from operations by other chords.

The proof for $\ee_{lm}$ and $\es_{ml}$ is also similar to the proof of Proposition~\ref{prop:domino}. This time, after \textsc{end}$(c_l)$ in the case of $\ee_{lm},$ the $h+1$ entry of that row is $0.$ After \textsc{end}$(c_m),$ the minor $\{l,m\}\times\{h,h+1\}$ will have proportional columns. The subsequent steps will not change that. Similarly, for $\es_{ml}$, after \textsc{start}$(c_{\ast l})$ the entry $i+1$ of the row of $c_l$ is $0.$ After \textsc{end}$(c_{\ast m})$, the minor $\{m,l\}\times\{i,i+1\}$ will have proportional columns. The subsequent steps will not change that.

By repeating the argument of Proposition~\ref{prop:domino} it is easily seen that not only every $\star'\in\Var\setminus\{\star\}$ is nonzero, it also has the expected sign of Corollary~\ref{cor:signs_of_1by1_2by2}.
\end{proof}
The generation algorithm can be translated into an algorithmic decorated permutation, which we denote by $\pi_{D,\star}$ according to the recipe of Definition~\ref{sigma_and_the_algorithm} {(recall that $\pi_D$ is the permutation associated with $D$). When the decorated permutation of a chord diagram $D'$ is denoted differently than $\pi_{D'}$, for example $\rshs{i}{\pi_D}$ which will appear later, we shall denote $\pi_{D',\star}$ with a $\star$ subscript, for example $\rshs{i}{\pi_D}_\star$}. Recall Definition~\ref{def:nicer_form_pi_alpha}.
It will be convenient in what follows to use slightly different notations for the same permutation, and to write it in terms of the chords $c_i,$ enumerated according to their starting point, rather than $c'_i$ which are enumerated according to their endpoints. We will write $c_i=(a_i,a_i+1,b_i,b_i+1),$ and then the permutation $\pi_D$ reads 
\begin{equation}\label{eq:good_old_pi}
   (p_1~q_1)~(p_2~q_2) \cdots (p_{2k}~q_{2k}) \prod_{\substack{c_i~\text{is ordered according}\\\text {to increasing order of ends}}}((a_i+1)~b_i~(b_i+1))
\end{equation}
where, $p_i,q_i$ are as in Definition~\ref{def:nicer_form_pi_alpha} and since we work with non commutative variables our convention is that the multiplication is from left, i.e. $\prod_{i=1}^3x_i=x_3x_2x_1.$ In case of same-end chords, the parent is considered to end after the child.
We also recall the notation $c_{j\ast}$ to denote the last sticky descendant in the sticky chain from $c_j,$ and $a_{j\ast}$ is its first marker. We denote by $a^\ast_j$ either $n,$ if $c_j$ is a top chord, or $a_h+1,$ where $c_h$ is the parent of $a_h,$ otherwise. 
\begin{obs}\label{obs:perm_for_codim_1}
In cases $\star\in\Var_D$ for which the generation algorithm was defined, the algorithmic permutation $\pi_{D,\star}$ is obtained as follows.
\begin{itemize}
\item $\pi_{D,\p_j}$ is obtained from \eqref{eq:good_old_pi} by omitting the transposition $((a_j+1)~a^\ast_j).$
\item $\pi_{D,\Sl_j}$ is obtained from \eqref{eq:good_old_pi} by omitting the transposition $(a_j~(a_{j\ast}+1)),$ where $c_{j\ast}$ is the last sticky descendant in the sticky chain from $c_j.$
\item $\pi_{D,\El_j}$ is obtained from \eqref{eq:good_old_pi} by replacing the cycle $((a_j+1)~b_j~(b_j+1))$ by $((a_j+1)~(b_j+1)).$
\item $\pi_{D,\Er_j},~\pi_{D,\ee_{j,j'}}$ are obtained from \eqref{eq:good_old_pi} by replacing the cycle $((a_j+1)~b_j~(b_j+1))$ by $((a_j+1)~b_j).$
\item $\pi_{D,\Sr_j},~\pi_{D,\es_{j',j}}$ are obtained from \eqref{eq:good_old_pi} by replacing $((a_j+1)~b_j~(b_j+1))$ by $(a_j~b_j~(b_j+1)),$ replacing $((a_j+1)~a^\ast_j)$ by $(a_j~a^\ast_j),$ and removing the transposition $(a_j~({a}_{j\ast}+1)).$
\end{itemize}
\end{obs}
\begin{lemma}\label{lem:reduced_for_codim_1}
Let $D \in \mathcal{CD}_{n,k}$ be a chord diagram. Let $\star\in\Var_D$ be an element which is not of the form $\Sr_i$ for a sticky child, $\Sl_i$ for a chord which has a sticky child, or $\p_i.$ Denote by $S$ the positroid cell that arises from $\textsc{construct-matrix}^\star\left(D,\, \{s_l,u_l,v_l,w_l\}_{l=1}^k \right)$. Every point in $S$ is obtained from a unique choice of the $4k-1$ positive variables. Hence, $S$ is $4k-1$-dimensional. 

Moreover, in these cases the positroid cell generated by the algorithm coincides with $\partial_\star S_D,$ and $\pi_{D,\star}$ is the associated permutation. In addition $\partial_\star S_D$ is a codimension one boundary of $S_D.$ 
\end{lemma}

The restrictions we make in this statement are not always necessary, but they suffice for our needs, and simplify the proof.

\begin{proof}
The proof follows the lines of the proof of Lemma~\ref{lem:reduced_rep}. The first statement and the 'Hence' part are proven by constructing an inverse map, again in an inductive process. That is, for a given point in $S,$ which must be an output of \textsc{construct-matrix}$^\star$, we determine the real parameters of the $x_l$ and $y_l$ operations used to construct it.
Finding these values for $k=1$ is straightforward. 
\\\textbf{More than one top chord:}
\\Suppose we have shown the claim for less than $k$ chords. 
Assume first that $\D$ has at least two top chords.
For $\star\in\{\Sl_l,~\Sr_l,~\El_l,~\Er_l,~\ee_{lm},~\es_{ml}\},$ 
write $c_l=(i,i+1,h,h+1)$ and denote by $c_p=(q,q+1,j,j+1)$ the top chord from which $c_l$ descends. Let $\D_b$ be the subdiagram made of chords which end no later than $(j,j+1).$ 
Let $\D_c$ be the subdiagram made of the other chords, with index set $[n]\setminus[j-1].$
Suppose that $\D_b$ has $k'$ chords. We consider first the case $k'<k.$

We can write
\[C=\begin{pmatrix}C'\\C''\end{pmatrix},\]where $C'$ consists of the first $k'$ rows and $C''$ of the last $k-k'$ rows. Let $V',V''$ be their row spans, respectively.
The recursive structure of the algorithm
$\textsc{construct-matrix}^\star\left(D,\, \{s_r,u_r,v_r,w_r\}_{r=1}^k \right),$ allows to write, as in the proof of Lemma~\ref{lem:reduced_rep}, $C',C''$ explicitly:
\[C'=
\pre_{n-1}\cdots\pre_{j+2}\;\textsc{construct-matrix}^\star\left(D,\, \{s_r,u_r,v_r,w_r\}_{r\in[k']} \right),\]
and 
\begin{equation*}C''\;=\;\pre_{j-1}\cdots\pre_1\;x_j(w)\;\generate\left(D_c,\,\{s_r,u_r,v_r,w_r\}_{r\in[k]\setminus[k']}\right),\end{equation*}where $w$ is the sum of $w_r$ variables that appear in $x_j$ operations for $C'.$

We claim that $V',V''$ are uniquely determined from $V.$ Indeed, by the same argument of Lemma~\ref{lem:reduced_rep}, we see that if this is not the case then either $V'$ or $V''$ contains a vector in $\Span\{\e_j,\e_{j+1},\e_n\}.$
$V''$ is the row span of a standard domino matrix. Therefore, by the same argument of Lemma~\ref{lem:more than 4}, \[V''\cap\Span\{\e_j,\e_{j+1},\e_n\}=0.\] We want to show that also $C'$ does not contain a vector in the span of $\{\e_j,\e_{j+1},\e_n\}.$ If $\star\in\{\Sl_p,\Sr_p,\El_p,\Er_p\},$ which means $p=l,$ then there is $i'\in\{i,i+1\}$ such that $(C'_{p})^{i'}\neq0.$ As in the proof of Lemma~\ref{lem:more than 4} we can find $I\in\binom{[j-1]}{k'}$ which intersects $\{i,i+1\}$ in $i'$ for which $\det((C')^I)>0,$ which implies, again as in that proof, that there is no linear combination of the rows of $C'$ which lies in $\Span\{\e_j,\e_{j+1},\e_n\}.$
If $\star\notin\{\Sl_p,\Sr_p,\El_p,\Er_p\},$ and $\star\neq \Sl_l$, then $\det((C')^I)>0,$ where $I$ is be the set of first markers of each chord. This is seen by performing row operations which cancel the dominoes inherited to chords from their parent. These operations do not change the starting dominoes and transform $(C')^I$ to an upper triangular with nonzero diagonal. 
If $\star=\Sl_l,$ then we take $I$ to be the set of first markers of each chord, except for $l$ and the maximal sticky chain which descends from it. For these chords we pick the second marker, and again the same reasoning shows that the determinant is nonzero. Thus, in this case we also see that $V'$ does not contain a vector in $\Span\{\e_j,\e_{j+1},\e_n\}.$ 

By induction we can recover from $V'$ the parameters used to construct $C'.$ We then apply $x_j(-w)$ to $V''$ and use induction again to recover the remaining parameters for $C''$.

The next case is $k'=k.$ Since we assume that there is more than one top chord, $c_p$ must be the last top chord.
This time let $\D_c$ be the subdiagram of $c_p$ and its descendants, and $\D_b$ the subdiagram of chords which precede $c_p$, with marker set $[q+1]\cup\{n\}.$
Redefine $k'$ to be the number of rows in $\D_b,$ and set
\[C=\begin{pmatrix}C'\\C''\end{pmatrix},\] where $C'$ consists of the first $k'$ rows, and $C''$ the last $k-k'$ rows. $V',V''$ are the row spans.

We use the recursive structure of the algorithm again to write
\[C'=\pre_{n-1}\cdots\pre_{q+2}\;\textsc{construct-matrix}\left(D,\, \{s_r,u_r,v_r,w_r\}_{r\in[k']} \right)\]
and 
\begin{equation*}C''\;=\;\pre_{q-1}\cdots\pre_1\;x_q(w)\;\generate^{\star'}\left(D_c,\,\{s_r,u_r,v_r,w_r\}_{r\in[k]\setminus[k']}\right),\end{equation*}where $w$ is the sum of $w_r$ variables that appear in $x_q$ operations for $C',$ and $\star'=\star$ \emph{unless }$c_l=c_p$ is a top chord and $\star= \es_{ml}.$ In that exceptional case $\star'=\Sr_l,$ since after the algorithm finishes to construct the row of $c_l$ its entry labeled $i+1=q+1$ is $0.$ It becomes non zero only by applications of $x_q$ in later stages, whose total effect on this row is the effect of $x_q(w)$.

Again we show that $V',V''$ can be determined from $V.$
We need to show that neither $V'$ nor $V''$ contain a non zero vector in $\Span\{\e_q,\e_{q+1},\e_n\}.$ This time $V'$ is the row span of a standard domino matrix, hence for $V'$ this is done as in Lemma~\ref{lem:reduced_rep}. For $V'',$ we split as above to the two cases depending on whether $\star$ belongs to $\{\Sl_p,\Sr_p,\El_p,\Er_p\}$ or not, and the same conclusion holds. 
As in the previous case we use induction to retrieve the parameters of the algorithm for $C',$ and then for $C''$ after applying $x_q(-w)$ to $V''.$
\\\textbf{A single top chord:}
\\The final case to consider is that $\D$ has a single top chord $c_p=c_1$, where we may assume, perhaps after removing zero columns, a long chord $(1,2,n-2,n-1).$
If $\star\notin\{\Sl_1,\Sr_1,\El_1,\Er_1\},$ then as in the analogous part in the proof of Lemma~\ref{lem:reduced_rep}
we can write \[C=\begin{pmatrix}C_1\\C'\end{pmatrix}\]
where 
\[C'=\pre_n\mathrm{rescale}_1 (u_1)\;x_1(1)\;\mathrm{rotate}_{k-1}\;\generate^\star\left(D_b,\,\{s_h,u_h,v_h,w_h\}_{h\in[k]\setminus[1]}\right),\]if $\star\neq \ee_{1m}$
and
\[C'=\pre_n\mathrm{rescale}_1 (u_1)\;x_1(1)\;\mathrm{rotate}_{k-1}\;\generate\left(D_b,\,\{s_h,u_h,v_h,w_h\}_{h\in[k]\setminus[1]}\right),\]
otherwise.
The reason for this exception is that in case $\ee_{1m}=0$ the missing parameter is $w_1$ from \textsc{end}$(c_1).$

As above, the next step is to show that there is no non zero vector in the intersection between the linear span of rows $2,\ldots,k$ and $V'=\Span\{\e_1,\e_2,\e_{n-2},\e_{n-1}\}.$
In order to show that this intersection is zero, first recall that by assumption $\star$ is not of the form $\Sr_l$, for any $c_l$ which is a sticky child. In this case, we can find $I\in\binom{\{3,4,\ldots,n-3,n\}}{k}$ such that $\det(C^I)\neq 0,$ showing that no vector in $\Span\{\e_1,\e_2,\e_{n-2},\e_{n-1}\}$ is spanned by $C$'s rows. We construct $I$ as follows. First, we add $n$ to $I.$ Then for every chord in the sticky chain from $c_1$ we take its second marker. Then, if $\star\neq \Sl_l$ for $l\in\{2,\ldots, k\}$ we add to $I$ the first marker of any other chord. If $\star=\Sl_l,$ for some $c_l,$ then also for this chord, which by assumption does not have a sticky child we take the second marker, and for other chords the first. The first row of $C^I$ has a single nonzero entry at $n.$ Removing this row and the column $n$ leaves us with a matrix, that like in the proof of Lemma~\ref{lem:more than 4} is upper triangular with nonzero diagonal after some row operations.

Thus, the intersection is $0,$ hence $C_1,~V'$ can be determined from $V.$ As in the proof of Lemma~\ref{lem:reduced_rep} we calculate $s_1,u_1,v_1$ from $C_1,$ then undo $\pre_n,~\mathrm{rescale}_1 (u_1),~x_1(1),~\mathrm{rotate}_{k-1},$ use induction to recover the parameters $\{s_h,u_h,v_h,w_h\}_{h\in[k]\setminus[1]}$ and then recover $w_1.$

We are left with the case $\star\in\{\Sl_1,\Sr_1,\El_1,\Er_1\}.$ The proof of Lemma~\ref{lem:more than 4} extends to this case, and we can find a set $I\subseteq\binom{\{3,4,\ldots,n-3,n\}}{k}$ with $\det(C^I)\neq0.$ This shows as above that indeed  $C_1,$ and $V'=\Span\{2,\ldots,k\}$ are uniquely determined. Again we can recover the parameters for $c_1,$ except, possibly, $w_1.$
If $\star\in\{\El_1,\Er_1\}$ then we act just as above. 
If $\star=\Sr_1,$ we erase the first row and $n$th column, cyclically shift the first column to be the last, possibly changing the sign of this column to preserve positivity, and re-index it as $n.$
If $\star=\Sl_1,$ then by assumption $c_1$ does not have a child which sticks to it. We erase the first row and $n$th column. Note that the first column is now a zero column. We then move the second column to be the $n$th column, possibly changing the sign of this column to preserve positivity.
In all these cases, the resulting vector space is precisely the row span of \[\generate\left(D_b,\,\{s_h,u_h,v_h,w_h\}_{h\in[k]\setminus[1]}\right),\] we use induction to calculate its parameters and then calculate $w_1$ if $\star\neq\Er_1.$
\\\textbf{The 'More-over' part:}
\\Arguing exactly as in Section~\ref{analysis}, based on Summary~\ref{subsec:plabic}, we deduce from the dimension calculation above that $\pi_{D,\star}$ is the associated permutation.

By definition $\partial_\star S_D$ contains the (non empty) positroid cell $S$ generated by the algorithm.
{By Proposition~\ref{prop:minimal_set_of_ineqs}, the set of nonvanishing Pl\"ucker coordinates for each point in $\partial_\star S_D$ are the same: For every $D$-extended domino matrix $C,$ every Pl\"ucker coordinate $\langle C^I\rangle$ can be written as a polynomial $Q_I$ with nonnegative coefficients in the set $\tVar_D.$ Every element of $\partial_\star S_D$ has a $(D,\star)$-matrix representation $C$, and a Pl\"ucker coordinate $\langle C^I\rangle$ is non zero precisely if the polynomial $Q_I$ does not become identically zero under the substitution $\star=0.$} {But since the points of $\partial_\star S_D$ all have the same collection of non vanishing Pl\"ucker coordinates, and since $\partial_\star S_D$ contains $S,$ this collection must coincide with the corresponding collection for the positroid cell $S.$ Hence the two spaces must coincide, as $S$ is the subset of all points in the non negative Grassmannian with that collection of non vanishing Pl\"ucker coordinates. Moreover, by Proposition~\ref{prop:minimal_set_of_ineqs} again, the collection of non vanishing Pl\"ucker coordinates for $\partial_\star S_D$ is contained in the corresponding collection for $S$. Therefore $S=\partial_\star S_D$ must be a $4k-1$ dimensional boundary stratum of~$S_D.$}
\end{proof}

The main result of this section is the next proposition.

\begin{prop}
\label{prop:bdries_either_paired_or_SA}
Each boundary stratum of a BCFW cell $S_a$ is either contained in $\SA$ or is a boundary stratum of another BCFW cell $S_b.$
\end{prop}

Before we get to the proof, we define shift operations on chords in a chord diagram. These operations are crucial to the pairing of boundaries.

\begin{definition}
\label{def:shifts}
Let $c_i=(h,h+1,l,l+1)$ be a chord in a chord diagram $\D.$ 
\\\textbf{Left shifting the start of $c_i$ (unobstructed):} 
\\If $h>1,$ no chord starts at $(h-1,h)$ or ends at $(h,h+1)$ we define $\lshs{i}{D}$ as the diagram obtained from $\D$ by replacing the $i$th chord with a new $i$th chord $(h-1,h,l,l+1).$
\\\textbf{Right shifting the start of $c_i$ (unobstructed):} \\If $c_i$ is not short, and no chord starts at $(h+1,h+2)$ we define $\rshs{i}{D}$ as the diagram obtained from $\D$ by replacing the $i$th chord with a new $i$th chord $(h+1,h+2,l,l+1).$ 
If $c_i$ is short, $l<n-2$ and no other chord starts or ends at $(l,l+1)=(h+2,h+3)$ we define $\rshs{i}{D}$ as the diagram obtained from $\D$ by replacing the $i$th chord with a new $i$th chord $(h+1,h+2,h+3,h+4).$ 
\\\textbf{Right shifting the start of $c_i$ (obstructed):} \\
Suppose $c_i$ is short and $l<n-2.$ 
\begin{itemize}
    \item If no other chord ends at $(l,l+1)=(h+2,h+3)$ but $c_i$ has a sibling $(h+2,h+3,j,j+1)$ then 
    we define $\rshs{i}{D}$ as the diagram obtained from $\D$ by replacing the $i$th chord with a new $i$th chord $(h+1,h+2,j,j+1).$ 
    \item If no other chord starts at $(l,l+1)=(h+2,h+3)$ but some ancestors of $c_i$ end there, 
    we define $\rshs{i}{D}$ as the diagram obtained from $\D$ by removing the $i$th chord, pushing the ends of all of its same-end ancestors one marker backwards, that is to end at $(h+1,h+2),$ and adding a new short $i$ chord $(h+1,h+2,h+3,h+4).$ 
    \item If there is a chord $(h+2,h+3,j,j+1)$ and some ancestors of $c_i$ end at $(h+2,h+3)$, 
    we define $\rshs{i}{D}$ as the diagram obtained from $\D$ by removing the $i$th, pushing the ends of all of its same-end ancestors one marker backwards, that is to end at $(h+1,h+2),$ and adding a new $i$th chord $(h+1,h+2,j,j+1).$ 
\end{itemize}
\textbf{Left shifting the end of $c_i$ (unobstructed):} \\If $c_i$ is not short and it has no same-end descendant
$\lshs{i}{D}$ is defined from $\D$ by replacing the $i$th chord by $(h,h+1,l-1,l).$
If $c_i$ is short, $h>1$ and no chord starts at $(h-1,h)$ or ends at $(h,h+1),$ $\lshe{i}{D}$ is defined from $\D$ by replacing the $i$th chord by $(h-1,h,h+1,h+2).$
\\\textbf{Left shifting the end of $c_i$ (obstructed):} 
\begin{itemize}
    \item Suppose $c_i$ is not short, and it has a same-end child which starts at $(j,j+1),$ for $j>h+1.$
    In this case $\lshe{i}{D}$ is obtained from $\D$ by replacing the $i$th chord with a new $i$th chord $(h,h+1,j,j+1).$
    \item Suppose $c_i$ is not short, and it has a same-end child which starts at $(h+1,h+2),$ but no chord ends at $(h,h+1)$ or starts at $(h-1,h)$ and $h>1.$
    In this case $\lshe{i}{D}$ is obtained from $\D$ by replacing the $i$th chord with a new $i$th short chord $(h-1,h,h+1,h+2).$
    \item Suppose $c_i$ is not short, it has a same end child which starts at $(h+1,h+2),$ and there are chords ending at $(h,h+1).$
    In this case $\lshe{i}{D}$ is obtained from $\D$ by removing the $i$th chord, pushing the endpoint of each chord whose end is $(h,h+1)$ to end at $(h+1,h+2)$ and add a new $i$th short chord $(h-1,h,h+1,h+2).$
    \item If $c_i$ is short and some chord ends at $(h,h+1)$ then $\lshe{i}{D}$ is obtained from $\D$ by removing the $i$th chord, pushing the endpoint of each chord whose end is $(h,h+1)$ to end at $(h+1,h+2)$ and add a new $i$th short chord $(h-1,h,h+1,h+2).$
\end{itemize}
\textbf{Right shifting the end of $c_i$ (unobstructed):} \\If $l<n-2,$ the parent of $c_i$ does not end at $(l,l+1)$ and no sibling of $c_i$ starts at $(l,l+1)$
    we define $\rshe{i}{D}$ as the diagram obtained from $\D$ by replacing the $i$th chord with a new $i$th chord $(h,h+1,l+1,l+2).$
\\\textbf{Right shifting the end of $c_i$ (obstructed):} \\If the parent of $c_i$ does not end at $(l,l+1)$ but a sibling of $c_i$ starts at $(l,l+1)$ and ends at $(j,j+1),$
    we define $\rshe{i}{D}$ as the diagram obtained from $\D$ by replacing the $i$th chord with a new $i$th chord $(h,h+1,j,j+1).$

{Note that the shift remains undefined for some combinations of direction, end and chord diagram. This is intentional for reasons that will become apparent after Lemma \ref{lem:SA cases}. The undefined situations match the cases 2-6 in that lemma: 
$\lshs{i}{D}$ is undefined in case $\Sl_i$,
$\rshs{i}{D}$ is undefined in case $\Sr_i$,
$\rshe{i}{D}$ is undefined in cases $\El_i$ and $\ee_{ij}$, and
$\lshe{i}{D}$ is undefined in case $\Er_i$. 
}

In all the cases where shifts are defined we call the resulting chord diagram the \emph{shifted diagram}. {By abuse of notation, we shall denote the induced shift operations on corresponding positroid cells and decorated permutations in the same way, i.e. for the right shift of the end: $$\rshe{i}{S_D}:=S_{\rshe{i}{D}},\;\rshe{i}{\pi_D}:=\pi_{\rshe{i}{D}}.$$}
\end{definition}

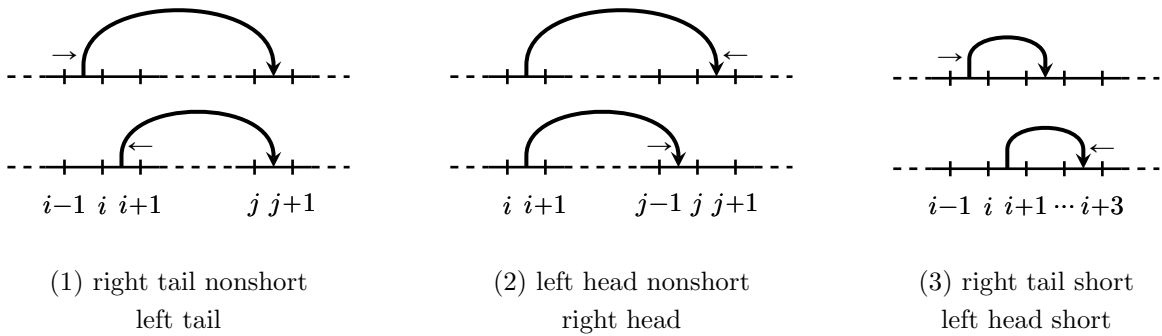
\begin{center}
\begin{tabular}{ccc}
\tikz[line width=1]{
\def\dh{1.5}
\foreach \h in {0,\dh}{
    \draw[dashed] (0.25,\h) -- (0.75,\h);
    \draw (0.75,\h) -- (2.25,\h);
    \draw[dashed] (2.25,\h) -- (3.25,\h);
    \draw (3.25,\h) -- (4.25,\h);
    \draw[dashed] (4.25,\h) -- (4.75,\h);
    \foreach \i/\j in {2/i{-}1,3/i,4/i{+}1,7/j,8/j{+}1}{
        \def\x{\i/2}
        \draw (\x,\h-0.1)--(\x,\h+0.1);
        \node at (\x,\h-0.25) {$\scriptscriptstyle\j$};
    }
}
\foreach \i/\j/\h in {3/7/0,2/7/\dh}{
    \def\x{\i/2+0.25}
    \def\y{\j/2+0.25}
    \draw[line width=1.5,-stealth] (\x,\h) -- (\x,\h+0.15) to[in=90,out=90] (\y,\h+0.15) -- (\y,\h);
}
\node at (2,0+0.25) {$\leftarrow$};
\node at (1,\dh+0.25) {$\rightarrow$};
}
& \hspace{0.75cm}
\tikz[line width=1]{
\def\dh{1.5}
\foreach \h in {0,\dh}{
\draw[dashed] (0.25,\h) -- (0.75,\h);
\draw (0.75,\h) -- (1.75,\h);
\draw[dashed] (1.75,\h) -- (2.75,\h);
\draw (2.75,\h) -- (4.25,\h);
\draw[dashed] (4.25,\h) -- (4.75,\h);
\foreach \i/\j in {2/i,3/i{+}1,6/j{-}1,7/j,8/j{+}1}{
\def\x{\i/2}
\draw (\x,\h-0.1)--(\x,\h+0.1);
\node at (\x,\h-0.25) {$\scriptscriptstyle\j$};}}
\foreach \i/\j/\h in {2/6/0,2/7/\dh}{
\def\x{\i/2+0.25}
\def\y{\j/2+0.25}
\draw[line width=1.5,-stealth] (\x,\h) -- (\x,\h+0.15) to[in=90,out=90] (\y,\h+0.15) -- (\y,\h);}
\node at (3,0.25) {$\rightarrow$};
\node at (4,\dh+0.25) {$\leftarrow$};
}
\hspace{0.75cm} & 
\tikz[line width=1]{
\def\dh{1.5}
\foreach \h in {0,\dh}{
\draw[dashed] (0.25,\h) -- (0.75,\h);
\draw (0.75,\h) -- (3.25,\h);
\draw[dashed] (3.25,\h) -- (3.75,\h);
\foreach \i/\j in {2/i{-}1,3/i,4/i{+}1,5/i{+}2,6/i{+}3}{
\def\x{\i/2}
\draw (\x,\h-0.1)--(\x,\h+0.1);
\node at (\x,\h-0.25) {$\scriptscriptstyle\j$};}}
\foreach \i/\j/\h in {3/5/0,2/4/\dh}{
\def\x{\i/2+0.25}
\def\y{\j/2+0.25}
\draw[line width=1.5,-stealth] (\x,\h) -- (\x,\h+0.25) to[in=90,out=90] (\y,\h+0.25) -- (\y,\h);}
\node at (3,0.25) {$\leftarrow$};
\node at (1,\dh+0.25) {$\rightarrow$};
}
\\[1em]
(1) right start nonshort & 
(2) left end nonshort & 
(3) right start short \\
left start &
right end &
left end short \\[1em]
\end{tabular}
\end{center}

\begin{center}
\begin{tabular}{ccc}
\hspace{0.5cm}
\tikz[line width=1]{
\def\dh{1.5}
\foreach \h in {0,\dh}{
\draw[dashed] (0.25,\h) -- (0.75,\h);
\draw (0.75,\h) -- (1.75,\h);
\draw[dashed] (1.75,\h) -- (2.75,\h);
\draw (2.75,\h) -- (3.75,\h);
\draw[dashed] (3.75,\h) -- (4.75,\h);
\draw (4.75,\h) -- (5.75,\h);
\draw[dashed] (5.75,\h) -- (6.75,\h);
\foreach \i/\j in {2/i,3/i{+}1,6/j,7/j{+}1,10/l,11/l{+}1}{
\def\x{\i/2}
\draw (\x,\h-0.1)--(\x,\h+0.1);
\node at (\x,\h-0.25) {$\scriptscriptstyle\j$};}}
\foreach \i/\j/\h/\c in {2/5.8/0/black,6.2/10/0/gray,6/9.8/\dh/gray,2/10.2/\dh/black}{
\def\x{\i/2+0.25}
\def\y{\j/2+0.25}
\draw[line width=1.5,-stealth,\c] (\x,\h) -- (\x,\h+0.15) to[in=90,out=90] (\y,\h+0.15) -- (\y,\h);}
\node at (2.9,0.25) {$\rightarrow$};
\node at (5.5,\dh+0.5) {$\leftarrow$};
}
\hspace{0.5cm} & \hspace{0.5cm}
\tikz[line width=1]{
\def\dh{1.5}
\foreach \h in {0,\dh}{
\draw[dashed] (1.25,\h) -- (1.75,\h);
\draw (1.75,\h) -- (3.75,\h);
\draw[dashed] (3.75,\h) -- (4.75,\h);
\draw (4.75,\h) -- (5.75,\h);
\draw[dashed] (5.75,\h) -- (6.75,\h);
\foreach \i/\j in {4/i,5/i{+}1,6/i{+}2,7/i{+}3,10/j,11/j{+}1}{
\def\x{\i/2}
\draw (\x,\h-0.1)--(\x,\h+0.1);
\node at (\x,\h-0.25) {$\scriptscriptstyle\j$};}}
\foreach \i/\j/\h/\c in {4/5.8/0/black,6.2/10/0/gray,6/9.8/\dh/gray,5/10.2/\dh/black}{
\def\x{\i/2+0.25}
\def\y{\j/2+0.25}
\draw[line width=1.5,-stealth,\c] (\x,\h) -- (\x,\h+0.15) to[in=90,out=90] (\y,\h+0.15) -- (\y,\h);}
\node at (2,0.25) {$\rightarrow$};
\node at (5.5,\dh+0.5) {$\leftarrow$};
} \hspace{0.5cm}
\\[1em]
(4) left end obstructed nonsticky & 
(5) left end obstructed sticky \\
right end obstructed &
right start obstructed short \\[1em]
\end{tabular}
\end{center}

\begin{center}
\begin{tabular}{ccc}
\hspace{1cm}
\tikz[line width=1]{
\def\dh{1.7}
\foreach \h in {0,\dh}{
\draw[dashed] (1.25,\h) -- (1.75,\h);
\draw (1.75,\h) -- (4.25,\h);
\draw[dashed] (4.25,\h) -- (4.75,\h);
\foreach \i/\j in {4/i,5/i{+}1,6/i{+}2,7/i{+}3,8/i{+}4}{
\def\x{\i/2}
\draw (\x,\h-0.1)--(\x,\h+0.1);
\node at (\x,\h-0.25) {$\scriptscriptstyle\j$};}}
\foreach \i/\j/\h/\c in {4/5.8/0/black,5.2/7/\dh/black}{
\def\x{\i/2+0.25}
\def\y{\j/2+0.25}
\draw[line width=1.5,-stealth,\c] (\x,\h) -- (\x,\h+0.15) to[in=90,out=90] (\y,\h+0.15) -- (\y,\h);}
\foreach \y in {0.5,0.7,0.9}{
\draw[line width=1.5,-stealth,gray] (2-\y,\y) to[in=90,out=0] (6.35/2+\y/4,0.1) -- (6.35/2+\y/4,0);
\draw[line width=1.5,-stealth,gray] (2-\y,\y+\dh) to[in=90,out=0] (5/2+\y/4,\dh+0.1) -- (5/2+\y/4,\dh);}
\node at (2,0.25) {$\rightarrow$};
\node at (4,\dh+0.25) {$\leftarrow$};
}
\hspace{1cm} & \hspace{1cm}
\tikz[line width=1]{
\def\dh{1.7}
\foreach \h in {0,\dh}{
\draw[dashed] (1.25,\h) -- (1.75,\h);
\draw (1.75,\h) -- (3.75,\h);
\draw[dashed] (3.75,\h) -- (4.75,\h);
\draw (4.75,\h) -- (5.75,\h);
\draw[dashed] (5.75,\h) -- (6.75,\h);
\foreach \i/\j in {4/i,5/i{+}1,6/i{+}2,7/i{+}3,10/j,11/j{+}1}{
\def\x{\i/2}
\draw (\x,\h-0.1)--(\x,\h+0.1);
\node at (\x,\h-0.25) {$\scriptscriptstyle\j$};}}
\foreach \i/\j/\h/\c in {4/5.7/0/black,6.3/10/0/gray,6/9.8/\dh/gray,5.2/10.2/\dh/black}{
\def\x{\i/2+0.25}
\def\y{\j/2+0.25}
\draw[line width=1.5,-stealth,\c] (\x,\h) -- (\x,\h+0.15) to[in=90,out=90] (\y,\h+0.15) -- (\y,\h);}
\foreach \y in {0.5,0.7,0.9}{
\draw[line width=1.5,-stealth,gray] (2-\y,\y) to[in=90,out=0] (6.2/2+\y/4,0.1) -- (6.2/2+\y/4,0);
\draw[line width=1.5,-stealth,gray] (2-\y,\y+\dh) to[in=90,out=0] (5/2+\y/4,\dh+0.1) -- (5/2+\y/4,\dh);}
\node at (2,0.25) {$\rightarrow$};
\node at (5.5,\dh+0.5) {$\leftarrow$};
}
\hspace{1cm}
\\[1em]
(6) left end short rolling ends & 
(7) left end obstructed sticky rolling ends \\
right start short rolling ends &
right start obstructed short rolling ends \\[1em]
\end{tabular}
\end{center}

\begin{rmk}\label{rmk:below_figs}
We can also describe the shift operation, whenever it is defined, as applying one of the four basic shifts appearing in the upper and lower rows of diagrams (1),(2) above, and in case one of the two basic rules for chord diagrams - that the start and end cannot intersect, and that two different chords cannot intersect - is violated, we resolve it in a systematic way.
Let us briefly detail this procedure.

{First shift naively (as in (1) or (2)) in the desired direction either the end or start of the chord $c_i=(a,b,c,d)$ to $c'_i = (a,b,c\pm1,d\pm1)$ or $(a\pm1,b\pm1,c,d)$ respectively, possibly violating the chord diagram rules.}
For example, 
if we want to shift the start of $c_i$ to the left we perform (1), top row, etc. 
In case $c'_i$ is too short, or is being crossed by other chords, we resolve these issues as follows. 
\begin{itemize}
    \item if $c'_i$ is short we move the other domino of $c'_i$ one step in the same direction (i.e. if we moved the start into an overlap with the end, we now move the end in the same direction, and vice versa);
    \item if there are chords which cross $c'_i$ and start to its left: we move the ends of all such crossing chords so that if
they used to be same-end to $c_i$
, they are now head-to-tail with $c'_i$ and vice-versa;
\item if there are
chords that cross $c'_i$ and end to its right, we move the end of $c'_i$
in a similar manner to exchange
same-end and head-to-tail.
\end{itemize} 
Each of these steps may cause new problems of the above types. We repeat the resolution steps until a legal chord diagram is attained.

For example, case (3), upper row, is obtained by left shifting the start as in case (1), and then, since the resulting chord is too short, we also left-shift its end.
In case (4), bottom row, we first right shift the end of the chord $(i,i+1,j,j+1),$ getting $(i,i+1,j+1,j+2),$ as in case (2), but then the shifted chord intersects the chord $(j,j+1,l,l+1)$. We resolve this by moving the end of the chord $(i,i+1,j+1,j+2)$ to be same end with $(j,j+1,l,l+1)$, resulting in the chord $(i,i+1,l,l+1)$.
\end{rmk}
The following observation is straightforward {and can be verified by looking at the figures above Remark \ref{rmk:below_figs}}.
\begin{obs} ~
\label{obs:shifts_involutive}
\begin{itemize}
    
    \item 
    If the left shift of the start of $c_i$ is defined, then in $\lshs{i}{D}$ the right shift of the start of the $i$th chord is defined, and
    \[\D=\rshs{i}{\lshs{i}{D}}.\]
    
    If $c_i$ is not short, and the right shift of its start is defined, then in $\rshs{i}{D}$ the left shift of the end of the $i$th chord is defined, and 
    \[\D=\lshs{i}{\rshs{i}{D}}.\]
    These cases correspond to Figure (1) in the diagram above.
    
     \item 
    If $c_i$ is not short, has no sticky and same-end child, and the left shift of its end is defined, then in $\lshe{i}{D}$ the right shift of the $i$th chord's end is defined, and
    \[\D=\rshe{i}{\lshe{i}{D}}.\]
    
    If right shift of the end of $c_i$ is defined, then in $\rshe{i}{D}$ the left shift of the $i$th chord's end is defined, and
    \[\D=\lshe{i}{\rshe{i}{D}}.\]
    These cases correspond to Figures (2),(4) in the diagram above.
    \item 
    If $c_i$ is short and the right shift of its start is defined, then in $\rshs{i}{D}$ the left shift of the end of the $i$th chord is defined, and \[\D=\lshe{i}{\rshs{i}{D}}.\]
    
    If $c_i$ is short {or has a sticky same-end child}, and the left shift its end of $c_i$ is defined, then in $\lshe{i}{D}$ the right shift of the $i$th chord's start is defined, and
    \[\D=\rshs{i}{\lshe{i}{D}}.\]
    These cases correspond to Figures (3),(5),(6),(7) in the diagram above.
\end{itemize}
\end{obs}
\begin{proof}[Proof of Proposition~\ref{prop:bdries_either_paired_or_SA}]
The proof follows from the following three lemmas.
\begin{lemma}\label{lem:at least one star vanishes}
$\partial S_D\subseteq\SA\cup\bigcup_{\star\in\Var_D}\widetilde{\partial_\star S_D}= \SA\cup\bigcup_{\star\in\Var_D}\overline{\partial_\star S_D},$ where $\overline{\partial_\star S_D}$ is the topological closure of $\partial_\star S_D.$
\end{lemma}

\begin{lemma}
\label{lem:SA cases}
{If $D,\star$ satisfy one of the following, then $\widetilde{\partial_\star S_D}\subseteq\SA$:}
\begin{enumerate}
    \item $\star=\p_i.$
    \item $\star=\Sl_i,$ where $c_i$ is either a short chord which ends at $(n-2,n-1)$ or $c_i$ has a sticky child.
    \item $\star=\Sr_i$ for $c_i$ which either starts at $(1,2)$ or $c_i$ is a sticky child.
    \item $\star=\El_i,$ when $c_i$ ends at $(n-2,n-1).$
    \item $\star=\Er_i$ when $c_i$ is short and either starts at $(1,2)$ or is a sticky child.
    \item $\star=\ee_{ij},$ if $c_j$ is a sticky and same-end child of $c_i,$ and $c_i$ either starts at $(1,2)$ or is a sticky child.
\end{enumerate}
\end{lemma}
\begin{lemma}\label{lem:paired cases}
The following codimension one boundaries of BCFW cells are boundaries of two different BCFW cells, hence also their closures are. Let $D$ be a chord diagram with chords $\{c_i\}_{i=1}^k$, whose associated positroid cell is $S$. Denote by $\{c'_i\}_{i=1}^k$ the corresponding chords of the shifted chord diagram in each case.
\begin{enumerate}
    \item Unobstructed cases:
    \begin{itemize}
    \item If the right shift of the end of $c_i$ is unobstructed then 
    \[\partial_{\El_i}S=\partial_{\Er_i}\rshe{i}{S}.\] 
    \item If the right shift of the start of $c_i$ is unobstructed and $c_i$ is not short, 
    \[\partial_{\Sl_i}S=\partial_{\Sr_i}\rshs{i}{S}.\]
    \item If the right shift of the start of $c_i$ is unobstructed and $c_i$ is short,  
    \[\partial_{\Sl_i}S=\partial_{\Er_i}\rshs{i}{S}.\] 
    
    \end{itemize}
    
    \item If $c_i$ is short and there is another chord $c_j$ which starts where $c_i$ ends 
    or equivalently, if in $\rshs{i}{D}$, $c'_j$ is a child of $c'_i$ which is both sticky and same ended,
    then 
    $\ee_{ij}\in\Var_{\rshs{i}{D}}.$
    In this case
     \[\partial_{\Sl_i}S = \partial_{\ee_{ij}}\rshs{i}{S}.\]

    \item If $c_i$ is short, ends before $(n-2,n-1)$ and ends where its parent end, but no chord starts where it ends, 
    or equivalently, if in $\rshs{i}{D}$, $c'_i$ is short and there is a chord ending where $c'_i$ starts, 
    then $\Sl_i\in \Var_D$ and $\Er_i\in\Var_{\rshs{i}{D}}.$  In this case
    \[\partial_{\Sl_i}S = \partial_{\Er_{i}}\rshs{i}{S}.\]
    
    \item If $c_i$ has a sibling $c_j$ which starts where $c_i$ ends (as usual we consider two top chords also as siblings), or equivalently, if in $\rshe{i}{D}$, $c'_i$ has a same-end child which is not sticky, then $\es_{ij}\in\Var_D,$ and $\ee_{ij}\in\Var_{\rshe{i}{D}}.$ 
    In this case 
    \[\partial_{\es_{ij}}S=\partial_{\ee_{ij}}\rshe{i}{S}.\]
    \end{enumerate}
\end{lemma}
Denote by $D_a$ the chord diagram of $S_a$, and $\Var_a = \Var_{D_a}$.
By Lemma~\ref{lem:at least one star vanishes} each boundary stratum of $S_a$ is either contained in $\SA,$ or is characterized by the vanishing of at least one $\star\in\Var$. All possible cases for the vanishing of $\star\in\Var_a$ are covered by Lemma~\ref{lem:SA cases} and Lemma~\ref{lem:paired cases}, as we now show.
\begin{itemize}
\item\emph{Boundaries of type $\p_i$}  are covered in Lemma~\ref{lem:SA cases}, case 1. 

\item\emph{Boundaries of type $\Sl_i$} are covered in case 1 of Lemma~\ref{lem:paired cases}, if the right shift of the start is unobstructed.
Otherwise, if $c_i$ is short, then either $c_i$ ends at $(n-2,n-1),$ (Lemma~\ref{lem:SA cases}, case 2; in this case the shift is not defined), or it ends where some other chord ends or starts (cases 2,3 of Lemma~\ref{lem:paired cases}; the shift is obstructed). If $c_i$ is not short, but the right shift of its start is not unobstructed then $c_i$ must have a sticky child (case 2 of Lemma~\ref{lem:SA cases} again), and the shift is undefined.

\item\emph{Boundaries of type $\Sr_i$} are covered in case 1 of Lemma~\ref{lem:paired cases}, if the left shift of the start is unobstructed.
Otherwise, either $c_i$ starts at $(1,2),$ or it is a sticky child (case 3 of Lemma~\ref{lem:SA cases}; in both cases the shift is not defined), or it starts where some other chord $c_j$ ends.
In this case $\Sr_i\notin\Var.$ 

\item\emph{Boundaries of type $\El_i$} are covered in case 1 of Lemma~\ref{lem:paired cases}, if the right shift of the end is unobstructed.
The right shift of the end is obstructed when it has a sibling which starts where $c_i$ ends, or when $c_i$ is the same-end child of $c_j$. In these situations $\El_i\notin\Var.$
The last possibility is that this chord ends at $(n-2,n-1).$ This case is treated in Lemma~\ref{lem:SA cases}, case 4 (and the shift is not defined). 

\item\emph{Boundaries of type $\Er_i$} are covered in case 1 of Lemma~\ref{lem:paired cases}, if the left shift of the end is unobstructed.
There are several cases in which this shift is obstructed or just not defined.
If $c_i$ is short it could either be that it starts at $(1,2)$ or it is a sticky child (case 5 of Lemma~\ref{lem:SA cases}), in these cases the shift is not defined, or it starts where some other chord ends (case 2 of Lemma~\ref{lem:paired cases}, this is an obstructed case hence $\Er_i\notin\Var$). If $c_i$ is not short then the left shift of the end is obstructed only if $c_i$ has a same-end child. In this case $\Er_i\notin\Var,$

\item\emph{Boundaries of type $\es_{ij}$} are covered in Lemma~\ref{lem:paired cases}, case 4.

\item\emph{Boundaries of type $\ee_{ij}$} are covered in Lemma~\ref{lem:paired cases}, case 4, when $c_j$ is not a sticky child of $c_i.$
If it is a sticky child, and $c_i$ is either a sticky child itself, or starts at $(1,2)$ then this case is covered in Lemma~\ref{lem:SA cases}, case 6.
The remaining possibilities are handled in Lemma~\ref{lem:paired cases}, case 2.
\end{itemize}

Since these are all the possibilities, the conclusion now follows.

\end{proof}
We turn to the proofs of the lemmas.
\begin{proof}[Proof of Lemma~\ref{lem:SA cases}]
Assume that $c_i$ starts at $(h,h+1)$ and ends at $(l,l+1).$

For the first item, if $\p_i=0$ then the $i$th row a linear combination of $\e_{h},\e_{h+1},\e_l,\e_{l+1}.$

For the second item, if $c_i$ is short and ends at $(n-2,n-1)$ then $\Sl_i=0$ implies that, {after performing row operations involving only the chain of ancestors of $c_i$, if $c_i$ is not a top chord,} the resulting $i$th row a linear combination of $\e_{n-3},\e_{n-2},\e_{n-1},\e_{n}.$ 
If $c_i$ has a sticky child $c_j$ then if $c_j$ ends at $(l',l'+1)$ then the vanishing of $\Sl_i,$ which affects also $c_i$'s children, in the domino they inherit from $c_i$, makes the $j$th row a linear combination of $\e_{h+1},\e_{h+2},\e_{l'},\e_{l'+1}.$

For the third item, if $h=1$ then also $i=1$ and $c_1$ is a top level chord. Thus if $\Sr_1=0$ then the first row is a linear combination of $\e_{1},\e_l,\e_{l+1},\e_n.$ If $c_i$ is a sticky child, and $\Sr_i=0$ then the $i$th row a linear combination of $\e_{h-1},\e_{h},\e_l,\e_{l+1}.$

When $c_i$ ends at $(n-2,n-1),$ setting $\El_i=0$ implies that, {after performing row operations involving only the chain of ancestors of $c_i$, if $c_i$ is not a top chord,} the resulting $i$th row is a linear combination of $\e_{h},\e_{h+1},\e_{n-1},\e_{n}.$

When $c_i$ is short and starts at $(1,2)$ putting $\Er_i=0$ makes the $i$ row a linear combination of $\e_1,\e_2,\e_3,\e_n,$ since in this case $c_i$ is also a top chord. If $c_i$ is short and a sticky child then after substituting $\Er_i=0$ the resulting $i$th row lies in the span of $\e_{h-1},\e_h,\e_{h+1},\e_{h+2}.$

Finally, in the last case, if $c_i$ starts at $(1,2),$ then $c_i=c_1$ is a top chord, and the submatrix $C^{l,l+1}_{1,j}$ has rank one. By adding to $C_1$ an appropriate multiple of $C_j$ we can make the $l,l+1$ entries of the first row $0,$ and the resulting row is in the span of $\e_{1},\e_2,\e_3,\e_n.$
Similarly, if $c_i$ is a sticky child, the same argument shows that some linear combination of $C_i$ and $C_j$ belongs to $\Span(\e_{h-1},\e_h,\e_{h+1},\e_{h+2}).$
\end{proof}

\begin{proof}[Proof of Lemma~\ref{lem:at least one star vanishes}]
Every $V\in\overline{S}_D,$ the closure of $S_D,$ is the limit of a converging sequence $\{V_i\}_{i=1}^\infty\subset S_D.$ Let $C_i$ be domino matrix representatives for $V_i,$ where we use the freedom to scale each row by an element in $\R_+$ to guarantee that as $i\to\infty$ the entries of $C_i$ remain bounded, and at least one of them does not tend to $0.$ By passing to a subsequence if necessary we may assume that each entry of $C_i$ tends to a limit. Let $C_\infty$ be the limit matrix. Note that this matrix's rows need not to span $V,$ but must be included in $V$. This matrix will have at least one nonzero entry in each row, and it will be in the extended $\D$-domino form.

If for $C_\infty$ all elements of $\tVar_D$ remain nonzero, then by Proposition~\ref{prop:minimal_set_of_ineqs} $C_\infty\in S_D.$ 

If, for some $1\leq i\leq k,$ the starting domino of $(C_\infty)_i$ is zero, then this row has support contained in the set of indices $(j,j+1,h,h+1)$ where $(j,j+1)$ is the end of $c_i,$ and $(h,h+1)$ is $(n-1,n)$ if $c_i$ is a top chord, and otherwise it is the start of $c_i$'s parent.
As in the proof of Lemma~\ref{lem:SA cases}, in this case $V\in\SA.$
Similarly, if a chord $c_i$ has a sticky child $c_{j}$ and either the first domino entry of $(C_\infty)_i$ vanishes, or the second domino entry of $(C_\infty)_j$ vanishes, then by Lemma~\ref{lem:SA cases} $V\in \SA.$ 

When those situations do not occur, $C_\infty$ has full rank. Indeed, let $i_j$ be the index of the first non vanishing entry of the start of $c_j.$ By our assumption on $C_\infty,$ for every $j\in[k]$ the index $i_j$ is defined, and $i_{j+1}>i_j.$ Set $I=\{i_1,\ldots,i_k\}.$ Then the matrix $(C_\infty)_I$ is invertible, since, if we apply row operations, going top-bottom, which cancel the domino inherited to chords from their parents, the resulting matrix is upper triangular, with non zero diagonal elements.

Since $C_\infty$ is an extended domino matrix, and some element of $\tVar_D$ vanishes, by Lemma~\ref{lem:eliminate_non_necessary}, at least one element of $\Var_D$ vanishes.
Thus, $V\in \widetilde{\partial_\star S_D},$ for some $\star\in\Var_D.$

If $\star$ is one of the cases considered in Lemma~\ref{lem:SA cases}, then $V\in\SA.$
We will finish the proof by showing that in the remaining cases \begin{equation}\label{eq:widetilde_vs_ovrline}
\widetilde{\partial_\star S_D}\setminus\SA=\overline{\partial_\star S_D}\setminus\SA,
\end{equation}
and thus $V\in\overline{\partial_\star S_D}.$

For any point $\widetilde{\partial_\star S_D}\setminus\partial_\star S_D$ the set of vanishing Pl\"ucker coordinates strictly contains the vanishing Pl\"ucker coordinates of the positroid cell $\partial_\star S_D,$ by Proposition~\ref{prop:minimal_set_of_ineqs} and the fact a positroid cell is determined by its set of vanishing Pl\"ucker coordinates. Thus it is a true boundary point of $\partial_\star S_D$ and is therefore contained in the closure. {Thus \[\widetilde{\partial_\star S_D}\subseteq\overline{\partial_\star S_D}\Rightarrow \widetilde{\partial_\star S_D}\setminus\SA\subseteq\overline{\partial_\star S_D}\setminus\SA.
\]}
For the other containment, by the above analysis, any point of $\overline{\partial_\star S_D}\setminus\SA$ has an extended domino form, {given by the matrix $C_\infty$ defined above}. Clearly for any such point $\star=0 $. {Since, by the definition of $\widetilde{\partial_\star S_D},$ Definition \ref{def:various_generalized_domino_forms}, every point of the nonnegative Grassmannian which has an extended $D-$domino representation with $\star=0$ belongs to  $\widetilde{\partial_\star S_D},$ it follows that
$\overline{\partial_\star S_D}\setminus\SA\subseteq \widetilde{\partial_\star S_D}\setminus\SA.$}
\end{proof}

\begin{proof}[Proof of Lemma~\ref{lem:paired cases}]
In all cases considered in the lemma, the boundary strata are the ones which are associated to the decorated permutations of Observation~\ref{obs:perm_for_codim_1} by Lemma~\ref{lem:reduced_for_codim_1}. It is therefore be enough to compare the permutations for pairs of strata that we claim that are equal. We first write each permutation of the shifted chord diagram using the permutation of the original chord diagram, \emph{and with the notations of the former} and then examine, using Observation~\ref{obs:perm_for_codim_1} how these permutations change when we move to the boundary strata. The final step is to carefully compare the resulting permutations.
\\\textbf{First item: }
We start with the unobstructed cases.
\begin{itemize}
\item The diagram $\rshe{i}{D}$ is obtained from $\D$ by shifting the end of $c_i$ to the right. The associated permutations, $\pi_D,~\rshe{i}{\pi_D}$ in their \eqref{eq:good_old_pi} forms, differ in a single cycle: the cycle $((a_i+1)~b_i~(b_i+1))$ in $\pi_D$ is replaced by $((a_i+1)~(b_i+1)~(b_i+2)).$ By Observation~\ref{obs:perm_for_codim_1}, $\pi_{D,\El_i}$ is obtained from $\pi_D$ by replacing $((a_i+1)~b_i~(b_i+1))$ by $((a_i+1)~(b_i+1)),$ while ${\rshe{i}{\pi_D}}_{\Er_i}$ is obtained from $\pi_D$ by replacing $((a_i+1)~(b_i+1)~(b_i+2))$ by $((a_i+1)~(b_i+1)),$ hence they are the same.
\item When $c_i$ is not short, the diagram $\rshs{i}{D}$ is obtained from $\D$ by shifting the start of $c_i$ to the right. 
Since we can do the shift, before the shift $c_i$ had no sticky child. 
The associated permutations, $\pi_D,~\rshs{i}{\pi_D}$, in their \eqref{eq:good_old_pi} forms, differ in three places: the cycle $((a_i+1)~b_i~(b_i+1))$ in $\pi_D$ is replaced by $((a_i+2)~b_i~(b_i+1)),$ the transposition $((a_i+1)~a^\ast_i)$ in $\pi_D$ is replaced by $((a_i+2)~a^\ast_i),$ and the transposition $(a_i~(a_i+1))$ in $\pi_D$ is replaced by $((a_i+1)~t),$ where $t=a_j+1$ where $c_j$ is the last descendant in a sticky chain starting from the \emph{shifted} $c_i.$

By Observation~\ref{obs:perm_for_codim_1}, $\pi_{D,\Sl_i}$ is obtained from $\pi_D$ by omitting the transposition $(a_i~(a_i+1)).$ 
Similarly, $\rshs{i}{\pi_D}_{\Sr_i}$ is obtained from $\rshs{i}{\pi_D}$ by erasing 
$(a~t),$ replacing $((a_i+2)~b_i~(b_i+1)),$ by $((a_i+1)~b_i~(b_i+1)),$ and replacing $((a_{i}+2)~a^\ast_i)$ by $((a_{i}+1)~a^\ast_i).$ The resulting permutations are again the same.
\item
The last unobstructed case to consider is when $c_i$ is short and we consider the $\Sl_i$ boundary. Note that now $b_i=a_i+2.$ $\rshs{i}{D}$ is obtained from $\D$ by shifting all markers of $c_i$ one position to the right. The effect on the permutation is shifting all indices by $+1,$ except for $a^\ast_i$ which does not change. 
According to Observation~\ref{obs:perm_for_codim_1}, $\pi_{D,\Sl_i}$ is obtained from $\pi_D$ by omitting $(a_i~(a_{i\ast}+1))=(a_i~(a_i+1)).$
$\rshs{i}{\pi_D}_{\Er_i}$ is obtained from $\rshs{i}{\pi_D}$ by 
removing $((a_i+2)~(b_i+1)~(b_i+2))=(b_i~(b_i+1)~(b_i+2))$ and adding $(b_i~(b_i+1))$ instead.
Note that since $c_i$ is short, its two transpositions are consecutive in the order of \eqref{eq:good_old_pi}. This holds also after the shift. More precisely, for $\D$ we have
\[\pi_D=\sigma_1 ((a_i~(a_i+1))((a_i+1)~a^\ast_i)\sigma_2((a_i+1)~b_i~(b_i+1))\sigma_3,\]
for $\rshs{i}{D}$
\begin{align*}\rshs{i}{\pi_D}=&\sigma_1 ((a_i+1)~(a_i+2))((a_i+2)~a^\ast_i)\sigma_2(b_i~(b_i+1)~(b_i+2))\sigma_3\\
&\quad\quad\quad\quad\quad\quad\quad\quad=\sigma_1 ((a_i+1)~b_i)(b_i~a^\ast_i)\sigma_2(b_i~(b_i+1)~(b_i+2))\sigma_3,\end{align*}
and hence 
\[\pi_{D,\Sl_i}=\sigma_1 (((a_i+1)~a^\ast_i)\sigma_2((a_i+1)~b_i~(b_i+1))\sigma_3
,\]
\[\rshs{i}{\pi_D}=\sigma_1((a_i+1)~b_i)(b_i~a^\ast_i)\sigma_2(b_i~(b_i+1))\sigma_3=\sigma_1((a_i+1)~a^\ast_i)(b_i~(a_i+1))\sigma_2(b_i~(b_i+1))\sigma_3\]
For some permutations $\sigma_1,\sigma_2,\sigma_3$. Now, $\sigma_2$ is made from transpositions $((a_j+1)~a^\ast_j)$ and cycles $((a_j+1)~b_j~(b_j+1))$ for chords $c_j$ which end after the end of $c_i.$
For any such chord $b_j>b_i=a_i+2,$ otherwise the right shift would have been obstructed. From the same reason $a_j+1,a^\ast_j\neq b_i.$ Also $a_j\neq a_i+1$ since $a_i+1$ cannot be the first marker of any chord, and $a^\ast_j=a_i+1$ would imply that $c_j$ is a child of $c_i,$ which is impossible as $c_i$ is short. Thus, $\sigma_2$ and $(b_i~(a_i+1))$ commute, and we can write
\[\rshs{i}{\pi_D}=\sigma_1((a_i+1)~a^\ast)\sigma_2(b_i~(a_i+1))(b_i~(b_i+1))\sigma_3=\sigma_1((a_i+1)~a^\ast)\sigma_2((a_i+1)~b_i~(b_i+1))\sigma_3\]
hence $\pi_{D,\Sl_i}=\rshs{i}{\pi_D}.$
\end{itemize}
\textbf{Second item: }
We consider the case where $c_i$ is short, and there is another chord $c_j$ which starts where $c_i$ ends.
Let $c_{i_r},\ldots,c_{i_2},c_{i_1}=c_i$ be the longest same-end chain whose lowest element is $c_{i},$ the possibility $r=1$ is not excluded. Note that $c_{i_r}$ and $c_j$ are siblings, and we write $a^\ast$ for the second marker of their parent, or $n$ if they are top chords. In addition, \[a_{i_1}+2=b_{i_1}=\ldots=b_{i_r}=a_j.\]
By \eqref{eq:good_old_pi}, $\pi_D=$
\begin{align*}&\sigma_1   
(a_{i_1}~(a_{i_1}+1))((a_{i_1}+1)~(a_{i_2}+1))((a_{i_2}+1)~(a_{i_3}+1))\cdots((a_{i_r}+1)~a^\ast)(a_j~({a}_{j\ast}+1))
\sigma_2
((a_j+1)~a^\ast)\\
&\cdot\sigma_3((a_j+1)~b_j~(b_j+1))\sigma_4((a_{i_r}+1)~b_{i_r}~(b_{i_r}+1)) \cdots((a_{i_2}+1)~b_{i_2}~(b_{i_2}+1)) ((a_{i_1}+1)~b_{i_1}~(b_{i_1}+1))   \sigma_5\end{align*}
for some $\sigma_1,
\ldots,\sigma_5.$
$\rshs{i}{D}$ is obtained from $\D$ by removing $c_i,$ moving the endpoints of $c_{i_2},\ldots,c_{i_r}$ one marker backwards, and then adding a new $i$th chord which starts one marker before $c_j$ and ends with $c_j.$ Note that now $c_j$ becomes a sticky child, hence both have the same element ${a}_{j\ast}+1$ in the leftmost transposition. In addition the new $i$th chord becomes a sibling of $c_{i_r},$ with the same $a^\ast.$
Putting together, with the notations for $\pi_D,$
\begin{align*}\rshs{i}{\pi_D}=\sigma_1 
((a_{i_2}+1)~&(a_{i_3}+1))\cdots((a_{i_r}+1)~a^\ast)
((a_{i_1}+1)~({a}_{j*}+1))
(a_j~({a}_{j*}+1))\\\cdot\sigma_2((a_j+1)~(a_{i_1}+2))&((a_{i_1}+2)~a^\ast)
\sigma_3 ((a_{i_1}+2)~b_j~(b_j+1))((a_j+1)~b_j~(b_j+1))\\&\cdot\sigma_4((a_{i_r}+1)~(b_{i_r}-1)~b_{i_r}) \cdots((a_{i_2}+1)~(b_{i_2}-1)~b_{i_2})   \sigma_5.\end{align*}
By Observation~\ref{obs:perm_for_codim_1}, $\pi_{D,\Sl_i}$ is obtained from $\pi_D$ by omitting the transposition $(a_{i_1}~(\tilde{a}_{i_1}+1)),$ hence
\begin{align*}&\pi_{D,\Sl_i}=\sigma_1   
((a_{i_1}+1)~(a_{i_2}+1))((a_{i_2}+1)~(a_{i_3}+1))\cdots((a_{i_r}+1)~a^\ast)(a_j~({a}_{j\ast}+1))
\sigma_2
((a_j+1)~a^\ast)\\
&\cdot\sigma_3((a_j+1)~b_j~(b_j+1))\sigma_4((a_{i_r}+1)~b_{i_r}~(b_{i_r}+1)) \cdots((a_{i_2}+1)~b_{i_2}~(b_{i_2}+1)) ((a_{i_1}+1)~b_{i_1}~(b_{i_1}+1))   \sigma_5\end{align*}
By the same observation, $\rshs{i}{\pi_D}_{\ee_{i,j}}$ is obtained from $\rshs{i}{\pi_D}$ by replacing the cycle $((a_{i_1}+2)~b_j~(b_j+1))$ by the transposition $((a_{i_1}+2)~b_j),$ thus obtaining
\begin{align*} \rshs{i}{\pi_D}_{\ee_{i,j}}=\sigma_1 &
((a_{i_2}+1)~(a_{i_3}+1))\cdots((a_{i_r}+1)~a^\ast)
((a_{i_1}+1)~({a}_{j\ast}+1))
(a_j~({a}_{j\ast}+1))\cdot\\\sigma_2&((a_j+1)~(a_{i_1}+2))((a_{i_1}+2)~a^\ast)
\cdot\\\sigma_3& ((a_{i_1}+2)~b_j)((a_j+1)~b_j~(b_j+1))\cdot\\\sigma_4&((a_{i_r}+1)~(b_{i_r}-1)~b_{i_r}) \cdots((a_{i_2}+1)~(b_{i_2}-1)~b_{i_2})   \sigma_5.\end{align*}
We wish to compare the permutations.
We apply commutation relations on $\pi_{D,\Sl_i}.$
Since \[(x~y)(y~z)=(y~z)(x~z),\] we have
\begin{align*}
    &((a_{i_1}+1)~(a_{i_2}+1))((a_{i_2}+1)~(a_{i_3}+1))\cdots((a_{i_r}+1)~a^\ast)=\\
    \quad&((a_{i_2}+1)~(a_{i_3}+1))\cdots((a_{i_r}+1)~a^\ast)((a_{i_1}+1)~(a_{i_r}+1)).
\end{align*}

Second, since $a_{i}+1=b_i-1$ we have
\[((a_{i_t}+1)~b_{i}~(b_{i}+1))((a_{i_1}+1)~b_{i}~(b_{i}+1))=
((a_{i_1}+1)~b_{i}~(b_{i}+1))((a_{i_1}+1)~(b_{i}-1)~b_{i}),\]we can iterate this relation to write
\begin{align*}
    &((a_{i_r}+1)~b_{i_r}~(b_{i_r}+1)) \cdots((a_{i_2}+1)~b_{i_2}~(b_{i_2}+1)) ((a_{i_1}+1)~b_{i_1}~(b_{i_1}+1)) =\\
    \quad&((a_{i_1}+1)~b_{i_1}~(b_{i_1}+1))((a_{i_r}+1)~(b_{i_r}-1)~b_{i_r}) \cdots((a_{i_2}+1)~(b_{i_2}-1)~b_{i_2})  .
\end{align*}
Thus,
\begin{align*}
\pi_{D,\Sl_i}=\sigma_1 &   
((a_{i_2}+1)~(a_{i_3}+1))\cdots((a_{i_r}+1)~a^\ast)((a_{i_1}+1)~a^\ast)(a_j~({a}_{j\ast}+1))\cdot\\\sigma_2&((a_j+1)~a^\ast)
\cdot\\\sigma_3&((a_j+1)~b_j~(b_j+1))\cdot\\\sigma_4&((a_{i_1}+1)~b_{i_1}~(b_{i_1}+1))((a_{i_r}+1)~(b_{i_r}-1)~b_{i_r}) \cdots((a_{i_2}+1)~(b_{i_2}-1)~b_{i_2})
\sigma_5\end{align*}
Now, $\sigma_4$ is made of cycles of the form $((a_h+1)~b_h~(b_h+1)),$ for $c_h$ which start after $c_j$ and end no later than $c_j.$ In this case $a_h+1,b_h,b_h+1$ are all larger than $a_{i_1}+1,b_{i_1},b_{i_1}+1,$ since the largest is $b_{i_1}+1=a_j+1,$ hence
\begin{align*}
\pi_{D,\Sl_i}=\sigma_1 &   
((a_{i_2}+1)~(a_{i_3}+1))\cdots((a_{i_r}+1)~a^\ast)((a_{i_1}+1)~a^\ast)(a_j~({a}_{j\ast}+1))\cdot\\\sigma_2&((a_j+1)~a^\ast)
\cdot\\\sigma_3&((a_j+1)~b_j~(b_j+1))((a_{i_1}+1)~b_{i_1}~(b_{i_1}+1))\cdot\\\sigma_4&((a_{i_r}+1)~(b_{i_r}-1)~b_{i_r}) \cdots((a_{i_2}+1)~(b_{i_2}-1)~b_{i_2})
\sigma_5\end{align*}
Thus, in order to compare with $\rshs{i}{\pi_D}_{\ee_{i,j}}$
it is enough to compare
\begin{align*}
    &\pi_1=((a_{i_1}+1)~a^\ast)(a_j~({a}_{j\ast}+1))\sigma_2((a_j+1)~a^\ast)
\sigma_3((a_j+1)~b_j~(b_j+1))((a_{i_1}+1)~b_{i_1}~(b_{i_1}+1)),\\
&\pi_2=((a_{i_1}+1)~({a}_{j\ast}+1))(a_j~({a}_{j\ast}+1))\sigma_2((a_j+1)~(a_{i_1}+2))((a_{i_1}+2)~a^\ast)
\cdot
\sigma_3 \\ & \;\;\;\;\;\;\;\;\;\;\;\;
((a_{i_1}+2)~b_j)((a_j+1)~b_j~(b_j+1)),
\end{align*}
which we can simplify to
\begin{align*}
&\pi_1=((a_{i}+1)~a^\ast)((a_i+2)~({a}_{j\ast}+1))\sigma_2((a_i+3)~a^\ast)
\cdot\\
&\;\;\;\;\;\;\;\;\;\;\;\;\;\;\sigma_3((a_i+3)~b_j~(b_j+1))((a_{i}+1)~(a_i+2)~(a_{i}+3)),\\
&\pi_2=((a_{i}+1)~({a}_{j\ast}+1))((a_i+2)~({a}_{j\ast}+1))\sigma_2((a_i+3)~(a_{i}+2))((a_{i}+2)~a^\ast)\cdot\\
&\;\;\;\;\;\;\;\;\;\;\;\;\;\;
\sigma_3 ((a_{i}+2)~b_j)((a_i+3)~b_j~(b_j+1)),
\end{align*}
Now,
\[((a_i+3)~b_j~(b_j+1))((a_{i}+1)~(a_i+2)~(a_i+3))=
((a_i+1)~(a_i+2))
((a_{i}+2)~b_j)((a_i+3)~b_j~(b_j+1)),\]
Thus, we can write $\pi_1$ as
\begin{align*}
    &((a_{i}+1)~a^\ast)((a_i+2)~({a}_{j\ast}+1))\sigma_2((a_i+3)~a^\ast)
\sigma_3((a_i+3)~b_j~(b_j+1))((a_{i}+1)~(a_i+2)~(a_i+3))=\\
&((a_{i}+1)~a^\ast)((a_i+2)~({a}_{j\ast}+1))\sigma_2((a_i+3)~a^\ast)
\sigma_3((a_i+1)~(a_i+2))
((a_{i}+2)~b_j)((a_i+3)~b_j~(b_j+1)).
\end{align*}
$\sigma_3$ is product of transpositions of the form $((a_h+1)~a^\ast_h),$ for $c_h$ which end (strictly) to the right of $(a_i+2,a_i+3),$ transpositions of the form $(a_h~({a}_{h\ast}+1)),$ for $c_h$ which starts (strictly) to the right of $(a_i+2,a_i+3),$ and cycles $((a_h+1)~b_h~(b_h+1))$ for $c_h\neq c_i,c_j$ which end at $(b_j,b_j+1)$ or to the right of it. Such cycles do not have a common support with $((a_i+1)~(a_i+2)).$ This is clear if they also start to the right of $c_j.$ If they start to the left of it, then they are ancestors of both $c_i$ and $c_j,$ and then $a_h+1\leq a_i.$ So we can move $((a_i+1)(a_i+2))$ past $\sigma_3.$
and the comparison of $\pi_1,\pi_2$ simplifies to the comparison of
\begin{align*}
    &\pi'_1=((a_{i}+1)~a^\ast)((a_i+2)~({a}_{j\ast}+1))\sigma_2((a_i+3)~a^\ast)
((a_i+1)~(a_i+2)),\\
&\pi'_2=((a_{i}+1)~({a}_{j\ast}+1))((a_i+2)~({a}_{j\ast}+1))\sigma_2((a_i+3)~(a_{i}+2))((a_{i}+2)~a^\ast).
\end{align*}
The final simplification we can do is as follows. $\sigma_2$ is made of transpositions of the forms $(a_h~({a}_{h\ast}+1),$ or $((a_h+1)~a^\ast_h)$ for $c_h$ which descend from $c_j.$
All indices involved in such transpositions satisfy $a_h,{a}_{h\ast}+1,a^\ast_h\geq a_i+3.$ They are all smaller than $n,$ since neither of them is a top chord, and if $a^\ast\neq n$ then they are also greater than $a^\ast,$ since $a_i\geq a^\ast$ in this case. Thus, $\sigma_2$ commutes with $((a_i+1)~(a_i+2)),~((a_{i}+2)~a^\ast).$
Using
\[((a_i+3)~a^\ast)
((a_i+1)~(a_i+2))=
((a_i+1)~(a_i+2))((a_i+3)~a^\ast) \]and\[((a_i+3)~(a_{i}+2))((a_{i}+2)~a^\ast)=((a_{i}+2)~a^\ast)((a_i+3)~a^\ast),\]
we can finally reduce the comparison of $\pi'_1,\pi'_2$ to the comparison of
\begin{align*}
    &\pi''_1=((a_{i}+1)~a^\ast)((a_i+2)~({a}_{j\ast}+1))
((a_i+1)~(a_i+2)),\\
&\pi''_2=((a_{i}+1)~({a}_{j\ast}+1))((a_i+2)~({a}_{j\ast}+1))((a_{i}+2)~a^\ast).
\end{align*}
Both $\pi''_1,\pi''_2$ equal to
\[(
(a_i+2)~a^\ast~(a_i+1)~({a}_{j\ast}+1)
)\]
and the second case of the lemma is also proven.
\\\textbf{Third item: }
This case is similar to the previous one, but the absence of $c_j$ simplifies the argument.
Let $c_{i_1}=c_i,c_{i_2},\ldots,c_{i_r},~r\geq 1,$ be the longest chain of children which stick to their parent from the left, $c_{i}$ being the last in the chain. We write $a^\ast$ for the second marker $c_{i_r}$'s parent, or $n$ if $c_{i_r}$ is a top chord. In addition, \[a_{i_1}+2=b_{i_1}=\ldots=b_{i_r}\leq n-2.\]
By \eqref{eq:good_old_pi}
\begin{align*}\pi_D=\sigma_1 &   
(a_{i_1}~(a_{i_1}+1))\cdot((a_{i_1}+1)~(a_{i_2}+1))((a_{i_2}+1)~(a_{i_3}+1))\cdots((a_{i_r}+1)~a^\ast)
\\
&\cdot\sigma_2((a_{i_r}+1)~b_{i_r}~(b_{i_r}+1)) \cdots((a_{i_2}+1)~b_{i_2}~(b_{i_2}+1)) ((a_{i_1}+1)~b_{i_1}~(b_{i_1}+1))   \sigma_3\end{align*}
for some $\sigma_1,
\sigma_2,\sigma_3.$
$\rshs{i}{D}$ is obtained from $\D$ by removing $c_i,$ moving the endpoints of $c_{i_2},\ldots,c_{i_r}$ one marker backwards, and then adding a new $i$th chord which starts at $(a_i+1,a_i+2)$ and ends at $(a_i+3,a_i+4).$ The new $i$th chord is a sibling of $c_{i_r},$ with the same $a^\ast.$
Putting together, with the notations for $\pi_D,$
\begin{align*}\rshs{i}{\pi_D}=\sigma_1 &   
((a_{i_2}+1)~(a_{i_3}+1))\cdots((a_{i_r}+1)~a^\ast)
((a_{i}+1)~(a_i+2))
((a_{i}+2)~a^\ast)\\
&\cdot\sigma_2 ((a_{i}+2)~(a_{i}+3)~(a_i+4))((a_{i_r}+1)~(b_{i_r}-1)~b_{i_r}) \cdots((a_{i_2}+1)~(b_{i_2}-1)~b_{i_2})   \sigma_3.\end{align*}
By Observation~\ref{obs:perm_for_codim_1}, $\pi_{D,\Sl_i}$ is obtained from $\pi_D$ by omitting the transposition $((a_{i_1}~({a}_{i_1}+1)),$ hence
\begin{align*}\pi_{D,\Sl_i}=
\sigma_1 &   
((a_{i_1}+1)~(a_{i_2}+1))((a_{i_2}+1)~(a_{i_3}+1))\cdots((a_{i_r}+1)~a^\ast)
\\
&\cdot\sigma_2((a_{i_r}+1)~b_{i_r}~(b_{i_r}+1)) \cdots((a_{i_2}+1)~b_{i_2}~(b_{i_2}+1)) ((a_{i_1}+1)~b_{i_1}~(b_{i_1}+1))   \sigma_3
\end{align*}
By the same observation, $\rshs{i}{\pi_D}_{\Er_{{i}}}$ is obtained from $\rshs{i}{\pi_D}$ by replacing the cycle $((a_{i}+2)~(a_{i}+3)~(a_i+4))$ by the transposition $((a_{i}+2)~(a_i+3)),$ thus obtaining
\begin{align*}\rshs{i}{\pi_D}_{\Er_{{i}}}=
\sigma_1 &   
((a_{i_2}+1)~(a_{i_3}+1))\cdots((a_{i_r}+1)~a^\ast)
((a_{i}+1)~(a_i+2))
((a_{i}+2)~a^\ast)\\
&\cdot\sigma_2 ((a_{i}+2)~(a_{i}+3))((a_{i_r}+1)~(b_{i_r}-1)~b_{i_r}) \cdots((a_{i_2}+1)~(b_{i_2}-1)~b_{i_2})   \sigma_3.\end{align*}
We wish to compare the permutations.
As in the previous case we can rewrite, using the commutation relations, $\pi_{D,\Sl_i}$ as
\begin{align*}
    \sigma_1 &   
((a_{i_2}+1)~(a_{i_3}+1))\cdots((a_{i_r}+1)~a^\ast)((a_{i}+1)~a^\ast)
\\
&\cdot\sigma_2((a_{i}+1)~b_{i}~(b_{i}+1))((a_{i_r}+1)~(b_{i_r}-1)~b_{i_r}+1) \cdots((a_{i_2}+1)~(b_{i_2}-1)~b_{i_2})    \sigma_3
\end{align*}
Thus, it is enough to compare
\[((a_{i}+1)~a^\ast)\sigma_2((a_{i}+1)~b_{i}~(b_{i}+1)),~\text{and }((a_{i}+1)~(a_i+2))
((a_{i}+2)~a^\ast)\sigma_2 ((a_{i}+2)~(a_{i}+3)).\]
Note that
\begin{align*}((a_{i}+1)~a^\ast)\sigma_2((a_{i}+1)~b_{i}~(b_{i}+1))&=((a_{i}+1)~a^\ast)\sigma_2((a_{i}+1)~(a_{i}+2)~(a_{i}+3))\\&=((a_{i}+1)~a^\ast)\sigma_2((a_{i}+1)~(a_{i}+2))((a_i+2)~(a_{i}+3)).
\end{align*}
As in the previous case we can commute $\sigma_2$ and $((a_i+1)~(a_i+2)).$ Indeed $\sigma_2$ is the product of transpositions and cycles. The first type of transpositions are $((a_h+1)~a^\ast_h)$ for $h$ which ends strictly to the right of the endpoint of $c_i,$ and in this case $a_h>a_i+2,$ and $a^\ast_h$ is either $n$ or smaller than $a_i.$ The second type is $(a_h~({a}_{h\ast}+1))$ for chords which start to the right of the endpoint of $c_i.$ By assumption their $a_h\geq a_i+3.$
The third type is cycles $((a_h+1)~b_h~(b_h+1)),$ for $c_h$ which end strictly to the right of the endpoint of $c_i.$ This means that they either start to its right, and then $a_h,b_h>a_i+2,$ or that they start before $c_i$ and end after. In this case $a_h+1<a_i,$ since we assume $r\geq 1$ and $b_h\geq a_i+3,$ since $c_h$ is not one of the $c_{i_j}$s. All these commute with  $((a_i+1)~(a_i+2)).$
Thus, we are left to verify
\[((a_{i}+1)~a^\ast)((a_{i}+1)~(a_{i}+2))\sigma_2((a_i+2)~(a_{i}+3)),~((a_{i}+1)~(a_i+2))
((a_{i}+2)~a^\ast)\sigma_2 ((a_{i}+2)~(a_{i}+3)),\]
which indeed holds.
\\\textbf{Fourth Item: }
$c_i$ has a sibling $c_j$ which starts where $c_i$ ends.
We write $a^\ast$ for the second marker of their parent, or $n$ if they are top chords and note that $a_{j}=b_i.$
By \eqref{eq:good_old_pi}
\begin{align*}\pi_D=\sigma_1 &   
((a_{i}+1)~a^\ast)(a_j~({a}_{j\ast}+1))
\sigma_2
((a_j+1)~a^\ast)\\
&\cdot\sigma_3((a_j+1)~b_j~(b_j+1))\sigma_4((a_{i}+1)~b_{i}~(b_{i}+1))   \sigma_5\end{align*}
for some $\sigma_1,
\ldots,\sigma_5.$
$\rshe{i}{D}$ is obtained from $\D$ by moving the end of $c_i$ to that of $c_j.$ Thus,
\begin{align*}\rshe{i}{\pi_D}=\sigma_1 &   
(a_j~({a}_{j\ast}+1))
\sigma_2
((a_j+1)~(a_i+1))((a_{i}+1)~a^\ast)\\
&\cdot\sigma_3((a_{i}+1)~b_{j}~(b_{j}+1)) ((a_j+1)~b_j~(b_j+1))\sigma_4  \sigma_5\end{align*}
By Observation~\ref{obs:perm_for_codim_1} $\pi_{D,\es_{i,j}}$ is obtained from $\pi_D$ by replacing $((a_j+1)~b_j~(b_j+1))$ by $(a_j~b_j~(b_j+1)),$ replacing $((a_j+1)~a^\ast),$ by $(a_j~a^\ast)$ and removing the transposition $(a_j~({a}_{j\ast}+1)),$
Thus, the permutation reads
\begin{align*}\pi_{D,\es_{i,j}}=\sigma_1 &   
((a_{i}+1)~a^\ast)
\sigma_2
(a_j~a^\ast)\\
&\cdot\sigma_3(a_j~b_j~(b_j+1))\sigma_4((a_{i}+1)~b_{i}~(b_{i}+1))   \sigma_5.\end{align*}
Similarly, by Observation~\ref{obs:perm_for_codim_1} again, $\rshe{i}{\pi_D}_{\ee_{i,j}}$ is obtained from $\rshe{i}{\pi_D}$ by replacing $((a_i+1)~b_j~(b_j+1))$ with $((a_i+1)~b_j),$ yielding
\begin{align*}\rshe{i}{\pi_D}_{\ee_{ij}}=\sigma_1 &   
(a_j~({a}_{j\ast}+1))
\sigma_2
((a_j+1)~(a_i+1))((a_{i}+1)~a^\ast)\\
&\cdot\sigma_3((a_{i}+1)~b_{j}) ((a_j+1)~b_j~(b_j+1))\sigma_4  \sigma_5\end{align*}
$\sigma_2$ is composed of transpositions of the forms $(a_h~({a}_{h\ast}+1)),$ or $((a_h+1)~a^\ast_h)$ for chords $c_h$ which descend from $c_j.$ Thus, $\sigma_2$ commutes with $((a_i+1)~a^\ast).$ 
Similarly, $\sigma_4$ is the product of cycles of the form $((a_h+1)~b_h~(b_h+1)),$ for $c_h$ which descend from $c_j,$ and thus it commutes with $((a_i+1)~b_i~(b_i+1)).$

The only terms in $\sigma_2$ which may not commute with $(a_j~({a}_{j\ast}+1))$ are of the form $(a_h~({a}_{h\ast}+1)),$ for the sticky descendants of $c_j$. These are the leftmost terms in $\sigma_2,$ and we can write
\[\sigma_2=((a_j+1)~({a}_{j\ast}+1))\cdots({a}_{j\ast}~({a}_{j\ast}+1))\sigma'_2,\]
where $\sigma'_2$ is composed of the remaining terms.
It is therefore a simple check that
\[(a_j~({a}_{j\ast}+1))\sigma_2=\sigma_2(a_j~(a_j+1)).\]
With these simplifications verifying that $\pi_{D,\es_{i,j}},~\rshe{i}{\pi_D}_{\ee_{i,j}}$ agree reduces to showing that
\[\pi'=((a_{i}+1)~a^\ast)
(a_j~a^\ast)\sigma_3(a_j~b_j~(b_j+1))((a_{i}+1)~b_{i}~(b_{i}+1)),~\text{and}\]\[
\pi''=(a_j~({a}_j+1))
((a_j+1)~(a_i+1))((a_{i}+1)~a^\ast)\sigma_3((a_{i}+1)~b_{j}) ((a_j+1)~b_j~(b_j+1))\]agree.
Using $a_j=b_i$ can write
\[(a_j~b_j~(b_j+1))((a_{i}+1)~b_{i}~(b_{i}+1))=
(a_j~(a_j+1)~(a_i+1))((a_i+1)~b_j~(b_j+1)).\] 
We also have \[((a_i+1)~b_j)((a_j+1)~b_j~(b_j+1))=
((a_i+1)~(a_j+1))((a_i+1)~b_j~(b_j+1)).\]
Now $\sigma_3$ is composed again from three types of components. Cycles $((a_h+1)~b_h~(b_h+1)),$ for $c_h$ which end (weakly) after $c_j.$ For such triplets $b_h\geq b_j,$ and $a_h<a_i.$ Transposition $((a_h+1)~a^\ast_h)$ for $c_h$ which end weakly after $c_j.$ Again $a_h<a_i,$ and $a^\ast_h$ is either even smaller or equals $n.$ The last type is transpositions $(a_h~({a}_{h\ast}+1)),$ for $c_h$ which start weakly to the right of where $c_i,c_j$ end. 
Together with the above equalities we can rewrite $\pi',\pi''$ as
\[\pi'=((a_{i}+1)~a^\ast)
(a_j~a^\ast)(a_j~(a_j+1)~(a_i+1))\sigma_3((a_i+1)~b_j~(b_j+1))\]
\[\pi''=(a_j~({a}_j+1))
((a_j+1)~(a_i+1))((a_{i}+1)~a^\ast)((a_i+1)~(a_j+1))\sigma_3((a_i+1)~b_j~(b_j+1)),\]
and these permutations equal, since the part of $\pi',\pi''$ to the left of $\sigma_3$ equals
\[(a_i+1)(a_j~(a_j+1)~a^\ast)\]for both. As needed.
\end{proof}
\begin{rmk}
In \cite{agarwala2023cancellation} an interesting analysis of boundaries and their cancellation was performed for another related physical model. In that model some other class of chord diagrams were used for indexing cells. Their techniques are very different from ours, but we believe that both techniques can be used to address either problem.
\end{rmk}

\section{Surjectivity}\label{sec:surj}
In this section we prove the surjectivity, hence showing that the images of the BCFW cells triangulate the amplituhedron.
We first prove the following lemma and two propositions. The lemma allows extending the amplituhedron map's domain to a small neighborhood of the nonnegative Grassmannian, while the propositions concern the amplituhedron map on $\SA$ and away from $\SA$.
{\begin{lemma}\label{lem:smooth_extension}The amplituhedron map can be extended, via the same formula to a smooth submersion
\[\Z:B\to \Gr_{k,k+m},~~V\mapsto \Z(V)=V\cdot Z,\]
where $B\subset\Gr_{k,n}$ is an open neighborhood of $\Grnn{k}{n}.$
\end{lemma}}
{\begin{proof}
The map $M \mapsto MZ$ is a smooth submersion from $\Mat_{k\times n}$ to $\Mat_{k\times(k+m)}$, as the restriction to rows shows.
If $M$ has nonnegative $k\times k$ minors then $\mathrm{rank}\,MZ = k$ (see, for example, \cite[Section 15]{lam2014totally}). By continuity this holds in a small neighborhood of $M$ in $\Mat_{k\times n}$ as well. After modding out the left $\GL_k(\R)$ action, $M \mapsto MZ$ descends to the amplituhedron map $\Z:B\to\Gr_{k,k+m},$ where $B\subset\Gr_{k,n}$ is some open neighborhood of $\Grnn{k}{n},$
which is a smooth submersion. 
\end{proof}}
\begin{prop}
\label{prop:SA_real_bdry}
For every $Z \in \Mat^{>}_{n \times (k+4)}$, the set $\Z(\SA)$ is contained in the topological boundary of the amplituhedron $\Ampl_{n,k,4}(Z)$.
\end{prop}

\begin{proof}
By definition $\Z(\SA)\subseteq\Ampl_{n,k,4}(Z).$ We need to show that every open neighborhood in $\Gr_{k,k+4},$ of any point $p\in\Z(\SA),$ intersects $\Gr_{k,k+4}\setminus\Ampl_{n,k,4}(Z).$
Assume that $\langle i,i{\pl1},j,j{\pl1}\rangle$ vanishes at $p.$
Since $\langle i,i{\pl1},j,j{\pl1}\rangle$ has a constant sign $s$ on $\Ampl_{n,k,4}(Z)\setminus\Z(\SA)$ and vanishes on $\Z(\SA),$ by Lemma~\ref{obs:SA_and_bdry_twistors}, it is enough to find points in $\Gr_{k,k+4}$ arbitrarily close to $p$ for which $\langle i,i{\pl1},j,j{\pl1}\rangle$ is of sign $-s,$ since such points cannot belong to the amplituhedron.

Let $C\in \SA$ be a matrix representation for a preimage of $p.$ For $R=(R_{i}^j)_{i\in [k]}^{j\in [n]}\in\Mat_{k\times n},$ write $C(R)=C+R.$ The function $G(R)=\langle \Z(C(R)) Z_{i,i{\pl1},j,j{\pl1}}\rangle$ is a multilinear function of the variables $R_{i}^{j},~i\in [k],~j\in [n].$  Order the pairs $(a,b)_{a\in [k],b\in [n]}$ arbitrarily. Let $(a',b')$ be the minimal element in that order such that 
\[G(R)|_{R_{a}^{b}=0~\forall (a,b)>(a',b')}\notequiv0,\]
where $'>'$ is the aforementioned order on pairs. 
{In other words, $(a',b')$ is defined by the property that if we restrict $G(R)$ to the domain where $R_a^b=0,$ for all $(a,b)>(a',b')$, then it is not identically $0,$ but if we restrict it to the domain where $R_a^b=0,$ for all $(a,b)\geq (a',b')$ then it is. 
Such $(a',b')$ exists since the function $G(R)$ vanishes when all $R_{i}^{j}=0,$ but it is not identically zero, since any element of $\Gr_{k,n}$ can be be represented by $C(R)$ for some $R,$ and $\langle i,i{\pl1},j,j{\pl1}\rangle$ is not identically zero on the whole Grassmannian.
We can write 
\begin{equation}\label{eq:multilinear_arg}G(R)|_{R_{a}^{b}=0~ \forall (a,b)>(a',b')}=H(R)+R_{a'}^{b'}F(R),\end{equation}
where $F(R),H(R)$ are multilinear functions in $(R_{a}^b)$ for $(a,b)<(a',b').$
By the choice of $(a',b')$ it follows that $H\equiv 0$ while $F\notequiv0.$  We will find a sequence of $k\times n$ matrices $\{R(l)\}_{l=1,2,\ldots}$ with $\lim_{l\to\infty}R(l)= 0,$ or equivalently $C(R(l))\to C,$ for which \[\sgn(\langle \Z(C(R(l))) Z_{i,i{\pl1},j,j{\pl1}}\rangle)=-s,\]thus establishing the claim.} 

{Write \[M_{(a',b')}:=\{\overline{R}\in \Mat_{k\times n}|\overline R_a^b=0~\forall (a,b)\geq (a',b')\}\simeq \R^{\{(a,b)|~(a,b)<(a',b')\}}.\]
Since $F$ is a multilinear function which is not identically $0,$ it is non zero on a dense subset of $M_{(a',b')}.$ Let $\overline{R}(l)\in M_{(a',b')}$ be a sequence of matrices on which $F(\overline{R}(l))\neq 0,$ and $\overline{R}(l)\to 0$ as $l\to\infty.$ Define $R(l)\in \Mat_{k\times n}$ by \begin{equation}\label{eq:R_a^b}R(l)_a^b=\begin{cases}
 \overline{R}(l)_a^b,~~(a,b)<(a',b'),\\ 
 (-s)\sgn(F(\overline{R}(l))/l,~~(a,b)=(a',b')\\
 0,~~(a,b)>(a',b')
\end{cases}.\end{equation}
Then $C(R(l))\to C,$ and, by \eqref{eq:multilinear_arg} and \eqref{eq:R_a^b}, \[\sgn(\langle \Z(C(R(l))) Z_{i,i{\pl1},j,j{\pl1}}\rangle)=\sgn\left(\frac{(-s)\sgn(F(\overline{R}(l))}{l} F(\overline{R}(l))\right)=-s.\]
}
\end{proof}

\begin{prop}
\label{prop:inj_codim_1}
For every BCFW cell $S_D \in \mathcal{BCFW}_{n,k}$ and $Z \in \Mat^{>}_{n \times (k+4)}$, the amplituhedron map $\Z$ is injective on $\overline{ S_D}\setminus \SA,$ and the inverse map is smooth. For every stratum $S\subseteq\overline{ S_D}\setminus \SA$, the image $\Z(S)$ is an embedded smooth submanifold of $\Gr_{k,k+4}.$ 
\end{prop}

\begin{proof}
By Lemma 
\ref{lem:at least one star vanishes} and \eqref{eq:widetilde_vs_ovrline} every element of $\overline{S_D}\setminus \SA$ has an extended domino representation, hence belongs to the set $S^\diamondsuit_D$ defined in Proposition~\ref{prop:injectiveness_more_general}. The first claim is the a consequence of applying Proposition~\ref{prop:injectiveness_more_general}. In particular $\Z|_S:S\to\Gr_{k,k+4}$ is a topological embedding. It is left to show that $\Z|_S$ is an immersion. We sketch the argument.
Let $p\in S$ be an arbitrary point, and write $q=\Z(p).$ 
{Write $\mathcal{I}_D$ for the set of Pl\"ucker coordinates which vanish on elements of $S_D,
$ and $\mathcal{I}_S$ for the set of Pl\"ucker coordinates which vanish on elements of $S.$}
By Lemma \ref{lem:smooth_extension} the map $\Z$ extends to a smooth map, given by the same formula, from an open neighborhood of $\Grnn{k}{n}$ in $\Gr_{k,n}.$ In particular $\Z$ extends smoothly, and by the same expression to a neighborhood $W$ of $p$ in $\Gr_{k,n}.$ Denote also the extended map by $\Z.$ Similarly, since no boundary twistor vanishes on $q$, the inverse map constructed in Proposition~\ref{prop:injectiveness_more_general} extends smoothly, and by the same expression, to the function $F_D(Y):=C_D(Y,Z)$, which is defined on an open neighborhood $U\subset\Gr_{k,k+4}$ of $q,$ on which no boundary twistor vanishes. The subscript $D$ indicates it is calculated with respect to the $D$-domino form.
{Indeed, the procedure of Proposition \ref{prop:injectiveness_more_general} assigns for any $q'$ on which no boundary twistor vanishes, a \emph{matrix} $F_D(q')$ that is in a $D$-domino form in the sense of Definition~\ref{def:domino_entries}, but does not necessarily satisfy the sign constraints of Definition~\ref{def:domino_signs}. Since $F_D$ is a matrix-of-functionaries valued function it is in particular continuous and smooth, hence, if $q'$ is close enough to $q$ then since $F_D(q)$ has full rank, also $F_D(q')$ has. Thus, $F_D(q')$ represents a point in the Grassmannian, which has a $D$-domino representative only without the inequalities of Definition~\ref{def:domino_signs}, or $F_D(q')\in S_D^\diamondsuit$. By shrinking $U$ furthermore we may assume that $F_D(U)\subseteq W,$ so that the amplituhedron map is well defined on $F_D(U).$
}
Now, by Proposition~\ref{prop:injectiveness_more_general}
$F_D:U\to F_D(U)$ is a right inverse to $\Z.$ 
{Thus, $F_D(U)=S^\diamondsuit_D\cap\Z^{-1}(U),$ and is a smooth submanifold of $\Gr_{k,n}$ of dimension $4k.$ Note that Proposition \ref{prop:minimal_set_of_ineqs} implies that for every full rank element of $S^\diamondsuit_D$, and in particular for elements of $F_D(U),$ all Pl\"ucker coordinates contained in $\mathcal{I}_D$ vanish. 
}{
{We may assume, by making $U$ smaller once more, that for every point in $F_D(U),$ its set of vanishing Pl\"ucker coordinates is contained or equals $\mathcal{I}_S.$ } 
One can now pick, as explained in Summary \ref{subsec:plabic222}, a plabic graph with a perfect orientation for $S_D,$ such that after omitting edges one obtains a plabic graph for $S$ (see also \cite[Section 18]{postnikov2006total}). 
This data induces a diffeomorphism (the \emph{boundary measurement map}) $\Phi:V\to F_D(U)\cap \overline{S}_D$ where $V\subset[0,\infty)^{4k}$ is relatively open in $[0,\infty)^{4k},$ and $\Phi^{-1}(F_D(U)\cap S)$ is the subset of $V$ where certain coordinates vanish. It is not hard to show, if $U$ is small enough, that the boundary measurement map extends to a diffeomorphism $\Phi:B\to F_D(U),$ where $B\subseteq\R^{4k}$ is an open set with $B\cap [0,\infty)^{4k}= V.$ 
}{
Thus, $S\cap F_D(U)$ is a smooth submanifold of $F_D(U)$ of dimension $\dim(S),$ obtained as the diffeomorphic image under $\Phi$ of the zero locus of certain coordinate functions on $B$.
} 
$\Z|_{F_D(U)}:F_D(U)\to U$ is a smooth bijection with a smooth inverse, hence a diffeomorphism, and its restriction to $S\cap F_D(U)$ is therefore an immersion. Applying this argument for every $p\in S$ yields the result.
\end{proof}

\begin{rmk}
\label{rmk:bdry_comps_image_as_zero_loci}
$\partial_\star S_D$ are defined using vanishing of matrix entries or of $2\times2$ minors. The inverse problem procedure provides a formula for each matrix entry in terms of twistor coordinates. Thus, the image of $\partial_\star S_D$ under the amplituhedron map is contained in the zero locus of some function of the twistors. This function is of constant sign on $S_D.$
\end{rmk}

\begin{prop}\label{prop:criterion_for_surj}
Let $f:B\to N$ be a smooth submersion between two manifolds (\emph{without boundary}) $B,N,$ where the dimension of $N$ is $n.$ Let $L$ be a connected open subset of $B$ with a compact closure $\overline{L}\subset B.$ Denote $f(\overline{L})$ by $K.$
Let $\{S_a\}_{a\in A}$ be a collection of $n$-dimensional submanifolds \emph{without boundary} of $B,$ which are contained in $\overline{L}$ and satisfy the following properties:
\begin{enumerate}
    \item $f(S_a)\cap f(S_b)=\emptyset,~a\neq b.$
    \item $\overline{S_a}$ is compact. The (topological) boundary of each $S_a$ has a stratification $\partial S_a = \bigcup_{i=1}^{k_a}S_{a;i}$ where each $S_{a;i}$ is a submanifold (\emph{without boundary}) of $B$ of dimension at most $n-1.$ We further assume that for each $S_{a;i},$ its topological boundary is contained in the union of strata $S_{a;j}$ of smaller dimension.
    \item 
    Every $f(S_{a;i})$ is either contained in $\partial K,$ the topological boundary of $K,$ or it is a smooth submanifold of $K.$ Moreover, if $S_{a;i}$ is of dimension $n-1,$ and  $f(S_{a;i})\varsubsetneq\partial K,$ then $S_{a;i}=S_{b;j}$ for some other $b,j.$ Boundaries of the latter type are called \emph{shared boundaries}.
    \item For every $D_a,$ $f$ is injective on the union $S_a$ and its shared boundaries, and the inverse map is continuous.
\end{enumerate}
Then 
\[f(\bigcup_{a\in A}\overline{S_a})=K.\]
Moreover, $f(S_{a;i})\cap \partial K=\emptyset$, for every shared boundary $S_{a;i},$ and $\partial K$ is precisely the union of  $f(\overline{S_{a;i}}),$ taken on $S_{a;i}$ of dimension $n-1$ which are not shared boundaries.
\end{prop}
Our proof uses the following transversality result, whose proof is postponed to the end of the section.
\begin{lemma}\label{lem:transversality}
Let $M$ be a smooth connected manifold of dimension $n$ and $N_1,\ldots, N_r$ a finite collection of smooth submanifolds of dimensions at most $n-1.$ 
Suppose that the closure of each $N_i$ (inside $M$) is contained in the union of $N_i$ with other submanifolds from $\{N_1,\ldots, N_r\},$ of dimensions smaller than $\dim(N_i).$
Then for every \[p,q\in M\setminus\left(\bigcup_{i=1}^r N_i\right)\] there exists a smooth path connecting $p,q$ which
intersects each $N_i$ finitely many times, and, if $\dim(N_i)<n-1$ then it does not intersect $N_i$ at all.
\end{lemma}
\begin{proof}[Proof of Proposition \ref{prop:criterion_for_surj}]
Since $L$ is open and connected, and $f$ is a submersion, $f(L)$ is open and connected. 
Since $\overline{L}$ is compact and connected, so is $\overline{f(L)}=K.$ Thus, $\text{int}(K),$ the interior of $K,$ which is an open set containing $f(L)$ but contained in its closure, is connected as well.
Since $\overline{S_a}$ are all compact, so are $f(\overline{S_a}).$ By the third assumption and the smoothness of $f,$ any $S_{a;i}$ which does not map to $\partial K,$ maps to a smooth submanifold of dimension at most the dimension of $S_{a;i}.$
By the last item, $f$ restricts to a homeomorphism of manifolds \emph{with boundary} from the union of $S_a$ and its shared boundaries, to its image. In particular, since $\text{dim}(S_a)=\text{dim}(N),$ each restriction $f|_{S_a}$ is an open map.
The remaining boundaries of $S_a$ either map to $\partial K$ or to smooth submanifolds of codimension at least $2.$

Assume by way of contradiction that the images of $(\overline{S_a})_{a\in A}$ do not cover $K.$

Then one can find a point $p\in \text{int}(K)$ which is not in this union.
Take an arbitrary point $q$ in the image of some $S_a.$ Then 
$q\in\text{int}(K),$since $f|_{S_a}$ is open and $f|_{S_a}(S_a)=f(S_a)\subset K.$ 
We can apply Lemma~\ref{lem:transversality} with $M=\text{int}(K),$ which is a connected submanifold of $N,$ and with $N_1,\ldots, N_r$ being the images $f(S_{a;i}),$ for strata $S_{a;i}$ which do not map to $\partial K.$ The second and third items above guarantee that the assumptions of the lemma are satisfied. Thus, we can find a path which does not pass through any $f(S_{a;i}),$ for boundary strata $S_{a;i}$ of codimension two or more, and passes only finitely many times through images of $n-1$-dimensional boundary strata, which must be shared boundaries.
Denote this path by $u:[0,1]\to \text{int}(K),$ with $u(0)=q,~u(1)=p.$
Let $t^*\in[0,1]$ be the last time $u(t)$ lies in the image of $\bigcup\overline{S_a}$. Then $t^*<1$ since this union is closed. Moreover, $u(t^*)$ cannot belong to any $S_a$, since the image of any point in such a cell has a neighborhood in $\text{int}(K)$ which is also contained in the image of the cell. So $u(t^*)$ must lie in the image of a shared boundary $S_{a;i}=S_{b;j},~b\neq a.$ Let $q^*$ be the unique preimage of $u(t^*)$ in that boundary, let 
{$W\subset {S_a}\cup S_{a;i} \cup{S_b}$ be an open neighborhood of $q^*.$ Write}
\[W_a=W\cap( {S_a}\cup S_{a;i}),~W_b=W\cap({S_b}\cup S_{b;j}).\] 
Then \[W_a\cap S_{a;i}=W_b\cap S_{b;j}.\]
Write $U_a=f(W_a),~U_b=f(W_b)\subset K$. 
Then $u(t^*)$ lies in the interior of $U_a\cup U_b.$ Indeed, let $\R^{n}_{\pm}$ be the half spaces of $\R^{n}$ defined by first coordinate being non negative or non positive, and consider $\R^{n-1}$ as their subspace where the first coordinate is $0.$
By making $W$ smaller if necessary, we may assume that there are homeomorphisms of manifolds \emph{with boundary} \[g:(\R^{n}_+,\R^{n-1})\to (W_a,S_{a;i}\cap W_a),~h:(\R^{n}_-,\R^{n-1})\to (W_b,S_{b;j}\cap W_b)~\text{such that }g|_{\R^{n-1}}=h|_{\R^{n-1}},~g(0)=q^*.\]
Composing with the map $f$ we obtain maps
\[G:(\R^{n}_+,\R^{n-1})\to (U_a,\widehat{\partial U}),~H:(\R^{n}_-,\R^{n-1})\to (U_b,\widehat{\partial U})~\text{such that }G|_{\R^{n-1}}=H|_{\R^{n-1}},~G(0)=u(t^*),\]
{where \[\widehat{\partial{U}}=U_a\cap f(S_{a;i})=f(W_a\cap S_{a;i})=f(W_b\cap S_{b;j})=U_b\cap f(S_{b;j}).\]}
We now perform a \emph{gluing procedure}: we can glue $G$ and $H$ along the boundary $\R^{n-1}$, to obtain a map from $\R^{n}$ to $U_a\cup U_b.$ This map is continuous. It is injective since it is injective on each closed half space, by the previous paragraph, and also on $\R^{n}\setminus\R^{n-1}$ by the first assumption. Its inverse is continuous since it is continuous on the image of each closed half space, by the fourth assumption again. Thus, it is a homeomorphism onto its image, which must be open, and contains $u(t^*).$
But then for small enough $t'>t^*,~u(t')$ must still be in $U_a\cup U_b,$ hence in the union of the images of $\overline{S}_a,\overline{S}_b.$ A contradiction to the definition of $t^*.$ 

We turn to the 'More-over' part.
{Every point in a shared boundary $S_{a;i}$ maps to an internal point of $K,$ by the gluing procedure above. Thus the shared boundaries in the sense of this proposition indeed map to the interior of $K.$}  
{Since $f(\bigcup _{a\in A}S_a)\subseteq\text{int}(K),$ and $f(\bigcup _{a\in A}\overline{S_a})=K,$ which is what we just showed,}\[\partial K\subseteq f(\bigcup_{a\in A}\partial S_a).\] 
Write \[S'=\bigcup_{\substack{S_{a;i}~\text{is not an internal}\\\text{ boundary of $S_a$}}} S_{a;i},~~~~~S=\bigcup_{\substack{S_{a;i}~\text{is not a shared boundary} \\\text{of $S_a,$~~$\dim S_{a;i}=n-1$}}} S_{a;i}.\]
Then $S'\setminus S$ is the union of boundary components of dimensions at most $n-2.$
Assume by way of contradiction that there exists $p\in\partial K\setminus f(\overline{S}).$ 
Take a small neighborhood $U\subset N$ of $p$ which does not intersect $f(\overline{S})=\overline{f(S)}.$
Then \[U\cap\partial K\subseteq f(S'\setminus S).\]
Since $p$ is a boundary point, and $K$ is the closure of an open set, we can find $q_1\in\text{int}(U\cap K),$ and $q_2\in U\setminus K.$
As in the proof of the previous part of the proposition, Lemma~\ref{lem:transversality}, with {$M=U$ and $N_1,\ldots N_r$ are $f(S_{a;i})\cap U,$ for the different strata  $S_{a;i}$ (note that for the strata which compose $S,$ $f(S_{a;i})\cap U=\emptyset$),} guarantees the existence of a path in $U$ which connects $q_1,q_2$ and does not pass through $f(S'\setminus S).$ This path cannot escape from $K,$ as can be seen by repeating the argument from the previous part of the proposition. But this is a contradiction, since $q_2\notin K.$
\end{proof}
\begin{rmk}
The proof of the above proposition shows that the first item, the separation, can be weakened to just assuming that $S_a$ and $S_b$ are \emph{locally separated} near their shared boundary, for every pair of cells that have share a boundary component, where by local separation we mean that the outward normals of the shared boundary, with respect to $S_a$ and $S_b,$ are pointing oppositely.
\end{rmk}

\begin{proof}[Proof of Theorem~\ref{thm:surj}]
The proof follows from Proposition~\ref{prop:criterion_for_surj}, once we verify its assumptions. 
More precisely, we verify that for $f=\Z,$ the amplituhedron map, $B\subset\Gr_{k,n}$ being some open neighborhood of $\overline{L}$ for $L=\Gr^{>}_{k,n},$ $N=\Gr_{k,k+4}$ and $\{S_a\}_{a\in A}$ being the BCFW cells, the assumptions of Proposition~\ref{prop:criterion_for_surj} hold. The conclusion is that the closure of images of the BCFW cells covers $K=\Ampl_{n,k,4}.$

We first verify the assumptions on $f=\Z,L=\Gr^{>}_{k,n},$ and its neighborhood $B.$ 
The open space $\Gr^{>}_{k,n}$ and its compact closure $\Grnn{k}{n}$ are connected. We take $B\subset\Gr_{k,n}$ to be the open neighborhood of $\Grnn{k}{n}$ guaranteed by Lemma \ref{lem:smooth_extension}. Hence the amplituhedron map extends to a smooth submersion  $\Z:B\to\Gr_{k,k+m}.$ 

Next, we need to verify the assumptions of Proposition~\ref{prop:criterion_for_surj} for the BCFW cells. Their dimensions indeed equal $\dim(\Gr_{k,k+4})=4k$. The first property is Theorem~\ref{thm:separation}. The second property is true for any positroid cell. The third property follows from Proposition~\ref{prop:SA_real_bdry} and Proposition~\ref{prop:bdries_either_paired_or_SA}. Proposition~\ref{prop:inj_codim_1} shows the fourth property, that $\Z$ is injective on the union of each BCFW cell and its \emph{shared boundaries} (in the sense of Proposition \ref{prop:criterion_for_surj}), those codimension one boundary strata which do not map to the boundary, and that $\Z^{-1}$ is injective on the image of this union.
The theorem follows.
\end{proof}

We have the following corollaries.
\begin{cor}
\label{cor:bdries}
The boundary of $\Ampl_{n,k,4}(Z)$ equals $\Z(\SA),$ which is, {by Remark \ref{rmk:for_SA_usage},} precisely the union of the intersections of $\Ampl_{n,k,4}(Z)$ with the hyperplanes \[\langle  i,i{\pl1},j,j{\pl1}\rangle = 0,\] for $i,j\in[n]$ such that $|\{i,i{\pl1},j,j{\pl1}\}|=4.$
Thus, $\text{int}(\Ampl_{n,k,4}(Z))=\Z(\Grnn{k}{n}\setminus\SA).$
\end{cor}
{\begin{proof}
The 'More-over' part of Proposition \ref{prop:criterion_for_surj}, applied to the amplituhedron, as in the first paragraph of the proof of Theorem \ref{thm:surj}, whose conditions were verified in the course of proof of Theorem \ref{thm:surj}, describes the topological boundary of $\Ampl_{j,k,4}(Z):$ it is the union of $\Z(S),$ over all codimension $1$ boundary strata $S$ of BCFW cells, which are not shared boundaries (in the sense of Proposition \ref{prop:criterion_for_surj}). By Proposition \ref{prop:bdries_either_paired_or_SA} all codimension $1$ boundaries strata of BCFW cells which are not internal are contained in $\SA.$ Since $\SA$ is closed this implies
\[\partial\Ampl_{n,k,4}\subseteq\Z(\SA).\]
On the other hand, Proposition \ref{prop:SA_real_bdry} implies the other containment
\[\partial\Ampl_{n,k,4}\supseteq\Z(\SA).\]
\end{proof}}
\begin{cor}
\label{good-triangulation}
The image of each codimension-one stratum of a BCFW cell that does not map to the boundary, is also the injective image of a codimension-one boundary of another BCFW cell.
\end{cor}
{
\begin{proof}
From the previous corollary every codimension $1$ boundary stratum of a BCFW cell that does not map to the boundary does not intersect $\SA$. From Proposition \ref{prop:bdries_either_paired_or_SA} every such boundary stratum is also a boundary stratum of another BCFW cell. By Proposition \ref{prop:inj_codim_1} the amplituhedron map is injective on these boundaries.
\end{proof}
}

We conclude this section with the proof of Lemma~\ref{lem:transversality}. The proof relies on the notion of transversality and on Thom's Parametric Transversality Theorem \cite{thom1954quelques} that we now recall. A proof of Thom's theorem can be found, for example, in~\cite{guillemin2010differential}. We assume that all manifolds of interest have a countable atlas.

\begin{definition}
Let $f:X\to Y$ be a smooth map between a smooth manifold with boundary $X$ and a smooth manifold $Y.$ Let $Q$ be a smooth submanifold of $Y.$ We say that $f$ is \emph{transverse} to $Q,$ and write $f\pitchfork Q$ if for every $x\in f^{-1}(Q)$
\[df_x(T_xX)+T_{f(x)}Q=T_{f(x)}Y,\]where $T_bB$ denotes the tangent space of $B$ at $b\in B,$ and $df_x$ is the differential map at $x,$ which maps $T_xX$ into $T_{f(x)}Y.$
We denote by $\partial f$ the restriction of $f$ to $\partial X.$
\end{definition}

\begin{definition}
Let $X$ be an $n$ dimensional manifold with an atlas $\{(U_\alpha,\phi_\alpha:U_\alpha\to\R^n)\}_{\alpha\in A},$ we say that a set $C\subseteq X$ is of measure $0,$ if for every $\alpha\in A,$ the set $\phi_\alpha(M\cap U_\alpha)$ is of Lebesgue measure $0$ in~$\R^n.$
If $Y\subset X$ is the complement of a measure $0$ subset, we say that almost every $x\in X$ belongs to $Y.$
\end{definition}

\begin{thm}[Thom's Parametric Transversality Theorem]\label{thm:thom}
Let $X$ be a smooth manifold with boundary, let $B,Y$ be smooth manifolds and let $Q$ be a submanifold of $Y$. Let $F:X\times B\to Y$ be a smooth map.
Suppose that $F\pitchfork Q$ and that $\partial F\pitchfork Q.$
Then for almost every $b\in B$ the map
\[F(-,b):X\times\{b\}\to Y\]
is transverse to $Q.$
\end{thm}
\begin{proof}[Proof of Lemma~\ref{lem:transversality}]
Fix $p,q\in M\setminus\left(\bigcup_{i=1}^r N_i\right)$ and $\epsilon>0.$
Take an arbitrary smooth embedding $u_0:(-\epsilon,1+\epsilon)\to M$ with $u(0)=p,u(1)=q.$
Extend $u_0$ to a diffeomorphism $U$ between $(-\epsilon,1+\epsilon)\times B_1,$ where $B_1\subset \R^{n-1}$ is the open unit ball, and a tubular neighborhood $W$ of $u_0((-\epsilon,1+\epsilon)).$
Define $F:[0,1]\times B_1\to M$ by
\[F(t,x)=U(t,t(1-t)x).\]
We claim that $F$ satisfies the assumptions of Theorem~\ref{thm:thom} with $X=[0,1],~B=B_1$ and $Q=N_i,$ for any given $i.$ Since $F(0,x)=p\notin N_i,~F(1,x)=q\notin N_i,$ it follows that $\partial F\pitchfork N_i.$ It is left to verify $F|_{(0,1)\times B_1}\pitchfork N_i.$ Since $U$ is a submersion, also $F$ is a submersion when $t\neq 0,1.$ Thus, \[dF_{(t,x)}\left(T_{(t,x)}\left((0,1)\times B_1\right)\right)=T_{F(t,x)}M.\]
Applying Theorem~\ref{thm:thom} we see that for each  $i,$ for almost all $b\in B_1,$ the path $F(-,b)\pitchfork N_i.$ Thus, we can find $b$ for which 
\[F(-,b)\pitchfork N_1,\ldots, N_r.\]
Write $u(t)=F(t,b).$
Then $u:[0,1]\to M$ is a smooth path which satisfies
\[u(0)=p,~u(1)=q,~ u\pitchfork N_i~\text{for }i=1,\ldots,r.\]
For any $i$ with $\dim(N_i)<n-1,$ the definition of transversality, and the fact that $du_{x}(T_x(0,1))$ is one dimensional implies that $u([0,1])\cap N_i=\emptyset.$

Suppose that $\dim(N_i)=n-1,$ but assume, towards contradiction, that $u([0,1])$ intersects $N_i$ infinitely many times. Let $t_1,t_2,\ldots$ be a converging sequence of such times, and let $t_*$ be its limit. 
Since $u$ is continuous, \[u(t_*)=\lim_{j\to\infty}u(t_j)\in \overline{N_i}.\]
Since $p,q\notin\overline{N_i},$ by the assumptions, $t_*\in(0,1).$ Since $\overline{N_i}$ is contained in the union of $N_i$ with submanifolds $N_j$ of smaller dimension, which are known not to intersect $u([0,1])$ we can deduce that $u(t_*)\in N_i.$

Because $u$ is smooth, and the sequence $\{t_i\}_{i=1}^\infty$ converges to $t_*,$ we can calculate the tangent to $u((0,1))$ at $u(t_*)$ along the sequence $\{u(t_i)\}_{i=1}^\infty\subset N_i,$ which shows that this tangent belongs to $T_{u(t_*)}N_i,$ contradicting the transversality. The proof follows.
\end{proof}

\section{Consequences}
\label{consequences}

In this final section we present two consequences of our results and methods to the topology and the structure of the amplituhedron. First, the interior of the amplituhedron $\Ampl_{n,k,4}(Z)$ is homeomorphic to an open ball, for every positive~$Z$. So far, this has been known for one particular $Z$, based on a general-$m$ result by Galashin, Karp and Lam \cite{galashin2022totally}, see also~\cite{bloch2023gradient}. 

Our second result gives a decomposition of $\Ampl_{n,k,4}(Z)$ as a union of $\Z$-images of a certain type of products of two positive Grassmannians. This kind of decomposition is less refined than the BCFW triangulation that we define in Section~\ref{bcfw} following the work of Karp, Williams, and Zhang~\cite{karp2020decompositions}, based on Arkani-Hamed et~al.~\cite{arkani2016Grassmannian}. However, this high-level decomposition more transparently reflects the outcome of one step in the BCFW recurrence, see Bai and He~\cite{bai2015amplituhedron}. A consequence of this decomposition and other tools developed in this work is a mechanism to obtain triangulations of $\Ampl_{n,k,4}(Z)$ from various triangulations of amplituhedra of smaller $k$ and $n$. This includes many more collections of cells beyond those discussed in this work, such as those obtainable from different ways to apply the BCFW recursion. These further results will appear in a separate paper \cite{even2023cluster}.

\subsection{The Interior is a Ball}
\label{ball}

Galashin, Karp and Lam \cite{galashin2022totally} proved the following result
\begin{thm}\label{thm:ball}
For every $n,k,m$ there is a special choice of a positive matrix $Z_*=Z_*^{n,k,m}$ such that the amplituhedron $\Ampl_{n,k,m}(Z_*)$ is homeomorphic to a closed ball. 
\end{thm}
An immediate conclusion is that for the same $Z_*$ the interior of the amplituhedron is homeomorphic to an open ball. We now show that when $m=4,$ this last fact generalizes for every positive $Z.$

\begin{thm}\label{thm:int_is_ball}
Let $Z$ be an arbitrary positive $n\times (k+4)$ matrix. Then the interior of $\Ampl_{n,k,4}(Z)$ is homeomorphic to an open ball.
\end{thm}
The proof relies on the following lemma, which is interesting in its own right.
\begin{lemma}\label{lem:inj_all_int}
The amplituhedron map is a homeomorphism from $M=\bigcup_{D\in\CD_{n,k}}\overline{S_D}\setminus\SA$ onto its image $\text{int}(\Ampl_{n,k,4}).$
\end{lemma}
\begin{proof}
$\Z$ is continuous on $M,$ and surjects on its image $\text{int}(\Ampl_{n,k,4})$ by Theorem~\ref{thm:surj} and Corollary~\ref{cor:bdries}.
We need to prove it is also injective and open.

$M$ is endowed with the topology induced from the Grassmannian $\Gr_{k,n}.$
Write $M_i$ for the open subset of $M$ defined as the union of all dimension $4k-j$ strata of $\overline{S_D}\setminus\SA,$ over all BCFW cells and $j=0,\ldots,i.$ Note that $M=M_{4k}.$
We prove the following three claims, whose combination immediately implies the lemma:
\begin{enumerate}
    \item $\Z|_{M_0}:M_0\to \Z(M_0)$ is a homeomorphism, and $\Z|_{M_1}:M_1\to \Z(M_1)$ is open.
    \item If for $i\geq 1,$ $\Z|_{M_i}$ is open, then $\Z|_{M_i}$ is also injective.
    \item If for $i\geq 2,$ $\Z|_{M_{i-1}}:M_{i-1}\to \Z(M_{i-1})$ is a homeomorphism, then $\Z|_{M_i}$ is an open map.
\end{enumerate}
We start with the first claim. 
{For each point of $p\in S_D\subseteq M_0$ the proof of Theorem~\ref{thm:injectiveness} shows that $\Z$ maps a small open neighborhood of $p$ in $S_D$ homeomorphically on its image, which is open in $\Gr_{k,k+4}$. Thus $\Z|_{M_0}$ is open. $\Z|_{M_0}$ is injective by Theorems \ref{thm:injectiveness} and \ref{thm:separation}.}

{$M_1\setminus M_0$ is the union of shared boundaries of BCFW cells. Consider $p\in S,$ where $S$ is a shared boundary of a BCFW cell $S_D$. Then, by Proposition \ref{prop:bdries_either_paired_or_SA} it is a codimension $1$ boundary of another BCFW cell $S_{D'}.$ By Proposition~\ref{prop:inj_codim_1} and Theorems \ref{thm:injectiveness} and \ref{thm:separation}, $\Z|_{S_D\coprod S\coprod S_{D'}}$ is injective. The proof of Proposition~\ref{prop:criterion_for_surj} uses a gluing procedure to construct, for every point $q^*\in S,$ a neighborhood in $S_D\coprod S\coprod S_{D'}$ which maps homeomorphically under $\Z$ to an open neighborhood of $\Z(q^*)$ in $\Gr_{k,k+4}.$ Thus $\Z|_{M_1}$ is open. }

For the second claim, assuming $\Z|_{M_i}$ is open, if $\Z(p_1)=\Z(p_2)$ for some $p_1\neq p_2\in M_i,$ then for any two neighborhoods $B_1,B_2\subset M_i$ of $p_1,p_2$ respectively, $\Z(B_1)\cap\Z(B_2)$ contains an open ball $B\subset\Z(M_i)$ around $\Z(p_1)=\Z(p_2).$ By making $B_1,B_2$ smaller we may assume $B_1\cap B_2=\emptyset.$ Since $\Z$ is a smooth map, the images of strata of positive codimension are nowhere dense, hence there are $q_1\in B_1\cap\left(\bigcup_{D\in\CD_{n,k}}{S_D}\right)=B_1\cap M_0$ and $q_2\in B_2\cap\left(\bigcup_{D\in\CD_{n,k}}{S_D}\right)=B_2\cap M_0$ respectively, with $\Z(q_1)=\Z(q_2),$ in contradiction to {the first claim, which shows that $\Z|_{M_0}$ is injective.}

Finally, for the third claim,
suppose that $\Z$ is a homeomorphism on $M_{i-1},$ for some $i\geq 2,$ but assume towards contradiction that for some point $p\in M_i$ we cannot find a small enough neighborhood of it in $M_i$ which is mapped under $\Z$ to an open set. Then $p\in M_i\setminus M_{i-1},$ since $M_{i-1}$ is open in $M_i$ and $\Z|_{M_{i-1}}$ is open. Let $S$ be the positroid cell containing $p.$ 

Take a smooth metric on  $\Gr_{k,n}$ and restrict to $\Grnn{k}{n}.$ For each point $x\in \Grnn{k}{n},$ and every $r>0$ small enough, the ball of radius $r$ around $x$ only intersects positroid cells that contain $x$ in their closures. Moreover, the metric can be chosen so that the intersection of the ball with any stratum is path connected.
Restrict the metric to $M.$ Take an open ball $B$ in $M$ around $p,$ which is small enough to be path connected and that its closure intersects only 
positroid cells $S'$ with $p\in \overline{S'}.$ By Proposition~\ref{prop:inj_codim_1}
$\Z$ is injective on $S\cap B,$ and moreover, for every $x$ which belongs to a positroid cell $S'\neq S$ with $S\subseteq \overline{S'},~S'\subseteq M_{i-1}$ $\Z(x)\notin\Z(S).$
By the assumption, $\Z$ is also injective on $B\setminus S\subset M_{i-1},$ and hence $\Z$ is injective on $B.$

By its construction $B$ is a connected and path connected component of $M\setminus\partial B,$ and $\Z(\partial B),\Z(\overline{B})$ are compact subsets of $\Ampl_{n,k,4},$ which do not intersect $\partial\Ampl_{n,k,4},$ as $B\cap\SA=\emptyset.$
Every connected component of $\text{int}(\Ampl_{n,k,4})\setminus\Z(\partial B)$ is open. We will therefore finish by showing that $\Z(B)$ is a path connected component of $\text{int}(\Ampl_{n,k,4})\setminus\Z(\partial B).$
Since $B$ is path connected, also $\Z(B)$ is. Let $C$ be the connected component of $\text{int}(\Ampl_{n,k,4})\setminus\Z(\partial B)$ which contains it. $C$ is an open manifold. It is enough to show that there is no $x\in C\setminus\Z(B).$ If there were such $x,$ then since $x\notin\Z(\partial B)$ as well, then $x\notin \Z(\overline{B}).$ Since $\Z(\overline{B})=\overline{\Z(B)}$ is compact, $x$ has an open connected neighborhood $U\subseteq C$ with $U\cap\Z(\overline{B})=\emptyset.$
Since $U$ is open it contains a point $\Z(q_1)$ for $q_1$ which belongs to a BCFW cell. Let $q_2\in B$ be another point of a BCFW cell. We can apply Lemma~\ref{lem:transversality}, {with $C$ as the smooth connected manifold $M$ from the statement of the lemma, and $N_i,~i=1,2,\ldots, r$ are the intersections of $C$ with the different strata of the different $\overline{S_D}\setminus\SA.$} The assumptions of the lemma are easily verified, exactly as in the proof of Proposition~\ref{prop:criterion_for_surj}. Thus, we can find a path in $C$ which connects $\Z(q_1)$ and $\Z(q_2),$ and does not intersect $C\setminus M_1.$ This path is therefore contained in $C\cap M_1.$

Since $\Z|_{M_1}$ is a homeomorphism, this path lifts to a path in $M,$ which connects $q_1,q_2$ but does not intersect $\partial B.$ This is a contradiction, since $q_1\in B,~q_2\notin B$ and $B$ is a path connected component of $M\setminus\partial B.$ And the proof follows.
\end{proof}
\begin{proof}[Proof of Theorem~\ref{thm:int_is_ball}]
Let $Z_*$ be as in Theorem~\ref{thm:ball}.
By Lemma~\ref{lem:inj_all_int} and Theorem~\ref{thm:ball}, $\Z_*$ maps $M$ homeomorphically on $\text{int}(\Ampl_{n,k,4}(Z_*)),$ which is an open ball. In particular $M$ is homeomorphic to an open ball. Now, let $Z$ be an arbitrary positive matrix. By using Lemma~\ref{lem:inj_all_int} for this $Z,$ $\Z(M)=\text{int}(\Ampl_{n,k,4}(Z))$ is also homeomorphic to an open ball.
\end{proof}

\subsection{A High-Level Decomposition}
\label{decomposition}

In this section we describe the decomposition of the amplituhedron corresponding to one step in the BCFW recurrence, which is stated in Theorem~\ref{thm:geom_recursion} below. It is given in terms of the following definition.

\begin{definition}\label{def:mid_emb}
Let $k_1,k_2,j$ and $n$ be nonnegative integers satisfying \[k_1+k_2\leq n-5,~k_1\leq j-2,~k_2\leq n-j-4.\]
For $s_1,s_2,t_1,t_2\in \R^+,$  and two matrices 
\[
L \;=\;
\begin{pmatrix}
L_{1}^{1} & \cdots & L_{1}^{j+1} & L_{1}^{n}\\
\vdots  & \ddots & \vdots & \vdots \\
L_{k_1}^{1} & \cdots & L_{k_1}^{j+1} & L_{k_1}^{n}
\end{pmatrix}\in\Mat_{k_1\times([j+1]\cup\{n\})}, ~~R \;=\;
\begin{pmatrix}
R_{1}^{j} & \cdots & R_{1}^{n-1}\\
\vdots  & \ddots & \vdots \\
R_{k_2}^{j} & \cdots & R_{k_2}^{n-1}
\end{pmatrix}\in\Mat_{k_2\times([n-1]\setminus[j-1])}
\]
we define the \emph{middle embedding} $\midemb_{j}(s_1,s_2,t_1,t_2,L,R)$ as the matrix

\begin{equation*}
\arraycolsep3pt
\left(\begin{array}{*{11}{c}}
L_{1}^{1} & \cdots & L_{1}^{j-1}& L_{1}^{j}{+}\frac{s_1}{s_2}L_1^{j+1}\!\! & L_{1}^{j+1}&0&\cdots&0&0&0 & \!(-1)^{k_2}L_{1}^{n}\\
\vdots  & \ddots & \vdots & \vdots&\vdots &\vdots & \ddots& \vdots & \vdots&\vdots&\vdots\\
L_{k_1}^{1} & \cdots& L_{k_1}^{j-1}& L_{k_1}^{j}{+}\frac{s_1}{s_2}L_{k_1}^{j+1} \!\!& L_{k_1}^{j+1} &0&\cdots &0&0&0& \!(-1)^{k_2}L_{k_1}^{n}\\[0.25em]
0       & \cdots & 0         & s_1 & s_2    &0   &\cdots&0 & (-1)^{k_2}t_1   &(-1)^{k_2}t_2&(-1)^{k_2}   \\
0&\cdots&0& R_1^{j} &\!\!R_1^{j+1}{+}\frac{s_2}{s_1}R_1^j & R_1^{j+2} & \cdots & R_1^{n-3}& R_1^{n-2}{+}\frac{t_1}{t_2}R_{1}^{n-1}\!& R_1^{n-1}&0\\
\vdots  & \ddots & \vdots  & \vdots & \vdots & \vdots&\ddots&\vdots&\vdots&\vdots&\vdots\\
0&\cdots&0& R_{k_2}^j & \!\!R_{k_2}^{j+1}{+}\frac{s_2}{s_1}R_{k_2}^j & R_{k_2}^{j+2}  & \cdots & R_{k_2}^{n-3}& R_{k_2}^{n-2}{+}\frac{t_1}{t_2}R_{k_2}^{n-1}\!& R_{k_2}^{n-1}&0
\end{array}\right)
\end{equation*}
\\In other words, \[\midemb_{j}(s_1,s_2,t_1,t_2,L,R)=\begin{pmatrix}L'\\v\\R'\end{pmatrix},\]
where $L'$ is the $k_1\times[n]$ matrix
\[L'=\pre_{j+2}\ldots\pre_{n-1}y_j(\frac{s_1}{s_2})L\]
$R'$ is the $k_2\times[n]$ matrix
\[R'=\pre_{1}\ldots\pre_{j-1}\pre_nx_j(\frac{s_2}{s_1})y_{n-2}(\frac{t_1}{t_2})R\]
and $v$ is the vector
\[(0,\ldots,0,s_1,s_2,0\ldots,0,(-1)^{k_2}t_1,(-1)^{k_2}t_2,(-1)^{k_2}),\]
where the non zero entries are at positions $j,j+1,n-2,n-1,n.$

{This map descends to a rational map \[\Gr_{k_1,[j+1]\cup\{n\}}\times\Gr_{1,5}\times\Gr_{k_2,[n-1]\setminus[j-1]}\dashrightarrow\Gr_{k,n},\] and preserves positivity when applied to nonnegative Grassmannians.}

Define the set $S_{j,n;k_1,k_2}$ to be the middle embedding $\midemb_{j}(\R_+^4,\Gr^{>}_{k_1,[j+1]\cup\{n\}},\Gr^{>}_{k_2,[n-1]\setminus[j-1]}).$ We extend the definition to the case $k_1=0,j=1$ by setting
\[S_{1,n;0,k_2}=\uemb_1
(\R_+^4\times\Gr_{k_2,[n]\setminus\{1\}}).\]
\end{definition}
The middle embedding operation was considered implicitly in \cite{arkani2014amplituhedron}, and explicitly in \cite{bai2015amplituhedron} as part of their explicit interpretation of the BCFW recursion. In matrix formulation it appeared in \cite{bourjaily2010efficient}.
The lower embedding can be seen to be the special case $k_2=0,~j=n-4.$

{Each $S_{j,n;k_1,k_2}$ is a positroid cell. The simplest way to show that is by using plabic graphs formalism, see \cite[Proposition 5.8]{even2023cluster} for a proof.
Lemma \ref{lem:iterative_construction_mid_emb} below provides an equivalent description of these sets, which manifests them being positroid cells. This latter description is the one we shall use in what follows.
}

\begin{thm}
\label{thm:geom_recursion}
For all $k\geq 0,~n\geq k+4$, and positive $Z \in \Mat^{>}_{n \times (k+4)}$, the sets 
$$ \Z(\pre_{n-1}\Gr^>_{k,[n]\setminus\{n-1\}}), \;\;\;\;\;\; \Z(S_{1,n;0,k-1}), \;\;\;\;\;\; \left\{\Z(S_{j,n;k_1,k_2})\right\}_{\substack{k_1 \geq 0, \, k_2 \geq 0,  \, k_1+k_2 = k-1 \\ k_1+2 \leq j \leq n-k_2-4\;\;\;\;\;\;\;\;\,}} $$ 
are disjoint and their union is dense in
$\Ampl_{n,k,4}(Z).$ 
\end{thm}

The following lemma provides alternative descriptions for the middle embedding, which are in the spirit of the constructions of Section~\ref{sec:domino}.
\begin{lemma}\label{lem:iterative_construction_mid_emb}
For $(k_1,j)\neq (0,1)$ the space $S_{j,n;k_1,k_2}$ can be constructed in the following two ways.
\begin{enumerate}
    \item Start with $\Gr^{>}_{k_2,[n]\setminus[j]}.$ Apply
    \begin{itemize}
        \item Upper embedding $\uemb_j$ 
        to $\R_+^4\times \Gr^{>}_{k_2,[n]\setminus[j]}.$
        \item Perform $\inc_1,\ldots,\inc_{k_1},$ and then $\pre_{k_1+1},\ldots,\pre_{j-1}.$
        \item For each $i=1,\ldots,k_1$ in that order apply $y_{i\mi1}(t_{i;1}),~t_{i;1}>0.$
        \item For each $h=1,\ldots,j+1-k_1,$ apply for $i=k_1,\ldots,1,$ in that order, the operation $x_{i+h-1}(t_{i;h+1})$ with $t_{\bullet;\bullet}>0.$ 
    \end{itemize}
    
    \item Start with $\Gr^{>}_{k_1,[j+1]\cup\{n\}}.$ Apply
    \begin{itemize}
    \item Lower embedding $\lemb_{n-2}$ 
    to $\R_+^4\times \pre_{n-1}\Gr^{>}_{k_1,[j+1]\cup\{n\}}.$
        \item Perform $\inc_{j+2},\ldots,\inc_{k_2+j+1},$ followed by $\pre_{k_2+j+2},\ldots,\pre_{n-3}.$
        \item For each $i=1,\ldots,k_2,$ in that order apply $y_{j+1+i\mi1}(t_{i;1}),~y_{j+1+i\mi2}(t_{i;2}),~t_{i;1},t_{i;2}>0.$
        \item For each $h=1,\ldots,n-2-k_2-j,$ apply for $i=k_2,\ldots,1,$ in that order $x_{j+i+h}(t_{i;h+2})$ with $t_{\bullet;\bullet}>0.$
    \end{itemize}
\end{enumerate}
\end{lemma}
\begin{proof}[Sketch of proof]
The two constructions in the statement of the lemma are positroid cells by their construction. {As mentioned above, by \cite[Proposition 5.8]{even2023cluster}, also each $S_{j,n;k_1,k_2}$ is a positroid cell.} The dimensions of these three positroid cells are easily seen to be $(j+2-k_1)k_1+(n-j-k_2)k_2+4.$ 
We claim that for each of these three positroid cells, the Pl\"ucker coordinate $P_I$ is non vanishing precisely if $I$ can be written as \[I=I_1\cup I_2\cup\{a\},~
I_1\in \binom{[j+1]\cup\{n\}}{k_1},~I_2\in \binom{[n-1]\setminus[j-1]}{k_2},~a\in\{j,j+1,n-2,n-1,n\}.\]
This is easily checked for $S_{j,n;k_1,k_2}.$

In order to analyze the two constructions, note that the effect of $x_i(t),~t>0$ on the set of non vanishing Pl\"ucker coordinates is as follows: $P_I\neq 0$ on $x_i(t)S$ precisely if $P_I\neq 0$ on $S,$ or $\{i+1\in I\}$ and $P_{I\setminus\{i+1\}\cup\{i\}}\neq 0$ on $S.$ The same rule holds for $y_i$ with the roles of $i,i+1$ interchanged.

Let $S_0$ be the output of the first two steps, i.e. the upper embedding and the sequence of $\pre_h,\inc_h.$ Then it is easily seen that the non vanishing Pl\"ucker coordinates for $S_0$ are precisely sets $I$ of the form
\[I=I_1\cup I_2\cup\{a\},~I_1=[k_1],~I_2\in \binom{[n-1]\setminus[j-1]}{k_2},~a\in\{j,j+1,n-2,n-1,n\}.\]
After the third step, i.e. the sequence of $y_h$ operations, by the above rule, the non vanishing coordinates $P_I$ are those with
\[I=I_1\cup I_2\cup\{a\},~I_1\in \binom{[k_1]\cup\{n\}}{k_1},~I_2\in \binom{[n-1]\setminus[j-1]}{k_2},~a\in\{j,j+1,n-2,n-1,n\}.\]
Similarly, after the $h$th part of the fourth step, i.e. the applications of 
$x_{i+h-1}(t_{i;h+1}),~i=k_1,\ldots,1,$ the set of non zero coordinates grows to 
\[I=I_1\cup I_2\cup\{a\},~I_1\in \binom{[k_1+h]\cup\{n\}}{k_1},~I_2\in \binom{[n-1]\setminus[j-1]}{k_2},~a\in\{j,j+1,n-2,n-1,n\}.\] In the end of the fourth  step, when $h=j+1-k_1,$ the resulting set of non vanishing Pl\"ucker coordinates is as stated.

The argument for the second positroid cell constructed above is similar. After the first two steps the non vanishing Pl\"ucker coordinates are $P_I$ for
\[I=I_1\cup I_2\cup\{a\},~I_1\in \binom{[j+1]\cup\{n\}}{k_1},~I_2=\{j+2,j+3,\ldots,j+k_2+1\},~a\in\{j,j+1,n-2,n-1,n\}.\]
After the third step the set of non zero coordinates becomes
\[I=I_1\cup I_2\cup\{a\},~I_1\in \binom{[j+1]\cup\{n\}}{k_1},~I_2\in \binom{[j+k_2+1]\setminus[j-1]}{k_2},~a\in\{j,j+1,n-2,n-1,n\}.\]
After the $h$th part of the fourth step this set further grows to be
\[I=I_1\cup I_2\cup\{a\},~I_1\in \binom{[j+1]\cup\{n\}}{k_1},~I_2\in \binom{[j+k_2+1+h]\setminus[j-1]}{k_2},~a\in\{j,j+1,n-2,n-1,n\},\]
and in the end of the fourth step, when $h=n-2-k_2-j,$ it becomes
\[I=I_1\cup I_2\cup\{a\},~I_1\in \binom{[j+1]\cup\{n\}}{k_1},~I_2\in \binom{[n-1]\setminus[j-1]}{k_2},~a\in\{j,j+1,n-2,n-1,n\}.\]
Thus, the cells are all equal.
\end{proof}
\begin{lemma}\label{lem:boundary ineqs for S_{j k_1k_2}}
Let $Z$ be a positive $n\times(k+4)$ matrix. At every point in $\Z(S_{i,n;k_1,k-k_1-1})$, 
\begin{itemize}
\item The twistors $\langle i+1,n-2,n-1,n\rangle,~\langle i,n-2,n-1,n\rangle,~\langle i,i+1,n-2,n\rangle$ have constant signs $(-1)^{k-k_1-1},~(-1)^{k-k_1}~$ and $-1,$ respectively, on $\Z(S_{i,n;k_1,k-k_1-1}).$ 
\item 
The functionaries 
$$ \favorite{i,i+1}{j,j+1}{n-2,n-1}{n} $$ are negative for $j\in\{i+1,\ldots,n-4\},$ while 
$$ \favorite{j,j+1}{i,i+1}{n-2,n-1}{n} $$ 
are positive for $j\in[i-1].$
\end{itemize}
\end{lemma}
\begin{proof}
The proof of the first item is completely analogous to the proof of Lemma~\ref{lem:non_zero_dets_for_domino},\eqref{it:non_zero_twistors_for_rows_with_5}.
Every point of $S_{i,n;k_1,k-k_1-1}$ contains a vector supported on positions $i,i+1,n-2,n-1,n,$ as seen from Definition~\ref{def:mid_emb}. As in Lemma~\ref{lem:non_zero_dets_for_domino},\eqref{it:non_zero_twistors_for_rows_with_5}, its $i$th, $i+1$th and $n-1$th entries are proportional to \[(-1)^{k-k_1-1}\langle i+1,n-2,n-1,n\rangle,~(-1)^{k-k_1}\langle i,n-2,n-1,n\rangle,~(-1)^{k-k_1}\langle i,i+1,n-2,n\rangle\] respectively. 
The first item follows from showing that these twistors are nonzero. 
As in the proof of Lemma~\ref{lem:non_zero_dets_for_domino},\eqref{it:non_zero_twistors_for_rows_with_5}, when a twistor coordinate $\langle Z_{\{i,i+1,n-2,n-1,n\}\setminus\{p\}}\rangle,~p=i,i+1,n-1,$ is expanded using Cauchy-Binet, all summands $\langle C^I\rangle\langle Z_{I\cup(\{i,i+1,n-2,n-1,n\}\setminus\{p\})}\rangle $ have the same sign. We therefore finish by showing that at least one such $\langle C^I\rangle,$ for $I\cap(\{i,i+1,n-2,n-1,n\}\setminus\{p\})=\emptyset,$ is nonzero. The existence of such $I$ was shown for every chord diagram in $\CD_{n,k}$ that contains a top chord $(i,i+1,n-2,n-1)$ with $k-k_1-1$ descendants in Lemma~\ref{lem:more than 4}. Since every such chord diagram is contained in the closure of $S_{i,n;k_1,k-k_1-1},$ the Pl\"ucker coordinate $P_I,$ for the same $I,$ is non vanishing also on $S_{i,n;k_1,k-k_1-1}.$

The proof of the second item when $(k_1,j)=(0,1)$ is just Example~\ref{ex:cor:upper_emb}. For $(k_1,j)\neq(0,1)$ the proof is completely analogous to the proof of Proposition~\ref{prop:separation_by_nesting_chord}, where any usage of Corollaries~\ref{cor:generation_left}-\ref{cor:generation_top} is replaced by Lemma~\ref{lem:iterative_construction_mid_emb}.
For the case $j>i$ we use the first construction of Lemma~\ref{lem:iterative_construction_mid_emb}. After its first step, the upper embedding of $\R_+^4\times\Gr^>_{k-k_1-1,[n]\setminus[i]},$ using Corollary~\ref{cor:upper_emb} we deduce that the functionaries $$ \favorite{i,i+1}{j,j+1}{n-2,n-1}{n} = - \favorite{j,j+1}{i,i+1}{n-2,n-1}{n}$$ for $j\in\{i+1,\ldots,n-4\}$ are all negative. The remaining steps involve applications of $\inc_h,~\pre_h,~x_h,~y_h$ which do not affect the functionary, or its sign, by Lemmas~\ref{lem:effect of  pre},~\ref{lem:effect_of_inc},~\ref{lem:effect_of_x_y}.

Similarly, for $j<i,$ we use the second construction of Lemma~\ref{lem:iterative_construction_mid_emb}. After the lower embedding step the functionary $$ \favorite{j,j+1}{i,i+1}{n-2,n-1}{n}$$ is positive on the resulting cell, by Corollary~\ref{cor:promotion_by_lower_emb}. The same reasoning as in the previous case shows that the functionary and its sign are preserved under the following steps of the construction.
\end{proof}
\begin{nn}
Denote by $\CD_{i,n;k_1,k_2}$ the subset of $\CD_{n,k_1+k_2+1}$ which consists of chord diagrams with a rightmost top chord $(i,i+1,n-2,n-1)$ which has $k_2$ descendants.
\end{nn}

\begin{prop}
\label{prop:bcfw_triang_ S_{j k_1 k_2}}
For every $Z \in \Mat^{>}_{n \times (k+4)}$,
$\Z(S_{i,n;k_1,k-k_1-1})$ is triangulated by the images of the BCFW cells that correspond to the chord diagrams in $\CD_{i,n;k_1,k-k_1-1}.$
\end{prop}

\begin{proof}
Write $k_2=k-k_1-1.$
Every cell $S_D$ which corresponds to a chord diagram $D\in \CD_{i,n;k_1,k_2}$ is clearly in the closure of $S_{i,n;k_1,k_2}.$ Thus, $\Z(S_D)\subseteq \overline{\Z(S_{i,n;k_1,k_2})}.$ Each such $S_D$ maps injectively into $\overline{\Z(S_{i,n;k_1,k_2})},$ by Theorem~\ref{thm:injectiveness}, and the images of every two cells $S_a,S_b$ corresponding to a pair of different diagrams $D_a,D_b\in \CD_{i,n;k_1,k_2}$ are disjoint, by Theorem~\ref{thm:separation}.

Write $S = \bigcup_{D_a\in \CD_{i,n;k_1,k_2}}S_a$, where $S_a:=S_{D_a}$. We would like to analyze codimension $1$ boundaries of $S.$ These boundaries come from codimension $1$ boundaries of cells $S_a$ in the union. By the treatment of Section~\ref{sec:precise_ineqs}, and in particular Proposition~\ref{prop:bdries_either_paired_or_SA}, we know that some of them belong to two cells in the union, some map to the boundaries, and the remaining are codimension $1$ boundaries of both $S_a$ and $S_b$ where the former is in the union, and the latter is not. In the last situation the diagram $D_a$ contains has a top chord $(i,i+1,n-2,n-1)$ with $k_2$ descendants, while the diagram $D_b$ does not have such a top chord. Since $S_a,~S_b$ share a boundary, $D_a,D_b$ must differ by a shift. There are three possibilities for the nature of this shift:
\begin{enumerate}
    \item It either moves the chord $c_l=(i,i+1,n-2,n-1),$ or
    \item the shift makes $c_l$ a descendant of another chord, or
    \item  the shift changes the number of chords which descend from the top chord $(i,i+1,n-2,n-1).$
\end{enumerate}

The first possibility may happen in one of three ways. It is either the result of a left shift of $c_l$'s start, or of an unobstructed right shift of $c_l$'s start or of a left shift of $c_l$'s end.
The first two ways, and the unobstructed version of the third way, correspond to the vanishings of $\beta_l,\alpha_l,\Er_l,$ or equivalently the twistors 
\[\langle i,n-2,n-1,n\rangle,~\langle i+1,n-2,n-1,n\rangle,~\langle i,i+1,n-2,n\rangle,\] respectively.

The left shift of the end of $c_l$ is obstructed precisely when $c_l$ has a same end child $c_h=(j,j+1,n-2,n-1).$ On the corresponding boundary $\ee_{l,h}$ vanishes. By adding the appropriate multiple of the row $C_l$ to the row $C_h$ that cancels the inherited domino, the resulting new $h$th row is supported on $j,j+1,n-2,n-1,n,$ and its entries can be solved in terms of twistors, as in the inverse problem procedure of Section~\ref{sec:inj}. If we denote by $C'$ the resulting matrix, then 
\[(C')_{h,l}^{n-2,n-1} = \begin{pmatrix}&\langle i, i+1, n-1, n\rangle &-\langle i, i+1, n-2, n\rangle\\&\langle j, j+1, n-1, n\rangle &-\langle j, j+1, n-2, n\rangle\end{pmatrix}.\]
The vanishing of $\ee_{l,h}$ is the vanishing of the determinant of the minor $C_{h,l}^{n-2,n-1},$ which is equivalent to the vanishing of the determinant of $(C')_{h,l}^{n-2,n-1},$ which can thus be written as\[\det\begin{pmatrix}&\langle i, i+1, n-1, n\rangle &-\langle i, i+1, n-2, n\rangle\\&\langle j, j+1, n-1, n\rangle &-\langle j, j+1, n-2, n\rangle\end{pmatrix} = \favorite{n-2,n-1}{i,i+1}{j,j+1}{n}=0.\]
By Lemma~\ref{obs:plucker_functionary} the last equality is equivalent to $\favorite{j,j+1}{i,i+1}{n-2,n-1}{n}=0.$ Note that in this case $j>i.$

We proceed to the second possibility, the situation where under the shift from $D_a$ to $D_b$ the chord $c_l$ ceases to be a top chord. This possibility happens precisely when in $D_a$ $c_l$ has a sibling $c_h=(j,j+1,i,i+1),$ and the shift under consideration is the right shift of $c_h$'s end. In this obstructed case $\es_{h,l}$ vanishes, and it is easy to see, as in the previous case, that this implies the vanishing of $\favorite{j,j+1}{i,i+1}{n-2,n-1}{n}.$ Note that in this case $j<i.$

Finally, the third possibility for the shift is impossible, since $(i,i+1,n-2,n-1)$ is a rightmost chord.

Any unmatched codimension $1$ boundary of $S$ 
maps to $\partial\Z(S_{i,n;k_1,k_2}).$
This is clear for boundaries of $S$ which are contained in $\SA.$ The $\Z$-image of an unmatched boundary $\partial_\star S_a$ of some $S_a,~D_a\in\CD_{i,n;k_1,k_2},$ which is also a boundary of $S_b,~D_b\notin\CD_{i,n;k_1,k_2},$ 
is contained in the zero locus of a twistor or functionary, by the above analysis. These twistors and functionaries are exactly those which appear in the statement of Lemma~\ref{lem:boundary ineqs for S_{j k_1k_2}} and that are shown to have a constant sign on $\Z(S_{i,n;k_1,k-k_2}).$ But every neighborhood of every $p\in\partial_\star S_a$ also intersects $S_b,$ and the twistor or functionary which vanishes on $\partial_\star S_a$ has the opposite sign on $\Z(S_b),$ showing that indeed $\Z(\partial_\star S_a)\subseteq\partial\Z(S_{i,n;k_1,k_2}).$

We are now in the setting of Proposition~\ref{prop:criterion_for_surj}, with $L=S_{i,n;k_1,k_2},$ and subspaces $S_a$ being the BCFW cells which correspond to $D_a\in\CD_{i,n;k_1,k_2}.$ We have verified the additional assumptions required for applying the proposition in the proof Theorem~\ref{thm:surj}. Thus we can deduce that
\[\Z(\overline{S_{i,n;k_1,k_2}})=\bigcup_{a:D_a\in \CD_{i,n;k_1,k_2}}\overline{S_a}.\]
\end{proof}
A corollary of the proof, and of Proposition~\ref{prop:criterion_for_surj} is that the boundary of the image of $\overline{S_{i,n;k_1,k_2}}$ is the union of zero loci of the boundary twistors $\langle i,i+1,j,j+1\rangle,$ and of the twistors and functionaries in the statement of Lemma~\ref{lem:boundary ineqs for S_{j k_1k_2}} above.
\begin{proof}[Proof of Theorem~\ref{thm:geom_recursion}]
Let $Z'$ be the matrix obtained from $Z$ by omitting the $n-1$th row.
By Theorems~\ref{thm:injectiveness}-\ref{thm:surj} $\mathcal{BCFW}_{[n]\setminus\{n-1\},k},$ the BCFW cells for the chord diagrams with $k$ chords on the marker set $[n]\setminus\{n-1\},$ triangulate $\Ampl_{k,4,[n]\setminus\{n-1\}}(Z').$ Clearly \[\Z(\pre_{n-1}\Gr^{\geq}_{k,[n]\setminus\{n-1\}})=\Ampl_{k,4,[n]\setminus\{n-1\}}(Z').\] Moreover, $\mathcal{CD}_{[n]\setminus\{n\},k},$ the chord diagrams which correspond to  $\mathcal{BCFW}_{[n]\setminus\{n-1\},k},$ are in natural bijection with chord diagrams on the marker set $[n],$ without a chord ending at $(n-2,n-1).$ $\pre_{n-1}$ maps homeomorphically every BCFW cell of $\mathcal{BCFW}_{[n]\setminus\{n-1\},k}$ to the cell of $\mathcal{BCFW}_{n,k}$ corresponding to it under this bijection.  Thus, the latter cells triangulate $\Z(\pre_{n-1}\Gr^{\geq}_{k,[n]\setminus\{n-1\}}).$

The remaining cells triangulate the different $\Z(S_{j,n;k_1,k-1-k_1}),$ by Proposition~\ref{prop:bcfw_triang_ S_{j k_1 k_2}}. 
Theorem~\ref{thm:surj} implies that the union of closures of these sets covers the amplituhedron.

Regarding disjointness, we first show that $\Z(\pre_{n-1}\Gr^{>}_{k,[n]\setminus\{n-1\}})$ and every $\Z(S_{j,n;k_1,k-1-k_1})$ are open. {For $\Z(\pre_{n-1}\Gr^{>}_{k,[n]\setminus\{n-1\}})$ this follows since
\[\Z(\pre_{n-1}\Gr^{>}_{k,[n]\setminus\{n-1\}})=\widetilde{Z'}(\pre_{n-1}\Gr^{>}_{k,[n]\setminus\{n-1\}}),\] where $\widetilde{Z'}$ is the amplituhedron map with respect to $Z',$ and amplituhedron maps are submersive when restricted to the positive Grassmannian, by Lemma \ref{lem:smooth_extension}.}

{For $\Z(S_{j,n;k_1,k-1-k_1})$ we argue as follows. We assume $(k_1,j)\neq(0,1),$ but the treatment for this case is similar. Denote by $W_L,W_c,W_R$ the following subspaces of $\R^{k+4}:$
\[W_L=\Span(Z_1,\dots,Z_{j+1},Z_n),~W_c=\Span(Z_j,Z_{j+1},Z_{n-2},Z_{n-1},Z_n),~W_R=\Span(Z_j,\ldots, Z_{n-1}).\]
Let $M$ be space of triplets $(Y_L,\ell,Y_R)\in \Gr_{k_1}(W_L)\times\Gr_1(W_c)\times\Gr_{k-k_1-1}(W_R),$ which satisfy that, if we write $Y=Y_L+\ell+Y_R$, then $\dim(Y)=k,$ and 
\begin{equation}\label{eq:M_constraints}Y+W_R=Y+W_L=Y+W_c=\R^{k+4}.\end{equation}
$M$ is open in $\Gr_{k_1}(W_L)\times\Gr_1(W_c)\times\Gr_{k-k_1-1}(W_R),$ hence is a manifold.
We can write  \[\Z|_{S_{j,n;k_1,k-1-k_1}}=\Phi\circ\Psi,\]
where $\Phi:M\to\Gr_{k,k+4}$ takes $(Y_L,\ell,Y_R)$ to $Y_L+\ell+Y_R,$
and $\Psi:S_{j,n;k_1,k-1-k_1}\to M$ takes, in the notations of Definition \ref{def:mid_emb}, the vector space represented by $\Upsilon_j(s_1,s_2,t_1,t_2,L,R)$ to $(L'Z,vZ,R'Z),$ thought of as a triplet of vector spaces.
It is easy to verify that this map is well defined and independent of the choices of representatives, but it requires explanation why \eqref{eq:M_constraints} is satisfied.  
To see this, note that in order to show that $Y+W_R$ ($Y+W_L,Y+W_c$ resp.) spans $\R^{k+4},$ it is enough to find a non zero twistor $\langle Y i_1,i_2,i_3,i_4\rangle\neq 0$ where $i_1,\ldots, i_4\in \{n,1,2,\ldots, j+1\}$ ($\{j,j+1,\ldots,n-1\},$~$\{j,j+1,n-2,n-1,n\}$ resp.). For $W_L,W_R$ we can just take $i_1,\ldots,i_4$ to be four cyclically consecutive indices in 
$\{n,1,2,\ldots, j+1\}$ or $\{j,j+1,\ldots,n-1\},$ respectively.
By Lemma \ref{obs:SA_and_bdry_twistors}, and the easily checked fact that $S_{j,n;k_1,k-1-k_1}\cap\SA=\emptyset,$ we deduce that $\langle Y i_1,i_2,i_3,i_4\rangle\neq 0.$ For $W_c$ we can take $i_1,\ldots,i_4$ to be any subset of size four of $\{j,j+1,n-2,n-1,n\}$ and deduce the same, this time using Lemma \ref{lem:boundary ineqs for S_{j k_1k_2}}.
}

We now argue that $\Phi,\Psi$ are submersions, and hence also $\Z|_{S_{j,n;k_1,k-1-k_1}},$ and this will imply that $\Z({S_{j,n;k_1,k-1-k_1}})$ is open. Both claims are simple, but somehow lengthy. We now sketch the proof for $\Psi;$ the proof for $\Phi$ is an exercise in linear algebra, which we omit. Fix $p\in S_{j,n;k_1,k-1-k_1}.$ We need to show that \[d\Psi:T_pS_{j,n;k_1,k-1-k_1}\to T_{\Psi(p)}M\] surjects. If $p$ has a representative of the form $\Upsilon_j(s_1,s_2,t_1,t_2,L,R)$ then it is easy to see, even in the level of matrix representatives of the triplets of vectors spaces which form elements of $M,$ that variations at the $L-$directions, while keeping $R,s_1,s_2,t_1,t_2$ fixed, surject on the $Y_L-$directions of $T_{\Psi(p)}M$ and map to $0$ on the $\ell, Y_R$ directions. An analogous claim holds for the $R-$directions. Variations at the $(s_1,s_2,t_1,t_2)$ directions, while keeping $L,R$ fixed, surject on the $\ell-$directions of $T_{\Psi(p)}M.$ Thus $d\Psi$ is a surjection at $p$.  

Thus $\Z(\pre_{n-1}\Gr^{>}_{k,[n]\setminus\{n-1\}})$ and every $\Z(S_{j,n;k_1,k-1-k_1})$ are open (in $\Gr_{k,k+4}$) and hence contained in the interiors of their closures. 
As a consequence, if two such spaces $\Z(S_1),\Z(S_2)$ for $S_1\neq S_2$ intersect, then their intersection contains an open ball. 
Each $\Z(S_i)\subseteq \bigcup_{D\in A_i}\overline{\Z(S_D)},$ where the union is taken over a collection $A_i$ of chord diagrams, where $A_i=\CD_{j,n;k_1,k-k_1-1}$ if $S_i=S_{j,n;k_1,k-k_1-1}$ and otherwise $A_i=\CD_{n,k}\setminus\bigcup_{j,k_1}\CD_{j,n;k_1,k-k_1-1},$ and these set are all disjoint. Hence we can find BCFW cells $S_{a_1},S_{a_2}$ for $D_{a_1}\in A_1,D_{a_2}\in A_2$ with a non empty intersection, contradicting Theorem~\ref{thm:separation}.
\end{proof}
As a final comment we note that $\Z(\pre_{n-1}\Gr^{\geq}_{k,[n]\setminus\{n-1\}})$ is precisely the subspace of $\Ampl_{n,k,4}(Z)$ on which all the twistors $\langle i,i+1,n-2,n\rangle $ are nonnegative.

\bibliographystyle{alpha}
\bibliography{bib}
\end{document}